\documentclass[12pt]{article}

\usepackage[margin=0.9in]{geometry}

\usepackage{setspace}

\usepackage{amsmath, amssymb, amsthm}

\usepackage[cal=boondoxo,scr=rsfs,bb=ams]{mathalpha}

\usepackage{graphicx}
\usepackage{epsfig}
\usepackage{float}
\usepackage{subfigure}      %
\usepackage{epstopdf}

\usepackage{booktabs,tabularx,array}
\usepackage{threeparttable}
\usepackage{multirow}
\usepackage{adjustbox}

\usepackage{lscape}
\usepackage{ragged2e}
\usepackage{lipsum}   
\usepackage{verbatim} 
\usepackage{xcolor}
\usepackage{url}

\usepackage[authoryear,round]{natbib}
\usepackage{chapterbib}

\usepackage{enumitem}
\allowdisplaybreaks

\numberwithin{table}{section}
\numberwithin{figure}{section}

\theoremstyle{plain}
\newtheorem{theorem}{Theorem}[section]
\newtheorem{lemma}{Lemma}[section]
\newtheorem{corollary}{Corollary}[section]

\theoremstyle{definition}
\newtheorem{example}{Example}[section]
\theoremstyle{plain}
\newtheorem{assumption}{Assumption}[section]

\theoremstyle{remark}
\newtheorem{remark}{Remark}[section]

\newcommand{\mcU}{{\mathcal U}}

\newcommand{\cH}{{\dcal H}}
\newcommand{\cG}{{\mathcal G}}

\usepackage{bm}

\newcommand{\beps}{\bm{\varepsilon}} 

\newcommand{\bga}{\bm{\gamma}}

\newcommand{\bxi}{\bm{\xi}}

\newcommand{\cR}{{\mathcal R}}

\newcommand\cK{{K}}

\newcommand\mcD{{\mathcal D}}

\newcommand{\bz}{{\bf z}}
\newcommand{\bV}{{\bf V}}

\newcommand{\bB}{{\bf B}}

\newcommand{\by}{{\bf y}}
\newcommand{\bu}{{\bf u}}

\newcommand{\cP}{{\mathcal P}}

\newcommand{\cD}{{\mathcal D}}

\newcommand{\cC}{{\mathcal C}}
\newcommand{\cS}{{\mathcal S}}

\newcommand\bs{{\bf s}}
\newcommand{\bP}{{\bf P}}
\newcommand{\bC}{{\bf C}}

\newcommand\bX{{\bf X}}

\newcommand\bD{{\bf D}}
\newcommand\bE{{\bf E}}
\newcommand\bA{{\bf A}}
\newcommand\bI{{\bf I}}

\newcommand\bZ{{\bf Z}}

\newcommand\bg{{\bf g}}
\newcommand\bS{{\bf S}}

\newcommand\bU{{\bf U}}

\newcommand\bx{{\bf x}}
\newcommand\bv{{\bf v}}
\newcommand\bY{{\bf Y}}

\newcommand\be{{\bf e}}

\def\beq{\begin{equation}}
\def\eeq{\end{equation}}
\def\bals{\begin{align*}}
\def\eals{\end{align*}}
\def\bal{\begin{align}}
\def\eal{\end{align}}

\renewcommand{\P}{{\mathsf P}}
\newcommand{\E}{{\mathsf E}\hspace{0.1mm}}
\newcommand{\R}{\mathbb{R}}
\DeclareMathAlphabet{\dcal}{U}{dutchcal}{m}{n}
\DeclareMathOperator{\Cov}{Cov}
\DeclareMathOperator{\Var}{Var}
\DeclareMathOperator{\tr}{tr}
\DeclareMathOperator*{\argmin}{arg\,min}

\numberwithin{equation}{section}
\numberwithin{theorem}{section}

\newcommand{\inlinetag}[1]{%
  \refstepcounter{equation}%
  \label{#1}%
  \nolinebreak\hfill\textup{(\theequation)}%
}

\DeclareMathOperator{\diag}{diag}
\title{Online detection of distributional changes for  time series in metric spaces}

\author{
B.\,Cooper Boniece\thanks{Department of Mathematics, Drexel University. 
Email: cooper.boniece@drexel.edu. Research supported in part by NSF DMS-2413558.}
\and
Lajos Horv\'ath\thanks{Department of Mathematics, University of Utah. 
Email: horvath@math.utah.edu.}
\and
Lorenzo Trapani\thanks{Department of Economics, University of Pavia. 
Email: lorenzo.trapani@unipv.it.}
}

\begin{document}
\maketitle

\begin{abstract}
We propose an online testing framework for detecting distributional changes in serially dependent data with values in a separable metric space. Based on two-sample $U$-statistics, the framework encompasses sequential analogs of energy distance and maximum mean discrepancy (MMD) procedures while accommodating temporal dependence. We establish asymptotic theory for finite and open-ended monitoring horizons that characterizes the full asymptotic run-length distribution under $H_0$ and yields asymptotic false-alarm control. We further establish new spectral approximation results for kernel matrices formed from serially dependent observations, and use them to construct a feasible Monte Carlo calibration procedure. Our flexible window construction encompasses classical, Page-type, and full-scan historical-baseline monitoring and can achieve short detection delays for both early and  late changepoints, without requiring sub-Gaussianity or high-order moments of the raw observations. Simulations show reliable false-alarm control across linear, nonlinear, high-dimensional, and functional time-series models and further demonstrate that, over a broad range of alternatives and changepoint locations, the proposed method can achieve substantially shorter delays than recent procedures
designed specifically for rapid detection. Applications to foreign exchange rates,  electricity-market curves, and daily air transportation networks illustrate the methodology across scalar, functional, and network-valued time series.
\end{abstract}
\doublespacing
\section{Introduction}
Structural stability is a central premise in much of time series analysis.  Early detection of departures from  stability is therefore crucial for maintaining the
reliability of inferential procedures \citep{chu:stinchcombe:white:1996,aue:kirch:2024}.  Although mean shifts are often of primary interest, such a shift can be a manifestation of a broader structural break in the data-generating mechanism that may affect higher marginal moments or other structural aspects of the data (e.g. \citealp{casini:perron:2019}). For example,  changes in the level, variability, and cyclical dynamics of industrial-production growth have been documented jointly across major economies \citep{giordani:kohn:vandijk:2007}; and complex structural changes in energy demand and consumption have been documented following COVID-19 \citep{jiang:yeevan:klemes:2021}. Thus, even when changes in level are the principal concern, a distributionally sensitive procedure capable of capturing several aspects of a structural change can potentially offer more power than one targeting mean shifts. A distributional procedure may also be preferable to mean- or covariance-based methods when the type of a break is unknown in advance or is only of secondary interest.  Accordingly, there is a growing body of work on distributional change-point detection in online and offline settings, including empirical CDF-based methods \citep{inoue:2001,kojadinovic:verdier:2021,fu:hong:wang:2023,holmes:kojadinovic:verhoijsen:2024}, characteristic-function-based methods \citep{huskova:meintanis:2006,horvath:rice:vanderdoes:2026}, energy-distance approaches \citep{matteson:james:2014,biau:bleakley:mason:2016,boniece:horvath:trapani:2025}, and kernel methods \citep{arlot:celisse:harchaoui:2019,li:xie:dai:song:2019,wei:xie:2026}, among others. 

The need for distributional changepoint detection may be especially
pronounced for complex data types. Modern time series may be high-dimensional or functional, or may consist of objects such as networks or probability distributions, for which a 
vector-space structure may not adequately reflect the geometry relevant to the application or may be entirely absent \citep{zhang:zhu:shao:2026}. In such domains, structural breaks may not be adequately described by a change in any single marginal feature. For example, changes in network structure associated with major external events have been documented in evolving transportation and communication systems
\citep{athreya:lubberts:park:priebe:2025,
wang:li:madridpadilla:yu:rinaldo:2026}. In functional settings, the U.S. yield-curve illustration of
\citet{bardsley:horvath:kokoszka:young:2017} around the 2008 collapse of Lehman Brothers exhibits changes in both curve shape and stochastic variability, while \citet{stoehr:aston:kirch:2021} find changes in the covariance structure of resting-state functional MRI sequences for which no mean change is detected. These considerations motivate general monitoring procedures that accommodate broad data types while remaining sensitive to complex changes in the underlying distribution.

In this work, we propose a general framework for monitoring time series that take values in a metric space. Our framework is based on degenerate two-sample $U$-statistics that encompass sequential MMD and energy-distance procedures as special cases while also permitting genuinely indefinite kernels. Whereas existing theory for sequential MMD and energy-distance  changepoint methods has focused largely on independent observations,   to the best of our knowledge, we provide
the first general theory for this broader class of sequential distributional
monitors for degenerate kernels under temporal dependence.\footnote{Sequential monitoring based on nondegenerate two-sample $U$-statistics
under dependence was studied
by \citet{kirch:stoehr:2022,kirch:stoehr:2022a}; their theory does not cover the degenerate statistics considered here.}

We establish a unified null theory over a broad class of window schemes, including classical,  Page-type, and full-scan historical-baseline procedures, over both finite and open-ended monitoring horizons. The resulting limits  characterize the full asymptotic run-length distribution of the procedure, yielding asymptotic false alarm control as well as approximations of run length summaries on any finite monitoring horizon. We further establish validity of our procedures  when kernel bandwidths or finite-dimensional preprocessing transformations are estimated from the historical sample for a subclass of bounded kernels.

To establish theoretical validity of our feasible calibration procedure, we provide new spectral approximation results for kernel matrices formed from serially dependent observations, building on the random kernel-matrix approximation framework of \citet{koltchinskii:gine:2000}. These results permit consistent estimation of the long-run covariance operator of an auxiliary kernel-induced Hilbert-space-valued process that governs the asymptotic behavior of the monitoring statistic under $H_0$, which may be of independent interest.

Under alternatives with vanishing break size, we provide theoretical consistency against all possible distributional break types under appropriate choice of kernel. We also show a full-type monitoring scheme yields detection delays whose dependence on the changepoint location can be made arbitrarily mild for appropriate choice of kernels, yielding near-logarithmic delays for fixed alternatives.  A key feature of the proposed framework is that such detection-delay bounds require relatively weak tail assumptions on the raw observations, in contrast with recent short-delay methods developed for weighted CUSUM-type procedures that assume sub-Gaussianity or existence of high-order moments on the raw data
\citep{kutta:dornemann:2025,bastian:kutta:2025,
yu:madridpadilla:wang:rinaldo:2023}. In our framework, the corresponding moment condition is instead imposed on a kernel-induced process and is automatically met for bounded kernels; the remaining data-level requirements reflect a novel tradeoff between kernel regularity and temporal dependence, and are readily met by many standard linear and nonlinear time-series models.

Simulations demonstrate reliable calibration across linear, nonlinear, high-dimensional, and functional time series models, together with favorable delay performance across structural, mixed, and pure mean breaks. Notably, despite its omnibus nature, the implemented scheme often outperforms recent mean-targeted procedures designed for fast detection even under pure mean shifts, with especially pronounced gains under more complex structural breaks. Applications to foreign exchange rates,  electricity-market curves, and daily air transportation networks illustrate the practical utility of the framework for monitoring complex time-series data.

The remainder of the paper is organized as follows. Section~\ref{s:setting} introduces the
monitoring framework and assumptions; Section~\ref{s:asymp} develops the asymptotic theory;
Section~\ref{s:implementation} presents spectral estimation and feasible calibration; and
Sections~\ref{s:simulations} and~\ref{s:data} report the simulations and data illustrations, respectively.  All proofs are given in the Appendix.

\textbf{Notation}.  
For positive sequences $a_m,b_m$, write $a_m\ll b_m$ if $a_m=o(b_m)$ and
$a_m\gg b_m$ if $b_m=o(a_m)$.  We write $a_m\lesssim b_m$ if
$a_m\leq Cb_m$ for all sufficiently large $m$ and some constant $C<\infty$ independent
of $m$; $a_m\asymp b_m$ means $a_m\lesssim b_m$ and $b_m\lesssim a_m$. We write $\mathbb O(r)$ for the group of $r\times r$ orthogonal matrices.

\section{Setting and monitoring statistics}\label{s:setting}
Let $\bX=\{\bX_j,j\in\mathbb Z\}$ be a time series taking values in a separable metric space $(\mathcal X,\rho)$, observed at times $j=1,2,\ldots$. We aim to detect changes in the marginal distribution of $\bX_j$. Working in the monitoring framework initiated by \cite{chu:stinchcombe:white:1996}, we suppose that an initial historical baseline (or ``training") sample $\bX_1,\ldots,\bX_m$ is available prior to the onset of monitoring.  Throughout, we assume the following:

\begin{assumption}\label{a:baseline}
    $\bX_j\sim F$, for $j=1,\ldots,m.$
\end{assumption}

Assumption~\ref{a:baseline} asserts that the historical observations are marginally homogeneous, with common distribution $F$. Conditions on their temporal dependence, which in particular imply stationarity under the null hypothesis, are introduced in Section~\ref{s:assumptions}. After the historical period, we monitor the sequence
$\bX_{m+1},\bX_{m+2},\ldots$ for distributional changes. We consider the sequential testing problem
\begin{equation}\label{e:H_0}
H_0:\quad \bX_j\sim F,\qquad j=m+1,m+2,\ldots,
\end{equation}
against
\begin{equation}\label{e:H_A}
H_A:\quad \text{there exists }~k_*\geq1~\text{ such that}~~
\bX_j\sim
\begin{cases}
F, & j\leq m+k_*,\\
F_*, & j>m+k_*,
\end{cases}
\end{equation}
where $F_*\neq F$. Thus, under $H_0$ the marginal distribution remains unchanged throughout the monitoring period, whereas under $H_A$ it changes at the unknown time $m+k_*$. The dependence conditions used to establish the asymptotic theory are stated separately in Section~\ref{s:assumptions}.

Our approach is based on monitoring schemes built from two-sample 
$U$-statistics. For a chosen kernel $h:\mathcal X \times \mathcal X\to \R$ (not necessarily positive semidefinite), and every set of integers  $0\leq a<b\leq c<d$, we set
\begin{equation}\label{e:U-master}
\begin{aligned}
\mcU\big((a,b]; (c,d]\big)
  &= \frac{2}{(b-a)(d-c)}
       \sum_{i=a+1}^b \sum_{j=c+1}^d h(\bX_i,\bX_j)   \\
  &\quad - \binom{b-a}{2}^{-1}
       \sum_{a<i<j\leq b} h(\bX_i,\bX_j)
        - \binom{d-c}{2}^{-1}
       \sum_{c<i<j\leq d} h(\bX_i,\bX_j).
\end{aligned}
\end{equation}

Thus, $\mcU((a,b]; (c,d])$ compares observations in the time window $(a,b]$ 
against those in $(c,d]$ through the kernel $h$. Intuitively, under $H_0$, the quantity $\mcU((a,b],(c,d])$  typically remains  ``small'' in magnitude over any chosen interval pairs  $(a,b]$ and $(c,d]$;  conversely, if for some pair $(a,b]$, $(c,d]$ the quantity $|\mcU((a,b],(c,d])|$ is ``large," a changepoint may be present somewhere in $(a,d]$.  We provide details on the class of kernels $h(\bx,\by)$ below in Assumption \ref{a:h}.

Once a kernel $h$ is fixed, a variety of monitoring procedures based on
$\mcU$ can be constructed by specifying which pairs of time windows are
compared at each time $m+k$.  To accommodate a broad class of
practically relevant procedures, we consider schemes in which the comparison
window always ends at the current observation.  Formally, a monitoring scheme
is a collection of index sets $\{\mathcal S_k:k\ge1\}$ satisfying
\begin{equation}\label{e:schemeclass}
\mathcal S_k \subseteq \mathcal C_k
:=
\Big\{(\ell_1,\ell_2)\in \mathbb N^2: 0\leq \ell_1 \leq \ell_2 \leq k-2\Big\},
\end{equation}
where each $(\ell_1,\ell_2)\in \mathcal S_k$ corresponds to comparing data occurring during
$
(0,m+\ell_1]$  against  $(m+\ell_2,m+k].
$ Thus, the first window always contains the historical sample $(0,m]$, while
$\ell_1$ controls whether the reference window remains fixed or expands into
the monitoring period, and $\ell_2$ determines how far back from the current time
the comparison window begins.

Given a scheme $\{\mathcal S_k\}$ and a weight function $g_m(k,\ell_1,\ell_2)>0$, we define the detector
\begin{equation}\label{e:detector}
\mcD_m(k)
=
m^{-1}
\max_{(\ell_1,\ell_2)\in \mathcal S_k}
\frac{(k-\ell_2)^2}{g_m(k,\ell_1,\ell_2)}
\big|
\mcU\big((0,m+\ell_1],(m+\ell_2,m+k]\big)
\big|.
\end{equation}
We consider weights of the form
$$
    g_m(k,\ell_1,\ell_2)
    =
    g(t,L,r),
    \qquad
    t=\frac{k}{m},\quad
    L=\frac{m+\ell_1}{m},\quad
    r=\frac{k-\ell_2}{m},
$$
where, for some \(0<\beta<1\),
\begin{equation}\label{a:g}
    g(t,L,r)
    =
    (1+t)^{1-\beta}r^\beta
    \left(1+\frac{r}{L}\right)
    \log^{2-\beta}(e+t),
    \qquad t\geq0,~~ r\geq 0,~~L\geq 1.
\end{equation}
Given a critical value $c_\alpha>0$, define the stopping time
\begin{equation}\label{e:stopping}
\tau_m(c_\alpha) =\inf\left\{k\geq2:\mathcal D_m(k)>c_\alpha\right\},
\end{equation}
where $\inf\varnothing=\infty$. In the open-ended setting, the
procedure signals a change if $\tau_m(c_\alpha)<\infty$. If monitoring is terminated
after $M_m$ observations, a change is detected during the monitoring period
if and only if
$
\tau_m(c_\alpha)\leq M_m,
$ or, equivalently, whenever
$
\max_{2\leq k\leq M_m}\mathcal D_m(k)>c_\alpha.$

The class \eqref{e:schemeclass} includes many common schemes that have been employed in CUSUM-type settings.  For example it includes the classical scheme \citep{chu:stinchcombe:white:1996}
$
\mathcal S_k^{\mathrm{cl}} =\{(0,0)\},$
 which compares the fixed
historical sample $(0,m]$ with all observations $(m,m+k]$ available from the onset of monitoring.  Owing to its simplicity, $\mathcal S_k^{\rm cl}$ offers among the simplest limit behavior under $H_0$ and serves as a useful reference scheme. However, particularly for late changes, $\mathcal S_k^{\rm cl}$ may lead to long detection delays, since for $k>k_*$, the monitoring window $(m,m+k]$ can potentially contain a large number of pre-change observations; hence, for faster detection, other schemes may be desirable.  Central examples covered by \eqref{e:schemeclass} include the Page- and full-type schemes:
\begin{equation}\label{e:D_full}
\mathcal S_k^{\mathrm{Page}}
    =\{(0,\ell):0\leq\ell\leq k-2\},\qquad 
\mathcal S_k^{\mathrm{full}}
    =\{(\ell,\ell):0\leq\ell\leq k-2\}.
\end{equation}
The Page-type scheme $\mathcal S_k^{\mathrm{Page}}$ (e.g., \citealp{page1955test,fremdt2015page}) compares this historical window $(0,m]$ while varying the
start of the comparison window $(m+\ell,m+k]$. The ``full"  scheme (c.f.~\citealp{aue:kirch:2024,gosmann:kley:dette:2021}) compares the contiguous windows $(0,m+\ell]$ and $(m+\ell,m+k]$. Each scheme choice entails a different tradeoff among power, detection delay, among other practical considerations; the full scope of admissible schemes is discussed in Section \ref{s:asymp}.

\begin{remark} The boundary function $g$ in \eqref{a:g} is inspired by the weighting approach
of \cite{kutta:dornemann:2025}, which was designed to reduce detection delays
over long monitoring horizons. Our choice differs in that it also depends on the
effective left-window length $m+\ell_1$. This refinement changes the scaling of
the detector when the left window grows, helping to further improve delay times.
\end{remark}

\subsection{Assumptions}\label{s:assumptions}
Throughout, we  let $\dcal H$ denote the Hilbert space
\begin{equation}\label{e:def_H}
\dcal H = {\mathcal L}^2(F)
   = \Bigl\{g:\mathcal X\to\R \text{ measurable} : \int g^2(\bx)\,F(d\bx)<\infty\Bigr\},
\end{equation}
 equipped with the inner product $\langle f,g \rangle_{\dcal H}= \int f g dF$.  We also define  the \textit{degenerate} part of $h$:
 \begin{equation}\label{e:degenH}
 \overline h(\bx,\by) = h(\bx,\by) - \E h(\bx,\bY) -  \E h(\bX,\by) + \E h(\bX,\bY) , \qquad \bX,\bY\stackrel{iid} \sim F.
  \end{equation}
  \begin{assumption}\label{a:h}
In \eqref{e:U-master}, we assume $h:\mathcal X\times \mathcal X\to \R$ is a measurable function satisfying $h(\bx,\by)=h(\by,\bx)$.  In addition:
\begin{enumerate}[label=(\roman*)]

\item 
$
\iint h^2(\bx,\by) F(d\bx)F(d\by)<\infty$ and $\int h^2(\bx,\bx) F(d\bx)<\infty.
$

\item Let $A:\cH\to \cH$ denote the operator
\begin{equation}\label{e:def_A}
Ag(\bx)=  \int \overline h(\bx,\by)g(\by)F(d\by),
\end{equation}
and let $\{(\phi_\ell,\lambda_\ell)\}_{\ell\ge1}$ be the eigenpairs of $A$,  with 
$\|\phi_\ell\|_{\cH}=1$, and 
ordered so that $|\lambda_1|\ge|\lambda_2|\ge\cdots$
We assume
\begin{equation}\label{a:traceclass}
\sum_{\ell=1}^\infty|\lambda_\ell|<\infty.
\end{equation}

\item For some $\alpha,C>0$,
$
\left(\displaystyle\sum_{\ell\ge1} |\lambda_\ell|(\phi_\ell(\bx)-\phi_\ell(\bx'))^2\right)^{1/2}
\leq C \rho(\bx,\bx')^\alpha.
$ \inlinetag{a:alphaholder}

\end{enumerate}
\end{assumption}
We provide some examples of kernels satisfying Assumption \ref{a:h} below.  Assumption~\ref{a:h}(i) guarantees the operator $A$ defined in condition (ii) is well-defined on $\cH$; it also contains a diagonal integrability requirement, used to control terms involving $h(\bX_i,\bX_i)$ when diagonal contributions are added or removed in proofs. Assumption~\ref{a:h}(ii) asserts that the operator $A$ defined by the degenerate kernel $\overline h$ is trace-class. In particular, under this condition
\begin{equation}\label{e:h_expansion}\overline h(\bx,\by)=\sum_{\ell=1}^\infty \lambda_\ell \phi_\ell(\bx)\phi_\ell(\by),\qquad F\times F\text{-a.e.}
\end{equation}
(see, for example \citealp{eagleson:1979a}).  Lastly, Assumption~\ref{a:h}(iii) requires a weighted H\"older continuity of the eigenfunctions, where the weights are determined by the eigenvalues of $A$. Though this assumption may not be easily verified for an arbitrary kernel, it is satisfied by a broad array of kernels $h$ commonly used in nonparametric testing, including distance-type kernels and many positive semidefinite kernels under mild conditions on the marginal distribution $F$, as our examples below show.

We also require the following high-level moment-type assumption regarding $\bX_1$ and $h$. 
Let
\begin{equation}\label{e:def_Y(x)}
Y(\bx) =\sum_{\ell=1}^\infty\sqrt{|\lambda_\ell|}\,\phi_\ell(\bx)\phi_\ell  \in \cH.
\end{equation}
\begin{assumption}\label{a:Y_moment}
We assume 
$
\E \|Y(\bX_1)\|_{\cH}^p < \infty$, for some $ p > \displaystyle\frac{2}{1-\beta}\vee 4.$
\inlinetag{a:p}
\end{assumption}

Note the map $\bx \mapsto Y(\bx)$, which depends on both $F$ and the chosen kernel $h$, ``embeds" the $\mathcal X$-valued data into $\cH$, producing an auxiliary $\cH$-valued time series $\{Y(\bX_j),j\in \mathbb Z\}$. Hence, Assumption~\ref{a:Y_moment} is a moment condition on the
kernel-induced object $Y(\bX_1)$, rather than directly on $\bX_1$.  In
spectral form, it is equivalent to
$
\E\left(\sum_{\ell}|\lambda_\ell|\,\phi_\ell(\bX_1)^2\right)^{p/2}< \infty.
$ 
This condition is used to apply the H\"olderian invariance principles that underpin the limit behavior of $\cD_m(k)$ under $H_0$; (cf.~\citealp{rackauskas:suquet:2005,kutta:dornemann:2025}).  The next examples
show how Assumptions~\ref{a:h} and \ref{a:Y_moment} translate into verifiable kernel- and data-level conditions
for common kernel classes.

\begin{example}\label{ex:cnd_kernels}
Suppose $h(\bx,\by)$ is a semimetric of negative type (see \citealp{sejdinovic:etal:2013}).
Then Assumptions~\ref{a:h} and \ref{a:Y_moment} are satisfied provided
$$
\E h(\bX_1,\bx_0)^{p/2}<\infty,
\quad
|h(\bx,\by)| \leq C\rho(\bx,\by)^{2\alpha},
$$
for some $\bx_0$, where $\alpha$ and $p$ are in \eqref{a:alphaholder} and \eqref{a:p}. In particular, for $\bx,\by\in \R^d$, they hold for the energy-distance  (\cite{szekely:rizzo:2022}) $h(\bx,\by)=\|\bx-\by\|^{2\alpha}$ when $\bX_1$ has $\alpha p$ moments.
\end{example}

\begin{example}\label{ex:psd_kernels}
Suppose $h(\bx,\by)$ is a positive definite kernel, and let $\alpha$ and $p$ be as in \eqref{a:alphaholder} and \eqref{a:p}.  Then, Assumptions~\ref{a:h} and \ref{a:Y_moment} are satisfied whenever
$$
\E h(\bX_1,\bX_1)^{p/2}<\infty,
\quad
h(\bx,\bx)+h(\by,\by)-2h(\bx,\by) \leq C\rho(\bx,\by)^{2\alpha}.
$$
\end{example}
The distinction between requiring
moments of $Y(\bX_1)$ instead of $\bX_1$ is worth noting: if $h$ is a bounded kernel belonging to one of the classes in Examples \ref{ex:cnd_kernels} or \ref{ex:psd_kernels}, Assumption~\ref{a:Y_moment} is automatically
satisfied for every finite $p$.

Concerning the dependence structure of the data itself, we impose an $L^p$-$m$
approximability condition, a standard weak dependence framework \citep{hormann:kokoszka:2010}. 
\begin{assumption}\label{a:Lpm} 
\begin{itemize}
\item[(i)]It holds that $
\bX_n = f(\varepsilon_n,\varepsilon_{n-1},\ldots)$ 
where $\varepsilon_n$ are i.i.d. elements taking values in a measurable space $S$, and $f:S^{\infty}\to \mathcal X$ is a deterministic measurable function.  
\item[(ii)] For each $n\geq1$, let $\{\varepsilon_t^{(n)}\}_{t\in\mathbb Z}$ be an i.i.d.~copy of $\{\varepsilon_t\}$, and define, for each $r\geq 1$,
\begin{equation}\label{e:X_m}
\bX_n^{(r)}= f(\varepsilon_n,\ldots,\varepsilon_{n-r+1},
\varepsilon_{n-r}^{(n)},\varepsilon_{n-r-1}^{(n)},\ldots),
\quad n \in \mathbb Z.
\end{equation}
Set 
$\vartheta _{q}(r):= \big(\E\rho(\bX_1,\bX_1^{(r)})^{q}\big)^{1/q}.$ We assume, for some $\bx_0\in \mathcal  X$, that 
$
 \E\rho(\bx_0,\bX_1)^{\alpha p}<\infty$ and that, for some $0<\eta<1$, $ \sum_{r\geq 1}\left(\vartheta_{p\alpha}(r)^{\alpha}\right)^{\eta} <\infty$%
 ~where $p$ is given in Assumption \ref{a:Y_moment} and $\alpha$ is given in Assumption \ref{a:h}.
 \end{itemize}
\end{assumption}

\begin{remark}\label{r:dependence_moments}
For many bounded kernels (e.g.~when the kernels in
Examples \ref{ex:cnd_kernels} and \ref{ex:psd_kernels} are bounded),
Assumption~\ref{a:Y_moment} imposes no moment condition on the marginal
distribution of $\bX_1$. The remaining data-level moment requirement comes
from Assumption~\ref{a:Lpm}: for $\beta$ close to one, the condition
$p>2/(1-\beta)$ gives, with $q=\alpha p$,
$$
    \E\rho(\bx_0,\bX_1)^q<\infty,
    \quad
    \text{for some } q>\frac{2\alpha}{1-\beta}.
$$
Thus, if $\alpha$ can be chosen small relative to $1-\beta$, the required
raw moment order on $\bX_1$ can remain small even when $\beta$ is close to one. The ability to take $\alpha$ small depends on the strength of the
approximability condition.  Assumption~\ref{a:Lpm} requires summability of the $L^{\alpha p}$-$r$ approximation coefficients, which is mild for many
short-memory models.  For example, if
$\vartheta_{\alpha p}(r)\leq z^r$ for some $z\in(0,1)$, then the summability over $r$ is automatic once the corresponding moment exists. Thus, for many bounded kernels, direct high-order moment assumptions on $\bX_1$ can effectively be replaced by approximability conditions on the
dependence structure.
\end{remark}

\begin{example}\label{ex:Lpm_standard}
Assumption~\ref{a:Lpm} is satisfied by many standard causal time series models,
including linear processes and
GARCH-type models under standard conditions.  In these
settings the approximation coefficients $\vartheta_q(r)$ often decay geometrically, see \cite{horvath:rice:2024}.
\end{example}

Finally, we require the spectral expansion to hold at dependent pairs of observations:

\begin{assumption}\label{a:ac}
The eigenfunctions in \eqref{e:h_expansion} may be chosen so that, for every
$i\neq j$,
$
\overline h(\bX_i,\bX_j)
=\sum_{\ell=1}^\infty
\lambda_\ell\phi_\ell(\bX_i)\phi_\ell(\bX_j)
$ almost surely.
\end{assumption}
This condition allows the $F\times F$-almost-everywhere expansion in
\eqref{e:h_expansion} to be used for serially dependent observations. It
holds, in particular, if the law of each $(\bX_i,\bX_j)$ with $i\neq j$
is absolutely continuous with respect to $F\times F$. Under
Assumption~\ref{a:h}(iii), it also holds if $\overline h$ is continuous on
$\operatorname{supp}(F)\times\operatorname{supp}(F)$. These sufficient
conditions, including examples of linear and nonlinear causal models with
innovations admitting a density, are established in Lemmas~\ref{l:ac}
and~\ref{l:ac_examples} in the Supplement.

\section{Asymptotic theory}\label{s:asymp}

For the asymptotic results below, we restrict attention to monitoring schemes that admit a continuous-time embedding.

\begin{assumption}\label{a:admissiblescheme} 
The following conditions hold on the scheme $\{\mathcal S_k\}$ in \eqref{e:schemeclass}:
\begin{enumerate}[label=(\roman*)]
\item 
There exists a family $\{\mathcal S(t):t\geq0\}$ such that
$\mathcal S(t)\subseteq \{(u,v): 0 \leq u \leq v \leq t\}$ for every
$t\geq0$ and, for all $m,k\geq1$,
\begin{equation}\label{e:S(t)_cond}
\mathcal S_k
=
\Big\{(\lfloor mu\rfloor,\lfloor mv\rfloor):(u,v)\in \mathcal S(k/m)\Big\}\cap \mathcal C_k.
\end{equation}
\item  
For each $t\geq 0$, the region $\mathcal S(t)\subseteq \R^2$ is compact. Moreover, for some compact set $K$, there exists a 
continuous function
$
h_{\cS}:[0,\infty)\times K \to \R^2
$
such that
\begin{equation}\label{e:cont_param_cond}
\mathcal S(t)=\{h_{\cS}(t,r): r\in K\}, \qquad t\geq 0
\end{equation}
and
\begin{equation}\label{e:cont_param_cond2}
\lim_{\delta\downarrow 0}\sup_{\substack{s,t\geq 0\\ |s-t|\leq \delta}}\sup_{r\in K}
\big|h_{\cS}(t,r)-h_{\cS}(s,r)\big| =0.
\end{equation}
\end{enumerate}
\end{assumption}
Assumption \ref{e:S(t)_cond} captures a wide array of practical schemes belonging to the class \eqref{e:schemeclass}. Assumption \ref{a:admissiblescheme}(i) requires that the scheme $\mathcal S_k$ admits a deterministic continuous-time embedding described by compact regions $\mathcal S(t)\subseteq \R^2$; Assumption \ref{a:admissiblescheme}(ii) requires the family of regions $\{\mathcal S(t),t\geq 0\}$ to be continuous, in an appropriate sense.

\begin{example}\label{ex:schemes}
The classical scheme as well as the schemes \eqref{e:D_full} satisfy
Assumption~\ref{a:admissiblescheme}.  For example, the classical scheme is simply $\mathcal S^{\mathrm{cl}}(t)= \{(0,0)\}$, and $h_{\cS}(t,0)=(0,0).$ The continuous-time embeddings of the Page-type and full schemes are:
$$
\begin{aligned}
\mathcal S^{\mathrm{Page}}(t)
&= \{(0,s):0\leq s\leq t\},
&\quad K&=[0,1],
&\quad h_{\cS}(t,r)&=(0,rt), \\[0.5em]
\mathcal S^{\mathrm{full}}(t)
&= \{(s,s):0\leq s\leq t\},
&\quad K&=[0,1],
&\quad h_{\cS}(t,r)&=(rt,rt).
\end{aligned}
$$
\end{example}
The general formulation of Assumption~\ref{a:admissiblescheme} permits many practically motivated restrictions. For example, comparison windows may be required to omit an initial period of order $m$, as may be useful when observations immediately following the training sample are unavailable or unreliable; it also permits variants that restrict windows  to a sparse grid to reduce computational cost. In the sequel,  we use the shorthand
$$
    g_t(u,v):=g(t,1+u,t-v),
$$
where $1+u$ and $t-v$ correspond to the limiting left- and right-window lengths,
respectively.

\begin{theorem}\label{thone} Let $M_m\in\mathbb N\cup\{\infty\}$ denote the monitoring horizon, where
$M_m\equiv \infty$ corresponds to open-ended monitoring. Assume
$M_m/m \to a\in(0,\infty].$
Suppose $H_0$ and Assumptions~\ref{a:baseline}--\ref{a:ac} and
\ref{a:admissiblescheme} hold.  Then,
\begin{align*}
\lim_{m\to \infty}\P\left\{\sup_{2\leq k\leq M_m}\mathcal D_m(k)>c\right\}=\P\left\{ \sup_{0<t<a}\Gamma(t)\geq c \right\},
\end{align*}
at each continuity point $c$ of $\sup_{0<t<a}\Gamma(t)$, where
\begin{align*}
\Gamma(t)
&=\sup_{(u,v)\in\mathcal S(t)}\frac{1}{g_t(u,v)}\left|
\sum_{\ell=1}^\infty \lambda_\ell \left(\Big(\frac{(t-v)^2}{1+u}+ (t-v)\Big)-\left(
 \frac{t-v}{1+u}\widetilde W_\ell(u) - \big(\widetilde W_\ell(t)-\widetilde W_\ell(v)\big)\right) ^2
\right)\right|.%
\end{align*}
Above, $\widetilde W_\ell(t)=W_\ell(t+1)$, where the processes $\{W_\ell(t), t\geq 0, \ell\geq 1   \}$  are jointly Gaussian with $EW_\ell(t)=  0$  and $EW_\ell(t)W_{\ell'}(s)=\mathcal g_{\ell,\ell'}(s\wedge t)$,  and 
with
\begin{equation}\label{e:def_gell}
\mathcal g_{\ell,\ell'}=\sum_{r\in \mathbb Z}\Cov\left(\phi_{\ell}(\bX_0), \phi_{\ell'}(\bX_r) \right).%
\end{equation}
Consequently,  if $M_m/m\to a\in(0,\infty]$, then 
$$\lim_{m\to\infty}\P\{\tau_m(c)\leq M_m\}= \P\left\{\sup_{0<t<a}\Gamma(t)\geq c\right\}.
$$

\end{theorem}
In words, Theorem~\ref{thone} shows that the monitoring statistic admits a tractable asymptotic limit under $H_0$, which serves as the basis for calibration of the procedure.  In Section \ref{s:implementation}, we discuss implementation based on Theorem \ref{thone}.   We illustrate the limit process for several monitoring schemes in the next examples.
\begin{example}
Consider the classical scheme $\mathcal S_k^{\mathrm{cl}} =\{(0,0)\}$ and the schemes in \eqref{e:D_full}.  Under $H_0$, the
limits take the following forms.

\begin{enumerate}[label=(\roman*)]

\item For the classical scheme,
$$
\Gamma^{\mathrm{cl}}(t)
=
\frac{1}{g_t(0,0)}
\left|
\sum_{\ell=1}^{\infty} \lambda_\ell
\left[
t(1+t)
-
\left((1+t)W_\ell(1)-W_\ell(1+t)\right)^2
\right]
\right|.
$$

\item For the Page scheme \eqref{e:D_full},
$$
\begin{aligned}
\Gamma^{\mathrm{Page}}(t)
&=
\sup_{0\leq s\leq t}
\frac{1}{g_t(0,s)}
\Bigg|\sum_{\ell=1}^{\infty}\lambda_\ell
\big[
(t-s)(1+t-s)\\[-0.4ex]
&\hspace{5.5cm} 
-\left((t-s)W_\ell(1)+W_\ell(1+s)-W_\ell(1+t)\right)^2
\big]\Bigg|.
\end{aligned}
$$

\item For the full-type scheme \eqref{e:D_full},
$$
\Gamma^{\mathrm{full}}(t)
=\sup_{0\leq s\leq t}
\frac{1}{g_t(s,s)} \left| \sum_{\ell=1}^{\infty}\lambda_\ell
\left[\frac{(t-s)^2}{1+s}+(t-s)
-\left(\frac{1+t}{1+s}W_\ell(1+s)-W_\ell(1+t)\right)^2
\right]\right|.
$$
\end{enumerate}
\end{example}

\subsection{Consistency and delay bounds}

Given $F_*$ as in \eqref{e:H_A}, define the following \textit{discrepancy measure}:
\begin{equation}
\mathfrak D_h(F,F_*) =2\E h(\bX,\bX_*) - \E h(\bX,\bX')-\E h(\bX_*,\bX_*'),
\end{equation}
where $(\bX,\bX')\sim F\times F$ and $(\bX_*,\bX_*')\sim F_*\times F_*$ are independent.   When $\mathcal X=\R^d$ and $h(\bx,\by)=\|\bx-\by\|$, $\mathfrak D_h$ coincides with the energy distance \citep{szekely:rizzo:2022}, and more generally, when $h$ is a semimetric of strong negative type  \citep{sejdinovic:etal:2013}, the discrepancy $\mathfrak D_h$ has the property that
\begin{equation}\label{e:omnibus}
\mathfrak D_h(F,F_*)\neq 0 \iff F\neq F_*,
\end{equation}
i.e. $\mathfrak D_h(F,F_*)$ can distinguish between arbitrary distributions on $\mathcal X$.

We next state our main assumption under $H_A$. Write
$\overline h_*(\bx,\by)=h(\bx,\by)-\E h(\bx,\bX_*)-\E h(\bX_*,\by)
+\E h(\bX_*,\bX_*')$, where $\bX_*,\bX_*'\stackrel{\mathrm{iid}}{\sim}F_*$.
Thus $\overline h_*$ is the degenerate version of $h$ under $F_*$.
\begin{assumption}\label{a:Ha_dependence,moments}
The following conditions hold under $H_A$.
\begin{enumerate}[label=(\roman*)]
\item The pre-change segment $\bX_1,\ldots,\bX_{m+k_*}$ is taken from a stationary process satisfying Assumptions~\ref{a:Lpm} and~\ref{a:ac}.
\item The kernel $h$ is either positive semidefinite or a conditionally negative definite semimetric.
\item The integrals $\iint |h(\bx,\by)|\,F(d\bx)F_*(d\by)$ and
$\iint |h(\bx,\by)|\,F_*(d\bx)F_*(d\by)$ are finite, and, for
$\bX_*\sim F_*$, $\E|\overline h(\bX_*,\bX_*)|\leq C$ for a constant
$C$ independent of $m$.
\item For every $n\geq1$ and some $C<\infty$ independent of $m$ and $n$,
\[
\E\left|\sum_{r=1}^{n}\sum_{s=1}^{n}
\overline h_*(\bX_{m+k_*+r},\bX_{m+k_*+s})\right|\leq Cn.
\]
\end{enumerate}
\end{assumption}

Part (i) requires the pre-change sequence to satisfy the conditions previously assumed under $H_0$. Part (ii) allows the centered kernel $\overline h$ to be factored, up to a sign, for both pre- and post-change observations. Part (iii) ensures that $\mathfrak D_h(F,F_*)$ and $\overline h_*$ are well defined and controls diagonal terms. Part (iv) controls the centered post-change partial sums. It permits a transition in the dependence structure, provided the post-change marginals remain $F_*$ and the stated bound continues to hold.

We next give a simple consistency statement, together with its implied
delay bound, for the full monitoring scheme. 
\begin{theorem}\label{t:delay_full}
Let $\tau_m^{\rm full}$ denote the stopping time \eqref{e:stopping} based on the full scheme \eqref{e:D_full}.  Suppose $H_A$ and Assumptions~\ref{a:h}, \ref{a:Y_moment}, and
\ref{a:Ha_dependence,moments} hold.  Let $\{k_m,m\geq 1\}$ be any sequence of monitoring times with $k_m>k_*$
and set
$$
d_m=k_m-k_*,  \qquad  q_m=\frac{d_m}{m+k_*},  \qquad  t_m=\frac{k_m}{m}.
$$
If
$$
(m+k_*)|\mathfrak D_h(F,F_*)|  \left(\frac{q_m}{1+q_m}\right)^{2-\beta}  \log^{-(2-\beta)}(e+t_m)
\to\infty,
$$
then $\P\left(\tau_m^{\mathrm{full}}\leq k_m\right)\to1.$
\end{theorem}

\begin{corollary}\label{c:delay_full_rate}
Suppose $H_A$ and Assumptions~\ref{a:h}, \ref{a:Y_moment}, and
\ref{a:Ha_dependence,moments} hold, and let $\tau_m^{\rm full}$ be the
stopping time in Theorem~\ref{t:delay_full}. Assume that
$k_*$ is bounded by a polynomial in $m$, and that
$ (m+k_*)|\mathfrak D_h(F,F_*)| (\log m)^{-(2-\beta)} \to\infty.$
Then
$$
    \big(\tau_m^{\mathrm{full}}-k_*\big)_+  =
    O_\P\!\left(
        (m+k_*)^{\frac{1-\beta}{2-\beta}}
        |\mathfrak D_h(F,F_*)|^{-\frac{1}{2-\beta}}
        \log m
    \right).
$$
\end{corollary} Note the magnitude of the change is allowed to depend on $m$.  
For fixed alternatives (i.e., $|\mathfrak D_h(F,F_*)|\geq C>0$) and $k_*\asymp m^\theta$
Corollary~\ref{c:delay_full_rate} gives
$$
   \big(\tau_m^{\mathrm{full}}-k_*\big)_+  = O_\P\!\left(m^{s (\theta\vee 1)}\log m\right),\qquad s=\frac{1-\beta}{2-\beta}.
$$
This improves the late-change dependence (when $\theta>1)$ of the weighted CUSUM method of
\cite{kutta:dornemann:2025}, whose delay bound carries a factor of order
$m^{\theta-1}$, up to logarithmic terms.  Moreover, observe that the exponent factor $s=(1-\beta)/(2-\beta)$ decreases as 
$\beta$ tends to one, so the delay can be nearly logarithmic in $m$ when $\beta\approx 1$.  Thus, although the bound for $\tau_m^{\textrm{full}}-k_*$ remains
polynomial in $m$, its polynomial dependence can be made arbitrarily mild with many bounded kernels under mild data-level conditions (see Remark \ref{r:dependence_moments}).

This behavior is close in spirit to the logarithmic delays obtained by recent mean-targeted monitoring procedures
(\citealp{yu:madridpadilla:wang:rinaldo:2023,bastian:kutta:2025}). However, those sharp mean-monitoring guarantees rely on strong moment conditions on the raw data, which are not strictly required by our framework with appropriate choice of kernels.

\section{Implementation and spectral estimation}\label{s:implementation}

Our approaches to obtaining critical values for the test are based on spectral approximations of the operator $A$ in \eqref{e:def_A}.   Let
\begin{equation}\label{e:gram}
H_m=m^{-1}(H_{ij})_{1\leq i,j\leq m},\quad H_{i,j}=h(\bX_i,\bX_j),
\end{equation}
and define its centered counterpart:
\begin{equation}\label{e:centered_gram}
\overline H_{m} = C_m H_m C_m,\quad C_m= I_m - \frac{1}{m}\mathbf 1 \mathbf 1^\top, \qquad \mathbf 1 = (1,\ldots,1)^\top,
\end{equation}
so that, for any $1\leq i,j\leq m $, it holds that $(\overline H)_{ij}=m^{-1}\overline h_m(\bX_i,\bX_j),$ with
$$
\overline h_m(\bx,\by) = h(\bx,\by)
- \frac{1}{m}\sum_{j=1}^m h(\bx,\bX_j)
- \frac{1}{m}\sum_{i=1}^m h(\bX_i,\by)
+ \frac{1}{m^2}\sum_{i=1}^m\sum_{j=1}^m h(\bX_i,\bX_j).
$$
Throughout this section, we let
\begin{equation}\label{e:eigpairs_Hm}
(\widehat\lambda_{\ell,m},\mathbf v_{\ell,m}),\quad \ell=1,\ldots, m %
\end{equation}
denote the eigenvalue--eigenvector pairs of the matrix $\overline H_m$,
where the eigenvectors $\{\mathbf v_{i,m}\}$ are chosen orthonormal
and the eigenvalues are ordered so that
$|\widehat\lambda_{1,m}| \geq \ldots \geq |\widehat\lambda_{m,m}|.$

The function $\overline h_m(\bx,\by)$ is an empirical version of the true degenerate kernel $\overline h$, and accordingly, the matrix \eqref{e:centered_gram} serves as an approximation of $A$.  For kernels under Assumption \ref{a:h}, our next theorem extends Theorem 3.1 of \cite{koltchinskii:gine:2000} to serially dependent observations, and shows the eigenvalues $\widehat \lambda_{\ell,m}$ can be used to approximate $\lambda_\ell$. 

\begin{theorem}\label{t:eigen_consistency_theo}  Suppose Assumptions \ref{a:baseline}--\ref{a:ac} hold. Let $\widehat \lambda_{\ell,m}$ as in \eqref{e:eigpairs_Hm}.  Then,
\begin{equation}\label{e:def_vareps}
\min_{\pi\in\Pi_m}\left(\sum_{\ell=1}^{m} (\widehat \lambda_{\pi(\ell),m}-\lambda_{\ell})^2\right)^{1/2} =O_P(\varepsilon_m),
\end{equation}
where $\Pi_m$ is the set of permutations of $\{1,\ldots,m\}$, and
$\varepsilon_m = m^{-1/2}+\sum_{\ell=m+1}^\infty |\lambda_\ell|.$

\end{theorem}

Our estimation approach makes use of the leading $L$ principal eigenvalues and eigenvectors of $\overline H_m$. To that end, we require the following separation condition regarding the choice of $L$ in relation to the separation of the top $L+1$  eigenvalues of $A$.  To fix notation, ahead, for a $b=b(L)\leq L$, we write
\begin{equation}\label{e:def_muj}
\mu_1,\ldots,\mu_b,\qquad |\mu_1|\geq \ldots \geq |\mu_b|
\end{equation}
for the distinct values among $\lambda_1,\ldots\lambda_L$.
\begin{assumption}\label{a:eigsep}
$L=L_m$ is a sequence of integers $1\leq L \leq m$ along which, 
(i) $|\lambda_L| - |\lambda_{L+1}| \geq \delta_m,$ and (ii) 
$
\displaystyle{\min_{1 \leq k < k' \leq b}}
(|\mu_{k}|-|\mu_{k'}|) \geq \delta_m,$ 
where $\delta_m\geq 0$ is a sequence satisfying 
\begin{equation}\label{e:cond_deltam}
   \delta_m\gg \varepsilon_m,\qquad \sqrt{m}\,\delta_m^{3/2}\to\infty,
\end{equation}
where $\varepsilon_m$ is given by \eqref{e:def_vareps}. 
\end{assumption}
Note that second condition in \eqref{e:cond_deltam}  implies that $\delta_m\gg m^{-1/3}$.  Hence, under Assumption \eqref{a:eigsep},
$|\lambda_L|\geq \delta_m$, and thus $$
L_m\delta_m \leq \sum_{\ell=1}^{L_m}|\lambda_\ell|\leq \sum_{\ell=1}^\infty |\lambda_\ell| \implies L_m\lesssim  \delta_m^{-1} \ll m^{1/3}.
$$
The next theorem shows the $j$--th entries of the $L$ principal eigenvectors of $\overline H_m$ $(v_{1,m}(j),\ldots,v_{L,m}(j))^\top$ $\overline H_m$  can serve as suitable approximations of the eigenfunctions $(\phi_1(\bX_j),\ldots,\phi_{L}(\bX_j))^\top$, up to a rotation within each eigenspace.

\begin{theorem}\label{t:spectral_consistency} 
Suppose Assumptions \ref{a:baseline}--\ref{a:ac}  and \ref{a:eigsep} hold. Let 
$$
V_{m}^L=(\mathbf v_{1,m}: \ldots : \mathbf v_{L,m}),\qquad \Phi_m^L = (\boldsymbol \Phi_{1,m}:\cdots:\boldsymbol \Phi_{L,m}),
$$
  where ${\boldsymbol \Phi}_{\ell,m}=m^{-1/2}\left(\phi_\ell(\bX_1),\ldots,\phi_\ell(\bX_m)\right)^\top$.  With $\mu_1,\ldots,\mu_b$ as in \eqref{e:def_muj}, define 
$\mathcal B_k=\{\ell:\lambda_\ell=\mu_k\}$. Then there exists a sequence of block-orthogonal matrices
$$
O_{L,m}=\mathrm{diag}(O_{1,L,m},\ldots,O_{b,L,m}),
\qquad 
O_{k,L,m}\in\mathbb O(|\mathcal B_k|),
$$
such that
$$
\|V_m^L - \Phi_m^L O_{L,m}\|_{F} 
=O_P\left(\frac1{\sqrt{m}\,\delta_m^{3/2}}\right).
$$
\end{theorem}

\begin{remark} Note that  for any fixed $L$ under which $|\lambda_L|-|\lambda_{L+1}|>0,$ one may take $\delta_m\equiv \delta>0$ sufficiently small and still satisfy Assumption \ref{a:eigsep}, in which case the conclusion of Theorem \ref{t:spectral_consistency} holds at a $1/\sqrt m$ convergence rate.  Hence, the separation factor $1/\delta_m^{3/2}$ can be viewed as a cost in the alignment of $V_m^L$ and $\Phi_m^L$ incurred when allowing $L\to\infty$.
\end{remark}

We are now in a position to describe implementations of the test.  Ours is based on approximations of the long run covariance of the vector $\bxi_j=(\phi_1(\bX_j),\phi_2(\bX_j),\ldots)^\top$:
$$
\mathcal G =\{  \mathcal g_{\ell,\ell'},~1\leq \ell,\ell'<\infty\},
$$
where $\mathcal g_{\ell,\ell'}$ are as in \eqref{e:def_gell}. Recalling $\mathbf v_{i,m}=(v_{i,m}(1),\ldots,v_{i,m}(m))^\top$  define, for each $1\leq j \leq m$,
\begin{equation}\label{e:def_Vj}
\widehat{\bxi}_{L,j}=\sqrt m\big(v_{1,m}(j),\ldots,v_{L,m}(j)\big)^\top\in\R^L .
\end{equation}
Next, define
$$
\widehat{\bga}_{m}(r)=\frac1m\sum_{i=1}^{m-r}   \widehat{\bxi }_{L,i} \big( \widehat{\bxi }_{L,i+r})^\top,\quad r\geq 0,
$$
and set $\widehat{\bga}_m(-r)=\widehat{\bga}_m(r)^\top.$ To estimate $\mathcal G_L=\{  \mathcal g_{\ell,\ell'},~1\leq \ell,\ell'\leq L\},$ we use 
\begin{align}\label{e:def_LRC_estimator}
\widehat{\cG}_{L,m}=\sum_{r=-(m-1)}^{m-1}K \left(\frac{r}{\mathfrak b_m}\right)\widehat{\bga}_{m}(r),
\end{align}
where $K:\R\to\R$ is the lag-window kernel and $\mathfrak b_m>0$ its bandwidth\footnote{Note the lag-window kernel $K$ and its bandwidth $\mathfrak b_m$ play a separate role from the data kernel $h$ and, when applicable, its bandwidth, which we later denote by $\sigma>0.$}. The following assumption is standard in the literature:
\begin{assumption}\label{ker1} The function $K:\mathbb R\to\mathbb R$ is bounded, even, supported on $[-c,c]$
for some $c>0$, and continuous at zero with $K(0)=1$.
\end{assumption}
Our next assumption requires the bandwidth be chosen appropriately relative to $L$.
\begin{assumption}\label{ker2}\; As $m\to\infty$,%
 $\mathfrak b_m\to \infty$ and $\mathfrak b_m\left(\frac{1}{\sqrt{m}\delta_m^{3/2}} +  \sqrt L \varepsilon_m\right) \to 0,$
 where $\delta_m$ and $\varepsilon_m$ are as in Assumption \ref{a:eigsep}.
\end{assumption}

Since a finite-sample lag-window estimator need not be positive
semidefinite, we apply a trace-preserving positive-semidefinite
correction. For a symmetric matrix $B$, let $B^+$ denote the
Frobenius-nearest positive-semidefinite matrix to $B$ with $\tr(B^+)=\max\{\tr(B),0\}$.  Set
\begin{equation}\label{e:cov_plus}
\widehat\Omega_{L,m}
=\left(
|\widehat\Lambda_{L,m}|^{1/2}
\widehat{\mathcal G}_{L,m}
|\widehat\Lambda_{L,m}|^{1/2}\right)^+, \qquad \widehat\Lambda_{L,m}
=
\operatorname{diag}(\widehat\lambda_{1,m},\ldots,
\widehat\lambda_{L,m}).
\end{equation}
As shown in the Appendix, Theorem \ref{t:spectral_consistency}, together with \eqref{e:cov_plus}, furnishes a consistent estimator of the long-run covariance operator of the $\cH$-valued sequence $\{Y(\bX_j),j\in \mathbb Z\}$ defined by \eqref{e:def_Y(x)}.  Based on this result, we obtain the following theorem.

\begin{theorem}\label{t:LRC_sim_method} 
Suppose $H_0$ and Assumptions~\ref{a:baseline}--\ref{a:ac} and 
\ref{a:eigsep}--\ref{ker2} hold, and that
$$
L_m\to \ell_{\max}=\sup\{\ell\ge1:\lambda_\ell\neq0\} \in \mathbb N \cup\{\infty\}.$$
Conditionally on $\mathcal F_\infty=\sigma(\bX_1,\bX_2,\ldots)$, let
$\widehat{\mathbf Z}_{L,m}(t)=(\widehat Z_{1,L,m}(t),\ldots,\widehat Z_{L,L,m}(t))^\top$ be a centered $L$-dimensional Brownian motion
with covariance matrix $\widehat\Omega_{L,m}$ as in \eqref{e:cov_plus}. For $(u,v)\in\mathcal S(t)$, write
$$
\widehat Q_{\ell,L,m}(t;u,v)
=
\frac{t-v}{1+u}\widehat Z_{\ell,L,m}(1+u)
-\widehat Z_{\ell,L,m}(1+t)
+\widehat Z_{\ell,L,m}(1+v).
$$
With $\widehat \lambda_{\ell,m}$ as in \eqref{e:eigpairs_Hm}, set
$$
\widehat \Gamma_m(t)
=
\sup_{(u,v)\in\mathcal S(t)}
\frac{1}{g_t(u,v)}
\Bigg|\bigg(\sum_{\ell=1}^{L_m}\widehat\lambda_{\ell,m}\bigg)
\left(\frac{(t-v)^2}{1+u}+t-v\right)
-\sum_{\ell=1}^{L_m}
\operatorname{sgn}(\widehat\lambda_{\ell,m})
\widehat Q_{\ell,L,m}^2(t;u,v)
\Bigg|.
$$
Then, for any $a\in (0,\infty]$, and each continuity point $c>0$ of the distribution of
$\sup_{0<t<a}\Gamma(t)$, writing
$\P_\infty(\cdot)=\P(\cdot\mid\mathcal F_\infty)$, as $m\to\infty$,
$$\P_\infty\left\{ \sup_{0<t<a}\widehat \Gamma_m(t)\geq c \right\}\to\P\left\{ \sup_{0<t<a}\Gamma(t)\geq c \right\},\quad \textnormal{in probability.}
$$
\end{theorem}
Theorem \ref{t:LRC_sim_method} establishes validity of Monte Carlo approximation of the quantiles of $\sup_{0<t<a}\Gamma(t)$, and hence, feasibility of the test. We now provide an example to illustrate the requirements of Assumptions \ref{a:eigsep}--\ref{ker2}. The supporting calculations are given in the Supplement.

A worked example for polynomial eigenvalue decay is given in the Supplement (Example~\ref{ex:lambda}).

In some applications, it may be desirable to select a kernel bandwidth or to
transform the original data using information from the training sample. The following
corollary justifies such procedures for Lipschitz-type transformations.

\begin{corollary}\label{c:transformation}
For $\theta_0\in\R^\nu$, let
$\widehat\theta_m=\widehat\theta_m(\bX_1,\ldots,\bX_m)$ satisfy
$\sqrt m\,\|\widehat\theta_m-\theta_0\|=O_\P(1)$.
Let $(\theta,\bx)\mapsto t_\theta(\bx)$ be a measurable map taking values
in a separable metric space $(\mathcal Y,d)$, and set
$$
h_\theta(\bx,\by)
=
k(t_\theta(\bx),t_\theta(\by)),
$$
where $k:\mathcal Y \times \mathcal Y\to \R$ is bounded and satisfies the conditions of either
Example~\ref{ex:cnd_kernels} or Example~\ref{ex:psd_kernels}.
Suppose the hypotheses of Theorems~\ref{thone} and
\ref{t:LRC_sim_method} hold with $h=h_{\theta_0}$.
Suppose further that
$d(t_\theta(\bx),t_\theta(\by))\leq C\rho(\bx,\by)$,  and there is an $\bx_0 \in \mathcal X$ and a neighborhood $N$ of $\theta_0$ such that
$d(t_\theta(\bx),t_\tau(\bx))
\leq C\{1+\rho(\bx,\bx_0)\}\|\theta-\tau\|$, and
$|h_\theta(\bx,\by)-h_\tau(\bx,\by)|
\leq C\|\theta-\tau\|\{q(\bx)+q(\by)\}$
for all $\theta,\tau\in N$ and $\bx,\by\in\mathcal X$, where 
$ q(\bx)\geq 1$ with $\E q(\bX_1)^2<\infty$.
If the quantities $p,\alpha,\eta$ in Assumption~\ref{a:Lpm} satisfy
$p\alpha(1-\eta)>\nu$, then the conclusions of Theorems~\ref{thone} and \ref{t:LRC_sim_method} remain valid when
$\cD_m(k)$ and $\widehat\Gamma_m$ are constructed using
$h_{\widehat\theta_m}$, and with $A$ and 
$\Gamma$ defined using $h_{\theta_0}$.
\end{corollary}

\begin{remark}\label{r:transformation_examples}
Corollary~\ref{c:transformation} covers several practical data-driven
choices.  For $\mathcal X=\R^d$, bandwidth selection corresponds to
$t_\sigma(\bx)=\bx/\sigma$, where $\widehat\sigma_m$ is any
$\sqrt m$-consistent estimator of $\sigma_0>0$ and $\nu=1$.
PCA projection onto an estimated principal subspace is also permitted:
if $\widehat P_{r,m}$ is the orthogonal projector onto the leading $r$
sample principal components and
$\sqrt m\|\widehat P_{r,m}-P_r\|_F=O_\P(1)$, as follows under suitable
covariance-consistency and eigengap conditions, one may take
$t_P(\bx)=P\bx$ and $\nu=r(d-r)$.  Finally, the corollary encompasses the many parameter-stability problems, such as those considered in \cite{chu:stinchcombe:white:1996} and related works: for
$Y_i=g_{\theta_0}(\bX_i)+\varepsilon_i$, taking
$t_\theta(Y_i,\bX_i)=Y_i-g_\theta(\bX_i)$ yields historical-fit residuals, including one-step-ahead prediction errors, and hence a
distributional monitor for identifiable changes in $\theta_0$.
\end{remark}

\section{Simulation study}\label{s:simulations}

We study the finite-sample behavior of the proposed monitoring procedure under
the null hypothesis and under several alternatives.  The simulation setups are designed
to assess three aspects of the method: calibration of the test under various types of data,
power against distributional changes that are not necessarily mean shifts, and
detection delay for early and late changepoints.  For simplicity throughout this section, we take $\mathcal X =\R^d$ and $\rho(\bx,\by)=\|\bx-\by\|.$

  Throughout the simulation study we use two embeddings of the data:
\begin{equation}\label{e:raw,lag1}
    \bZ_t^{\mathrm{raw}}=\bX_t,\qquad    \bZ_t^{\mathrm{lag}}=(\bX_t^\top,\bX_{t-1}^\top)^\top .
\end{equation}
The raw version applies
the kernel directly to the observations, whereas  $\bZ_t^{\mathrm{lag}}$ is included to make the detector sensitive to changes in
short-range serial dependence  while retaining the same distributional form of the
statistic.  In both cases we use the bounded kernel
\begin{equation}\label{e:kernelchoice}
    h_{}(\mathbf x ,\by)
    =
    \left(1-\exp\left\{-\left(\frac{\|\bx-\by\|}{\sigma}\right)^{2\alpha}\right\}\right)^{1/2},
\end{equation}
 with $\alpha=1/4$  (which satisfies \eqref{e:omnibus}; see \citealp[Thm.~4.3]{boniece:horvath:trapani:2026}) and $\sigma$ chosen by the median heuristic from the historical sample.   We refer to these as the raw and lag-1 embeddings.

For computational efficiency, we use a geometric discretization of the Page and
full monitoring schemes.  In the notation of Assumption \eqref{e:S(t)_cond}, this corresponds to
restricting the scheme parameter to
$
    K_\rho=\{1-\rho^j:j=0,1,2,\ldots\}\cup\{1\}\subseteq[0,1]$ where $0<\rho<1.
$
Thus
\begin{equation}\label{e:geo_schemes}
\begin{gathered}
\mathcal S^{\mathrm{Page}}_\rho(t)
=
\left\{\bigl(0,(1-\rho^j)t\bigr):j=0,1,2,\ldots\right\}\cup\{(0,t)\},\\
\mathcal S^{\mathrm{full}}_\rho(t)
=
\left\{\bigl((1-\rho^j)t,(1-\rho^j)t\bigr):j=0,1,2,\ldots\right\}\cup\{(t,t)\}.
\end{gathered}
\end{equation}
In both cases, the right window in the continuous-time scheme has length
$\rho^j t$.  Hence, in finite samples, in the notation of
\eqref{e:detector}, we restrict the right-window length $k-\ell$
to a recursive geometric grid. For $k\geq2$, set $s_0(k)=k$ and
$$
s_{j+1}(k)=\max\{2,\lfloor\rho s_j(k)\rfloor\},
$$
stop when $s_j(k)=2$, and scan the distinct lengths $k-\ell=s_j(k)$.
This reduces the number of comparisons at time $k$ from $O(k)$ to
$O(\log k)$, while retaining a multi-scale grid concentrated at
recently available windows. In the experiments we set $\rho=0.9$;
the grid then has $55$ candidate windows at $k=2,000$ and $71$ at
$k=10,000$.
\subsection{Null calibration}
\label{s:simulations:null}

We first examine the finite-sample size of the proposed calibration procedure under four data-generating mechanisms: two linear processes and two with nonlinear dynamics. These are: an AR(1) process with autoregressive coefficient $\phi=0.50$ and standardized Laplace innovations; the same AR(1) process with standardized $t_3$ innovations; a threshold autoregressive (TAR) process,
$X_t=\kappa_tX_{t-1}
+(1-\kappa_t^2)^{1/2}\varepsilon_t$, with $\kappa_t=0.20\mathbf 1_{\{X_{t-1}\leq 0\}}+0.60\mathbf 1_{\{X_{t-1}>0\}}$, and standardized Laplace innovations; and a GARCH$(1,1)$ process with parameters
$(\omega,a,b)=(0.45,0.05,0.50)$ and standardized Laplace
innovations. 

The nominal level is $\alpha=0.05$. In
Table~\ref{tab:null-size-cL6-cbw5-newweight}, we report empirical rejection
probabilities for historical sample sizes $m\in\{100,300,500\}$ and
monitoring horizons $M=m$ and $M=5m$. For each null model, we consider both
the raw and lag-1 embeddings in \eqref{e:raw,lag1}. We use the bounded
energy kernel, set the weight parameter to $\beta=0.95$, and take
$\rho=0.9$. We report only results for the scheme
$S_\rho^{\mathrm{full}}$ in \eqref{e:geo_schemes}, since the corresponding
results for $S_\rho^{\mathrm{Page}}$ were similar.

\begin{table}[h!tb]
\centering
\caption{Empirical rejection probabilities under $H_0$ at nominal level
$\alpha=0.05$ for the calibration rule
$L_m=\lceil 2.5(\log m)^{4/3}\rceil$ and
$\mathfrak b_m=\lceil 5m^{1/3}\rceil$.
Results are based on 2000 Monte Carlo replications and 1000 simulations
to obtain each critical value.}
\label{tab:null-size-cL6-cbw5-newweight}

\small
\setlength{\tabcolsep}{4pt}
\begin{tabular}{llcccccccc}
\toprule
& & \multicolumn{4}{c}{$M=m$}
  & \multicolumn{4}{c}{$M=5m$} \\
\cmidrule(lr){3-6}\cmidrule(lr){7-10}
Embedding & $m$
& AR-Lap. & AR-$t_3$ & TAR & GARCH
& AR-Lap. & AR-$t_3$ & TAR & GARCH \\
\midrule
Raw
& 100 & .051 & .044 & .050 & .032
      & .047 & .051 & .062 & .021 \\
& 300 & .044 & .050 & .053 & .042
      & .056 & .055 & .071 & .045 \\
& 500 & .048 & .047 & .056 & .053
      & .054 & .062 & .065 & .050 \\
\midrule
Lag-1
& 100 & .037 & .035 & .045 & .019
      & .050 & .051 & .066 & .018 \\
& 300 & .047 & .054 & .046 & .041
      & .059 & .057 & .064 & .041 \\
& 500 & .049 & .040 & .047 & .040
      & .056 & .052 & .052 & .042 \\
\bottomrule
\end{tabular}
\end{table}

Critical values are computed using the Gaussian spectral approximation
described in Section~\ref{s:implementation}. The long-run covariance
operator is estimated using a Bartlett kernel, with bandwidth
$\mathfrak b_m$, and the spectral approximation $L_m$ chosen as
\begin{equation}\label{e:Lm_bm_default}
L_m=\big\lceil 2.5(\log m)^{4/3}\big\rceil,
\qquad
\mathfrak b_m=\big\lceil 5m^{1/3}\big\rceil.
\end{equation}

This rate rule was selected from a small calibration grid to provide
stable size across the representative null models. Although no
finite-sample tuning rule can be universal, this is our recommended
default when using the kernel in \eqref{e:kernelchoice} and those with infinite spectrum but with low effective dimension (see Section \ref{s:hidimcalib}). Each entry in
Table~\ref{tab:null-size-cL6-cbw5-newweight} is based on 2000 Monte Carlo
replications using each 1000 simulations of the limit to 
compute the critical value.

\subsubsection{Calibration in higher dimensions}\label{s:hidimcalib}

The framework accommodates a broad class of kernels and dependence models, whose finite-sample sizes need not be equally close to the nominal level at a given historical sample size. Calibration depends on the spectral
structure of the kernel and the underlying dependence, rather than on data dimension alone. In particular, rough kernels can have diffuse spectra, especially as the underlying model becomes more complex, making estimation of the long-run covariance more demanding. To illustrate, the Supplement reports empirical effective ranks
(\citealp{lopes:jacob:wainwright:2011}),
$\widehat r_{\rm eff}=(\sum_{\ell=1}^m|\widehat\lambda_\ell|)^2/
\sum_{\ell=1}^m\widehat\lambda_\ell^2$.
This quantity measures the diffusiveness of the spectrum of $\overline H_m$ and provides an empirical measure of the effective dimension of $A$. At $m=500$, the bounded kernel \eqref{e:kernelchoice} with $\alpha=1/4$ has median empirical effective rank
$46.7$ for scalar AR(1) data and $473$ for a dense $d=50$ VAR(1).

Hence, to facilitate calibration in higher dimensions, we consider three smoother kernels with more concentrated spectra: Gaussian, rational quadratic,
and Mat\'ern $3/2$, respectively,
\begin{equation}\label{e:highdim_psd_kernels}
h_{\rm G}(\bx,\by)=e^{-r^2/2},\quad
h_{\rm RQ}(\bx,\by)=(1+r^2)^{-1/2},\quad
h_{\rm M}(\bx,\by)=(1+\sqrt{3}\,r)e^{-\sqrt{3}\,r}.
\end{equation}
Here $r=\|\bx-\by\|_2/(c\widehat\sigma_m)$, where
$\widehat\sigma_m$ is the median pairwise distance in the historical sample,
held fixed during monitoring, and $c\in\{1,2\}$.\footnote{We refer to \citet{sriperumbudur:etal:2010}, who show that each of these kernels is distribution distinguishing.} For the dense VAR(1) model at
$c=1$, their median empirical effective ranks are $72.7$, $107.0$, and
$136.0$, respectively, compared with $473$ for the rough kernel \eqref{e:kernelchoice}.

We take $m=M=500$ and retain
$L_m=\lceil2.5(\log m)^{4/3}\rceil=29$ terms for calibration. Under the regularity
assumptions, smoother kernels yield faster decay bounds in the
$\mathcal L^p$-approximability coefficients compared to rougher kernels, motivating shorter temporal bandwidths. We therefore use Bartlett long-run covariance estimation with bandwidth $\mathfrak b_m=1+\lfloor4(m/100)^{2/9}\rfloor$, under the conventional Newey--West rule of thumb. We impose a minimum right-window length of 10, which slightly reduced conservatism at this sample size without affecting the asymptotics. 
Table~\ref{tab:highdim-psd-calibration} reports rejection probabilities
for two VAR models in $d=50$ and two functional principal-component
models discretized in $d=500$; model specifications are given in the Supplement.

\begin{table}[h!tb]
\centering
\caption{Empirical rejection probabilities under $H_0$ at nominal level
$\alpha=0.05$ for the kernels in \eqref{e:highdim_psd_kernels}.
Here $m=M=500$ and $L_m=29$. Results are based on 2000 Monte Carlo replications and 1000 simulations
to obtain each critical value.}\vspace{1ex}
\label{tab:highdim-psd-calibration}\footnotesize
\begin{tabular}{lccccc}
\toprule
Kernel & $c$ & VAR(1) & Threshold VAR(1) & Linear FPC & Nonlinear FPC \\
 & & $d=50$ & $d=50$ & $d=500$ & $d=500$ \\
\midrule
Gaussian & 1 & .0410 & .0335 & .0435 & .0290 \\
         & 2 & .0515 & .0400 & .0475 & .0345 \\
\addlinespace
Rational quadratic & 1 & .0335 & .0270 & .0335 & .0195 \\
                   & 2 & .0460 & .0365 & .0460 & .0335 \\
\addlinespace
Mat\'ern $3/2$ & 1 & .0280 & .0195 & .0295 & .0145 \\
              & 2 & .0385 & .0325 & .0425 & .0275 \\
\bottomrule
\end{tabular}
\end{table}

Rejection frequencies are close to or below the nominal level, with
greater conservatism at $c=1$. These results illustrate that calibration
can remain practical in high dimensions when the kernel spectrum is
more concentrated. However, we caution that a common historical sample
size may not suffice for every kernel and dependence model.  Joint data-adaptive selection of the kernel bandwidth, the temporal
bandwidth, and the spectral truncation level is an important direction for future work.

\subsection{Power and detection delay at a fixed break location}
\label{s:simulations:fixed-break}

We next compare the finite-sample power and detection delay of the proposed procedures under a range of alternatives. Throughout this subsection, the historical sample size is fixed at $m=300$, monitoring continues for $M=800$ periods, and the changepoint occurs at monitoring time $k_*=300$.

Across three different classes of DGPs (AR(1), VAR(1), and GARCH(1,1)), we consider three families of alternatives, which we describe below. In addition to the schemes in \eqref{e:geo_schemes}, we include the classical scheme
for reference and set $\beta=0.95$ throughout. As benchmarks, we use the weighted CUSUM
procedure of \citet{kutta:dornemann:2025}, denoted WC, and the mean-targeted TWIN and
nonparametric ECDF-based NP-TWIN procedures of \citet{bastian:kutta:2025}. NP-TWIN is
omitted from the multivariate VAR experiments because the empirical distribution function is not easily computed in higher dimensions. All benchmark methods are implemented
with their recommended default parameters.

For a matched comparison of detection delays, each method is separately calibrated to nominal $5\%$ rejection over the monitoring horizon $M$ using $2,000$ paths from the corresponding null model. Since monitoring stops at the first alarm, we report the distribution of $\tau_m-k_*$ conditional on $\{k_*<\tau_m\leq M\}$ and give the post-change detection probability $\P\{\tau_m\leq M\mid\tau_m>k_*\}$ in parentheses. This probability makes the frequency of successful post-change detection visible alongside the conditional delay; pre-change alarms are excluded from both quantities because they would artificially shorten the reported delays. Their frequencies are reported separately in Table~\ref{tab:fixed-break-prechange-false-alarms} of the Supplement.

The first alternative family, corresponding to Figure~\ref{f:alt_fixedbreak_iso}, contains alternatives in which the mean is unchanged and the break occurs through serial dependence or volatility structure. Specifically, we consider a persistence change in an AR(1) model, a low-rank factor-structure change in a VAR(1) model, and a change in volatility dynamics in a GARCH(1,1) model.

\begin{figure}[h!]
\begin{center}
\includegraphics[width=0.9\linewidth]{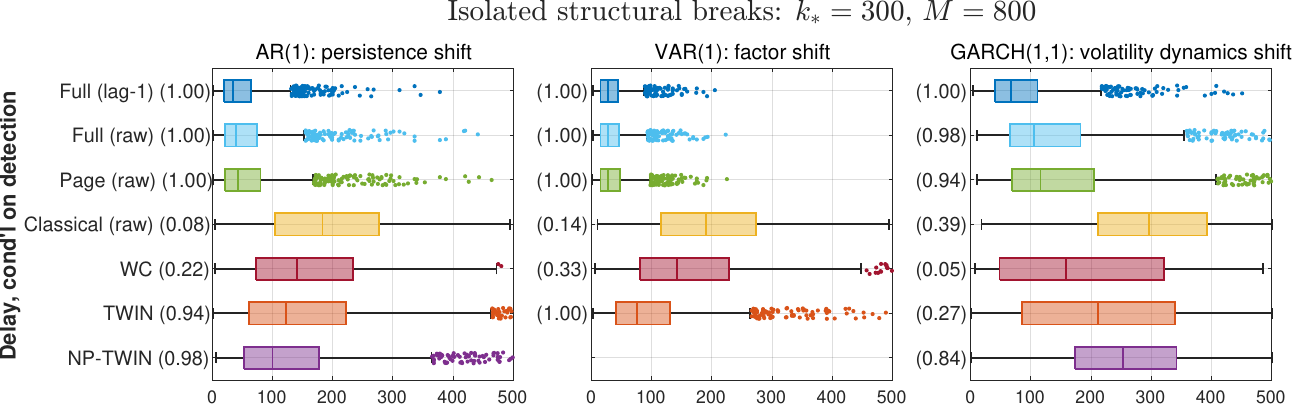}
\end{center}
\caption{Empirical delay distributions for isolated structural breaks.%
\label{f:alt_fixedbreak_iso}
}
\end{figure}

The second family, shown in Figure~\ref{f:alt_fixedbreak_mean}, considers a scalar AR(1) model with $\phi=0.50$ before and after the break and a change only in the mean, from $\mu=0$ to $\mu\in\{0.5,1\}$.

\begin{figure}[h!]
\begin{center}
\includegraphics[width=0.7\linewidth]{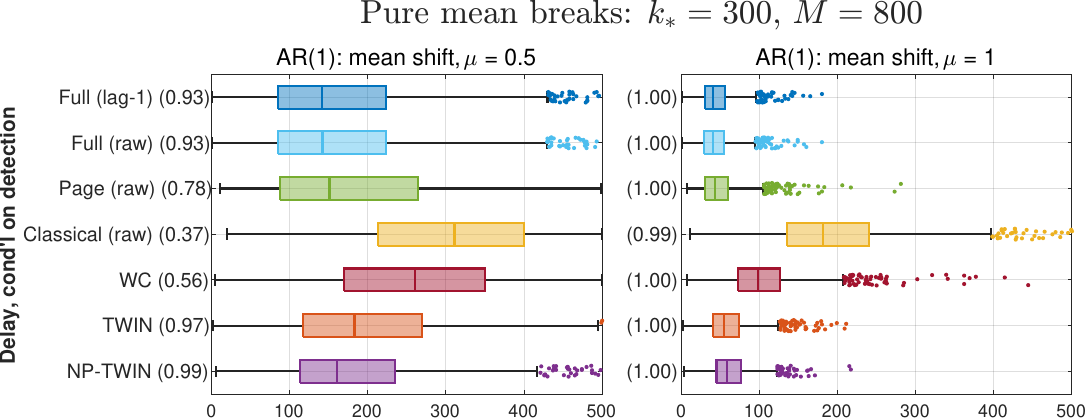}
\end{center}
\caption{
Conditional delay distributions for pure mean breaks.%
\label{f:alt_fixedbreak_mean}}
\end{figure}

The third family, shown in Figure~\ref{f:alt_fixedbreak_mixed}, combines location changes with structural changes.  Full model specifications for each of the three families of alternative scenarios are given in Section \ref{supp:sim} of the Supplement.

\begin{figure}[h!]\
\begin{center}
\includegraphics[width=0.9\linewidth]{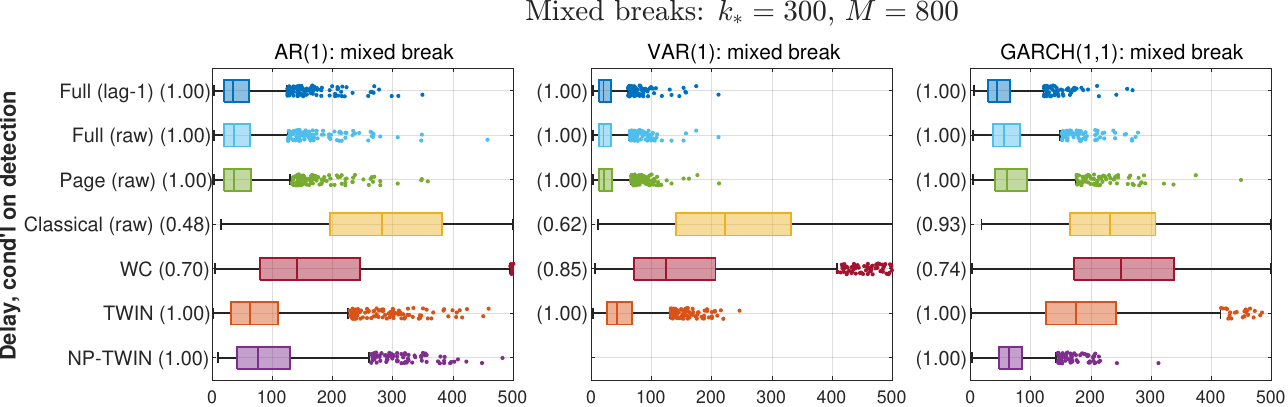}
\caption{
Delay distributions for mixed-type breaks.%
\label{f:alt_fixedbreak_mixed}}
\end{center}
\end{figure}

The fixed break location experiments broadly show favorable delay behavior for the proposed procedures across structural, pure mean, and mixed alternatives. Across the reported plots, the full kernel scheme is typically fastest conditional on post-change detection, while Page is slightly slower. Page may therefore best viewed as a reference scanning scheme, although it may be useful as a conservative scan if a gradual change is expected instead of a sharp change at $k_*$.  The classical scheme is slower and serves mainly as a benchmark. The lag-1 embedding is often slightly faster than the raw embedding, especially for GARCH.

The structural-break experiments give the clearest advantage to our proposed method. When the mean is fixed and the change occurs through persistence, factor structure, or volatility dynamics, mean-targeted benchmarks are not expected to perform especially well; their weaker performance reflects this mismatch rather than a general deficiency. More strikingly, under pure mean breaks, where those benchmarks are directly targeted, the proposed procedure does not appear to pay a meaningful delay penalty for being omnibus: in the reported plots it remains typically fastest while maintaining high post-change detection probability.

These findings show that the proposed  kernel procedure can have short delays and is broadly powerful at a moderate break location, but do not address how comparisons change as the break moves through the monitoring period; this is examined in the next section. 
\subsection{Delay across break locations}

We next examine detection delay as the changepoint moves across the monitoring period.
All experiments use a univariate AR(1) model with Laplace innovations. We take $m=100$ and
consider changepoints
$k^\star\in\{m^{1/2},m,m^{3/2},m^2\}$, with $M=2500$ for the first three cases and
$M=12000$ for $k^\star=m^2$. We consider a strong pure mean break and a mixed break
combining changes in location, persistence, and variance; full model specifications are given
in the Supplement.

We report the proposed lag-one procedure with $\beta=0.95$ and $\beta=0.99$, together
with WC, TWIN, and NP-TWIN. Raw embeddings and the Page and classical schemes are omitted
for succinctness. The corresponding pre-change alarm frequencies are reported in
Table~\ref{tab:delay-location-prechange-false-alarms} of the Supplement.

\begin{figure}[h!]
    \begin{center}
        \includegraphics[width=0.95\linewidth]{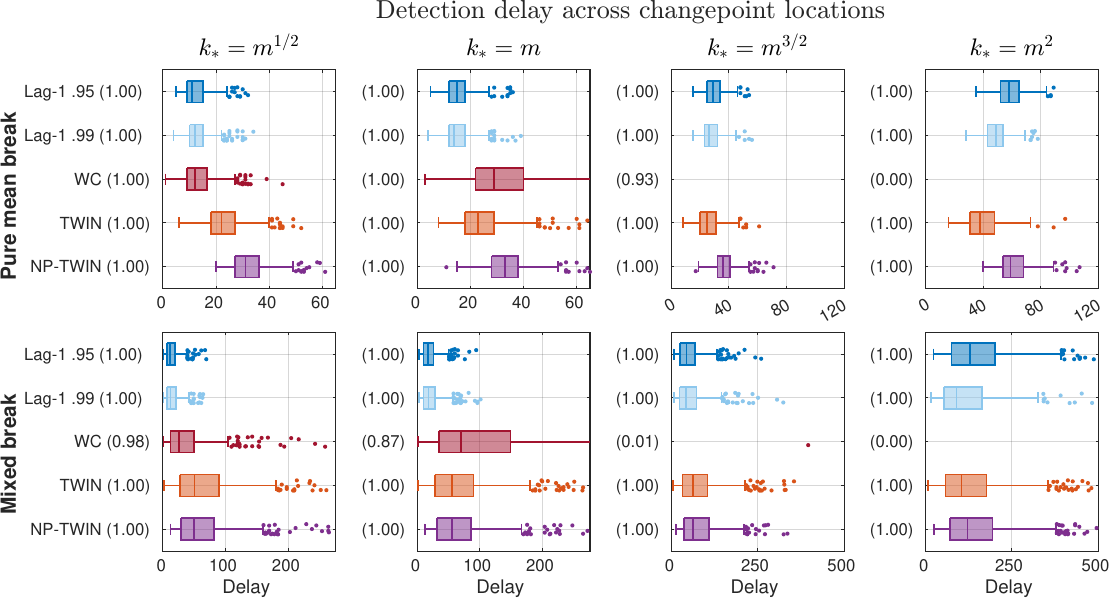}
    \end{center}
    \caption{
Empirical detection delays across changepoint locations for pure mean and mixed breaks in an AR(1) model. Parentheses give estimates of $\P\{\tau_m \leq M|\tau_m>k_*\}$.   Delays for WC in the $k_*=m^{3/2}$ pure mean case lie beyond the plot limits.
    }
    \label{fig:delay_puremean}
\end{figure}

Figure~\ref{fig:delay_puremean} shows that the favorable delay behavior of the proposed procedure persists as the changepoint moves to the later regimes $k^\star=m^{3/2}$ and $k^\star=m^2$. For pure mean changes, TWIN is strongest at the latest locations, as expected
from its mean-targeted design, whereas for mixed breaks the proposed procedure is competitive with, and often faster than, the benchmarks. Performance is similar for $\beta=0.95$ and $\beta=0.99$; together with the null results in Section~\ref{s:simulations:null}, this supports $\beta=0.95$ as a useful default for bounded kernels.

\begin{remark}
The proposed method is computationally more demanding than CUSUM-based detectors
such as WC and TWIN. Despite the geometric grid, our implementation requires
$O((m+M)^2)$ accumulated pairwise-kernel evaluations over a horizon $M$. Developing computationally efficient variants based on kernel approximations (e.g., \citealp{williams:seeger:2000,rahimi:recht:2007}) under dependence is left for future work.

\end{remark}

\section{Data illustrations}\label{s:data}

\paragraph{Exchange-rate returns.}
We first apply our procedure to daily USD/JPY log returns, using the geometric full scheme \eqref{e:geo_schemes} with the kernel in \eqref{e:kernelchoice} and the same tuning choices as in Section~\ref{s:simulations:fixed-break}. The historical sample comprises $m=125$ returns from December 31, 2024 to June 30, 2025, and monitoring begins on July 1, 2025; we use WC and TWIN as comparisons. Our procedure signals distributional instability in early June 2026, following an official yen-supporting foreign-exchange intervention by Japanese authorities in May\footnote{See https://www.mof.go.jp/english/policy/international\_policy/reference/feio/monthly/index.html}; in contrast, neither WC nor TWIN crosses its critical boundary.

\begin{figure}[h!]
    \begin{center}
        \includegraphics[width=0.49\linewidth]{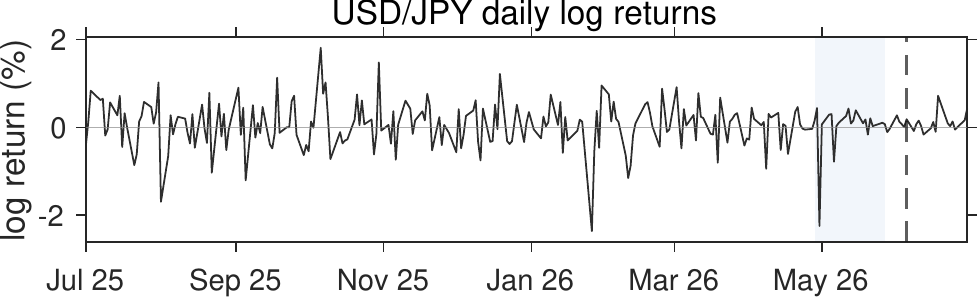}
        ~\includegraphics[width=0.49\linewidth]{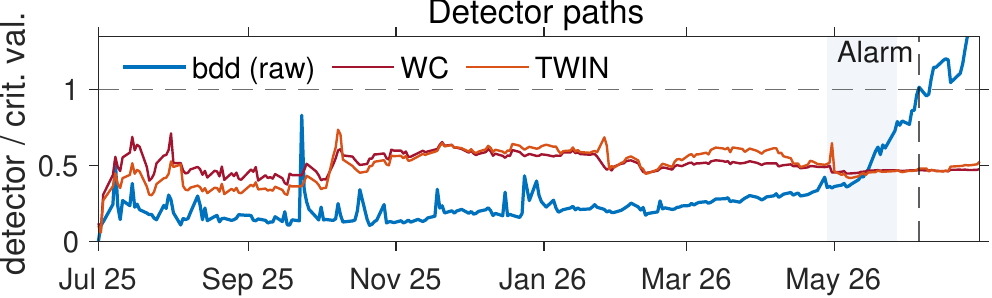}
    \end{center}
\caption{Daily USD/JPY log returns (left) and detector-to-critical-value ratios (right), with the horizontal dashed line marking the rejection threshold. The shaded region denotes the Japan Ministry of Finance reporting period, April 28--May 27, 2026, during which an official yen-supporting intervention took place.\label{fig:usdjpy}}
\end{figure}

\paragraph{Electricity prices and demand forecast errors.} We consider 24-hour day-ahead load forecasts and electricity-price curves from the Electric Reliability Council of Texas (ERCOT), which manages the power grid serving most of Texas.\footnote{Forecasts, realized demand, and day-ahead market and real-time market settlement-point prices were obtained from ERCOT's public archives. Real-time prices were averaged to the hourly frequency and matched to the corresponding day-ahead prices. We use system-wide load for Winter Storm Uri, Coast-zone load for Hurricane Nicholas, and prices at the settlement point \texttt{HB\_HUBAVG}.} ERCOT provides an hourly day-ahead forecast of system load, i.e., electricity demand on the grid. For operating day $t$, we combine the load forecast $\widehat D_{t|t-1}(u)$ at hour $u$ issued by ERCOT on day $t-1$ with the observed load $D_t(u)$ at hour $u$  to form
the 24-hour percentage forecast-error curve
$
\bE_t =\left\{100\,\frac{D_{t}(u)-\widehat D_{t\mid t-1}(u)}{\widehat D_{t\mid t-1}(u)},~u=1,\ldots,24\right\}.
$
We pair this with the 24-hour price-spread curve
$
\bS_t =\left\{ P_t^{\mathrm{RT}}(u)-P_t^{\mathrm{DA}}(u),~u=1,\ldots,24\right\},
$
where $P_t^{\mathrm{DA}}(u)$ is the day-ahead market price for hour $u$ and
$P_t^{\mathrm{RT}}(u)$ is the corresponding average of the real-time market prices in hour $u$; we monitor the bivariate functional series $(\bE_t,\bS_t)^\top$. We focus on two separate monitoring periods, surrounding Winter Storm Uri and Hurricane Nicholas, using $m=183$ historical observations and horizons of $M=45$ operating days. Before computing principal-component scores, we standardize each hourly coordinate using its historical mean and standard deviation. We use the geometric full scheme \eqref{e:geo_schemes} and a bounded rank-four kernel obtained by applying $\tanh(\cdot/5)$ to each of the four products of the first two standardized principal-component scores of $\bE_t$ and $\bS_t$, thereby targeting changes in their interaction; full construction and tuning details are given in Section~\ref{s:data_supplement}. As a benchmark, we apply WC to the full standardized bivariate mean curve. For Nicholas, our procedure signals on 13 September 2021, one day before the hurricane landfall, whereas WC does not signal, suggesting a change in the joint behavior of the forecast errors and day-ahead price spreads. For Winter Storm Uri, our procedure signals on 10 February 2021, three days before the annotated event date, and WC signals on 9 February 2021. Complete results are reported in Section~\ref{energy-supp} of the Supplement.

\begin{figure}[h!]
    \begin{center}
        \includegraphics[width=0.48\linewidth]
        {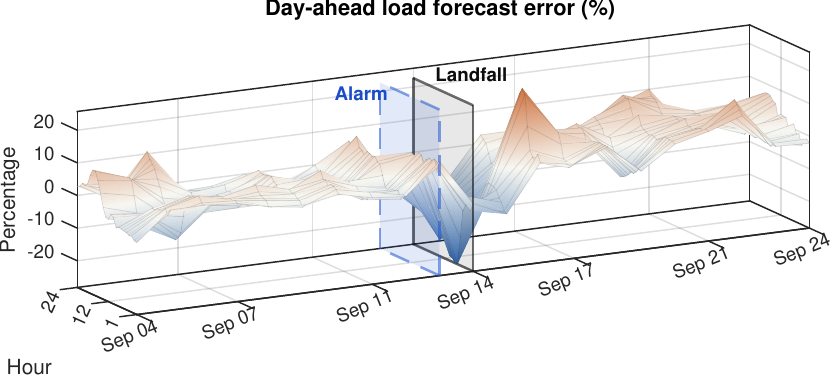}
        \quad
        \includegraphics[width=0.48\linewidth]
        {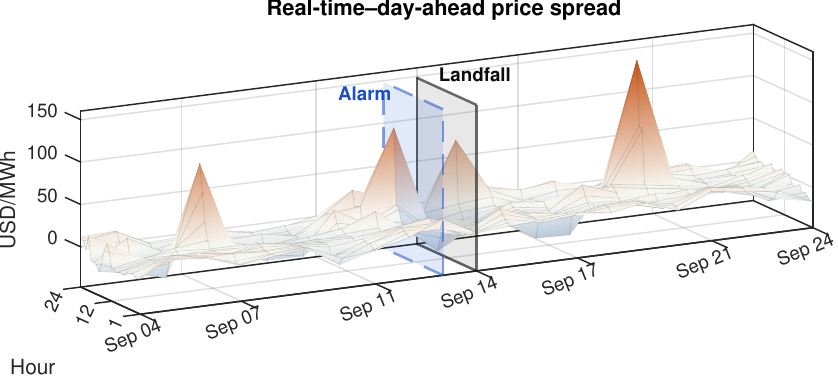}
    \end{center}
    \caption{Hourly ERCOT data around Hurricane Nicholas. The left panel shows the signed Coast-zone day-ahead load-forecast error, in percent, and the right panel shows the real-time--day-ahead price spread at the ERCOT 345-kV hub average, in USD/MWh. The blue and black planes mark the kernel alarm on 13 September 2021 and landfall on 14 September 2021, respectively.}
    \label{fig:ercot-nicholas-data}
\end{figure}

\paragraph{Air transportation networks.}
To further illustrate the flexibility of our method, we analyze daily U.S.\ domestic
flight networks constructed from the BTS Marketing Carrier On-Time Performance database
for American, Delta, United, and Southwest, restricting to the 50 busiest airports selected
from the associated historical sample. We consider three representations targeting different
aspects of the network. Let $\bA_t$ be the symmetric weighted adjacency matrix of operated flights on day $t$, with entry $A_{ij,t}$ equal to the number of flights between airports $i$ and $j$ in either direction.  We first
normalize by the total number of flights $c_t$ that day, and monitor the resulting route-share
network $\bP_t=c_t^{-1}\bA_t$ using total-variation distance. Second, to capture large-scale network structure,
we form the normalized graph Laplacian
$\mathbf L_t =\bI-\bD_t^{-1/2}\bA_t\bD_t^{-1/2}$ where $\bD_t=\diag(\bA_t\boldsymbol 1)$ and monitor its first five positive eigenvalues
using a bounded rank-five kernel. Finally, we monitor the four carrier-specific cancellation
rates using a bounded rank-four kernel. Full construction and kernel details are given in
Section~\ref{s:data_supplement} of the Supplement.

Similar to \citet{wang:li:madridpadilla:yu:rinaldo:2026}, we focus on the major airline
disruption surrounding COVID-19. We use $m=365$ historical days, from January 1--December
31, 2019, and monitor from March 1--June 30, 2020. Letting $\bY_t$ denote generically the
chosen representation, we remove a historical calendar mean $\widehat{\boldsymbol\mu}_t$
and fit
$(1-\phi\mathsf B)(1-\psi\mathsf B^7)
(\bY_t-\widehat{\boldsymbol\mu}_t)=\boldsymbol\varepsilon_t$,
with all quantities estimated from the historical sample and subsequently frozen. We apply
our procedure to the resulting residual sequence, as justified by
Corollary~\ref{c:transformation}. This allows us to retain daily resolution while accounting
for calendar and serial dependence; full preprocessing details are given in the Supplement.
The route-share, Laplacian, and cancellation-rate detectors signal on March 24, March 26,
and March 18, 2020, respectively. By contrast,
\citet{wang:li:madridpadilla:yu:rinaldo:2026} work at monthly frequency,
assume independence, and report a change in April 2020.

\begin{figure}[t]
\centering
\includegraphics[width=.48\textwidth]{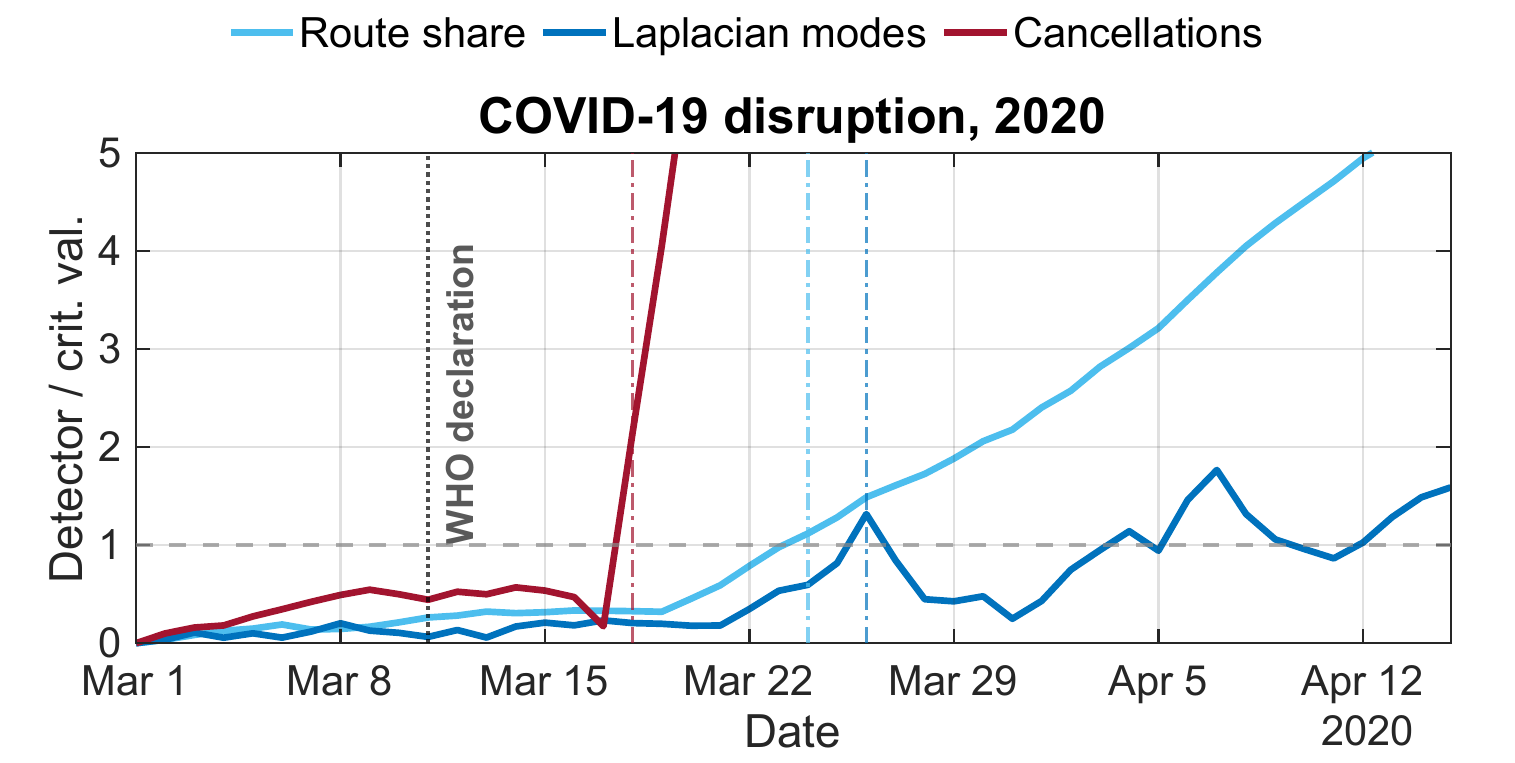}
\hfill
\includegraphics[width=.48\textwidth]{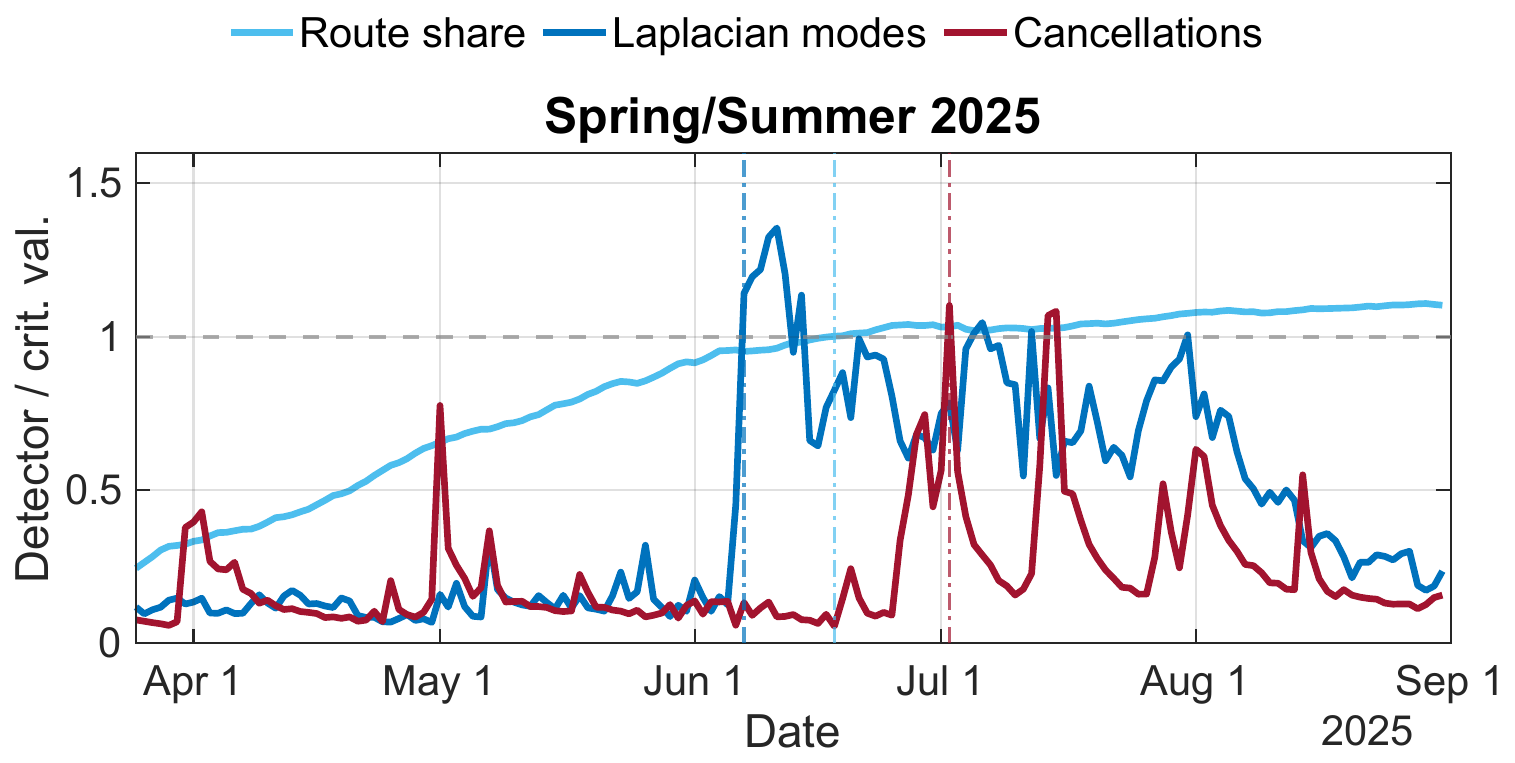}
\caption{Airline-network monitoring results for the COVID-19 disruption
(left) and the early-summer 2025 episode (right). Curves give the detector
divided by its critical value; the horizontal dashed line marks the
rejection threshold and the colored vertical lines mark the corresponding
first alarms. The dotted line in the left panel marks the WHO pandemic
declaration.}
\label{fig:airline}
\end{figure}

As a robustness exercise, we also consider Spring/Summer 2025, using March 1,
2024--February 28, 2025 as the historical sample and monitoring from March 1--August 31,
2025. The Laplacian and route-share detectors signal on June 7 and June 18, respectively,
while the cancellation-rate detector does not signal until July 2, illustrating that changes
in network topology and route allocation may be detected separately from cancellation-driven events.

\section{Conclusion}

We have developed a general framework for online detection of distributional changes in dependent time series taking values in metric spaces. The results show that, with suitable weighting, distributionally sensitive procedures can attain short detection delays over a broad range of alternatives while retaining comparatively weak moment requirements on the underlying observations.

Numerical results also indicate that finite-sample calibration becomes more difficult for models of higher complexity, and care should be taken in kernel choice. In the high-dimensional experiments, smoother kernels with more concentrated spectra improve calibration at a fixed historical sample size, although the choice also depends on the temporal dynamics. This suggests that, in high-dimensional settings, the choice of kernel should account not only for sensitivity to alternatives, but also for the complexity of the induced representation relative to the available historical sample. A more complete understanding of this finite-sample tradeoff, as well as data-driven selection of kernels and their complexity, would be useful topics for further study.

\clearpage
\thispagestyle{empty}

\begin{center}
    {\Large\bfseries Supplementary material}\\
\end{center}
\appendix

\paragraph{Notation.}
Unless stated otherwise, we retain the notation of the main paper. For a bounded
linear operator $T$ between Hilbert spaces, $\|T\|_{\mathrm{op}}$ denotes its
operator norm and $T^*$ its adjoint. For a compact operator $T$, we write
$\|T\|_{\mathrm{HS}}$ and $\|T\|_{\mathrm{tr}}$ for its Hilbert--Schmidt and
trace norms, respectively, whenever these are finite, and $\tr(T)$ for its trace.
For a finite-dimensional matrix $M$, $\|M\|_{\rm F}$ denotes the Frobenius norm.
For $x,y$ in a Hilbert space $\dcal K$ with inner product $\langle\cdot,\cdot\rangle_{\dcal K}$, $x\otimes y$ denotes the rank-one operator
$z\mapsto\langle y,z\rangle_{\dcal K} x$. For $a_m,b_m>0$, we write $a_m\lesssim b_m$ if
$a_m\leq Cb_m$ for all sufficiently large $m$ and some constant $C<\infty$ independent
of $m$; $a_m\asymp b_m$ means $a_m\lesssim b_m$ and $b_m\lesssim a_m$; $a_m\ll b_m$ means $a_m/b_m\to 0$, and $a_m\gg b_m$ means $b_m/a_m\to 0$. When $\mu,\nu$ are measures, $\mu\ll \nu$ means $\mu$ is absolutely continuous with respect to $\nu$.  Other relevant notation is introduced throughout.

\section{Miscellaneous statements}\label{miscellanea}

\begin{lemma}\label{l:kernel_examples}
Assumptions~\ref{a:h} and \ref{a:Y_moment} are satisfied in either of the following settings.

\begin{enumerate}[label=(\roman*)]
\item $h(\bx,\by)$ is positive definite, and
$$
\E h(\bX_1,\bX_1)^{p/2}<\infty,\qquad
h(\bx,\bx)+h(\by,\by)-2h(\bx,\by) \leq C\rho(\bx,\by)^{2\alpha}.
$$

\item $h(\bx,\by)$ is a semimetric of negative type, and
$$
\E h(\bX_1,\bx_0)^{p/2}<\infty,\qquad
|h(\bx,\by)| \leq C\rho(\bx,\by)^{2\alpha},
$$
for some $\bx_0\in\mathcal X$.
\end{enumerate}
\end{lemma}

\begin{proof}
First suppose $h$ is positive definite. Since
$
|h(\bx,\by)|\leq \sqrt{h(\bx,\bx)h(\by,\by)},
$
and $p>4$, the assumed moment condition gives
  $\E h(\bX_1,\bX_1)^2<\infty$; together with the preceding
  inequality, this verifies Assumption~\ref{a:h}(i). Further, by
\cite[Theorems 4.16 and 4.21]{steinwart:christmann:2008}, writing
$$
\varphi(\bx)=h(\cdot,\bx),
$$
there is a Hilbert space of functions $\dcal K$ with
$
h(\bx,\by)=\langle \varphi(\bx),\varphi(\by)\rangle_{\dcal K},
$
and $\langle f,\varphi(\bx)\rangle_{\dcal K}=f(\bx)$ for each  $f\in\dcal K. $
Since
$$
\E \|\varphi(\bX_1)\|_{\dcal K}
=
\E \sqrt{h(\bX_1,\bX_1)}
<\infty,
$$
we have $h_1:=\E h(\cdot,\bX_1)\in \dcal K$. Thus, with
$
\overline \varphi(\bx)=h(\cdot,\bx)-\E h(\cdot,\bX_1),
$
we obtain
\begin{align*}
\langle \overline \varphi(\bx),\overline \varphi(\by)\rangle_{\dcal K}
&=
\langle \varphi(\bx),\varphi(\by)\rangle_{\dcal K}
-\langle \varphi(\bx),h_1\rangle_{\dcal K}
-\langle \varphi(\by),h_1\rangle_{\dcal K}
+\langle h_1,h_1\rangle_{\dcal K}\\
&=
\overline h(\bx,\by),
\end{align*}
showing that $\overline h$ is positive definite, and hence $A$ is positive. Therefore,
\begin{align*}
\tr(|A|)
=
\tr(A)
=
\E \overline h(\bX_1,\bX_1)
\leq
\E h(\bX_1,\bX_1)
<\infty,
\end{align*}
so Assumption~\ref{a:h}(ii) holds. Now, expanding
$$
\overline h(\bx,\by)=\sum_{\ell=1}^\infty\lambda_\ell\phi_\ell(\bx)\phi_\ell(\by),
$$
we may write
$$
\overline h(\bx,\by)=\langle \Phi(\bx),\Phi(\by)\rangle_{\ell^2(\mathbb N)},
\qquad
\Phi(\bx)=\big(\sqrt{\lambda_1}\phi_1(\bx),\sqrt{\lambda_2}\phi_2(\bx),\ldots\big)^\top.
$$
Hence
\begin{align}
\sum_{\ell=1}^\infty |\lambda_\ell|(\phi_\ell(\bx)-\phi_\ell(\by))^2
&=
\|\Phi(\bx)-\Phi(\by)\|^2_{\ell^2(\mathbb N)}\notag\\
&=
\overline h(\bx,\bx)+\overline h(\by,\by)-2\overline h(\bx,\by)\label{e:canonical_h_identity_new}\\
&=
h(\bx,\bx)+h(\by,\by)-2h(\bx,\by)\notag\\
&\leq C\rho(\bx,\by)^{2\alpha},\notag
\end{align}
showing Assumption~\ref{a:h}(iii) holds. It remains to verify Assumption~\ref{a:Y_moment}. Since $\overline h$ is positive definite, it holds that $\|Y(\bx)\|_{\cH}^2=\overline h(\bx,\bx).$
Moreover,
\begin{align*}
\overline h(\bx,\bx)
&=h(\bx,\bx)-2\E h(\bx,\bX_1)+\E h(\bX_1,\bX_1')\\
&\leq h(\bx,\bx)+2\sqrt{h(\bx,\bx)}\big(\E h(\bX_1,\bX_1)\big)^{1/2}
+\E h(\bX_1,\bX_1'),
\end{align*}
and therefore
$$
\overline h(\bx,\bx)\leq C(1+h(\bx,\bx)).
$$
It follows that
$$
\E \|Y(\bX_1)\|_{\cH}^p
=
\E \overline h(\bX_1,\bX_1)^{p/2}
\leq
C\E\big(1+h(\bX_1,\bX_1)^{p/2}\big)
<\infty,
$$
so Assumption~\ref{a:Y_moment} holds, completing case (i).

For case (ii), suppose $h$ is a semimetric of negative type. Then, by Schoenberg's
theorem (see, e.g., \cite[Proposition 3]{sejdinovic:etal:2013}), there is a Hilbert space $\dcal K$
and an injective map $\varphi:\mathcal X\to \dcal K$ such that
\begin{equation}\label{e:schoenberg_new}
h(\bx,\by)=\|\varphi(\bx)-\varphi(\by)\|_{\dcal K}^2
=
\|\varphi(\bx)\|_{\dcal K}^2+\|\varphi(\by)\|_{\dcal K}^2
-2\langle \varphi(\bx),\varphi(\by)\rangle_{\dcal K}.
\end{equation}
Using
$
h(\bx,\by)\leq 2h(\bx,\bx_0)+2h(\by,\bx_0),$
we obtain
$$
\E h(\bX_1,\bX_1')^2
\leq
C\E h(\bX_1,\bx_0)^2
<\infty,
$$
so Assumption~\ref{a:h}(i) holds. Further,
\begin{align}
\E \|\varphi(\bX_1)\|_{\dcal K}^2
&\leq
2\|\varphi(\bx_0)\|_{\dcal K}^2
+2\E \|\varphi(\bX_1)-\varphi(\bx_0)\|_{\dcal K}^2\notag\\
&=
2\|\varphi(\bx_0)\|_{\dcal K}^2
+2\E h(\bX_1,\bx_0)
<\infty,\label{e:varphi_finite_var_new}
\end{align}
so $\E \varphi(\bX_1)$ is a well-defined element of $\dcal K$. Now
\begin{align*}
\E h(\bx,\bX_1)
&=
\|\varphi(\bx)\|_{\dcal K}^2+\E \|\varphi(\bX_1)\|_{\dcal K}^2
-2\langle \varphi(\bx),\E \varphi(\bX_1)\rangle_{\dcal K},\\
\E h(\bX_1,\bX_1')
&=
2\E \|\varphi(\bX_1)\|_{\dcal K}^2
-2\|\E \varphi(\bX_1)\|_{\dcal K}^2,
\end{align*}
from which it follows that
\begin{align}
\overline h(\bx,\by)
&=
h(\bx,\by)-\E h(\bx,\bX_1)-\E h(\bX_1,\by)+\E h(\bX_1,\bX_1')\notag\\
&=
-2\left\langle
\varphi(\bx)-\E\varphi(\bX_1),
\varphi(\by)-\E\varphi(\bX_1)
\right\rangle_{\dcal K}.\label{e:h_as_inner_new}
\end{align}
Hence $-\overline h$ is positive semidefinite. Therefore,
\begin{align*}
\tr(|A|)
&=
-\E \overline h(\bX_1,\bX_1)\\
&=
2\E \|\varphi(\bX_1)-\E\varphi(\bX_1)\|_{\dcal K}^2\\
&\leq
C\E \|\varphi(\bX_1)\|_{\dcal K}^2
<\infty
\end{align*}
by \eqref{e:varphi_finite_var_new}, showing Assumption~\ref{a:h}(ii). Finally, since $-\overline h$ is positive semidefinite, the same argument as above gives
\begin{align*}
\sum_{\ell=1}^\infty |\lambda_\ell|(\phi_\ell(\bx)-\phi_\ell(\by))^2
&=
-\overline h(\bx,\bx)-\overline h(\by,\by)+2\overline h(\bx,\by)\\
&=
2\|\varphi(\bx)-\varphi(\by)\|_{\dcal K}^2\\
&=
2h(\bx,\by)\\
&\leq C\rho(\bx,\by)^{2\alpha},
\end{align*}
so Assumption~\ref{a:h}(iii) holds. To verify Assumption~\ref{a:Y_moment}, note that
$$
\|Y(\bx)\|_{\cH}^2=-\overline h(\bx,\bx).
$$
By \eqref{e:h_as_inner_new},
$$
-\overline h(\bx,\bx)
=
2\|\varphi(\bx)-\E\varphi(\bX_1)\|_{\dcal K}^2
\leq
C\big(\|\varphi(\bx)-\varphi(\bx_0)\|_{\dcal K}^2+1\big)
=
C\big(h(\bx,\bx_0)+1\big).
$$
Consequently,
$$
\E \|Y(\bX_1)\|_{\cH}^p
\leq
C\E\big(h(\bX_1,\bx_0)^{p/2}+1\big)
<\infty,
$$
and Assumption~\ref{a:Y_moment} follows. 
\end{proof}

\begin{lemma}
\label{esempi-kernels}The conditions of Lemma \ref{l:kernel_examples}\textit{(ii)} are satisfied
by the following two kernels%
\begin{equation}
h\left( \mathbf{x,y}\right) =\left\Vert \mathbf{x-y}\right\Vert ^{2\varsigma
},\text{ \ \ }\varsigma \in \left( 0,1\right) ,  \label{energia}
\end{equation}%
assuming that $\E\left\Vert \mathbf{X}_{1}\right\Vert ^{\varsigma p}<\infty $
with $p$ defined in Assumption \ref{a:Y_moment}; and%
\begin{equation}
h\left( \mathbf{x,y}\right) =\arccos \left[ \frac{1+\left\langle \mathbf{x,y}%
\right\rangle }{\sqrt{1+\left\Vert \mathbf{x}\right\Vert ^{2}}\sqrt{%
1+\left\Vert \mathbf{y}\right\Vert ^{2}}}\right] .  \label{groth}
\end{equation}

\begin{proof}
Recall $\rho \left( \mathbf{x},\mathbf{y}\right) =\left\Vert \mathbf{x-y}%
\right\Vert $. We begin with (\ref{energia}). \citet{sejdinovic:etal:2013} show that $\left\Vert 
\mathbf{x-y}\right\Vert ^{2\varsigma }$ is a semimetric of negative type\
(see their Example 15). Also, it is immediate to see that $\E\left\vert
h\left( \mathbf{X}_{1},\mathbf{x}_{0}\right) \right\vert ^{p/2}=\E\left\Vert 
\mathbf{X}_{1}-\mathbf{x}_{0}\right\Vert ^{\varsigma p}<\infty $. Similarly,
the second equation in part \textit{(ii)} of Lemma \ref{l:kernel_examples} readily follows with $C=1$ and 
$\alpha =\varsigma $. Consider now (\ref{groth}); we begin by showing that
it is a semimetric of negative type. Let $\widetilde{\mathbf{x}}=\left( 1,%
\mathbf{x}^{\prime }\right) ^{\prime }$ and similarly define $\widetilde{%
\mathbf{y}}=\left( 1,\mathbf{y}^{\prime }\right) ^{\prime }$, and consider
the normalised vectors%
\[
\mathbf{u}_{x}=\frac{\widetilde{\mathbf{x}}}{\left\Vert \widetilde{\mathbf{x}%
}\right\Vert }\text{, \ \ }\mathbf{u}_{y}=\frac{\widetilde{\mathbf{y}}}{%
\left\Vert \widetilde{\mathbf{y}}\right\Vert };
\]%
clearly%
\[
h\left( \mathbf{x,y}\right) =\frac{\pi }{2}\left[ 1-\frac{2}{\pi }\arcsin
\left\langle \mathbf{u}_{x},\mathbf{u}_{y}\right\rangle \right] .
\]%
Let now $\mathbf{g}$ denote a $\left( d+1\right) $-variate, standard
Gaussian vector (and denote the sign function as $sgn$); using
Grothendieck's identity (see e.g. \citealp{lindenstrauss1968absolutely})%
\[
\frac{2}{\pi }\arcsin \left\langle \mathbf{u}_{x},\mathbf{u}%
_{y}\right\rangle =\E\left[ sgn\left\langle \mathbf{g},\mathbf{u}%
_{x}\right\rangle \cdot sgn\left\langle \mathbf{g},\mathbf{u}%
_{y}\right\rangle \right] ,
\]%
which is positive semidefinite by construction. Then, for any sequence $%
\left\{ a_{i},1\leq i\leq n\right\} $ such that $\sum_{i=1}^{n}a_{i}=0$%
\begin{eqnarray*}
&&\sum_{i,j=1}^{n}a_{i}a_{j}h\left( \mathbf{x}_{i}\mathbf{,x}_{j}\right)  \\
&=&-\frac{\pi }{2}\sum_{i,j=1}^{n}a_{i}a_{j}\E\left[ sgn\left\langle \mathbf{g%
},\mathbf{u}_{x_{i}}\right\rangle \cdot sgn\left\langle \mathbf{g},\mathbf{u}%
_{x_{j}}\right\rangle \right] =-\frac{\pi }{2}\E\left[ \sum_{i=1}^{n}a_{i}sgn%
\left\langle \mathbf{g},\mathbf{u}_{x_{i}}\right\rangle \right] ^{2},
\end{eqnarray*}%
which proves that $h\left( \mathbf{x,y}\right) $ is of negative type. Seeing
as $0\leq h\left( \mathbf{x,y}\right) \leq \pi $, the first condition of
part \textit{(ii)} of Lemma \ref{l:kernel_examples} holds immediately for all $p$. Finally, let $%
\psi $\ denote the angle between $\mathbf{u}_{x}$ and $\mathbf{u}_{y}$; we
have $\left\langle \mathbf{u}_{x},\mathbf{u}_{y}\right\rangle =\cos \psi $,
and therefore%
\begin{eqnarray*}
&&\left\Vert \mathbf{u}_{x}-\mathbf{u}_{y}\right\Vert ^{2} \\
&=&\left\Vert \mathbf{u}_{x}\right\Vert ^{2}+\left\Vert \mathbf{u}%
_{y}\right\Vert ^{2}-2\left\langle \mathbf{u}_{x},\mathbf{u}%
_{y}\right\rangle =2-2\cos \psi =4\sin ^{2}\frac{\psi }{2};
\end{eqnarray*}%
therefore%
\[
\rho \left( \mathbf{u}_{x},\mathbf{u}_{y}\right) =2\sin \left[ \frac{1}{2}%
\arccos \left\langle \mathbf{u}_{x},\mathbf{u}_{y}\right\rangle \right] ;
\]%
by the concavity of the $\sin \left( \cdot \right) $\ function, it can be
shown that  
\[
\psi \leq \frac{\pi }{2}\left\Vert \mathbf{u}_{x}-\mathbf{u}_{y}\right\Vert ;
\]%
further, by the Lipschitz continuity of the map $\mathbf{x\mapsto u}_{x}$,
it follows that $\left\Vert \mathbf{u}_{x}-\mathbf{u}_{y}\right\Vert \leq
\left\Vert \mathbf{x}-\mathbf{y}\right\Vert $. Now we have 
\[
h\left( \mathbf{x,y}\right) =\arccos \left\langle \mathbf{u}_{x},\mathbf{u}%
_{y}\right\rangle =\psi \leq \frac{\pi }{2}\left\Vert \mathbf{x}-\mathbf{y}%
\right\Vert ,
\]%
whence the proof of the last condition of part \textit{(ii)} of Lemma \ref{l:kernel_examples}
follows immediately with $C=\pi /2$ and $\alpha =1/2$.\footnote{%
In fact, since $h$ is bounded, the
result holds for any $0<\alpha \leq 1/2$.}
\end{proof}
\end{lemma}

  The following lemma gives a convenient sufficient condition under which Assumption \ref{a:ac} holds.
  
\begin{lemma} \label{l:ac} If for some $\sigma$-finite measure $\nu$, it holds that $F\ll \nu$ and also for each  $k\geq 1$, ${\P(\bX_k\in \cdot~|\bX_0=\bx)}\ll \nu $ for $F$-almost every $\bx\in \mathcal X$, then Assumption \ref{a:ac} holds under $H_0$.
\end{lemma}
\begin{proof}[Proof of Lemma \ref{l:ac}]  Let $N\subseteq  \mathcal X \times \mathcal X $ be any Borel set with  $(F\times F)(N)=0$.  Consider the sets $N_\bx=\{\by \in \mathcal X: (\bx,\by)\in N\}$. Fubini's theorem gives
$$
0=(F\times F)(N) = \int F(N_\bx) F(d\bx)
$$
implying $F(N_\bx)=0$ for $F$-a.e. $\bx$.  Let $f=dF/d\nu$, and $S=\{\bx:f(\bx)>0\}$. Then $F(N_{\bx})=\int_{N_\bx} f(\by)\nu(d\by)=0$ implies $\nu(S\cap N_\bx)=0$ for $F$-a.e. $\bx$.  However, since $\bX_k\sim F$
$$
0=\int_{S^c} f(\by)\nu(d\by)=F(S^c)=\P(\bX_k\in S^c)
=\int \P(\bX_k\in S^c\mid\bX_0=\bx)F(d\bx),
$$
showing $\P(\bX_k\in S^c\mid\bX_0=\bx)=0$ for $F$-a.e. $\bx$.  On the other hand, since 
$\P(\bX_k\in\cdot\mid\bX_0=\bx)\ll\nu$ for $F$-a.e.~$\bx$, using $\nu(S\cap N_\bx)=0$ we have
$
\P(\bX_k\in S\cap N_\bx\mid\bX_0=\bx)=0
$
for $F$-a.e.~$\bx$. Hence, splitting
$N_\bx=(S\cap N_\bx)\cup(S^c\cap N_\bx)$, we obtain $
\P(\bX_k\in N_\bx\mid\bX_0=\bx)=0$
for $F$-a.e.~$\bx$. Therefore,
$$
\P((\bX_0,\bX_k)\in N)
=
\int \P(\bX_k\in N_\bx\mid\bX_0=\bx)F(d\bx)
=0,
$$
which yields $\P( (\bX_k,\bX_0)\in \cdot) \ll F\times F$.
\end{proof}

\begin{lemma}\label{l:ac_examples}
Assumption~\ref{a:ac} is satisfied in each of the following settings.
\begin{enumerate}[label=(\roman*)]
\item  Assumption~\ref{a:h}(iii) holds, and $\overline h$ is continuous on
$\operatorname{supp}(F)\times\operatorname{supp}(F)$.
\item Suppose $\bX_t\in\R^d$ is a linear process of the form
$$
\bX_t=\sum_{j=0}^\infty A_j\beps_{t-j},
$$
where $A_j\in\R^{d\times d}$, $\det A_0\neq0$, and
$\{\beps_j\}$ are i.i.d.\ $\R^d$-valued random vectors such that
$\beps_1$ has a Lebesgue density.

\item Suppose $\bX_t$ is a causal and stationary solution to
$$
\bX_t
=
H(\bX_{t-1},\dots,\bX_{t-p})
+
\Sigma(\bX_{t-1},\dots,\bX_{t-p})\beps_t,
$$
where $H:(\R^d)^p\to\R^d$ and
$\Sigma:(\R^d)^p\to\R^{d\times d}$ are measurable functions satisfying
$\det\Sigma(\cdot)\neq0$, and $\{\beps_j\}$ are i.i.d.\
$\R^d$-valued random vectors such that $\beps_1$ has a Lebesgue density.
\end{enumerate}
\end{lemma}
\begin{proof} For case (i), recall $Y(\bx) =\sum_{\ell=1}^\infty\sqrt{|\lambda_\ell|}\,\phi_\ell(\bx)\phi_\ell.$ Under Assumption~\ref{a:h}(iii),
$Y:\text{supp}(F)\to\mathcal H$ is $\alpha$-H\"older continuous.  Now, let $J:\cH\to\cH$ be the operator defined by the relations
\begin{equation}\label{e:def_J}
J\phi_\ell = \text{sgn}(\lambda_\ell)\phi_\ell, \qquad \ell\geq 1.
\end{equation}

In particular, $K(\bx,\by):=\langle Y(\bx),JY(\by)\rangle_{\cH}$
is continuous on $\text{supp}(F)\times \text{supp}(F)$.  However, by orthogonality of $\phi_\ell$, 
$
\overline h(\bx,\by)=K(\bx,\by)$ 
 $F\times F\text{-a.e.}$ 
Since  $\overline h$ is continuous on $\text{supp}(F)\times \text{supp}(F)$, then
$\overline h-K$ is continuous there and vanishes $F\times F$-a.e.; and hence it vanishes identically on $\text{supp}(F)\times \text{supp}(F)$.  Thus the spectral
expansion holds everywhere on $\text{supp}(F)\times \text{supp}(F)$.  Since $(\bX_i,\bX_j)\in \text{supp}(F)\times \text{supp}(F)$ a.s., Assumption~\ref{a:ac} follows.

For cases (ii) and (iii), we proceed by verifying the hypotheses of Lemma \ref{l:ac} are satisfied.  First recall that if $\bX$ and $\bY$ are two independent variables in $\R^d$, and $\bX$ has a density, then so does $\bX+\bY$. 

Consider the setting (ii). Since $\det A_0\neq 0$, $A_0\beps_0$ is absolutely continuous. Hence, since
$$
\bX_0 = A_0\beps_0 + \sum_{j\geq 1}A_j\beps_{-j},
$$
and the two terms are independent, it follows that $F(\cdot)=\P(\bX_0\in \cdot)$ is absolutely continuous. Similarly, for $k\geq 1$, write
$$
\bX_k = \sum_{j=0}^{k-1} A_j \beps_{k-j} +
\sum_{m=0}^{\infty} A_{m+k} \beps_{-m}
=
\bU_k + \bV_k.
$$
Since $A_0\beps_k$ appears as a summand in $\bU_k$, and the $\beps_j$ are independent, it follows that $\bU_k$ is absolutely continuous. Moreover, $\bU_k$ is independent of $\sigma(\beps_0,\beps_{-1},\ldots)$, and hence independent of both $\bX_0$ and $\bV_k$. Therefore, for $F$-a.e.\ $\bx$,
$$
\P(\bX_k\in \cdot \mid \bX_0=\bx)
=
\P(\bU_k + \bV_k \in \cdot \mid \bX_0=\bx),
$$
which is absolutely continuous. This verifies Lemma \ref{l:ac} in case (ii).

Now consider the setting (iii).   Since $\det \Sigma(\bx_1,\ldots,\bx_p)\neq 0$, $\Sigma(\bx_1,\ldots,\bx_p)\beps_0$ is absolutely continuous for any choice of $\bx_1,\ldots,\bx_p$. Since
$$
\bX_0 =
H(\bX_{-1},\ldots,\bX_{-p}) +
\Sigma(\bX_{-1},\ldots,\bX_{-p})\beps_0,$$
conditioning on $(\bX_{-1},\ldots,\bX_{-p})$ shows that $F$ is absolutely continuous.  Next, let $\mathcal F_k$ be the $\sigma$-algebra generated by $\{\bX_k,\bX_{k-1},\ldots\}$. Since $\bX_k$ is causal, $\beps_k$ is independent of $\mathcal F_{k-1}$. Therefore,  $\P(\bX_k\in \cdot \mid \mathcal F_{k-1})$ is absolutely continuous with probability $1$. 
Finally, since
$$
\P(\bX_k \in \cdot \mid \bX_0)
=
\E\big[ \P(\bX_k \in \cdot \mid \mathcal F_{k-1}) \,\big|\, \bX_0 \big],
$$
it follows that $\P(\bX_k \in \cdot \mid \bX_0)$ is absolutely continuous. This verifies the hypotheses of Lemma \ref{l:ac} are satisfied in setting (iii).
\end{proof}

\section {Preliminary lemmas}\label{secone}

We first collect some facts used throughout the proofs.  For a generic function $f:\mathcal X\times  \mathcal X \to \R$, write
\begin{align*}
\mcU\big(f;(a,b]; (c,d]\big)
  &= \frac{2}{(b-a)(d-c)}
       \sum_{i=a+1}^b \sum_{j=c+1}^d f(\bX_i,\bX_j)   \\
  &\quad - \frac{1}{\displaystyle \binom{b-a}{2}}
       \sum_{a<i<j\leq b} f(\bX_i,\bX_j)
        - \frac{1}{\displaystyle\binom{d-c}{2}}
       \sum_{c<i<j\leq d} f(\bX_i,\bX_j),
\end{align*}
so that $\mcU((a,b],(c,d])=\mcU(h;(a,b],(c,d])$.
It is easily verified that whenever $f(\bx,\by)=f_0(\bx)+f_0(\by)$ for some function $f_0:\mathcal X\to \R$, it holds that $\mcU(f;(a,b]; (c,d])=0$. In particular, since $h(\bx,\by)- \overline h(\bx,\by)= \E h(\bx,\bX)+\E h(\bX,\by)- \E h(\bX,\bY)$, we have  $\mcU(h;[a,b]; [c,d])=\mcU(\overline h;[a,b]; [c,d])$, and we may assume without loss of generality that $\overline h=h$ (i.e., that $h$ is $F$-degenerate) throughout the Appendix.

Further, we note that by degeneracy of $h$, with $\phi(\bx)\equiv 1$, we have $A \phi(\bx)= \E h(\bx,\bX_1)=0$ {$F(d\bx)$-a.e.}, i.e., $(\phi,\lambda)=(1,0)$ is an eigenpair for $A$, and
by orthogonality of the eigenfunctions, we  have that
\begin{equation}\label{e:meanzero}
\lambda_\ell \E \phi_\ell(\bX_1)=0,\qquad \ell\geq 1.
\end{equation}
\begin{lemma}\label{l:prod_kern}
Let $f(\bx,\by)=\phi(\bx)\phi(\by)$ for some $\phi:\mathcal X\to\R$. Define
$$
S_{a,b}=\sum_{i=a+1}^b \phi(\bX_i), \qquad Q_{a,b}=\sum_{i=a+1}^b \phi(\bX_i)^2,
$$
and write $n_1=b-a$,  $n_2=d-c$.
Then
\begin{align}
\mathcal U(f;(a,b];(c,d])
&= \frac{2}{n_1n_2} S_{a,b} S_{c,d}
   - \frac{S_{a,b}^2 - Q_{a,b}}{n_1(n_1-1)}
   - \frac{S_{c,d}^2 - Q_{c,d}}{n_2(n_2-1)}.
\notag %
\end{align}
Equivalently, writing
$\bar\phi_{a,b}=n_1^{-1}S_{a,b}$ and $\bar\phi_{c,d}=n_2^{-1}S_{c,d}$,
\begin{align}
\mathcal U(f;(a,b];(c,d])
&= -\big(\bar\phi_{a,b}-\bar\phi_{c,d}\big)^2 \notag + \frac{Q_{a,b}}{n_1(n_1-1)} -\frac{\bar\phi_{a,b}^2}{n_1-1}\\
& \qquad + \frac{Q_{c,d}}{n_2(n_2-1)}   -\frac{\bar\phi_{c,d}^2}{n_2-1}.
\label{e:U-product-means}
\end{align}
\end{lemma}

\begin{proof}
The result follows from expanding each term in $\mathcal U(f,(a,b],(c,d])$, and the elementary relation valid for all real numbers $u_i$:
$$
2 \sum_{a\leq i<j \leq b} u_iu_j = \left(\sum_{i=a}^{b}u_i\right)^2 -  \sum_{i=a}^{b}u_i^2.
$$
\end{proof}
Throughout the appendix, for a normed space $E$, we let
$$
C([0,T],E)
$$
denote the space of continuous $E$-valued functions (e.g., $E=\mathbb R$ or $E=\cH$). 

In the spirit of \cite{kutta:dornemann:2025}, our arguments under $H_0$ rely on an invariance principle in a suitable H\"older space of $\cH$-valued functions, which we now introduce.  For $x:[0,T]\to \cH$, define
\begin{equation}\label{e:holdersemi}
\|x\|_{\infty,T} = \sup_{0\leq t\leq T}\|x(t)\|_\cH, 
\qquad  
[x]_{\gamma,T}=\sup_{0\leq s<t\leq T}\frac{\|x(t)-x(s)\|_\cH}{|t-s|^\gamma},
\end{equation}
and set
\begin{align*}
\|x\|_{\gamma,T}
:=
\|x\|_{\infty,T}
+[x]_{\gamma,T}.
\end{align*}
Then $C^\gamma_0([0,T],\cH)$ denotes the space of continuous functions $x:[0,T]\to\cH$ such that
\begin{equation*}
\|x\|_{\gamma,T}<\infty
\qquad\text{and}\qquad
\sup_{\substack{0\leq s<t\leq T\\ t-s\leq \delta}}
\frac{\|x(t)-x(s)\|_\cH}{|t-s|^\gamma}
\;\longrightarrow\; 0,
\quad\text{as }\delta\downarrow 0.
\end{equation*}
Note that $C^\gamma_0([0,T],\cH)$, equipped with the norm $\|\cdot\|_{\gamma,T}$, is a closed separable subspace of
$$
C_\gamma([0,T],\cH)
=
\{x:[0,T]\to \cH~~\text{continuous}:\|x\|_{\gamma,T}<\infty\},
$$
see, e.g., \cite{rackauskas:suquet:1998,rackauskas:suquet:2005}.

\begin{lemma}\label{l:main_invariance_lemma} Suppose $V_j$ is a stationary Bernoulli shift sequence in a separable Hilbert space $\widetilde \cH$:
$$
V_j = \widetilde f(\varepsilon_j, \varepsilon_{j-1},\ldots,),
$$
and let $V_j^{(r)}$ be the usual $r$-dependent coupled version of $V_j$.  Suppose that $\E \|V_j\|_{\widetilde \cH}^q<\infty$,  for some $q>2$,  $\E V_j=0$, and for some $0<\eta<1$ %
$$
\sum_{r=1}^\infty\left(\E\|V_1-V^{(r)}_1\|^q_{\widetilde{\cH}}\right)^{\eta /q}<\infty. %
$$
Then, 
\begin{enumerate}[label=(\roman*)]
\item The series 
$$
\Sigma_V=\sum_{r \in \mathbb Z}\Cov( V_0,V_r)
$$
converges absolutely in the trace norm, and defines a covariance operator on $\widetilde \cH$. 
\item For every $\delta>0$ with $2+\delta<q$, and all $N\geq 1$,
$$
\E\max_{1\leq k \leq N} \Big\|\sum_{j=1}^k V_j \Big\|_{\widetilde \cH}^{2+\delta} \leq C N^{(2+\delta)/2}.
$$
\item With
$$
Z_m^V(t)=\frac{1}{\sqrt m}\sum_{j=1}^{\lfloor mt \rfloor} V_j + \frac1{\sqrt m}\left(mt -  \lfloor mt\rfloor\right) V_{\lceil mt \rceil},
$$
for any $T>0$, and  $0<\gamma<\frac12-\frac1q$, for each $m$ we can define a Gaussian process $G^V_m(t)$ such that
$$
\|Z_m^V- G^V_m\|_{\gamma,T}\stackrel\P \to 0,%
$$
where $G^V_m\stackrel{d}=G^V$, where $G^V$ is an $\widetilde{\cH}$-valued Brownian motion with $\Var(G^V(1))=  \Sigma_V$.

\end{enumerate}
\end{lemma}
\begin{proof}  For (i), simply note for $r>0$,
\begin{align*}
\| \Cov ( V_0,V_r) \|_{\text{tr}} &= \| \Cov ( V_0,V_r-V_r^{(r)}) \|_{\text{tr}}\\
& \leq \E \|V_0\|_{\widetilde \cH}\| V_r-V_r^{(r)}\|_{\widetilde \cH}\\
& \leq \left(\E \|V_0\|_{\widetilde\cH}^2 \E \|V_r-V_r^{(r)}\|^2_{\widetilde \cH}\right)^{1/2}.
\end{align*}
Similarly, recalling $^*$ denotes the adjoint of an operator, for $r<0$ we have 
$$\| \Cov ( V_{r},V_0) \|_{\text{tr}}= \| \Cov ( V_0,V_{|r|})^*\|_{\text{tr}}=\| \Cov ( V_0,V_{|r|} )^*\|_{\text{tr}}.$$%
 Hence,
$$
\sum_{r\in \mathbb Z} \| \Cov ( V_0,V_r) \|_{\text{tr}}\leq  \E\|V_0\|_{\widetilde\cH}^2
       +C\sum_{r=1}^\infty \left(\E \|V_r-V_r^{(r)}\|^2_{\widetilde \cH}\right)^{1/2} <\infty.
$$
Thus, $\Sigma_V$ is  well-defined, $\|\Sigma_V\|_{\text{tr}}<\infty$, and is easily seen to be positive semidefinite i.e., $\Sigma_V$ is a well-defined covariance operator on $\widetilde \cH$.  

 For (ii), since $\widetilde{\cH}$ is a separable Hilbert space, we may find a linear isometric embedding
$U:\widetilde{\cH}\to\mathcal L^2[0,1]$. Then, if we set $\breve V_j := U V_j$ and $\breve V_j^{(r)} := U V_j^{(r)}$, $\|\breve V_j-\breve V_j^{(r)}\|_{\mathcal L^2[0,1]}=\|V_j-V_j^{(r)}\|_{\widetilde{\cH}}$ and
$\|\breve V_j\|_{\mathcal L^2[0,1]}=\| V_j\|_{\widetilde{\cH}}$.  Thus, $\breve V_j$ satisfies the hypotheses of \cite[Theorem 3.3]{berkes:horvath:rice:2013}, from which we deduce
$$
\E \Bigg\|\sum_{j=1}^N \breve V_j\Bigg\|_{\mathcal L^2[0,1]}^{2+\delta}=\E \Bigg\|\sum_{j=1}^N V_j\Bigg\|_{\widetilde{\cH}}^{2+\delta} \leq C N^{1+\delta/2}.
$$
Then, for any integers $1\leq n_1\leq n_2$,
$$
\E \Bigg\|\sum_{j=n_1}^{n_2} V_j\Bigg\|^{2+\delta} \leq C(n_2-n_1+1)^{1+\delta/2}= C\mathfrak g^{2+\delta}(n_1,n_2),
$$
with $\mathfrak g(n_1,n_2)=C(n_2-n_1+1)^{1/2}$. Using the bound $\sqrt{n_1}+\sqrt{n_2}\leq \sqrt{2(n_1+n_2)}$, the function $\mathfrak g$ satisfies, for $i\leq j \leq k$, $\mathfrak g(i,j) +\mathfrak g(j+1,k) \leq \sqrt2\mathfrak g(i,k)$.  Hence we may apply \cite[Theorem 3.1]{moricz:serfling:stout:1982} to arrive at (ii).  The statement (iii) then follows from (ii) by Corollary 3.11 of \cite{kutta:dornemann:2025}, and an application of the Skorokhod-Dudley-Wichura theorem  (e.g., \citealp[p. 48]{shorack:wellner:1986}.)

\end{proof}
\section {Proofs for Section \ref{s:asymp}}
\subsection{Proofs under the null hypothesis}\label{s:H0_proofs}

Throughout, to simplify arguments and notation, we assume
\begin{equation*}
\lambda_\ell \neq 0 ,\quad \ell \geq 1,
\end{equation*}
as the case with finitely many $\lambda_\ell\neq 0$ is entirely analgous and will follow from minor adjustments.
We also define
\begin{equation}\label{e:def_Yj}
Y_j(\cdot)=\sum_{\ell =1}^\infty\sqrt{|\lambda_\ell|} \phi_\ell(\bX_j)\phi_\ell(\cdot),\quad j \in \mathbb Z
\end{equation}
and, for each $L\geq 1$,
\begin{equation}\label{e:def_YjL}
Y_{L,j}(\cdot)=\sum_{\ell =1}^L \sqrt{|\lambda_\ell|} \phi_\ell(\bX_j)\phi_\ell(\cdot).%
\end{equation}
We recall that under Assumption \ref{a:h}(iii) %
\begin{equation}\label{a:h3}
\left(\sum_{\ell=1}^\infty|\lambda_\ell|(\phi_\ell(\by)-\phi_\ell(\by'))^2 \right)^{1/2} \leq C\rho(\by,\by')^\alpha.
\end{equation}
The next lemma shows the $\cH$-valued processes $Y_j$, $Y_{L,j}$ inherit $\mathcal L^p$ $m$--approximability from $\bX$.
\begin{lemma}\label{l:dep_transfer}
Under $H_0$, 
\begin{enumerate}[label=(\roman*)]
\item For any $p>0$, with $Y^{(r)}_j =\sum_{\ell=1}^\infty\sqrt{|\lambda_\ell|} \phi_\ell(\bX_j^{(r)})\phi_\ell(\cdot)$, we have
$$
\left(\E \|Y_1-Y^{(r)}_1\|^p_{\cH}\right)^{1/p}\leq C \vartheta_{p\alpha}(r)^\alpha,$$
and
$$
\left(\E \big|\sqrt{|\lambda_\ell|}\phi_\ell(\bX_1)-\sqrt{|\lambda_\ell|}\phi_\ell(\bX_1^{(r)})\big|^p\right)^{1/p}  \leq  C \vartheta_{p\alpha}(r)^\alpha.
$$
\item   
The long run covariances, given by 
\begin{equation}\label{e:def_Sigma}
\Sigma = \sum_{r\in \mathbb Z} \Cov(Y_0,Y_r),\qquad \Sigma_L = \sum_{r \in \mathbb Z } \Cov( Y_{L,0},Y_{L,r}),
\end{equation}
are well-defined and satisfy
$$
\|\Sigma-\Sigma_L\|_{\text{tr}} \to 0,\qquad L\to\infty.
$$
Moroever,
\begin{align}\label{e:Sigma_rep}
\langle \phi_\ell,\Sigma\phi_{\ell'}\rangle_{\cH}= \sqrt{|\lambda_\ell\lambda_{\ell'}|}\sum_{r\in \mathbb Z}\Cov(\phi_{\ell'}(\bX_0),\phi_\ell(\bX_r)).
\end{align}
\end{enumerate}
\end{lemma}
\begin{proof} For (i), note
\begin{align*}
  \left(\E \| Y_{1}-Y_{1}^{(r)}\|_{\mathcal H}^p \right)^{1/p}&=\left(\E\left(\sum_{\ell=1}^\infty |\lambda_\ell| (\phi_\ell(\bX_1)-\phi_\ell(\bX_1^{(r)}))^2 \right)^{p/2}\right)^{1/p}\\
  & \leq C \left(\E \rho(\bX_1,\bX_1^{(r)})^{p\alpha}\right)^{1/p}\\%
  &\leq C\vartheta_{p\alpha}(r)^\alpha.
\end{align*}
Since $\big|\sqrt{|\lambda_\ell|}\phi_\ell(\bX_1)-\sqrt{|\lambda_\ell|}\phi_\ell(\bX_1^{(r)})\big|=| \langle \phi_\ell, Y_1-Y_{1}^{(r)}\rangle|$, we have the second statement in (i).

For (ii), under Assumption \ref{a:Lpm}, for some $p>4$, $0<\eta<1$, we have
$$
\sum_{r=1}^\infty\left(\E\|Y_j-Y^{(r)}_j\|^p_{{\cH}}\right)^{\eta /p}~\leq ~\sum_{r=1}^\infty\left(\vartheta_{p\alpha}(r)^{\alpha}\right)^{\eta} <\infty,
$$
hence 
Lemma \ref{l:main_invariance_lemma}(i) implies $\Sigma$ and $\Sigma_L$ are well-defined.
Furthermore, note
\begin{align*}
&\| \Cov ( Y_0,Y_r)-  \Cov ( Y_{L,0},Y_{L,r}) \|_{\text{tr}}\\
& \leq \| \Cov ( Y_0-Y_{L,0},Y_r)\|_{\text{tr}}+  \| \Cov ( Y_{L,0},Y_r-Y_{L,r}) \|_{\text{tr}}\\
& = \| \Cov ( Y_0-Y_{L,0},Y_r-Y_{r}^{(r)})\|_{\text{tr}}+  \| \Cov \big( Y_{L,0},Y_r-Y_{L,r}-(Y_r^{(r)}-Y_{L,r}^{(r)})\big)\|_{\text{tr}}\\
& \leq\big(\E \|Y_0-Y_{L,0}\|_{\cH}^2\big)^{1/2} \vartheta_{2\alpha}(r)^\alpha +  \left(\E \|Y_0 \|_{\cH}^2 \E\|Y_r-Y_{L,r}-(Y_r^{(r)}-Y_{L,r}^{(r)})\|_{\cH}^2\right)^{1/2}.
\end{align*}
Moreover,%
$$
\E\| Y_r-Y_{L,r}-(Y_r^{(r)}-Y_{L,r}^{(r)})\|_{\cH}^2=   \sum_{\ell=L+1}^\infty |\lambda_\ell|  \E (\phi_\ell(\bX_1)-\phi_\ell(\bX_1^{(r)}))^2 \to 0,\qquad L\to\infty,\\
$$
and also
$$
\sum_{\ell=L+1}^\infty |\lambda_\ell|  \E (\phi_\ell(\bX_1)-\phi_\ell(\bX_1^{(r)}))^2 \leq \sum_{\ell=1}^\infty |\lambda_\ell|  \E (\phi_\ell(\bX_1)-\phi_\ell(\bX_1^{(r)}))^2 \leq C\vartheta_{2\alpha}(r)^{2\alpha}.
$$
Similarly we have $\E \|Y_0-Y_{L,0}\|_{\cH}^2\to 0$ as $L\to\infty$. Since $p>4$ and $0<\eta<1$,  for all large $r$,  $\vartheta_{2\alpha}(r)^\alpha\leq \left(\vartheta_{p\alpha}(r)^{\alpha}\right)^{\eta}$, which is summable.  Hence by dominated convergence, as $L\to\infty$,
\begin{align*}
\|\Sigma-\Sigma_L\|_{\rm tr} & \leq \sum_{r\in \mathbb Z}  \| \Cov ( Y_0,Y_r)-  \Cov ( Y_{L,0},Y_{L,r}) \|_{\text{tr}}  \to 0.
\end{align*}
Since $\langle Y_j,\phi_\ell\rangle_{\cH}=\sqrt{|\lambda_\ell|}\phi_\ell(\bX_j)$, expression \eqref{e:Sigma_rep} is immediate.
\end{proof}

For the next few lemmas, we define, for any $0\leq k_1<k_2$, 
\begin{equation}\label{e:barphi_ell}
S_{\ell,k_1,k_2}=\sum_{i=k_1+1}^{k_2} \phi_\ell(\bX_i),\quad \bar\phi_{\ell,k_1,k_2}=(k_2-k_1)^{-1}S_{\ell,k_1,k_2},
\end{equation}
so that from Lemma \ref{l:prod_kern} and Assumption \ref{a:h}, we have:
\begin{align}
\mathcal U(h;(a,b];(c,d])= \tr(A)\left(\frac{1}{b-a} + \frac{1}{d-c}\right) -\sum_{\ell=1}^\infty\lambda_\ell \big(\bar\phi_{\ell,a,b}-\bar\phi_{\ell,c,d}\big)^2 \notag\\
+ \cR(h,(a,b]) +\cR(h,(c,d]),\label{e:Udecomp}%
\end{align}
with
\begin{align}\notag
\cR(h,(k_1,k_2]) &= \frac{1}{(k_2-k_1)(k_2-k_1-1)}\sum_{i=k_1+1}^{k_2}
\sum_{\ell=1}^\infty\lambda_\ell\left(\phi_\ell(\bX_i)^2-1 \right)\\
 &\qquad- \frac{1}{k_2-k_1-1}\sum_{\ell=1}^\infty \lambda_\ell\bar\phi_{\ell,k_1,k_2}^2+ \frac{\tr(A)}{(k_2-k_1)(k_2-k_1-1)}.\label{e:def_cR}
\end{align}

Also let $J:\cH\to\cH$ be the operator \eqref{e:def_J}, i.e., 
$$
J f = \sum_{\ell=1}^\infty \text{sgn}(\lambda_\ell) \langle f,\phi_\ell\rangle_{\cH}\phi_\ell
$$
It is easily seen that $\|Jf\|_{\cH}\leq \|f\|_{\cH}$.    We set
\begin{equation}\label{e:def_QJ}
Q_J(f) = \langle f, Jf\rangle_{\cH}, \qquad f\in \cH,
\end{equation}
and note from Cauchy-Schwarz, we have $|Q_J(f)|\leq \|f\|_{\cH}\|Jf\|_{\cH}\leq \|f\|^2_{\cH}.$

The next lemma shows the remainder terms  $\cR(h,(a,b]) +\cR(h,(c,d])$ appearing in \eqref{e:Udecomp} are asymptotically negligible.
\begin{lemma}\label{lemma_remainder}For every $x>0$, as $m\to\infty$,
\begin{equation}\label{e:remainder_negligb}
\P\left\{m^{-1}\max_{k\geq 2} \max_{(\ell_1,\ell_2) \in \cS_k }\frac{(k-\ell_2)^2}{g_m(k,\ell_1,\ell_2)}|\cR(h,(0,m+\ell_1]) +\cR(h,(m+\ell_2,m+k])| >x\right\} \to 0.
\end{equation}
\end{lemma}
\begin{proof}
We deal with each of the terms
$\cR(h,(0,m+\ell_1])$ and $\cR(h,(m+\ell_2,m+k])$ separately. First we set up some notation. Throughout, write
$$
    t=\frac{k}{m},\qquad L=1+\frac{\ell_1}{m},\qquad R=\frac{k-\ell_2}{m}.
$$
Since $\ell_1\leq \ell_2$, we have $L+R\leq 1+t$.  Also, note that with $J$ and $Q_J$ as in \eqref{e:def_QJ},
$$
J Y_i(\cdot) =  \sum_{\ell=1}^\infty  \text{sgn}(\lambda_\ell)\sqrt{|\lambda_\ell|} \phi_\ell(\bX_i)\phi_\ell(\cdot),\quad Q_J(Y_i) = \sum_{\ell=1}^\infty   \lambda_\ell\phi^2_\ell(\bX_i).
$$
Set
\begin{align}\label{e:T_i}
T_i = \sum_{\ell=1}^\infty \lambda_\ell\big(\phi_\ell(\bX_i)^2-1\big)= Q_J(Y_i)-\E Q_J(Y_i).%
\end{align}
Then, with $n=k_2-k_1$,  for $n\geq2$,
\begin{equation}\label{e:Rbound_basic}
|\cR(h,(k_1,k_2])|
\leq \frac{1}{n-1}\left|\frac1n\sum_{i=k_1+1}^{k_2}T_i\right|
+ \frac{1}{n-1}\sum_{\ell=1}^\infty |\lambda_\ell|\,\bar\phi_{\ell,k_1,k_2}^2
+ \frac{|\tr(A)|}{n(n-1)}.
\end{equation}
So, consider $\cR(h,(0,m+\ell_1])$.  Put $b=m+\ell_1=mL$.
By \eqref{a:g},
\begin{equation}\label{e:g_basic_bound}
    \frac{R^2/L}{g(t,L,R)}
    \leq
    \frac{R+R^2/L}{g(t,L,R)}
    =
    \frac{R^{1-\beta}}
    {(1+t)^{1-\beta}\log^{2-\beta}(e+t)}
    \leq 1 ,
\end{equation}
i.e. $R^2/g(t,L,R)\leq L$. 
Then,
\begin{align}
&m^{-1}\max_{(\ell_1,\ell_2)\in\cC_k}
\frac{(k-\ell_2)^2}{g_m(k,\ell_1,\ell_2)}
|\cR(h,(0,b])|
\notag\\
&\qquad
=
m\max_{(\ell_1,\ell_2)\in\cC_k}
\frac{R^2}{g(t,L,R)}
|\cR(h,(0,b])|\notag \\
 & \qquad\leq  \sup_{b\geq m }b|\cR(h,(0,b])|\notag\\
 &\qquad \leq C\left(
\sup_{b\geq m}\left|\frac{1}{b}\sum_{i=1}^{b}T_i\right|
+\sup_{b\geq m}\sum_{\ell=1}^\infty |\lambda_\ell|\,\bar\phi_{\ell,0,b}^2
+\frac{|\tr(A)|}{m}
\right).\label{e:cR(h,(0,b])_bound1}
\end{align}
Using \eqref{e:Rbound_basic} with $k_1=0$, $k_2=b=m+\ell_1$, and $n=b$, and using
\eqref{e:Rbound_basic2_term2}, we obtain
\begin{align}\label{e:Rbound_basic2}
&m^{-1}\max_{k\geq2}\max_{(\ell_1,\ell_2)\in\cC_k}
\frac{(k-\ell_2)^2}{g_m(k,\ell_1,\ell_2)}
|\cR(h,(0,m+\ell_1])|
\notag\\
&\qquad\leq
C\sup_{b\geq m}
\left(
\left|\frac1b\sum_{i=1}^b T_i\right|
+ \left\|\frac1b\sum_{i=1}^bY_i\right\|_{\cH}^2 +\frac{|\tr(A)|}{b}
\right).
\end{align}

To complete the argument for $\cR(h,(0,b])$, it suffices to show each term in \eqref{e:Rbound_basic2} tends to zero in probability.   For the first term, since $\{Y_i\}$ is stationary and ergodic, so is $\{T_i\}$, and immediately from the ergodic theorem we obtain
 \begin{equation}\label{e:Ti_ergodic}
\sup_{b\geq m}\left|\frac{1}{b}\sum_{i=1}^{b}T_i\right| \stackrel{\text{a.s.}}\to 0\quad m\to \infty.
\end{equation}
For the second term in \eqref{e:Rbound_basic2}, using Parseval,
\begin{align}
\notag\sum_{\ell=1}^\infty |\lambda_\ell|\,\bar\phi_{\ell,k_1,k_2}^2 & = \frac1{(k_2-k_1)^2} \sum_{\ell=1}^\infty\left(\sum_{j=k_1+1}^{k_2} \langle Y_j,\phi_\ell\rangle_{\cH} \right)^2\\
\notag& = \frac1{(k_2-k_1)^2} \sum_{\ell=1}^\infty\ \sum_{i,j=k_1+1}^{k_2} \langle Y_i,\phi_\ell\rangle_{\cH}\langle Y_j,\phi_\ell\rangle_{\cH}\\
\notag& = \frac1{(k_2-k_1)^2}  \sum_{i,j=k_1+1}^{k_2} \langle Y_i,Y_j\rangle_{\cH}\\
& =  \Bigg\|  \frac1{(k_2-k_1)}\sum_{j=k_1+1}^{k_2} Y_j\Bigg\|_{\cH}^2.\label{e:Rbound_basic2_term2}
\end{align}
Again by the ergodic theorem we have

 \begin{equation}\label{e:Yi_ergodic}
\sup_{b\geq m}\sum_{\ell=1}^\infty |\lambda_\ell|\,\bar\phi_{\ell,0,b}^2 =\sup_{b\geq m}\Bigg\|  \frac1{b}\sum_{j=1}^b Y_j\Bigg\|_{\cH}^2 \stackrel{\text{a.s.}}\to 0, \qquad m\to \infty. %
\end{equation}
Since $|\tr(A)|<\infty$, the third term in \eqref{e:Rbound_basic2} tends to zero, giving
$$
m^{-1} \sup_{k\geq 2}\max_{(\ell_1,\ell_2)\in\cC_k}\frac{(k-\ell_2)^2}{g_m(k,\ell_1,\ell_2)}|\cR(h,(0,m+\ell_1])| =o_\P(1).
$$
We now show
\begin{equation}\label{e:2nd_remainder_neglig}
m^{-1} \sup_{k\geq 2}
\max_{(\ell_1,\ell_2)\in\cC_k}
\frac{(k-\ell_2)^2}{g_m(k,\ell_1,\ell_2)}
|\cR(h,(m+\ell_2,m+k])|
=o_\P(1),
\end{equation}
by considering $\sup_{2\leq k<m}(\ldots)$ and $\sup_{k\geq m}(\ldots)$ separately.   When  $k\geq m$,
$$
m g_m(k,\ell_1,\ell_2)
\geq
C\,k^{1-\beta}(k-\ell_2)^\beta.
$$
Hence,
\begin{align}
&m^{-1} \sup_{k\geq m}
\max_{(\ell_1,\ell_2)\in\cC_k}
\frac{(k-\ell_2)^2}{g_m(k,\ell_1,\ell_2)}
|\cR(h,(m+\ell_2,m+k])|
\notag\\
&\quad\leq
C\sup_{k\geq m}\left\{
\frac{1}{k^{1-\beta}}
\max_{0\leq \ell_2<k}
\frac{1}{(k-\ell_2)^\beta}
\left|\sum_{i=m+\ell_2+1}^{m+k}T_i\right|\right\}
\notag\\
&
\quad\qquad +\sup_{k\geq m}\left\{
\frac{1}{k^{1-\beta}}
\max_{0\leq \ell_2<k}
\frac{1}{(k-\ell_2)^{1+\beta}}
\left\|\sum_{j=m+\ell_2+1}^{m+k}Y_j\right\|_{\cH}^2
\right\}
+
\frac{C|\tr(A)|}{m^{1-\beta}}\notag\\
& \quad=:\sup_{k\geq m} M_{1,m}(k) + \sup_{k\geq m} M_{2,m}(k)  +C\frac{|\tr (A)|}{m^{1-\beta}}.\label{e:Rcd_bound}
\end{align}
We now bound the two terms in \eqref{e:Rcd_bound}. Under Assumption \ref{a:Y_moment}, we can choose a real number $q$ with $\max\{2,(1-\beta)^{-1}\}<q<p/2$; in particular 
$
\beta<1-\frac1q.$
For the first term, write
$$
B_{1,v}=\max_{2^v\leq k<2^{v+1}} M_{1,m}(k),
\qquad v\in \mathbb N.
$$
Then
\begin{equation}\label{e:M1bound}
\max_{k\geq m} M_{1,m}(k)\leq \sum_{v=\lfloor \log_2 m\rfloor}^\infty B_{1,v}.
\end{equation}
Now, for $2^v\leq k<2^{v+1}$, setting $r=k-\ell_2$, we have
$$
M_{1,m}(k)
=
\frac{1}{k^{1-\beta}}
\max_{1\leq r\leq k}
\frac{1}{r^\beta}
\left|\sum_{i=m+k-r+1}^{m+k}T_i\right|.
$$
However, 
\begin{align}\notag
\max_{1\leq r\leq k}\frac{1}{r^\beta}
\Bigg|\sum_{i=m+k-r+1}^{m+k}T_i\Bigg|
&\leq
\max_{0\leq j\leq\lfloor \log_2 k\rfloor} \max_{\substack{2^j\leq r<2^{j+1}\\r\leq k}}\frac{1}{r^\beta}
\Bigg|\sum_{i=m+k-r+1}^{m+k}T_i\Bigg|\\
&\leq 
\sum_{j=0}^{v} 2^{-j\beta}
\max_{1\leq r\leq \min\{2^{j+1},k\}}
\Bigg|\sum_{i=m+k-r+1}^{m+k}T_i\Bigg|.\label{e:bound_M1mk}
\end{align}
Hence, applying the bound \eqref{e:bound_M1mk} to $B_{1,v}$, we obtain
\begin{align*}
B_{1,v}
&\leq  \max_{2^v\leq k < 2^{v+1}} \frac{1}{(2^v)^{1-\beta}}
\max_{1\leq r\leq k}\frac{1}{r^\beta}
\Bigg|\sum_{i=m+k-r+1}^{m+k}T_i\Bigg|\\
&\leq  2^{-v(1-\beta)} \max_{2^v\leq k < 2^{v+1}} 
\max_{1\leq r\leq k}\frac{1}{r^\beta}
\Bigg|\sum_{i=m+k-r+1}^{m+k}T_i\Bigg|\\
& \leq2^{-v(1-\beta)} \max_{2^v\leq k < 2^{v+1}} \sum_{j=0}^{v} 2^{-j\beta}
\max_{1\leq r\leq \min\{2^{j+1},k\}}
\Bigg|\sum_{i=m+k-r+1}^{m+k}T_i\Bigg|.\\
&\leq C 2^{-v(1-\beta)}
\sum_{j=0}^v 2^{-j\beta} b_{1}(v,j),
\end{align*}
where
\begin{equation}\label{e:D1(v,j)}
b_{1}(v,j)=\max_{2^v\leq k<2^{v+1}} \max_{1\leq r\leq 2^{j+1}}\Bigg|\sum_{i=m+k-r+1}^{m+k}T_i\Bigg|.
\end{equation}
Using the elementary bound $\max_{1\leq i \leq N}|x_i| \leq \Big(\sum_{i=1}^N |x_i|^q\Big)^{1/q}$ 
together with Lemma \ref{l:increment_maximal}, we obtain
$$
\E b_1(v,j)
\leq \Big(\E b_1(v,j)^q\Big)^{1/q}
\leq
\Bigg(
\sum_{k=2^v}^{2^{v+1}-1}
\E \max_{1\leq r\leq 2^{j+1}}
\Bigg|\sum_{i=m+k-r+1}^{m+k}T_i\Bigg|^q
\Bigg)^{1/q}
\leq C\,2^{v/q}2^{j/2}.
$$
Therefore
$$
\E B_{1,v} \leq
C\,2^{-v(1-\beta-1/q)}
\sum_{j=0}^v 2^{j(1/2-\beta)} \leq  C 2^{-v(1-\beta-1/q)}\times \begin{cases}
2^{v(1/2-\beta)}, & \beta<1/2\\
v+1 & \beta=1/2\\ C & \beta>1/2.
\end{cases}
$$
Hence, from \eqref{e:M1bound}, we obtain, as $m\to\infty$,
$$
\E \max_{k\geq m} M_{1,m}(k)\leq \sum_{v=\lfloor \log_2 m\rfloor}^\infty \E B_{1,v}\to 0,
$$
and therefore
$$
\max_{k\geq m} M_{1,m}(k)=o_\P(1).
$$
We now turn to the second term in \eqref{e:Rcd_bound}. Set
$$
B_{2,v}=\max_{2^v\leq k<2^{v+1}} M_{2,m}(k),
\qquad v\in \mathbb N.
$$
Then
$$
\max_{k\geq m} M_{2,m}(k)\leq \sum_{v=\lfloor \log_2 m\rfloor}^\infty B_{2,v}.
$$
Again writing $r=k-\ell_2$ we have
$$
M_{2,m}(k)
=
\frac{1}{k^{1-\beta}}
\max_{1\leq r\leq k}\frac{1}{r^{1+\beta}}
\Bigg\|\sum_{j=m+k-r+1}^{m+k}Y_j\Bigg\|_{\cH}^2.
$$
Arguging analogously to the case of $B_{1,v}$, we find
$$
B_{2,v}
\leq
C\,2^{-v(1-\beta)}
\sum_{j=0}^v 2^{-j(1+\beta)} b_2(v,j),
$$
where
$$
b_2(v,j)
=
\max_{2^v\leq k<2^{v+1}}
\max_{1\leq r\leq 2^{j+1}}
\Bigg\|\sum_{i=m+k-r+1}^{m+k}Y_i\Bigg\|_{\cH}^2.
$$
Using the same argument as above, together with Lemma \ref{l:increment_maximal},
$$
\E b_2(v,j)
\leq \Big(\E b_2(v,j)^q\Big)^{1/q}
\leq
\Bigg(
\sum_{k=2^v}^{2^{v+1}-1}
\E \max_{1\leq r\leq 2^{j+1}}
\Bigg\|\sum_{i=m+k-r+1}^{m+k}Y_i\Bigg\|_{\cH}^{2q}
\Bigg)^{1/q}
\leq C\,2^{v/q}2^{j}.
$$
Hence
$$
\E B_{2,v}
\leq
C\,2^{-v(1-\beta-1/q)}
\sum_{j=0}^v 2^{-j\beta}
\leq
C\,2^{-v(1-\beta-1/q)},
$$
since $\beta>0$. Again using $\beta<1-1/q$, we obtain
$$
\sum_{v=\lfloor \log_2 m\rfloor}^\infty \E B_{2,v}\to 0,
$$
giving 
$$
\max_{k\geq m} M_{2,m}(k)=o_\P(1).
$$
Finally, since $\beta<1$, clearly
$ m^{\beta-1}\tr(A)\to 0, $
and the right-hand side of \eqref{e:Rcd_bound} tends to zero in probability, establishing the
claim for the range $k\geq m$.

Now we turn to the range $2\leq k \leq m$. Since $g_m(k,\ell_1,\ell_2)\gtrsim ((k-\ell_2)/m)^\beta$ for $2\leq k \leq m$, from the bound \eqref{e:Rbound_basic}, we have
\begin{align}
&m^{-1}  \max_{2\leq k\leq m}\max_{(\ell_1,\ell_2)\in\cC_k}\frac{(k-\ell_2)^2}{g_m(k,\ell_1,\ell_2)}|\cR(h,(m+\ell_2,m+k])\notag\\
\notag& \leq C m^{\beta-1} \max_{2\leq k \leq m} \max_{0\leq \ell_2 < k} \frac{1}{(k-\ell_2)^\beta}\left( \Bigg|\sum_{i=m+\ell_2+1}^{m+k} T_i\Bigg|+   \frac{1}{(k-\ell_2)}  \Bigg\|  \sum_{j=m+\ell_2+1}^{m+k} Y_j\Bigg\|_{\cH}^2 +|\tr (A)|\right)\\
& =:m^{\beta-1}\left(\max_{2\leq k \leq m}\left(\widetilde M_{1,m}(k) +\widetilde  M_{2,m}(k)\right)   +|\tr (A) |\right)\label{e:Rcdbound_2}
\end{align}
Write
$$
\widetilde B_{1,v}=\max_{2^v\leq k<2^{v+1}} \widetilde M_{1,m}(k),
\qquad v\in \mathbb N.
$$
From the bound \eqref{e:bound_M1mk}, when $2^v\leq k < 2^{v+1}$, 
\begin{align*}
 \widetilde M_{1,m}(k)&\leq  \sum_{j=0}^{v} 2^{-j\beta}
 \max_{1\leq r\leq 2^{j+1}}
\Bigg|\sum_{i=m+k-r+1}^{m+k}T_i\Bigg|,
\end{align*}
Hence, 
\begin{align*}
\widetilde B_{1,v}\leq  \max_{2^v\leq k<2^{v+1}}\sum_{j=0}^{v} 2^{-j\beta} 
 \max_{1\leq r\leq 2^{j+1}}
\Bigg|\sum_{i=m+k-r+1}^{m+k}T_i\Bigg| \leq \sum_{j=0}^{v}2^{-j\beta} b_1(v,j),
\end{align*}
where $b_1(v,j)$ is as in \eqref{e:D1(v,j)}.  
  Thus,
  \begin{align*}
  \E\max_{2\leq k\leq m}\widetilde M_{1,m}(k)
  &\leq C\sum_{v=0}^{\lceil\log_2m\rceil}
  2^{v/q}\sum_{j=0}^{v}2^{j(1/2-\beta)}\\
  &\leq C\begin{cases}
  m^{1/2-\beta+1/q},&\beta<1/2,\\
  m^{1/q}\log m,&\beta=1/2,\\
  m^{1/q},&\beta>1/2,
  \end{cases}\\
  &=o(m^{1-\beta}),
  \end{align*}
  since $q>2$ and $1/q<1-\beta$.

Analogous reasoning shows  

$$
m^{\beta-1}\max_{2\leq k \leq m}\widetilde M_{2,m}(k)=o_\P(1).
$$
Hence,  \eqref{e:Rcdbound_2} tends to zero in probability, and  thus \eqref{e:2nd_remainder_neglig} holds.
\end{proof}

\begin{lemma}\label{l:increment_maximal}
Pick a real number $q$ with $\max\{2,(1-\beta)^{-1}\}<q<p/2$.  Then, uniformly over integers $a\geq 0$ and $N\geq 1$,
$$
\E \max_{1\leq r\leq N}\Bigg|\sum_{i=a+1}^{a+r}T_i\Bigg|^q \leq C N^{q/2},
$$
and
$$
\E \max_{1\leq r\leq N}\Big\|\sum_{i=a+1}^{a+r}Y_i\Big\|_{\cH}^{2q} \leq C N^{q}.
$$
\end{lemma}

\begin{proof}
The second bound follows immediately from Lemma \ref{l:main_invariance_lemma}(ii), applied to
$\{Y_i\}$.
For the first bound,
With $T_i^{(r)}=Q_J (Y_{i}^{(r)}) - \E Q_J (Y_{i}),$ using that $J^*=J$, $\|J\|_{\text{op}}=1,$ note
\begin{align*}
|T_1-T_1^{(r)}|&= |\langle Y_1,JY_1\rangle_{\cH}-\langle Y_1^{(r)},JY_1^{(r)}\rangle_{\cH}|\\
&= |\langle Y_1+Y_1^{(r)}, J( Y_1-Y_1^{(r)})\rangle_{\cH}| \\
&\leq (\|Y_1\|_{\cH} + \|Y_1^{(r)}\|_{\cH})\|Y_1-Y_1^{(r)}\|_{\cH}.
\end{align*}
Hence,  %
$$
\vartheta_{p/2}^T(r)
  :=\left(\E|T_1-T_1^{(r)}|^{p/2}\right)^{2/p}
  \leq C\left(\E\|Y_1\|_{\cH}^{p}\right)^{1/p}
  \left(\E\|Y_1-Y_1^{(r)}\|_{\cH}^{p}\right)^{1/p},
  \sum_{r=1}^{\infty}(\vartheta_{p/2}^T(r))^\eta<\infty.
$$
Lemma \ref{l:dep_transfer}(i) yields
$$
\sum_{r=1}^\infty\big(\vartheta_q^T(r)\big)^\eta<\infty.
$$
Hence Lemma \ref{l:main_invariance_lemma}(ii), applied to the scalar process $\{T_i\}$, gives
$$
\E \max_{1\leq r\leq N}\Bigg|\sum_{i=1}^rT_i\Bigg|^q \leq C N^{q/2}.
$$
By stationarity, the same bound holds with $\sum_{i=a+1}^{a+r}T_i$ in place of $\sum_{i=1}^rT_i$,
uniformly in $a\geq 0$.
\end{proof}

For the statements ahead, we set
\begin{align}\label{e:U0_1}
\mathcal U_0((a,b];(c,d])&= \tr(A)\left(\frac{1}{b-a} + \frac{1}{d-c}\right) -\sum_{\ell=1}^\infty\lambda_\ell \big(\bar\phi_{\ell,a,b}-\bar\phi_{\ell,c,d}\big)^2,
\end{align}
and
$$
\mcD_{m,0}(k)=m^{-1}\max_{(\ell_1,\ell_2) \in \mathcal S_k }\frac{(k-\ell_2)^2}{g_m(k,\ell_1,\ell_2)}\big|\mcU_0\big((0,m+\ell_1],(m+\ell_2,m+k]\big)\big|.
$$
Lemma \ref{lemma_remainder}, just established, shows the limit behavior of $\cD_m(k)$ is determined by $\cD_{m,0}(k)$.  Now, we set
\begin{equation}\label{e:Dm(t)}
\Gamma_m(t) = \mcD_{m,0}(\lfloor mt \rfloor\vee 2), \qquad t\geq 0.
\end{equation}
We proceed to analyze the behavior of $\Gamma_m$ as $m\to\infty$.

In the sequel, we fix the following notation. Define the linearly interpolated partial-sum process
\begin{equation}\label{e:Zm}
Z_m(t)=\frac{1}{\sqrt m}\sum_{j=1}^{\lfloor mt \rfloor} Y_j+ \frac1{\sqrt m}\left(mt-\lfloor mt\rfloor\right)Y_{\lceil mt\rceil},
\end{equation}
which is clearly an element of $C^\gamma_0([0,T],\cH)$ for any $T>0$ and $0<\gamma <1$.  Let
\begin{equation}\label{e:G}
G=\{G(t),t\geq 0\}
\end{equation}
denote an $\cH$-valued Brownian motion with $\Var G (1) = \Sigma$, where $\Sigma$ is defined in \eqref{e:def_Sigma}.  We remark that, on account of Assumption \ref{a:Lpm}, by Lemma \ref{l:main_invariance_lemma}, after extending the probability space if necesary, there exists a sequence $G_m$ of processes with $G_m\stackrel d = G$, satisfying
\begin{equation}\label{e:G_m}
\|Z_m-G_m\|_{\gamma,T} \stackrel \P \to 0,
\end{equation}
which we call upon several times below.

For convenience we also write
\begin{align}\label{e:def_Ybar}
\overline Y_{k_1,k_2} = \frac{1}{k_2-k_1} \sum_{j=k_1+1}^{k_2}Y_j ~= \sum_{\ell =1}^\infty\sqrt{|\lambda_\ell|} \bar \phi_{\ell,k_1,k_2}\phi_\ell(\cdot),
\end{align}
so that with $Q_J$ as in \eqref{e:def_QJ}, and $\mathcal U_0$ as in \eqref{e:U0_1}, we have
\begin{align}\label{e:U0_simplification}
\mathcal U_0((a,b];(c,d])%
&= \tr(A)\left(\frac{1}{b-a} + \frac{1}{d-c}\right)  - Q_J(\overline{Y}_{a,b} - \overline{Y}_{c,d}).
\end{align}

The next lemma, below, shows that $\Gamma_m(t)$ can be expressed as a transformation of $Z_m$, up to a discretization and error term.

\begin{lemma}\label{l:Gamma_m_error}

For   $x\in C^\gamma_0([0,1+T],\cH)$, and each  $0\leq t \leq T$, set
\begin{equation}\label{e:def_Psi}
(\Psi x)(t) = \sup_{(u,v)\in \mathcal S(t)} \frac{1}{g(t,1+u,t-v)}\left| \tr (A) \left( \frac{(t-v)^2}{1+u} + (t-v)\right) - Q_J \left( \frac{t-v}{1+u} \widetilde x(u) -  (\widetilde x(t)-\widetilde x(v)\big)\right)\right|.
\end{equation}
where $\widetilde x(t)=x(t+1)$ and $Q_J$ is given in \eqref{e:def_QJ}.
Then,
\begin{align}
\Gamma_m(t) & = (\Psi Z_m)(\lfloor mt \rfloor/m) + \mathcal E_m(t), \label{e:PsiZ_m+E_m}
\end{align}
where $Z_m(t)$ is as in \eqref{e:Zm}, and for every $\beta/2<\gamma<1/2,$
\begin{align}
  \sup_{0\leq t\leq T}|\mathcal E_m(t)|
  \lesssim m^{\beta-1}|\tr A|
       +(m^{-\gamma}+m^{\beta-2\gamma})
         \|Z_m\|_{\gamma,1+T}^2.\label{e:E_m_bound}
\end{align}

\end{lemma}

\begin{proof}
We first set up some notation.  Observe 
 from \eqref{e:U0_simplification},  for any integers $0\leq a<b\leq c<d$ we can reexpress
\begin{align}
& m^{-1}(d-c) ^2\big|\mcU_0\big((a,b],(c,d]\big)\big|\notag\\
&=m^{-1} \left| \tr(A)\left(\frac{(d-c)^2}{b-a} + (d-c) \right)  - (d-c)^2Q_J(\overline{Y}_{a,b} - \overline{Y}_{c,d})\right|\notag\\
  & =\left|\tr(A)\left(\frac{(d-c)^2}{m (b-a)} + \frac{(d-c)}m\right) -Q_J\left( \frac{d-c}{b-a} \frac{1}{\sqrt m}\sum_{j=a+1}^b Y_j -   \frac{1}{\sqrt m}\sum_{j=c+1}^d Y_j \right)\right|\notag
\end{align}
Also, for each $t,u,v\geq 0$, write 
$$
t_m =\lfloor tm\rfloor/m, \quad u_m = \lfloor um\rfloor/m, \qquad v_m = \lfloor vm\rfloor/m.
$$
Hence, when $\ell_1 = \lfloor um\rfloor$ and $\ell_2 = \lfloor v m\rfloor$, 
\begin{align}
&m^{-1} (k-\ell_2)^2\big|\mcU_0\big((0,m+\ell_1],(m+\ell_2,m+k]\big)\big|\notag\\
  & = \left|\tr(A)\left(\frac{(k-\ell_2)^2}{m(m+\ell_1)} + \frac{(k-\ell_2)}m\right) -Q_J\left( \frac{k-\ell_2}{m+\ell_1} \frac{1}{\sqrt m}\sum_{j=1}^{m+\ell_1} Y_j -   \frac{1}{\sqrt m}\sum_{j=m+\ell_2+1}^{m+k} Y_j \right)\right|\notag\\
  & = \left|\tr(A)\left(\frac{(t_m-v_m)^2}{1+u_m} +  (t_m-v_m)\right) -Q_J\left( \frac{t_m-v_m}{1+u_m} \widetilde Z_m(u_m) -   (\widetilde Z_m(t_m)-\widetilde Z_m(v_m)) \right)\right|\notag,
\end{align}
where $\widetilde Z_m(t)=Z_m(1+t)$.  So, for any $u,v,t\geq 0$ set
\begin{equation}\label{e:xi_m}
\xi_m(t,u,v) =\tr (A) \left( \frac{(t-v)^2}{1+u} + (t-v)\right) - Q_J \left( \frac{t-v}{1+u} \widetilde Z_m(u) -  (\widetilde Z_m(t)-\widetilde Z_m(v)\big) \right),
\end{equation} 
and
$$
f_m(t,u,v) = \frac{1}{g(t,1+u,t-v)}|\xi_m(t,u,v)|,
$$
where we set $f_m(t,u,t)=0.$
Since
\begin{equation}\label{e:def_set}
\mathcal S_{\lfloor mt \rfloor}= \{(\lfloor mu\rfloor,\lfloor mv\rfloor): (u,v)\in \mathcal S(t_m)\} \cap \mathcal C_{\lfloor mt\rfloor},
\end{equation} 
for $\lfloor mt\rfloor\geq2$, we have the bound
  \begin{align}
  |\Gamma_m(t)-(\Psi Z_m)(t_m)|
  &=\left|
  \sup_{(\ell_1,\ell_2)\in\mathcal S_{\lfloor mt\rfloor}}
  f_m(t_m,\ell_1/m,\ell_2/m)
  -\sup_{(u,v)\in\mathcal S(t_m)}f_m(t_m,u,v)
  \right|\notag\\
  &\leq
  \sup_{\substack{(u,v)\in\mathcal S(t_m)\\
  (\lfloor mu\rfloor,\lfloor mv\rfloor)
  \in\mathcal C_{\lfloor mt\rfloor}}}
  |f_m(t_m,u_m,v_m)-f_m(t_m,u,v)|\notag\\
  &\quad+
  \sup_{\substack{(u,v)\in\mathcal S(t_m)\\
  (\lfloor mu\rfloor,\lfloor mv\rfloor)
  \notin\mathcal C_{\lfloor mt\rfloor}}}
  |f_m(t_m,u,v)|.
  \label{e:twosupbound}
\end{align}
For $0\leq t<2/m$, the same bound follows directly from
$\Gamma_m(t)=\mcD_{m,0}(2)$ and $t_m\leq1/m$. It remains to bound $\mathcal E_m(t)$.  We have %
\begin{align*}
&\big| f_m(t_m,u_m,v_m)-f_m(t_m,u,v)\big|\\
& \leq  \frac{1}{g(t_m,1+u_m,t_m-v_m)}\big|\xi_m(t_m,u_m,v_m)-\xi_m(t_m,u,v)\big|\\
& \qquad + \big|\xi_m(t_m,u,v)\big|\left| \frac{1}{g(t_m,1+u_m,t_m-v_m)} - \frac{1}{g(t_m,1+u,t_m-v)}\right|\\
& = \mathcal E_{m,1}(t,u,v) + \mathcal E_{m,2}(t,u,v).
\end{align*}
We proceed to bound $ \mathcal E_{m,1}$.  Since $|(t_m-v_m)-(t-v)| \leq m^{-1}$ and $|u-u_m|\leq C m^{-1}$, we readily find
\begin{align*}
\left|\left( \frac{(t_m-v_m)^2}{1+u_m} + (t_m-v_m)\right) -\left(\frac{(t_m-v)^2}{1+u} + (t_m-v)\right) \right| \leq C m^{-1}.
\end{align*}
Similarly, with
$$
q_m(t,u,v)= \frac{t-v}{1+u} \widetilde Z_m(u) -  (\widetilde Z_m(t)-\widetilde Z_m(v)),
$$
since $\|J\|_{\text{op}}=1$, and
\begin{equation}\label{e:Q_lipschitz}
|Q_J(f)-Q_J(g)| =  |\langle f-g, Jf\rangle_{\cH} - \langle g, J(g-f)\rangle_{\cH}| \leq  (\|f\|_{\cH}+\|g\|_{\cH})\|f-g\|_{\cH},
\end{equation}
we find
\begin{align*}
&\left| Q_J\big(q_m(t_m,u_m,v_m)\big)-Q_J\big(q_m(t_m,u,v)\big)\right|\\
 &\qquad \leq\big( \| q_m(t_m,u_m,v_m)\|_{\cH} + \|q_m(t_m,u,v)\|_{\cH}\big)\big\| q_m(t_m,u_m,v_m)-q_m(t_m,u,v)\big\|_{\cH}.
\end{align*}
Now, using the crude bound $\|\widetilde Z_m(u)-\widetilde Z_m(u_m)\| \leq m^{-\gamma}[Z_m]_{\gamma,1+T}$,
\begin{align*}
\left\|q_m(t_m,u_m,v_m)-q_m(t_m,u,v)\right\|_{\cH} &\leq \left|  \frac{t_m-v_m}{1+u_m} -  \frac{t_m-v}{1+u}\right|\| \widetilde Z_m(u)\|_{\cH}\\
& \quad +  \frac{t_m-v_m}{1+u_m}\big\| \widetilde Z_m(u_m)-\widetilde Z_m(u)\big\|_{\cH}\\
& \quad  + \big\|\widetilde Z_m(v_m)-\widetilde Z_m(v)\big\|_{\cH}\\
& \lesssim  \frac{1}{m} \| Z_m\|_{\infty,1+T}+m^{-\gamma}[Z_m]_{\gamma,1+T}.
\end{align*}
Moreover, since 
$$
\|q_m(t,u,v)\|_{\cH} \leq (t-v) \|  Z_m\|_{\infty,1+T}+  (t-v)^\gamma [Z_m]_{\gamma,1+T},
$$
we obtain
\begin{align*}
&\left| Q_J\big(q_m(t_m,u_m,v_m)\big)-Q_J\big(q_m(t_m,u,v)\big)\right|\\
& \lesssim  m^{-1}\left((t_m-v_m) \| Z_m\|^2_{\infty,1+T}  + (t_m-v_m)^\gamma [Z_m]_{\gamma,1+T}\| Z_m\|_{\infty,1+T}\right)\\
& \qquad + \frac1{m^\gamma}\left((t_m-v_m) [Z_m]_{\gamma,1+T}\| Z_m\|_{\infty,1+T}    + (t_m-v_m)^\gamma [Z_m]_{\gamma,1+T}^2\right).
\end{align*}
Then, since $t_m-v_m\geq m^{-1}$, we have $g(t_m,1+u_m,t_m-v_m)\geq  C (t_m-v_m)^{\beta}\geq m^{-\beta}.$ Also, clearly  $t_m-v_m\leq t_m \leq T$, giving  the bound
 \begin{align}\label{e:E_m1_bound}
  \sup_{0\leq t\leq T}\sup_{\substack{(u,v)\in\mathcal S(t_m)\\
    (\lfloor mu\rfloor,\lfloor mv\rfloor) \in\mathcal C_{\lfloor mt\rfloor}}} \mathcal E_{m,1}(t,u,v)
  \lesssim m^{\beta-1}|\tr A|       +(m^{-\gamma}+m^{\beta-2\gamma})\|Z_m\|_{\gamma,1+T}^2.
  \end{align}
Now we proceed to bound $\mathcal E_{m,2}(t,u,v)$ and the remaining term in \eqref{e:twosupbound} with ${(\lfloor mu\rfloor,\lfloor mv\rfloor)\notin \mathcal C_{\lfloor mt\rfloor}}$.  First, since $|Q_J(f)| \leq \| f\|_{\cH}^2$, by arguments similar to those for $\mathcal E_{m,1}(t,u,v)$, we find 
\begin{align*}
  |\xi_m(t_m,u,v)| & \leq C \left| |\tr A| (t_m-v)  + (t_m-v)^2 \| \widetilde Z_m(u)\|^2_{\cH} + (t_m-v)^{2\gamma}[Z_m]_{\gamma,1+T}^2\right| .
\end{align*}
For $(u,v)\in \mathcal S(t_m)$ with $(\lfloor mu\rfloor,\lfloor mv\rfloor)\notin \mathcal C_{\lfloor mt\rfloor}$, $0\leq t_m-v\leq  m^{-1}$,  so we obtain
  \begin{equation}\label{e:nonC_mt_terms}
  \sup_{\substack{(u,v)\in\mathcal S(t_m)\\
    (\lfloor mu\rfloor,\lfloor mv\rfloor)
       \notin\mathcal C_{\lfloor mt\rfloor}}}
  |f_m(t_m,u,v)| \lesssim m^{\beta-1}|\tr A| +m^{\beta-2}\|Z_m\|_{\infty,1+T}^2 
  +m^{\beta-2\gamma}[Z_m]_{\gamma,1+T}^2.
  \end{equation}
Turning to $\mathcal E_{m,2}(t,u,v)$, set $L_m=1+u_m\geq1$.
Since $|u-u_m|\leq m^{-1}$ and
$\partial_L\{L/(L+r)\}=r/(L+r)^2$, we have
\begin{align*}
&\left|
\frac{1}{g(t_m,1+u_m,t_m-v_m)}
-\frac{1}{g(t_m,1+u,t_m-v)}
\right|\\
&\quad\leq
C\left|
\frac{(t_m-v_m)^{-\beta}}{1+(t_m-v_m)/L_m}
-\frac{(t_m-v)^{-\beta}}{1+(t_m-v)/L_m}
\right|
+\frac{C}{m}(t_m-v)^{1-\beta}\\
&\quad\leq C
\begin{cases}
(t_m-v)^{-\beta}+m^\beta,
&0<t_m-v\leq m^{-1},\\
m^{-1}(t_m-v)^{-(1+\beta)},
&t_m-v>m^{-1}.
\end{cases}
\end{align*}
(The first case uses $t_m-v_m\geq m^{-1}$, while the second follows from the mean value theorem. The additional term arising from $u-u_m$ is absorbed into these bounds since $0<t_m-v\leq T$).
Hence, when $t_m-v\leq 1/m$, since  for any $s>\beta$, one has $(t_m-v)^s ((t_m-v)^{-\beta} + m^{\beta} )\leq 2m^{-(s-\beta)}$, we find
\begin{align}\label{e:E_m2_bound1}
|\mathcal E_{m,2}(t,u,v)|  \leq C \left| |\tr A| m^{\beta-1}+ m^{\beta-2} \| \widetilde Z_m(u)\|^2_{\cH} + m^{\beta-2\gamma}[Z_m]_{\gamma,1+T}^2\right|.
\end{align}
Similarly, when $t_m-v>1/m$,  we readily deduce
\begin{align}\label{e:E_m2_bound2}
|\mathcal E_{m,2}(t,u,v)|  \leq C \left| |\tr A| m^{\beta-1}+ m^{-1} \| \widetilde Z_m(u)\|^2_{\cH} + m^{\beta-2\gamma}[Z_m]_{\gamma,1+T}^2\right| .
\end{align}
Taking suprema in \eqref{e:E_m2_bound1} and \eqref{e:E_m2_bound2} and combining with \eqref{e:E_m1_bound} and \eqref{e:nonC_mt_terms} gives the bound \eqref{e:E_m_bound}.
\end{proof}

The next lemma shows $\Psi$ in \eqref{e:def_Psi} is continuous at each point  $x\in C^\gamma_0([0,1+T],\cH)$, in an appropriate sense.   Together with \eqref{e:G_m} and
Lemma~\ref{l:Gamma_m_error}, this will imply weak convergence
of $\Gamma_m$ to $\Psi G$ in the uniform topology on $[0,T]$, for any $T>0$.

\begin{lemma} \label{l:Psi_continuous}  Let $T>0$  and suppose $\gamma>\beta/2$.  Then,
\begin{enumerate}[label=(\roman*)]
\item If $x\in C^\gamma_0([0,1+T],\cH)$, then $\Psi x \in C([0,T],\R)$.
\item For any $x,y\in C^\gamma_0([0,1+T],\cH)$,%
\begin{equation}\label{e:Psi_lipschitz}
\sup_{0\leq t \leq T}\left |(\Psi x)(t) -(\Psi y)(t)\right| \leq C_{T,\gamma}\|x-y\|_{\gamma,1+T} \big(\|x\|_{\gamma,1+T}+\|y\|_{\gamma,1+T}\big).
\end{equation}
\end{enumerate}
\end{lemma}
\begin{proof}
For any $x\in C^\gamma_0([0,1+T],\cH)$, and  $0\leq u \leq v \leq T$ define (where again we write $\widetilde x(u)=x(1+u)$ for convenience)
$$
H_x(t,u,v)
=  \frac{1}{g(t,1+u,t-v)}\bigg|\tr(A)\Big(\frac{(t-v)^2}{1+u}+(t-v)\Big)
- Q_J\Big(\frac{t-v}{1+u}\,\widetilde x(u) - (\widetilde x(t)-\widetilde x(v))\Big)\bigg|,
$$
so that $(\Psi x)(t)=\sup_{(u,v)\in\mathcal S(t)}H_x(t,u,v)$.  We first show (i).   By Assumption~\ref{a:admissiblescheme}(ii), we have 
$$
(\Psi x)(t)=\sup_{r\in K} H_x\big(t, h_{\cS}(t,r)\big),
$$
where $K$ is compact and $h_{\cS}$ is uniformly continuous. Observe that the map $(t,u,v)\mapsto H_x(t,u,v)$ is continuous on the restricted region 
\begin{equation}\label{e:restricted_triangle}
\{(t,u,v): 0\leq u \leq v < t \leq T\}.
\end{equation}
Indeed, since $x$ is continuous and $Q_J$ is continuous,  the numerator of $\Psi$ is clearly continuous, and when $0<t-v<t$, then $g(t,1+u,t-v)>0$, giving continuity on the region \eqref{e:restricted_triangle}.  We claim continuity of $H$ extends to the closed region
\begin{equation}\label{e:full_triangle}
\{(t,u,v): 0\leq u \leq v \leq t \leq T\}.
\end{equation}
Indeed, similar to the proof of Lemma \ref{l:Gamma_m_error}, observe
\begin{equation}\label{e:bound_on_argument_of_Q}
\Bigg\|\frac{t-v}{1+u}\,\widetilde x(u) - (\widetilde x(t)-\widetilde x(v))\Bigg\|_{\cH} \leq (t-v)\|x\|_{\infty,1+T} + (t-v)^\gamma [x]_{\gamma,1+T},
\end{equation}
which, since $|Q_J(f)|=|\langle f,Jf\rangle_{\cH}|\leq\|f\|_{\cH}^2$,  gives%
$$
\Bigg|Q_J\Big(\frac{t-v}{1+u}\,\widetilde x(u) - (\widetilde x(t)-\widetilde x(v))\Big)\Bigg| \lesssim (t-v)^2\|x\|^2_{\infty,1+T} + (t-v)^{2\gamma} [x]^2_{\gamma,1+T}.
$$
Clearly, on the region \eqref{e:full_triangle},
$$
\Bigg|\frac{(t-v)^2}{1+u}+(t-v)\Bigg| \lesssim t-v.
$$
Hence, using that $g(t,1+u,t-v)\gtrsim (t-v)^\beta$, we readily find
$$
H_x(t,u,v) \lesssim  (t-v)^{1-\beta} +(t-v)^{2-\beta}\|x\|^{2}_{\infty,1+T}  +(t-v)^{2\gamma-\beta} [x]^2_{\gamma,1+T},
$$
i.e. $H_x(t,u,v)\to 0$ as $t-v\to 0$ uniformly on \eqref{e:restricted_triangle},  and hence $H$ is continuous on the compact region \eqref{e:full_triangle}.
Set
$$
R_x(t,r):=H_x\big(t,h_{\cS}(t,r)\big), \qquad (t,r)\in[0,T]\times K.
$$
Since $h_{\cS}$ is itself continuous and $[0,T]\times K$ is compact,  $R_x(\cdot,\cdot)$ is uniformly continuous on $[0,T]\times K$.
Therefore, if $t_n\to t$,
$$
|(\Psi x)(t_n)-(\Psi x)(t)| \leq \sup_{r\in K}|R_x(t_n,r)-R_x(t,r)|\to 0,
$$
giving (i).
We now show (ii). Let $x,y\in C^\gamma_0([0,1+T],\cH)$.  For any $(u,v)\in\mathcal S(t)$, set
$$
\Delta(t,u,v)
= \frac{t-v}{1+u}\,\widetilde y(u) - (\widetilde y(t)-\widetilde y(v))
- \bigg(\frac{t-v}{1+u}\,\widetilde x(u) - (\widetilde x(t)-\widetilde x(v))\bigg),
$$
where $\widetilde y(t)=y(1+t)$.  
Note
$$
\|\Delta(t,u,v)\| \leq (t-v)\|x-y\|_{\infty,1+T} + (t-v)^\gamma[x-y]_{\gamma,1+T}.
$$
From the bound \eqref{e:Q_lipschitz}, 
\begin{align*}
&\bigg|Q_J\Big(\frac{t-v}{1+u}\,\widetilde x(u) - (\widetilde x(t)-\widetilde x(v))\Big) - Q_J\Big(\frac{t-v}{1+u}\,\widetilde y(u) - (\widetilde y(t)-\widetilde y(v))\Big)\bigg|\\
&\qquad \leq \| \Delta(t,u,v)\|_{\cH} \left(\bigg\| \frac{t-v}{1+u}\,\widetilde y(u) - (\widetilde y(t)-\widetilde y(v)) \bigg\| + \bigg\| \frac{t-v}{1+u}\,\widetilde x(u) - (\widetilde x(t)-\widetilde x(v))\bigg\| \right)\\
& \qquad \leq \Big((t-v)\|x-y\|_{\infty,1+T} + (t-v)^\gamma[x-y]_{\gamma,1+T}\Big)\Big(\big((t-v)+(t-v)^\gamma\big) \big(\|x\|_{\gamma,1+T} +\|y\|_{\gamma,1+T}\big)\Big),
\end{align*}
where we used \eqref{e:bound_on_argument_of_Q} on the last line.  Hence, using that  $t-v \lesssim (t-v)^\gamma$ on the region \eqref{e:full_triangle}, we have
\begin{align*}
&|H_x(t,u,v) - H_y(t,u,v)|\\
 & \quad \leq \frac{1}{g(t,1+u,t-v)}\bigg|Q_J\Big(\frac{t-v}{1+u}\,\widetilde x(u) - (\widetilde x(t)-\widetilde x(v))\Big) - Q_J\Big(\frac{t-v}{1+u}\,\widetilde y(u) - (\widetilde y(t)-\widetilde y(v))\Big)\bigg|\\
 & \quad \lesssim \frac{(t-v)^{2\gamma}}{g(t,1+u,t-v)} \left(\|x-y\|_{\infty,1+T} +[x-y]_{\gamma,1+T}\right)\big(\|x\|_{\gamma,1+T} +\|y\|_{\gamma,1+T}\big)\\
  & \quad \lesssim  (t-v)^{2\gamma-\beta} \| x-y\|_{\gamma,1+T}\big(\|x\|_{\gamma,1+T} +\|y\|_{\gamma,1+T}\big),
\end{align*}
from which we conclude
\begin{align*}
\left |(\Psi x)(t) -(\Psi y)(t)\right|  \leq  \sup_{(u,v)\in \mathcal S(t)} |H_x(t,u,v) - H_y(t,u,v)|\lesssim\| x-y\|_{\gamma,1+T}\big(\|x\|_{\gamma,1+T} +\|y\|_{\gamma,1+T}\big).
\end{align*}
Taking the supremum over $0\leq t \leq T$, (ii) readily follows. 
\end{proof}
The next lemma establishes a uniform approximation of $\Gamma_m$ by $\Psi G_m$ on $[0,T]$.
\begin{lemma}\label{l:Dm_to_D} Let $\Gamma_m$ be as in \eqref{e:Dm(t)}, and let $G_m$ be the sequence \eqref{e:G_m}.  Then, for any $T>0$, as $m\to\infty$,
$$
\|\Gamma_{m} - \Psi { G}_{m}\|_{\infty,T}=o_\P(1).
$$
\end{lemma}
\begin{proof}
Write
\begin{align*}
&\|\Gamma_{m} - \Psi { G}_{m}\|_{\infty,T}\\
& \leq   \sup_{0\leq t \leq T}| \Gamma_{m}(t) - \Psi { Z}_{m}(\lfloor mt\rfloor/m)| +\sup_{0\leq t \leq T}|  \Psi { Z}_{m}(t) - \Psi { Z}_{m}(\lfloor mt\rfloor/m)|+\| \Psi Z_{m} -\Psi { G}_{m}\|_{\infty,T}.
\end{align*}
From Lemma \ref{l:Gamma_m_error} we have
$$
\sup_{0\leq t \leq T}| \Gamma_{m}(t) - \Psi { Z}_{m}(\lfloor mt\rfloor/m)| = \sup_{0\leq t \leq T}|\mathcal E_m(t)| =o_\P(1).
$$ Now, from Lemma \ref{l:Psi_continuous}, 
$$
\| \Psi Z_{m} -\Psi { G}_{m}\|_{\infty,T}\lesssim\|Z_m-G_m\|_{\gamma,1+T} \big(\|Z_m\|_{\gamma,1+T}+\|G_m\|_{\gamma,1+T}\big)\stackrel{\P}\to 0.
$$
Furthermore, %
\begin{align*}
&|(\Psi Z_{m})(\lfloor m t \rfloor/m) -(\Psi Z_{m})(t)|\\
&\leq \sup\{|\Psi Z_{m}(s)-\Psi Z_{m}(t)|: |s-t| \leq 1/m, s,t\in [0,T]\}\\
&\leq 2\| \Psi Z_{m} -\Psi { G}_{m}\|_{\infty,T} + \sup\{ |\Psi { G}_{m}(s)-\Psi { G}_m(t)|: |s-t| \leq 1/m, s,t\in [0,T]\}\\
& \stackrel\P \to 0,\
\end{align*}
where we used that %
\begin{align*}
&\P\left(\sup\{ |\Psi { G}_{m}(s)-\Psi { G}_m(t)|: |s-t| \leq 1/m, s,t\in [0,T]\}>\delta\right)\\
&\quad = \P\left(\sup\{ |\Psi { G}(s)-\Psi { G}(t)|: |s-t| \leq 1/m, s,t\in [0,T]\}>\delta\right)  \to 0,\quad m\to\infty,
\end{align*}
which holds because $ { G}$ is a.s.~continuous and Lemma \ref{l:Psi_continuous} gives that $ \Psi { G}(t)$ is a.s.~continuous.

\end{proof}
Next, we set
\begin{equation}\label{e:gammadef} \Gamma(t)= \Psi G(t),\quad t>0,
\end{equation}
where $G$ is as in \eqref{e:G}, and we set $\Gamma(0)=0$. 

\begin{lemma}\label{l:Gamma_tail}
For every $\delta>0$,
\begin{equation}\label{l:Gamma_tail_1}
\lim_{T\to\infty}\limsup_{m\to\infty}
\P\left\{\sup_{t\ge T}\Gamma_m(t)>\delta\right\}=0,
\end{equation}
and
\begin{equation}\label{l:Gamma_tail_2}
\lim_{T\to\infty}
\P\left\{\sup_{t\ge T}\Gamma(t)>\delta\right\}=0.
\end{equation}
\end{lemma}
\begin{proof}
First note that for any $0\leq a<b\leq c<d$, 
\begin{align}
\mcU_0((a,b];(c,d])
&= \frac{\tr(A)}{b-a}+\frac{\tr(A)}{d-c}
   - Q_J\!\big(\overline Y_{a,b}-\overline Y_{c,d}\big)\notag\\
&= \frac{\tr(A)}{b-a}+\frac{\tr(A)}{d-c}
   - Q_J(\overline Y_{a,b}) - Q_J(\overline Y_{c,d})
   +2\big\langle \overline Y_{a,b},J\overline Y_{c,d}\big\rangle_{\cH}.
\label{e:quad_expand}
\end{align}
Hence, with 
$$
a_m(\ell) =\frac{\tr(A)}{m+\ell}-Q_J\big(\overline Y_{0,m+\ell}\big),
$$
we have
\begin{align}
&\big|\mcU_0((0,m+\ell_1];(m+\ell_2,m+k])-a_m(\ell_1)\big|\notag\\
&\qquad \leq  
   \Big|\frac{\tr(A)}{k-\ell_2}- Q_J(\overline Y_{m+\ell_2,m+k}).\Big|
   +2\|\overline Y_{0,m+\ell_1}\|_{\cH}\,\|\overline Y_{m+\ell_2,m+k}\|_{\cH} \notag\\
&\qquad=: R_{1,m} (\ell_2,k)+2R_{2,m}(\ell_1,\ell_2,k).\label{e:R1R2R3}
\end{align}
In Lemma \ref{l:Rm1,Rm2,neglig}, we show
\begin{equation}\label{e:tailterm_remianders}
\limsup_{T\to\infty}\lim_{m\to\infty}\P\left( \sup_{k\geq mT }m^{-1}\max_{(\ell_1,\ell_2) \in \mathcal S_k }\frac{(k-\ell_2)^2}{g_m(k,\ell_1,\ell_2)} \left(R_{1,m} (\ell_2,k)+2R_{2,m}(\ell_1,\ell_2,k)\right) >\delta \right) =0.
\end{equation}
So we now study $a_m(\ell)$.  Let
$$
    t_m=\frac{k}{m},\qquad u_m=\frac{\ell_1}{m},\qquad v_m=\frac{\ell_2}{m}.
$$
Then
\begin{align}\label{a:lm_rescaled}
m^{-1}\frac{(k-\ell_2)^2}{g_m(k,\ell_1,\ell_2)}|a_m(\ell_1)|  =
\frac{(t_m-v_m)^2}{g(t_m,1+u_m,t_m-v_m)}
\left|
\frac{\tr(A)}{1+u_m}
-
Q_J\left(
\frac{\widetilde Z_m(u_m)}{1+u_m}
\right)
\right|.
\end{align}
Write $r_m=t_m-v_m$ and $L_m=1+u_m$. We have
\begin{align}\label{e:r2_over_g_tail}
\frac{r_m^2}{g(t_m,L_m,r_m)}
&=
\frac{r^{2-\beta}_m}
{(1+t_m)^{1-\beta}(1+r_m/L_m)\log^{2-\beta}(e+t_m)}
\notag\\
&\leq
\frac{L r^{1-\beta}}
{(1+t_m)^{1-\beta}\log^{2-\beta}(e+t_m)}
\leq
\frac{L}{\log^{2-\beta}(e+t_m)} ,
\end{align}
since $r\leq 1+t_m$.  Hence
\begin{align}\label{e:r2_over_Lg_tail}
\frac1{L_m} \frac{r_m^2}{g(t_m,L_m,r_m)}
\leq
\frac{1}{\log^{2-\beta}(e+t_m)} .
\end{align}
Thus, using $|Q_J(f)|\leq C\|f\|_{\cH}^2$, \eqref{a:lm_rescaled} gives
\begin{align}\label{e:am_tail_basic}
&m^{-1}\max_{(\ell_1,\ell_2)\in\mathcal S_k}
\frac{(k-\ell_2)^2}{g_m(k,\ell_1,\ell_2)}
|a_m(\ell_1)| \notag\\
&\qquad\leq
\frac{C|\tr(A)|}{\log^{2-\beta}(e+k/m)}
+
C\max_{0\leq \ell\leq k}
\frac{1}{\log^{2-\beta}(e+k/m)}
\frac{\|\widetilde Z_m(\ell/m)\|_{\cH}^2}{1+\ell/m}.
\end{align}
It remains to show that the second term is negligible uniformly for $k\geq mT$ as $T\to\infty$.
Recall that $\widetilde Z_m(u)=Z_m(1+u)$. Hence, with $t_m=k/m$,
\begin{align*}
\E \max_{0\leq \ell\leq k}
\frac{\|\widetilde Z_m(\ell/m)\|_{\cH}^2}{1+\ell/m}
& \leq \E  \max_{0 \leq s \leq t_m}
\frac{\| Z_m(1+s)\|_{\cH}^2}{1+s}\\
& \leq \sum_{j=0}^{\lceil \log_2 (1+t_m)\rceil }\E \max_{2^j\leq w \leq 2^{j+1}}\frac{\| Z_m(w)\|_{\cH}^2}{2^j}\\
& \leq C\log_2(1+t_m).
\end{align*}
Hence, we find
$$
 \frac{1}{\log^{2-\beta}(e+k/m)} \E \max_{1\leq \ell \leq k}
\frac{\|\widetilde Z_m(\ell/m)\|_{\cH}^2}{1+\ell/m} \leq \log^{\beta-1}(e+k/m).
$$
Splitting the range of $\ell$ into $\ell\leq mT$ and $\ell>mT$,
the preceding estimate and a dyadic argument for the range $\ell>mT$,
and using $k\geq\ell$, yields, by Markov's inequality,
\begin{equation}\label{e:Z_m_logtail_bound}
\lim_{T\to\infty}\limsup_{m\to\infty}\P\left\{ \sup_{k\geq mT} \max_{0\leq \ell\leq k}
\frac{1}{\log^{2-\beta}(e+k/m)}
\frac{\|\widetilde Z_m(\ell/m)\|_{\cH}^2}{1+\ell/m}>\delta \right\} = 0,
\end{equation}
giving
\begin{align}\label{e:am_tail_negligible}
\lim_{T\to\infty}\limsup_{m\to\infty}
\P\left\{
\sup_{k\geq mT}
m^{-1}\max_{(\ell_1,\ell_2)\in\mathcal S_k}
\frac{(k-\ell_2)^2}{g_m(k,\ell_1,\ell_2)}
|a_m(\ell_1)|
>\delta
\right\}=0 .
\end{align}
Combining \eqref{e:R1R2R3}, \eqref{e:tailterm_remianders}, and
\eqref{e:am_tail_negligible}, we obtain
$$
\lim_{T\to\infty}\limsup_{m\to\infty}
\P\left\{\sup_{t\ge T}\Gamma_m(t)>\delta\right\}=\lim_{T\to\infty}\limsup_{m\to\infty}
\P\left\{
\sup_{k\geq mT}\Gamma_m(k/m)>\delta
\right\}=0 .
$$
It remains to show \eqref{l:Gamma_tail_2}. With $\widetilde G(t)=G(1+t)$, first note
\begin{align*}
&\frac{1}{g(t,1+u,t-v)}\bigg|Q_J\Big(\frac{t-v}{1+u}\,\widetilde G(u) - (\widetilde G(t)-\widetilde G(v))\Big)\bigg|\\
&\quad=\frac{1}{g(t,L,r)}\bigg|Q_J\Big(\frac{r}{L}\,\widetilde G(u) - (\widetilde G(t)-\widetilde G(v))\Big)\bigg|\\
& \quad\lesssim\frac{r^2}{L^2g(t,L,r)} \|  G(L)\|_{\cH}^2  + \frac{1}{g(t,L,r)}\left\|\widetilde G(t)-\widetilde G(v)\right\|_{\cH}^2 \\
&  \quad \leq \frac{1}{L\log^{2-\beta}(e+t)} \| G(L)\|_{\cH}^2 + \frac{1}{(1+t)^{1-\beta}\log^{2-\beta}(e+t) }[G]^2_{\beta/2,1+T}.
\end{align*}
Arguing as in \eqref{e:Z_m_logtail_bound}, we find
$$
\lim_{T\to \infty}\P\left\{\sup_{t\geq T}\sup_{1\leq L \leq 1+t}\frac{1}{L\log^{2-\beta}(e+t)} \| G(L)\|_{\cH}^2>\delta\right\}  = 0.
$$
Now, by self-similarity of $G$,
$$
 [G]_{\beta/2,1+T}  \stackrel{d}= (1+T)^{(1-\beta)/2}[G]_{\beta/2,1}.
$$
Hence, by a dyadic argument,
\begin{align*}
&\P\left\{
\sup_{t\geq T}
\frac{[G]_{\beta/2,1+t}^2} {(1+t)^{1-\beta}\log^{2-\beta}(e+t)}>\delta \right\} \\
&\qquad\leq
\sum_{v=\lfloor\log_2 T\rfloor}^{\infty}
\P\left\{
\frac{[G]_{\beta/2,1+2^{v+1}}^2}
{2^{v(1-\beta)}(1+v)^{2-\beta}}
>C\delta
\right\}\\
&\qquad\leq
\frac{C}{\delta}
\sum_{v=\lfloor\log_2 T\rfloor}^{\infty}
\frac{1}{(1+v)^{2-\beta}}
\to0.\end{align*}
Thus, we conclude \eqref{l:Gamma_tail_2} holds.
\end{proof}

\begin{lemma}\label{l:Rm1,Rm2,neglig}
Expression \eqref{e:tailterm_remianders} holds.
\end{lemma}
\begin{proof} We treat $R_{1,m}(\ell_2,k)$ and $R_{2,m}(\ell_1,\ell_2,k)$ separately.  Write
$$
    r=k-\ell_2,\qquad L=m+\ell_1 .
$$
Then
$$
    m\,g_m(k,\ell_1,\ell_2)
    =
    (m+k)^{1-\beta}r^\beta
    \left(1+\frac{r}{L}\right)
    \log^{2-\beta}(e+k/m).
$$
First note
\begin{align}
m^{-1}\frac{r^2}{g_m(k,\ell_1,\ell_2)}
|R_{1,m}(\ell_2,k)|
&\leq
\frac{1}{m g_m(k,\ell_1,\ell_2)}
\left(
r|\tr A|
+
\Bigg\|\sum_{j=m+k-r+1}^{m+k}Y_j\Bigg\|_{\cH}^2
\right).
\label{e:R1_decomp}
\end{align}
Moreover,
$$
\frac{r}{m g_m(k,\ell_1,\ell_2)}
\leq
\frac{r^{1-\beta}}
{(m+k)^{1-\beta}\log^{2-\beta}(e+k/m)}
\leq
\log^{-(2-\beta)}(e+k/m),
$$
and hence
$$
\sup_{k\geq mT}\sup_{0\leq r\leq k}
\frac{r}{m g_m(k,\ell_1,\ell_2)}
\leq
\log^{-(2-\beta)}(e+T)\to0.
$$
For $k\geq 1$, write
 \begin{equation}\label{e:Mk_beta}
  M_\beta(k):=\max_{0\leq a<b\leq k}
  \frac{1}{(b-a)^\beta}
  \Bigg\|\sum_{j=a+1}^bY_j\Bigg\|_{\cH}^2.
  \end{equation}
Choose $\max\{2,(1-\beta)^{-1}\}<q<p/2$. For
$0\leq a<b\leq k$ with $2^v\leq b-a<2^{v+1}$, writing
$s=\lfloor a/2^v\rfloor$, we have
$$
\Bigg\|\sum_{j=a+1}^bY_j\Bigg\|_{\cH}
\leq
\Bigg\|\sum_{j=s2^v+1}^bY_j\Bigg\|_{\cH}
+
\Bigg\|\sum_{j=s2^v+1}^aY_j\Bigg\|_{\cH}
\leq
2\max_{1\leq r\leq3\cdot2^v}
\Bigg\|\sum_{j=s2^v+1}^{s2^v+r}Y_j\Bigg\|_{\cH}.
$$
Hence, by stationarity and Lemma~\ref{l:increment_maximal},
\begin{align*}
\E M_\beta(k)
&\leq C\sum_{v=0}^{\lfloor\log_2 k\rfloor}2^{-v\beta}
\left(
\sum_{s=0}^{\lfloor k/2^v\rfloor}
\E\max_{1\leq r\leq3\cdot2^v}
\Bigg\|\sum_{j=s2^v+1}^{s2^v+r}Y_j\Bigg\|_{\cH}^{2q}
\right)^{1/q}\\
&\leq C\sum_{v=0}^{\lfloor\log_2 k\rfloor}
2^{-v\beta}\left(k2^{-v}2^{vq}\right)^{1/q}\\
&=Ck^{1/q}\sum_{v=0}^{\lfloor\log_2 k\rfloor}
2^{v(1-\beta-1/q)}
\leq Ck^{1-\beta},
\end{align*}
where the first inequality also uses
$\max_i x_i\leq(\sum_i x_i^q)^{1/q}$ for $x_i\geq0$
and Jensen's inequality, and the last follows from
$1-\beta-1/q>0$. 
Therefore
\begin{align*}
&\E\sup_{k\geq mT}
\max_{1\leq r\leq k}
\frac{1}{m g_m(k,\ell_1,\ell_2)}
\Bigg\|\sum_{j=m+k-r+1}^{m+k}Y_j\Bigg\|_{\cH}^2  \\
&\qquad\leq
C\sum_{v=\lfloor\log_2(mT)\rfloor}^{\infty}
\frac{1}{2^{v(1-\beta)}\log^{2-\beta}(e+2^v/m)}
\E M_\beta(2^{v+2})  \\
&\qquad\leq
C\sum_{v=\lfloor\log_2(mT)\rfloor}^{\infty}
\frac{1}{\log^{2-\beta}(e+2^v/m)}
\lesssim
\sum_{j=\lfloor\log_2 T\rfloor}^{\infty}\frac{1}{(1+j)^{2-\beta}}
\to0,
\end{align*}
as $T\to\infty$, since $2-\beta>1$.

Putting this together with \eqref{e:R1_decomp}, we have
$$
\limsup_{T\to\infty}\lim_{m\to\infty}\P\left( \sup_{k\geq mT }m^{-1}\max_{(\ell_1,\ell_2) \in \mathcal S_k }\frac{(k-\ell_2)^2}{g_m(k,\ell_1,\ell_2)} R_{1,m} (\ell_2,k) >\delta \right) =0.
$$

For $R_{2,m}(\ell_1,\ell_2,k)$, we have
\begin{align*}
&m^{-1}\frac{r^2}{g_m(k,\ell_1,\ell_2)}
R_{2,m}(\ell_1,\ell_2,k)  \\
&\qquad\leq
\frac{r}{m g_m(k,\ell_1,\ell_2)}
\frac{1}{L}\bigg\|\sum_{j=1}^{L}Y_j\bigg\|_{\cH}
\bigg\|\sum_{j=m+k-r+1}^{m+k}Y_j\bigg\|_{\cH}.
\end{align*}
Using
$$
\frac{r}{m g_m(k,\ell_1,\ell_2)}
=
\frac{r^{1-\beta}}
{(m+k)^{1-\beta}(1+r/L)\log^{2-\beta}(e+k/m)},
$$
and recalling $x^a y^{1-a}\leq a x + (1-a)y$ for any $x,y\geq 0$ and $0\leq a \leq 1$ (e.g., \citealp[p.17]{hardy1952}), we have the bound
$$
\frac{L^{\beta/2}r^{1-\beta/2}}{L+r}\leq \frac{1}{L+r}\left(\frac{\beta}{2}L + \left(1-\frac\beta2\right)r\right) \leq  1.
$$
Hence, we obtain
\begin{align}
&m^{-1}\frac{r^2}{g_m(k,\ell_1,\ell_2)}
R_{2,m}(\ell_1,\ell_2,k)
\notag\\
&\qquad\leq
\frac{1}{(m+k)^{1-\beta}\log^{2-\beta}(e+k/m)}
\left(
\frac{1}{L^\beta}\Bigg\|\sum_{j=1}^{L}Y_j\Bigg\|_{\cH}^2
\right)^{1/2}
\left(
\frac{1}{r^\beta}\Bigg\|\sum_{j=m+k-r+1}^{m+k}Y_j\Bigg\|_{\cH}^2
\right)^{1/2}.
\label{e:Rm2_bound_new}
\end{align}
Taking the maximum over
$(\ell_1,\ell_2)\in\mathcal S_k$ in \eqref{e:Rm2_bound_new},
we obtain
$$
m^{-1}\max_{(\ell_1,\ell_2)\in\mathcal S_k}
\frac{(k-\ell_2)^2}{g_m(k,\ell_1,\ell_2)}
R_{2,m}(\ell_1,\ell_2,k)
\leq
\frac{M_\beta(m+k)}
{(m+k)^{1-\beta}\log^{2-\beta}(e+k/m)}.
$$
Now let \(2^v\leq k<2^{v+1}\), with \(v\geq \lfloor\log_2(mT)\rfloor\). Since
\(m+k\leq m+2^{v+1}\leq 2^{v+2}\) for such \(v\), and since
\(\log(e+k/m)\gtrsim \log(e+2^v/m)\), we have
\begin{align*}
&\E\max_{2^v\leq k<2^{v+1}}
m^{-1}
\max_{(\ell_1,\ell_2)\in\mathcal S_k}
\frac{(k-\ell_2)^2}{g_m(k,\ell_1,\ell_2)}
R_{2,m}(\ell_1,\ell_2,k)\\
&\qquad\leq
\frac{C}{2^{v(1-\beta)}\log^{2-\beta}(e+2^v/m)}
\E M_\beta(2^{v+2})
\end{align*}
Using $\E M_\beta(n)\leq Cn^{1-\beta}$, this gives
$$
\E\max_{2^v\leq k<2^{v+1}}
m^{-1}
\max_{(\ell_1,\ell_2)\in\mathcal S_k}
\frac{(k-\ell_2)^2}{g_m(k,\ell_1,\ell_2)}
R_{2,m}(\ell_1,\ell_2,k)
\leq
\frac{C}{\log^{2-\beta}(e+2^v/m)}.
$$
Therefore,
\begin{align*}
&\E\sup_{k\geq mT}
m^{-1}
\max_{(\ell_1,\ell_2)\in\mathcal S_k}
\frac{(k-\ell_2)^2}{g_m(k,\ell_1,\ell_2)}
R_{2,m}(\ell_1,\ell_2,k)\\
&\qquad\leq
C\sum_{v=\lfloor\log_2(mT)\rfloor}^{\infty}
\frac{1}{\log^{2-\beta}(e+2^v/m)}.
\end{align*}
Writing \(v=\lfloor\log_2 m\rfloor+j\), the last display is bounded by
$$
    C\sum_{j= \lfloor\log_2 T\rfloor}^\infty
    \frac{1}{(1+j)^{2-\beta}},
$$
which tends to zero as \(T\to\infty\), since \(2-\beta>1\). Thus,
$$
\limsup_{T\to\infty}\lim_{m\to\infty}\P\left( \sup_{k\geq mT }m^{-1}\max_{(\ell_1,\ell_2) \in \mathcal S_k }\frac{(k-\ell_2)^2}{g_m(k,\ell_1,\ell_2)} R_{2,m}(\ell_1,\ell_2,k)>\delta \right) =0.
$$
 This proves \eqref{e:tailterm_remianders}. %
\end{proof}

\begin{proof}[Proof of Theorem~\ref{thone}]

Let $G_m$ be as in \eqref{e:G_m}.  Fix $T>0$.  Then
\begin{align*}
\left|
\sup_{t\geq0}\Gamma_m(t)
-
\sup_{t\geq0}(\Psi G_m)(t)
\right|
&\leq
\sup_{0\leq t\leq T}
\left|\Gamma_m(t)-(\Psi G_m)(t)\right| \\
&\quad+
\sup_{t\geq T}\Gamma_m(t)
+
\sup_{t\geq T}(\Psi G_m)(t).
\end{align*}
By Lemma~\ref{l:Dm_to_D}, for each fixed $T<\infty$,
$$
\sup_{0\leq t\leq T}
\left|\Gamma_m(t)-(\Psi G_m)(t)\right|
=o_\P(1).
$$
Moreover, Lemma~\ref{l:Gamma_tail} gives
$$
\lim_{T\to\infty}\limsup_{m\to\infty}
\P\left\{\sup_{t\geq T}\Gamma_m(t)>\delta\right\}=0
$$
and, since \(G_m\stackrel d=G\),
$$
\lim_{T\to\infty}\limsup_{m\to\infty}
\P\left\{\sup_{t\geq T}(\Psi G_m)(t)>\delta\right\}
=
\lim_{T\to\infty}
\P\left\{\sup_{t\geq T}(\Psi G)(t)>\delta\right\}
=0.
$$
Consequently,
$$
\sup_{t\geq0}\Gamma_m(t)
-
\sup_{t\geq0}(\Psi G_m)(t)
=o_\P(1).
$$
Since \(G_m\stackrel d=G\), it follows that
$$
\sup_{t\geq0}\Gamma_m(t)
\stackrel d\to
\sup_{t\geq0}(\Psi G)(t)
=
\sup_{t\geq0}\Gamma(t).$$
This proves the claim when $M_m=\infty$.  If $M_m/m\to a\in(0,\infty)$, clearly
\begin{align*}
\left|\sup_{0<t<M_m/m}\Gamma_m(t) -\sup_{0<t<M_m/m}(\Psi G_m)(t)\right| &\leq \sup_{0<t<a+1}\left|\Gamma_m(t) -(\Psi G_m)(t)\right| = o_\P(1),
\end{align*}
which by continuity of $\Gamma(t)=\Psi G(t)$ implies $\sup_{0<t<M_m/m}\Gamma_m(t) \stackrel d \to \sup_{0<t<a}\Gamma(t)$. 

Now, since the $\dcal H$-valued Brownian motion $G$ has $\Var G(1)=\Sigma$ as in \eqref{e:def_Sigma},  by expanding $G$ in the basis $\{\phi_\ell\}$, and using the identity \eqref{e:Sigma_rep}, we have
$$
G(t) = \sum_{\ell=1}^\infty \langle G(t),\phi_\ell\rangle_{\cH}\phi_\ell\stackrel d = \sum_{\ell=1}^\infty \sqrt{|\lambda_\ell|}W_\ell(t)\phi_\ell,
$$
where $\{W_\ell(t),t\geq 0,\ell\geq 1\}$ is Gaussian with 
$$\Cov(W_\ell(t),W_{\ell'}(s))=(t\wedge s) \sum_{r\in \mathbb Z}\Cov(\phi_\ell(\bX_0),\phi_{\ell'}(\bX_r)).$$
In particular, note
$$
Q_J(G(t)) \stackrel d = \sum_{\ell=1}^\infty\lambda_\ell W_\ell(t)^2,\quad 
$$
and for $0\leq u\leq v\leq t$,  with $\widetilde G(t)=G(1+t),$
$$
 Q_J \left( \frac{t-v}{1+u} \widetilde G(u) -  (\widetilde G(t)-\widetilde G(v)\big)\right) \stackrel d= \sum_{\ell=1}^\infty \lambda_\ell \left( \frac{t-v}{1+u} \widetilde W_\ell(u) -  (\widetilde W_\ell(t)-\widetilde W_\ell(v)\big)\right)^2.
$$
Since with $g_t(u,v)=g(t,1+u,t-v)$, %
\begin{align*}
(\Psi G)(t) &= \sup_{(u,v)\in \mathcal S(t)} \frac{1}{g_t(u,v)} \left| \tr (A) \left( \frac{(t-v)^2}{1+u} + (t-v)\right) - Q_J \left( \frac{t-v}{1+u} \widetilde G(u) -  (\widetilde G(t)-\widetilde G(v)\big)\right)\right|\\
&\stackrel {\mathcal L}= \sup_{(u,v)\in \mathcal S(t)} \frac{1}{g_t(u,v)}\left|  \sum_{\ell=1}^\infty \lambda_\ell \left( \frac{(t-v)^2}{1+u} + (t-v)  -    \left( \frac{t-v}{1+u} \widetilde W_\ell(u) -  (\widetilde W_\ell(t)-\widetilde W_\ell(v)\big)\right)^2 \right)\right|,
\end{align*}
 the conclusion follows.
\end{proof}

\begin{lemma}\label{l:transformation}
For $\theta_0\in\R^\nu$, suppose
$\widehat\theta_m=\widehat\theta_m(\bX_1,\ldots,\bX_m)
\stackrel{\P}\to\theta_0$. Let
$(\theta,\bx)\mapsto t_\theta(\bx)$ be a measurable map taking values in $(\mathcal Y,d)$; and set
$h_\theta(\bx,\by)=k(t_\theta(\bx),t_\theta(\by))$, where $k(\cdot,\cdot)$ is bounded
and satisfies the conditions of Lemma \ref{l:kernel_examples}. Suppose that, for some neighborhood $N$ of
$\theta_0$, and for some $\bx_0\in \mathcal X,$
$d(t_\theta(\bx),t_\theta(\by))\leq C\rho(\bx,\by)$ and
$d(t_\theta(\bx),t_\tau(\bx))
\leq C\{1+\rho(\bx,\bx_0)\}\|\theta-\tau\|$
for all $\theta,\tau\in N$ and $\bx,\by\in\mathcal X$. If the quantities
$p,\alpha,\eta$ in Assumption~\ref{a:Lpm} satisfy
$p\alpha(1-\eta)>\nu$, then Theorem~\ref{thone} remains valid for the
detector based on $h_{\widehat\theta_m}$, with $\Gamma$ defined using
$h_{\theta_0}$.
\end{lemma}
\begin{proof}
We first consider the case when $k(\cdot,\cdot)$ is a semimetric of negative type as in Lemma \ref{l:kernel_examples}\textit{(ii)}. For $\by,\by'\in \mathcal Y$,
$$k_0(\by,\by')=-[k(\by,\by')-k(\by,0)-k(0,\by')]$$ is positive semidefinite (see, e.g.,\cite{sejdinovic:etal:2013}). Let $\mathcal K_0$ denote the RKHS associated with $k_0$, and $\varphi_0(\by)=k_0(\cdot,\by)\in \mathcal K_0$, so that $k_0(\by,\by')=\langle \varphi_0(\by),\varphi_0(\by')\rangle_{\cK_0}$. Since $k(\by,\by')\lesssim d(\by,\by')^{2\alpha}$, we have $k(\by,\by)=0$  and hence
\begin{align}\label{e:phi0-distance}
 \|\varphi_0(\by)-\varphi_0(\by')\|_{\mathcal K_0}^2 &=k_0(\by,\by)+k_0(\by',\by')-2k_0(\by,\by')
 =2k(\by,\by').
\end{align}
Since $k$ is bounded, $\mu_\theta=\E\varphi_0(t_\theta(\bX))\in \mathcal K_0$ is  well-defined. Let 
$$Y_\theta(\bx) = \varphi_0(t_\theta(\bx))-\mu_\theta,$$
 so that $\E Y_\theta(\bX_1)=\mathbf 0$.  Then, with $\overline h_\theta(\bx,\bx')=h_\theta(\bx,\bx')-\E h_\theta(\bx,\bX')-\E h_\theta(\bX,\bx)+\E h_\theta(\bX,\bX')$ for $\bX,\bX'\stackrel{iid}\sim F$, the proof of Lemma \ref{l:kernel_examples} shows that $-\overline h_\theta$ is  positive semidefinite.  We claim:
\begin{equation}\label{e:Y_sigma_commonK0}
\overline h_\theta(\bx,\bx') = -\langle Y_\theta(\bx),Y_\theta(\bx')\rangle_{\cK_0}.
\end{equation}
Indeed, writing $a_\theta(\bx)=k(t_\theta(\bx),\boldsymbol 0)$,
$$
\langle\varphi_0(t_\theta(\bx)),\varphi_0(t_\theta(\bx'))\rangle_{\mathcal K_0}= k_0(t_\theta(\bx),t_\theta(\bx'))=a_\theta(\bx)+a_\theta(\bx')-h_\theta(\bx,\bx').
$$
Since $\langle\varphi_0(t_\theta(\bx)),\mu_\theta \rangle_{\mathcal K_0} = \E k_0(t_\theta(\bx),t_\theta(\bX))$ and $\langle \mu_\theta,\mu_\theta \rangle_{\mathcal K_0}=\E k_0(t_\theta(\bX),t_\theta(\bX'))$, we find
\begin{align*}
&\langle\varphi_0(t_\theta(\bx))-\mu_\theta,\varphi_0(t_\theta(\bx'))-\mu_\theta \rangle_{\mathcal K_0}\\
&= k_0(t_\theta(\bx),t_\theta(\bx'))  - \E k_0(t_\theta(\bx),t_\theta(\bX))- \E k_0(t_\theta(\bX),t_\theta(\bx')) + \E k_0(t_\theta(\bX),t_\theta(\bX'))\\
&=- \left[k(t_\theta(\bx),t_\theta(\bx'))  - \E k(t_\theta(\bx),t_\theta(\bX))- \E k(t_\theta(\bX),t_\theta(\bx')) + \E k(t_\theta(\bX),t_\theta(\bX'))\right]\\
&= -\overline h_\theta(\bx,\bx'),
\end{align*}
giving \eqref{e:Y_sigma_commonK0}.   Moreover, boundedness of $k$ gives boundedness of $\sup_\theta |h_\theta|$ and hence of  $\sup_\theta |\overline h_\theta|$, implying
\begin{equation}\label{e:bdd_of_Y}
\sup_\bx\sup_\theta\|Y_\theta(\bx)\|^2_{\cK_0}=\sup_\bx\sup_\theta|\overline h_\theta(\bx,\bx)|<\infty.
\end{equation}
Moreover, from \eqref{e:phi0-distance},$ \|\varphi_0(\by)-\varphi_0(\by')\|_{\mathcal K_0}\lesssim d(\by,\by')^{\alpha}$.   Now, since $\widehat \theta_m \stackrel \P \to \theta_0$, with $N_\delta=\{\theta:\|\theta-\theta_0\|\leq \delta\}$, clearly $\P(\widehat \theta_m\in N_\delta )\to 1$.  Then, for any  $\theta,{\theta'}\in N_\delta$, using \eqref{e:phi0-distance}
\begin{align}\notag
\|Y_\theta(\bx)-Y_{\theta'}(\bx')\|_{\cK_0}& \leq \|\varphi_0(t_\theta(\bx))-\varphi_0(t_{\theta'}(\bx'))\|_{\cK_0}+ \|\mu_\theta-\mu_{\theta'}\|_{\cK_0}\\
\notag& \lesssim k(t_\theta(\bx),t_{\theta'}(\bx'))^{1/2} + \big(\E \| \varphi_0(t_\theta(\bX))-\varphi_0(t_{\theta'}(\bX))\|_{\cK_0}\big)\\%
\notag&\lesssim d(t_\theta(\bx),t_{\theta'}(\bx'))^{\alpha} +  \E d( t_\theta(\bX), t_{\theta'}(\bX))^\alpha \\
&\lesssim \rho(\bx,\bx')^{\alpha} + \|\theta-{\theta'}\|^{\alpha}\left( 1 + \rho(\bx_0,\bx)^\alpha +\rho(\bx_0,\bx')^\alpha +
\E\rho(\bx_0,\bX)^\alpha\right).\label{e:holder_of_Y}
\end{align}
Let $A_\theta$ denote the integral operator based on $\overline h_\theta$; since $-\overline h$ is positive semidefinite,%
$$
\tr A_\theta = \E \overline h_\theta(\bX,\bX)=-\E\|Y_\theta(\bX)\|^2_{\cK_0}.
$$
Hence, from \eqref{e:bdd_of_Y} and \eqref{e:holder_of_Y},
\begin{align*}
|\tr A_{\theta} - \tr A_{{\theta'}}| &= |\E \|Y_\theta(\bX_1)\|_{\cK_0}^2-\E \|Y_{\theta'}(\bX_1)\|_{\cK_0}^2|\\
& \leq \big|\E\big[(\|Y_\theta(\bX_1)\|_{\cK_0}-\|Y_{\theta'}(\bX_1)\|_{\cK_0})(\|Y_\theta(\bX_1)\|_{\cK_0}+\|Y_{\theta'}(\bX_1)\|_{\cK_0})\big]\big|\\
&\lesssim  \E\big|\|Y_\theta(\bX_1)\|_{\cK_0}-\|Y_{\theta'}(\bX_1)\|_{\cK_0}\big|\\
& \lesssim \|\theta-{\theta'}\|^\alpha.
\end{align*}
Hence, we have the uniform bounds
\begin{equation}\label{e:uniform_tracebounds}
\sup_{\theta,{\theta'} \in N_\delta}|\tr A_{\theta} - \tr A_{{\theta'}}|\lesssim \delta^\alpha,\qquad \sup_{\theta \in N_\delta}|\tr A_{\theta}|<\infty.
\end{equation}

Write $R_\theta$ for the remainder term in \eqref{e:Udecomp} based on $h_\theta$.  In Lemma \ref{l:uniform_theta_lemma}, we show the following extension of Lemma \ref{lemma_remainder} holds:
\begin{equation}\label{e:remainder_thetahat_bound}
m^{-1}\max_{k\geq2}
\max_{(\ell_1,\ell_2)\in\cS_k} \frac{(k-\ell_2)^2}{g_m(k,\ell_1,\ell_2)} \left|\cR_{\widehat \theta_m}((0,m+\ell_1])+\cR_{\widehat \theta_m}((m+\ell_2,m+k]) \right| =o_\P(1).
\end{equation}
 By hypothesis, with $\eta$ as in Assumption \ref{a:Lpm}, we may find a $\zeta\in(\eta,1)$, and an $s\in \big(2/(1-\beta)\vee 4, p\big)$ 
  such that
\begin{equation}\label{e:zeta,s,bounds}
 s\alpha(1-\zeta)>\nu, \qquad \frac\beta2<\gamma<\frac12-\frac1s.
\end{equation}
 Now, write $Y_{j,\theta}=Y_\theta(\bX_j),$ and let $Z_{m,\theta}$ be the partial sums process \eqref{e:Zm} with $Y_{j,\theta}$ in place of $Y_j$. In Lemma \ref{l:uniform_theta_lemma}, we have
\begin{align}\label{e:Z_m,sigma,holder}
\sup_m\E \sup_{\theta\in N_\delta}\| Z_{m,\theta}-Z_{m,\theta_0}\|^s_{\gamma,1+T} \leq C_T \delta^{\alpha(1-\zeta)s},\qquad 
\sup_m\E \sup_{\theta\in N_\delta}\| Z_{m,\theta}\|^s_{\gamma,1+T}<\infty.
\end{align}
Hence,
\begin{align*}
\P\left(\big\| Z_{m,\widehat \theta_m}-Z_{m,\theta_0}\big\|_{\gamma,1+T}> \varepsilon \right) &\leq \P(\| \widehat \theta_m-\theta_0\|>\delta) + \P\Big(\sup_{\theta \in N_\delta}\| Z_{m,\theta}-Z_{m,\theta_0}\|_{\gamma,1+T}>\varepsilon\Big)\\
&= o(1) + C_T\varepsilon^{-s}\delta^{\alpha(1-\zeta)s}.
\end{align*}
Taking $m\to\infty$ and then $\delta\to 0$, we find
\begin{equation}\label{e:Zmsigmahat_to_Zmsigma}
\big\| Z_{m,\widehat \theta_m}-Z_{m,\theta_0}\big\|_{\gamma,1+T}=o_\P(1).
\end{equation}
  Now, let $\Psi_\theta$ denote the transformation \eqref{e:def_Psi} with $\tr(A_\theta)$ in place of $\tr(A)$.  
From Lemma \ref{l:Psi_continuous}, we find, for any $x,y\in  C^\gamma_0([0,1+T],\cK_0)$,
$$
\|\Psi_\theta x-\Psi_{\theta'} y\|_{\infty,T} \lesssim  |\tr A_{\theta} - \tr A_{{\theta'}}| +  
|\|x-y\|_{\gamma,1+T} \big(\|x\|_{\gamma,1+T}+\|y\|_{\gamma,1+T}\big).
$$
where we used that
$$
\sup_{0\leq u\leq v < t \leq T}\frac{ 1}{g(t,1+u,t-v)}\left( \frac{(t-v)^2}{1+u} + (t-v)\right) \leq 1.
$$
The bounds \eqref{e:Zmsigmahat_to_Zmsigma} and \eqref{e:uniform_tracebounds}  yield
$$
\big\| \Psi_{\widehat\theta_m}Z_{m,\widehat\theta_m}- \Psi_{\theta_0}Z_{m,\theta_0}\big\|_{\infty,T}=o_\P(1).
$$
Now, let $\Gamma_{\theta,m}$ denote the process \eqref{e:Dm(t)} based on $ h_\theta$. Using the bound in the previous display, the bound \eqref{e:Z_m,sigma,holder},  together with Lemma \ref{l:Gamma_m_error}, we find
\begin{equation}\label{e:Gamma_theta_hat_T}
\| \Gamma_{\widehat \theta_m,m}-\Gamma_{\theta_0,m}\|_{\infty,T}=o_\P(1).
\end{equation}
In Lemma \ref{l:uniform_theta_lemma} we also show
\begin{equation}\label{e:Gamma_theta_hat_unif_infty}
\lim_{T\to\infty}\limsup_{m\to\infty}
\P\left\{\sup_{\theta\in N_\delta}\sup_{t\geq T}\Gamma_{m,\theta}(t)> \varepsilon \right\}=0.
\end{equation}
Since $\P(\widehat\theta_m\in N_\delta)\to1$, it follows that
\begin{equation}\label{e:Gamma_theta_hat_infty}
\lim_{T\to\infty}\limsup_{m\to\infty}
\P\left\{\sup_{t\geq T}\Gamma_{m,\widehat\theta_m}(t)>\varepsilon \right\}=0.
\end{equation}
Hence, by combining the bounds \eqref{e:remainder_thetahat_bound}, \eqref{e:Gamma_theta_hat_T}, \eqref{e:Gamma_theta_hat_infty}, the statement follows by repeating the proof of Theorem \ref{thone}, with $\Gamma_{m,\widehat\theta_m}$ in place of $\Gamma_m$.  
The case for $k(\cdot,\cdot)$ positive definite is nearly the same, taking instead $k_0=k$, and using that, in this case, $\overline h_\theta(\bx,\bx') = \langle Y_\theta(\bx),Y_\theta(\bx')\rangle_{\cK_0}$, and hence $J=I$. We omit the details.
\end{proof}

\begin{lemma}\label{l:uniform_theta_lemma}
Expressions \eqref{e:remainder_thetahat_bound}, \eqref{e:Z_m,sigma,holder}, and \eqref{e:Gamma_theta_hat_unif_infty} hold. 
\end{lemma}

\begin{proof} Write 
$$Y_{j,\theta}=Y_\theta(\bX_j), \qquad Y_{j,\theta}^{(r)}=Y_\theta(\bX^{(r)}_j),$$
$$
T_{j,\theta}= Q_J(Y_{j,\theta})-\E Q_J(Y_{j,\theta}),\qquad T^{(r)}_{j,\theta}=Q_J(Y_{j,\theta}^{(r)})-\E Q_J(Y_{j,\theta}).
$$
 For $\theta,{\theta'} \in N_\delta$, let
$$
D_{j,\theta,{\theta'}}=Y_{j,\theta}-Y_{j,{\theta'}}, \qquad D_{j,\theta,{\theta'}}^{(r)} =Y_{j,\theta}^{(r)}-Y_{j,{\theta'}}^{(r)}.
$$
Since $\E Y_{j,\theta}=0$, we have $\E D_{j,{\theta},\theta'}=0$. Using \eqref{e:phi0-distance}, we have the two bounds:
\begin{align}
\big\|D_{j,\theta,{\theta'}}
 -D_{j,\theta,{\theta'}}^{(r)}\big\|_{\mathcal K_0}
&\leq \| Y_{j,\theta}-Y_{j,{\theta'}}\|_{\cK_0} +  \| Y^{(r)}_{j,\theta}-Y^{(r)}_{j,{\theta'}}\|_{\cK_0} \notag\\
&\leq 
C_I\|\theta-{\theta'}\|^\alpha
\big(1+ \rho(\bx_0,\bX_j)^\alpha +\rho(\bx_0,\bX_{j}^{(r)})^\alpha +\E \rho(\bx_0,\bX_1)^\alpha\big),
\label{e:bandwidth-coupling-scale}\\
\big\|D_{j,{\theta},\theta'}
 -D_{j,{\theta},\theta'}^{(r)}\big\|_{\mathcal K_0}
&\leq \| Y_{j,\theta}-Y_{j,\theta}^{(r)}\|_{\cK_0} +  \| Y_{j,{\theta'}}-Y^{(r)}_{j,{\theta'}}\|_{\cK_0}\notag\\ 
&\leq C_I\rho(\bX_j,\bX_j^{(r)})^\alpha \label{e:bandwidth-coupling-data}.
\end{align}
Hence, since $\E\rho(\bx_0,\bX_1)^{p\alpha}<\infty$,
\begin{equation*}
\big(\E \big\|D_{j,\theta,{\theta'}}
 -D_{j,\theta,{\theta'}}^{(r)}\big\|_{\mathcal K_0}^p\big)^{1/p}  \lesssim \vartheta_{p\alpha}(r)^\alpha \wedge \|\theta-{\theta'}\|^\alpha
\end{equation*}
Now, with $\zeta,s$ as in \eqref{e:zeta,s,bounds}, set
$$
\widetilde D_{j,\theta,{\theta'}} = \|\theta-{\theta'}\|^{-\alpha(1-\zeta)}D_{j,\theta,{\theta'}}.
$$
Then,  using that $a\wedge b \leq a^\zeta b^{1-\zeta}$ for $a,b>0$, with $\eta'=\eta/\zeta\in(0,1)$, 
$$
\sum_{r=1}^\infty\big(\E \big\|\widetilde D_{j,\theta,{\theta'}}
 -\widetilde D_{j,\theta,\theta'}^{(r)}\big\|_{\mathcal K_0}^p\big)^{\eta'/p} \lesssim  \sum_{r=1}^\infty\left(\vartheta_{p\alpha}(r)^{\zeta \alpha}\right)^{\eta'} = \sum_{r=1}^\infty\vartheta_{p\alpha}(r)^{\eta \alpha} <\infty.
$$
Hence, Lemma \ref{l:main_invariance_lemma}(ii) applied to $\widetilde D_{j,\theta,{\theta'}}$ gives
\begin{equation}\label{D,j,theta,sigma_bound}
\E\max_{1\leq k \leq N} \Bigg\|\sum_{j=1}^k D_{j,\theta,{\theta'}}\Bigg\|_{\mathcal K_0}^{s}\leq  C N^{s/2}\|\theta-{\theta'}\|^{\alpha(1-\zeta)s},\quad \theta,{\theta'} \in N_\delta.
\end{equation}
Now, observe
$
N_\delta \subseteq \theta_0 + [-\delta,\delta]^\nu$; by shrinking $\delta$ if needed we may assume all bounds previously etablished over $\theta \in N_\delta$ also hold over $\theta_0+[-\delta,\delta]^\nu$. So, consider the finite cubic sets
$$
 \mathcal Q_r =\theta_0+\{-\delta,-\delta(1-2^{-r}),\ldots,  \delta(1-2^{-r}),\delta\}^\nu ,\qquad  \mathcal Q_0 =\{\theta_0\}.
$$
Then $\mathcal Q_{r-1}\subseteq \mathcal Q_r$, and $|\mathcal Q_r|\leq (1+ 2^{r+1})^\nu\lesssim 2^{r\nu}.$ For each $\theta \in N_\delta$, let $\pi_r(\theta)\in \mathcal Q_r$ be closest point to $\theta$ (chosen arbitrarily in the case of ties); note that
$$
\|\theta-\pi_r(\theta)\| \leq  \sqrt{(\delta 2^{- r})^2 + \ldots+(\delta 2^{- r})^2} = \sqrt{\nu}\delta 2^{- r} 
$$
Since $\pi_r(\theta)\to \theta$ and $Y_{j,\theta}$ is $\alpha$-H\"older continuous in $\theta$, by \eqref{e:holder_of_Y}, $Y_{j,\pi_R(\theta)}\to Y_{j,\theta}$ as $R\to\infty$, giving 
$$
 D_{j,\theta,\theta_0}=\lim_{R\to\infty}Y_{j,\pi_R(\theta)}-Y_{j,{\theta_0}} = \lim_{R\to\infty}\sum_{r=1}^R\big( Y_{j,\pi_r(\theta)}-Y_{j,\pi_{r-1}(\theta)}\big)=\sum_{r=1}^\infty D_{j,\pi_r(\theta),\pi_{r-1}(\theta)}
$$
We have
\begin{align*}
\bigg(\E\max_{1\leq k \leq N} \sup_{\theta \in N_\delta}\Bigg\|\sum_{j=1}^k D_{j,\theta,{\theta_0}}\Bigg\|_{\mathcal K_0}^{s} \bigg)^{1/s}&\leq \sum_{r=1}^\infty \bigg(\E\max_{1\leq k \leq N} \sup_{\theta \in N_\delta}\Bigg\|\sum_{j=1}^k D_{j,\pi_r(\theta),\pi_{r-1}(\theta)}\Bigg\|_{\mathcal K_0}^{s}\bigg)^{1/s} \\
&\leq \sum_{r=1}^\infty \bigg(\E\max_{1\leq k \leq N} \sup_{(q,q')\in \dcal N_\delta(r)}\Bigg\|\sum_{j=1}^k D_{j,q,q'}\Bigg\|_{\mathcal K_0}^{s}\bigg)^{1/s},
\end{align*}
where
$$
\dcal N_\delta(r)=\{(\pi_r(\theta),\pi_{r-1}(\theta)):\theta \in N_\delta\}.
$$
Note there are $|\mathcal Q_{r-1}|\lesssim 2^{r\nu}$ possible values of $\pi_{r-1}(\theta)$, and  upon fixing $q=\pi_{r-1}(\theta)$,  there are $3^\nu$ possible values of $\pi_r(\theta)$ (obtained by adding either $- \delta 2^{-r},0, \delta 2^{-r}$ to each coordinate of $q$).  Hence,
$$
|\dcal N_\delta(r)| \leq 3^\nu|\mathcal Q_{r-1}|  \lesssim 2^{\nu r}.
$$
Since $\|q-q'\|\leq \sqrt{\delta^2 2^{-2r}+\ldots+\delta^2 2^{-2r}}=\sqrt \nu \delta 2^{-r}$, for each $(q,q')\in \dcal N_\delta(r)$, we have,  via \eqref{D,j,theta,sigma_bound}:
\begin{align*}
\bigg(\E\max_{1\leq k \leq N} \sup_{(q,q')\in \dcal N_\delta(r)}\Bigg\|\sum_{j=1}^k D_{j,q,q'}\Bigg\|_{\mathcal K_0}^{s}\bigg)^{1/s}&\leq \bigg(\sum_{(q,q') \in \dcal N_\delta(r)}\E\max_{1\leq k \leq N}\Bigg\|\sum_{j=1}^k D_{j,q,q'}\Bigg\|_{\mathcal K_0}^{s}\bigg)^{1/s}\\
& \lesssim 2^{(\nu /s)r-r\alpha(1-\zeta)}\delta^{\alpha(1-\zeta)}N^{1/2}.
\end{align*}
Since $\nu/s-\alpha(1-\zeta)<0$ by \eqref{e:zeta,s,bounds}, we obtain, 
$$
\E\max_{1\leq k \leq N} \sup_{\theta \in N_\delta}\Big\|\sum_{j=1}^k D_{j,\theta,{\theta_0}}\Big\|_{\mathcal K_0}^{s} \lesssim \delta^{\alpha(1-\zeta)s}N^{s/2}.
$$
Using $Y_{j,\theta}= Y_{j,\theta_0} + D_{j,\theta,\theta_0}$,  and arguing similarly for $T_{j,\theta}$ (using that  $|Q_J(Y(\bx))-Q_J(Y(\bx'))|\leq2 \|Y(\bx)-Y(\bx')\|_{\cK_0}\sup_{\theta,\bx}\|Y_\theta(\bx)\|\lesssim \rho(\bx,\bx')^{\alpha})$ we obtain, for any integer $a$ (c.f. Lemma \ref{l:increment_maximal}):
\begin{equation}\label{e:global-sigma-maximal}
\E\sup_{\theta \in N_\delta}
\max_{1\leq k\leq N}
\left\|\sum_{j=a+1}^{a+k}Y_{j,\theta}\right\|^{s}
\lesssim N^{s/2},\qquad \E\sup_{\theta\in N_\delta}
\max_{1\leq k\leq N}
\left\|\sum_{j=a+1}^{a+k}T_{j,\theta}\right\|^{s}
\lesssim N^{s/2}.
\end{equation}
We claim this implies
\begin{equation}\label{e:remainder_theta_unif_bound}
m^{-1}\max_{k\geq2}
\max_{(\ell_1,\ell_2)\in\cS_k} \frac{(k-\ell_2)^2}{g_m(k,\ell_1,\ell_2)} \sup_{\theta \in N_\delta}\left|\cR_{\theta}((0,m+\ell_1])+\cR_{\theta}((m+\ell_2,m+k]) \right| =o_\P(1).
\end{equation}
 Indeed, with $q=s/2$ and
\begin{align*}
\bigg(\E  \bigg|\sup_{\theta\in N_\delta}\sup_{b\geq m}
\frac1b\Big|\sum_{j=1}^bT_{j,\theta}\Big|\bigg|^q
\bigg)^{1/q}&\leq \sum_{z=\lfloor\log_2m\rfloor}^\infty 2^{-z}
\bigg(\E \bigg|
\sup_{\theta\in N_\delta}
\max_{1\leq b\leq2^{z+1}}
\Big|\sum_{j=1}^bT_{j,\theta}\Big| \bigg|^q\bigg)^{1/q}\\
&\lesssim
\sum_{z=\lfloor\log_2m\rfloor}^\infty2^{-z/2}
\lesssim m^{-1/2},
\end{align*}
i.e. the bound \eqref{e:Ti_ergodic} in the proof of Lemma \ref{lemma_remainder} holds for $T_{j,\theta}$ uniformly in $\theta \in N_\delta$.  
The same argument shows the bound \eqref{e:Yi_ergodic} with $Y_{j,\theta}$ in place of $Y_j$ holds uniformly in $\theta\in N_\delta$; all remaining bounds in  the proof of Lemma \ref{lemma_remainder} rely only on maximal inequalities of Lemma \ref{l:increment_maximal}. Thus, by repeating the proof of Lemma \ref{lemma_remainder}, using  the uniform maximal inequalities \eqref{e:global-sigma-maximal} in place of the bounds in Lemma \ref{lemma_remainder}, we conclude \eqref{e:remainder_theta_unif_bound}. Since $\P(\widehat \theta_m \in N_\delta)\to 1$, we  obtain the first statement \eqref{e:remainder_thetahat_bound}.

We now turn to \eqref{e:Z_m,sigma,holder}. Recall $Z_{m,\theta}$ is the partial sums process \eqref{e:Zm} with $Y_{j,\theta}$ in place of $Y_j$. Let 
$$
D^{Z}_{m,\theta,\theta'}(t) = Z_{m,\theta}(t)-Z_{m,\theta'}(t).
$$
Note, for $0<u<t$
$$
 Z_{m,\theta}(t)-Z_{m,\theta}(u)
  =\frac1{\sqrt m}\left(
     \sum_{j=\lfloor mu\rfloor+1}^{\lfloor mt\rfloor}Y_{j,\theta}
     +(mt-\lfloor mt\rfloor)Y_{\lfloor mt\rfloor+1,\theta}
     -(mu-\lfloor mu\rfloor)Y_{\lfloor mu\rfloor+1,\theta}
    \right).
$$
From \eqref{D,j,theta,sigma_bound}, we have, for $u<t$, when $\lfloor mu\rfloor <\lfloor mt\rfloor$,
\begin{align}
&\bigg(\E\sup_{\theta\in N_\delta} \|D^{Z}_{m,\theta,\theta_0}(t)-D^{Z}_{m,\theta,\theta_0}(u)\|_{\cK_0}^s\bigg)^{1/s}\notag\\
&  \lesssim\bigg( \E\sup_{\theta\in N_\delta} \bigg\|\frac{1}{\sqrt m}\bigg[(\lfloor mu\rfloor +1 -mu) D_{\lfloor mu \rfloor +1 ,\theta,\theta_0}+ \sum_{j=\lfloor mu \rfloor+2}^{\lfloor mt \rfloor} D_{j,\theta,\theta_0} + (mt-\lfloor mt\rfloor) D_{\lceil mt\rceil,\theta,\theta_0}\bigg] \bigg\|_{\cK_0}^s\bigg)^{1/s}\notag\\
&  \lesssim\bigg( \E\sup_{\theta\in N_\delta} \bigg\|\frac{1}{\sqrt m}\sum_{j=\lfloor mu \rfloor+2}^{\lfloor mt \rfloor} D_{j,\theta,\theta_0} \bigg\|_{\cK_0}^s\bigg)^{1/s} + \frac{(mt-\lfloor mt\rfloor + \lfloor mu\rfloor +1 -mu)}{\sqrt m} \Big(\E\sup_{\theta\in N_\delta} \| D_{1,\theta,\theta_0}\|^s_{\cK_0}\Big)^{1/s}\notag\\
& \lesssim \delta^{\alpha(1-\zeta)} m^{-1/2}(m^{1/2}|t-u|^{1/2} + (m|t-u|\wedge 2) )\notag\\
& \lesssim \delta^{\alpha(1-\zeta)} |t-u|^{1/2}.
\label{e:D^Z_increment_bound}
\end{align}
When $u<t$ but $\lfloor mu\rfloor =\lfloor mt\rfloor$, the first sum vanishes and it follows that  \eqref{e:D^Z_increment_bound} still holds. Now we establish a H\"older bound for $t\mapsto D_{m,\theta,\theta_0}^Z(t)$.  So, let
$$
I_r =\ \{0, 2^{-r},\ldots, 1-2^{-r},1\}=:\{t_{0,r},\ldots,t_{2^r,r}\}.
$$
Write
$$
M_{m,r}=\max_{1\leq j \leq 2^r}\sup_{\theta \in N_\delta}\Big\|D^{Z}_{m,\theta,\theta_0}(t_{j,r})-D^{Z}_{m,\theta,\theta_0}(t_{j-1,r})\Big\|.
$$
The bound \eqref{e:D^Z_increment_bound} gives
$$
\big(\E M_{m,r}^s\big)^{1/s} \lesssim \Big(\sum_{j=1}^{2^r}  \delta^{\alpha(1-\zeta)s}|t_{j,r}-t_{j-1,r}|^{s/2}\Big)^{1/s} =  \delta^{\alpha(1-\zeta)} 2^{-r(1/2-1/s)}.
$$
Now, re-using notation, now write $\pi_r(t)$ for the nearest point in $I_r$ to $t$.Then, $\pi_R(t)\to t$ as $R\to\infty$, and by contruction $t\mapsto D_{m,\theta,\theta_0}(t)$ is continuous, so for any integer $r_0$,
\begin{align*}\| D^{Z}_{m,\theta,\theta_0}(t) -D^{Z}_{m,\theta,\theta_0}(\pi_{r_0}(t))\|_{\cK_0}  &=  \Big\|\lim_{R\to\infty}\sum_{r=r_0+1}^R\big(D^{Z}_{m,\theta,\theta_0}(\pi_r(t))-D^{Z}_{m,\theta,\theta_0}(\pi_{r-1}(t))\big)\Big\|_{\cK_0}\\
&  \leq  \sum_{r=r_0+1}^\infty\Big\|\big(D^{Z}_{m,\theta,\theta_0}(\pi_r(t))-D^{Z}_{m,\theta,\theta_0}(\pi_{r-1}(t))\big)\Big\|_{\cK_0}\\
& \leq \sum_{r=r_0+1}^\infty M_{m,r}.
\end{align*}
So, pick $r_0=r_0(t,u)$  such that
$$
2^{-(r_0+1)}<|t-u|\leq 2^{-r_0}.
$$
Then,
\begin{align*}
\frac{\|D^{Z}_{m,\theta,\theta_0}(t)-D^{Z}_{m,\theta,\theta_0}(u)\|_{\cK_0}}{|t-u|^\gamma} &\lesssim \frac{ \|D^{Z}_{m,\theta,\theta_0}(\pi_{r_0}(t))-D^{Z}_{m,\theta,\theta_0}(\pi_{r_0}(u))\|_{\cK_0} + 2\sum_{r=r_0+1}^\infty M_{m,r}}{|t-u|^\gamma} \\
& \lesssim   2^{r_0\gamma}\sum_{r=r_0}^\infty M_{m,r}\\
& \lesssim   \sum_{r=0}^\infty 2^{r\gamma}M_{m,r}.
\end{align*}
Hence, for the $C_0^\gamma([0,1],\cK_0)$ seminorm  \eqref{e:holdersemi} of $D^Z_{m,\theta,\theta_0}$, $\sup_{\theta \in N_\delta}\big[D^{Z}_{m,\theta,\theta_0}(u)\big]_{\gamma,1}\leq C  \sum_{r=0}^\infty 2^{r\gamma}M_{m,r} $. Hence,%
\begin{align*}
\Big(\E\sup_{\theta \in N_\delta}\left[D^{Z}_{m,\theta,\theta_0}(u)\right]_{\gamma,1}^{s}\Big)^{1/s} &\lesssim \sum_{r=0}^\infty 2^{r\gamma }\big(\E M_{m,r}^s\big)^{1/s}\\
 &\leq \delta^{\alpha(1-\zeta)}\sum_{r=0}^\infty2^{r\gamma -r(1/2-1/s)} \lesssim \delta^{\alpha(1-\zeta)},
 \end{align*}
where the last inequality is from \eqref{e:zeta,s,bounds}. The same argument shows, for any fixed $T$,
$$
\bigg(\E\sup_{\theta \in N_\delta}\left[D^{Z}_{m,\theta,\theta_0}(u)\right]_{\gamma,T}^{s}\bigg)^{1/s} \lesssim \delta^{\alpha(1-\zeta)}.
$$
Since $\|x\|_{\infty,T}\leq T^\gamma [x]_{\gamma,T}$, the same bound holds for $\|D^{Z}_{m,\theta,\theta_0}\|_{\gamma,T}$, giving the first bound in \eqref{e:Z_m,sigma,holder}.  The second bound follows from $Z_{m,\theta} = D^{Z}_{m,\theta,\theta_0} + Z_{m,\theta_0}$ and the triangle inequality.

For the last bound \eqref{e:Gamma_theta_hat_unif_infty},   we repeat the proof of \eqref{l:Gamma_tail_1}, taking 
$\sup_{\theta\in N_\delta}$ throughout, using the uniform-in-$\theta$ maximal inequalities for $Z_{m,\theta}$ and for $\sum_{j=1}^k Y_{j,\theta}$ above. 
Consequently, the arguments establishing
\eqref{e:Z_m_logtail_bound}, \eqref{e:tailterm_remianders}, and
\eqref{e:am_tail_negligible} hold uniformly over
$\theta\in N_\delta$, and combining them proves \eqref{e:Gamma_theta_hat_unif_infty}.
\end{proof}

\subsection{Proofs under the alternative}\begin{proof}[Proof of Theorem~\ref{t:delay_full}]

If $h$ is positive semidefinite, then $\overline h$ is positive semidefinite; if $h$ is a conditionally negative definite semimetric, then
$-\overline h$ is positive semidefinite (see proof of Lemma \ref{l:kernel_examples}). Thus, with $\dcal K_h$ the RKHS associated with $\overline h$ in the PSD case, (or with $-\overline h$ in the CND case),  and $\Phi(\bx)=\overline h(\bx,\cdot)\in \dcal K_h$, we have, for every $\bx,\by \in \mathcal X$,
\begin{equation}
\overline h(\bx,\by)
=\langle\Phi(\bx),J_h\Phi(\by)\rangle_{\dcal K_h},
\qquad
J_h=
\begin{cases}
I,  & h\text{ is positive semidefinite},\\
-I, & h\text{ is conditionally negative definite}.
\end{cases}
\label{e:centered-representation}
\end{equation}
Note that under Assumption \ref{a:h}(i),  $\int |\overline h(\bx,\bx)|F(d\bx)\leq \int |h(\bx,\bx)|F(d\bx) + 3\iint |h(\bx,\by)|F(d\bx)F(d\by)<\infty$. Moreover, since \eqref{e:centered-representation} gives $\| \Phi(\bx)\|_{\dcal K_h}^2 = |\overline h(\bx,\bx)|,$ it follows that $\E\| \Phi(\bX)\|_{\dcal K_h}<\infty$ and hence the expectation $\mu_F=\E \Phi(\bX)$ with $\bX\sim F$ is well-defined.  In fact, by degeneracy of $\overline h$,
$$
\|\mu_F\|_{\dcal K_h}
^2 =  \big|\langle\mu_F,J_h\mu_F\rangle_{\dcal K_h}\big|=\bigg|\iint \overline h(\bx,\by)F(d\bx)F(d\by)\bigg| = 0,
$$
and hence $\mu_F=0$.  For $w\in\dcal K_h$, write
\begin{equation}
Q_{J_h}(w)=\langle w,J_hw \rangle_{\dcal K_h}.
\label{eq:QJh}
\end{equation}
Set $W_j=\Phi(\bX_j) \in \dcal K_h$, so that $\overline h(\bX_i,\bX_j)=\langle W_i, J_hW_j\rangle_{\dcal K_h}$.
Assumption~\ref{a:Ha_dependence,moments}(iii) and
\eqref{e:centered-representation} give, for any $j\geq 1$
$$
\E\|W_{m+k_*+j}\|^2
=\E\left|\overline h(\bX_*,\bX_*)\right|\leq C, 
$$
so 
$
\Delta=\E W_{m+k_*+j}$
is well defined. Hence,
$$
\E W_j=
\begin{cases}
0, & j\leq m+k_*,\\
\Delta, & j>m+k_*.
\end{cases}
$$
Moreover, for the $F_*$-degenerate kernel $h_*$, it holds that
\begin{align}
\overline h_*(\bx,\by)
=
\overline h(\bx,\by)
-\E\overline h(\bx,\bX_*)
-\E\overline h(\bX_*,\by)
+\E\overline h(\bX_*,\bX_*'),
\label{e:post-centering-from-pre}
\end{align}
where $\bX_*,\bX_*'\stackrel{\mathrm{iid}}{\sim}F_*$; i.e., $\overline h_*$ is also the $F_*$-degenerate version of $\overline h$.  Thus, \eqref{e:centered-representation} and 
 \eqref{e:post-centering-from-pre} give
\begin{align*}
\overline h_*(\bx,\by)
&=
\langle\Phi(\bx),J_h\Phi(\by)\rangle_{\dcal K_h}
-\langle\Phi(\bx),J_h\Delta\rangle_{\dcal K_h}
-\langle\Delta,J_h\Phi(\by)\rangle_{\dcal K_h}
+\langle\Delta,J_h\Delta\rangle_{\dcal K_h}
\notag\\
&=
\big\langle
\Phi(\bx)-\Delta,
J_h(\Phi(\by)-\Delta)
\big\rangle_{\dcal K_h}.
\end{align*} Hence,  for $i,j>m+k_*$,
it holds that
\begin{equation}
\overline h_*(\bX_i,\bX_j)=\big\langle W_i-\Delta,J_h(W_j-\Delta)\big\rangle_{\dcal K_h}.
\label{e:post-centered-representation}
\end{equation}
Consequently, Assumption~\ref{a:Ha_dependence,moments}(iv) can be reexpressed as
\begin{equation}
\E\left\|\sum_{j=1}^{n}\big(W_{m+k_*+j}-\Delta\big)\right\|^{2}
\leq Cn.
\label{eq:post-Hilbert-bound}
\end{equation}
Finally, independence of $\bX_*$ and $\bX_*'$ immediately yields
\begin{equation*}
\mathfrak D_h(F,F_*)=-Q_{J_h}(\Delta).
\end{equation*}
Thus, by Assumption~\ref{a:Ha_dependence,moments}(iii),
$$
|\mathfrak D_h(F,F_*)|
=
\|\Delta\|_{\dcal K_h}^{2}
\leq
\E\|W_{m+k_*+1}\|_{\dcal K_h}^{2}
\leq C.
$$
Now, let $k_m>k_*$ and set
$$
d_m=k_m-k_*,
\qquad
b_m=m+k_*.
$$
Since (recalling $q_m=d_m/(m+k_*)$)
$$
\left(\frac{q_m}{1+q_m}\right)^{2-\beta}\leq q_m,
$$
and $|\mathfrak D_h(F,F_*)|\leq C$, the divergence condition in the hypothesis implies $d_m\to\infty$. 
Define
$$
\widetilde W_m^{(L)} =\frac{1}{b_m}\sum_{j=1}^{b_m}W_j,
\qquad
\widetilde W_m^{(R)}= \frac{1}{d_m}\sum_{j=1}^{d_m}
\big(W_{b_m+j}-\Delta\big).
$$
Using the identity
$$
2\sum_{i<j}\langle W_i,J_hW_j\rangle_{\dcal K_h}
=
Q_{J_h}\left(\sum_iW_i\right)
-
\sum_iQ_{J_h}(W_i)
$$
and the identity \eqref{e:centered-representation} gives
\begin{align}
&\mathcal U\big(h;(0,b_m];(b_m,m+k_m]\big)
        \notag\\
&\qquad
=
\mathfrak D_h(F,F_*)
+
2\left\langle
\Delta,
J_h\big(\widetilde W_m^{(L)}-\widetilde W_m^{(R)}\big)
\right\rangle_{\dcal K_h}
-
Q_{J_h}\big(\widetilde W_m^{(L)}-\widetilde W_m^{(R)}\big)
+
\mathcal R_m,
\label{e:full_main_expansion}
\end{align}
where
\begin{align*}
\mathcal R_m
={}&
\frac{1}{b_m(b_m-1)}
\sum_{j=1}^{b_m}Q_{J_h}(W_j)
-
\frac{Q_{J_h}(\widetilde W_m^{(L)})}{b_m-1}\\
&\quad+
\frac{1}{d_m(d_m-1)}
\sum_{j=1}^{d_m}Q_{J_h}(W_{b_m+j})
-
\frac{Q_{J_h}(\widetilde W_m^{(R)}+\Delta)}{d_m-1}.
\end{align*}
By the same bounds used under $H_0$, together with
Assumption~\ref{a:h}(i),
$$
\|\widetilde W_m^{(L)}\|_{\dcal K_h}
=
O_\P(b_m^{-1/2}),
\qquad
\frac{1}{b_m}\sum_{j=1}^{b_m}
\|W_j\|_{\dcal K_h}^{2}
=
O_\P(1).
$$
Now,
\eqref{eq:post-Hilbert-bound} gives
$$
\|\widetilde W_m^{(R)}\|_{\dcal K_h}
=
O_\P(d_m^{-1/2}),
$$
while Assumption~\ref{a:Ha_dependence,moments}(iii) and Markov's
inequality give
$$
\frac{1}{d_m}\sum_{j=1}^{d_m}
\|W_{b_m+j}\|_{\dcal K_h}^{2}
=
O_\P(1).
$$
Consequently,
$$
\|\widetilde W_m^{(L)}-\widetilde W_m^{(R)}\|_{\dcal K_h}
=
O_\P(b_m^{-1/2}+d_m^{-1/2}).
$$
Since $\|J_h\|_{\mathrm{op}}=1$, the Cauchy--Schwarz inequality gives
$$
\left|
\left\langle
\Delta,
J_h\big(\widetilde W_m^{(L)}-\widetilde W_m^{(R)}\big)
\right\rangle_{\dcal K_h}
\right|
\leq
\|\Delta\|_{\dcal K_h}
O_\P(b_m^{-1/2}+d_m^{-1/2}),
$$
and
$$
\left|
Q_{J_h}\big(\widetilde W_m^{(L)}-\widetilde W_m^{(R)}\big)
\right|
=
O_\P(b_m^{-1}+d_m^{-1}).
$$
Furthermore,
\begin{align*}
|\mathcal R_m|
&\leq
\frac{1}{b_m(b_m-1)}
\sum_{j=1}^{b_m}\|W_j\|_{\dcal K_h}^{2}
+
\frac{\|\widetilde W_m^{(L)}\|_{\dcal K_h}^{2}}{b_m-1}\\
&\quad+
\frac{1}{d_m(d_m-1)}
\sum_{j=1}^{d_m}\|W_{b_m+j}\|_{\dcal K_h}^{2}
+
\frac{\|\widetilde W_m^{(R)}+\Delta\|_{\dcal K_h}^{2}}{d_m-1}\\
&=
O_\P(b_m^{-1})+O_\P(d_m^{-1}).
\end{align*}
Here, for the last term, Jensen's inequality gives
$$
\|\widetilde W_m^{(R)}+\Delta\|_{\dcal K_h}^{2}
=
\left\|
\frac{1}{d_m}\sum_{j=1}^{d_m}W_{b_m+j}
\right\|_{\dcal K_h}^{2}
\leq
\frac{1}{d_m}\sum_{j=1}^{d_m}
\|W_{b_m+j}\|_{\dcal K_h}^{2}
=
O_\P(1).
$$
It follows from \eqref{e:full_main_expansion} that
\begin{align}
\mathcal U\big(h;(0,b_m];(b_m,m+k_m]\big)
&=
\mathfrak D_h(F,F_*)
+
O_\P\Big(
\|\Delta\|_{\dcal K_h}
(b_m^{-1/2}+d_m^{-1/2})
+b_m^{-1}+d_m^{-1}
\Big).
\label{e:full_U_asympt}
\end{align}
Write
$$
    t_m=\frac{k_m}{m},
    \qquad
    L_m=\frac{b_m}{m}=1+\frac{k_*}{m},
    \qquad
    \rho_m=\frac{d_m}{m},
    \qquad
    q_m=\frac{d_m}{m+k_*}=\frac{\rho_m}{L_m}.
$$
Then
\begin{align}
\mathcal D_m(k_m)
&\geq
\frac{d_m^2}{m\,g(t_m,L_m,\rho_m)}
\left|
\mathcal U\big(h;(0,b_m];(b_m,m+k_m]\big)
\right|.
\label{e:full_detector_lower_start}
\end{align}
Since
$$
    1+t_m=L_m(1+q_m),
    \qquad
    1+\frac{\rho_m}{L_m}=1+q_m,
$$
we have
\begin{align*}
g(t_m,L_m,\rho_m)
&=
\left(L_m(1+q_m)\right)^{1-\beta}
(L_mq_m)^\beta
(1+q_m)
\log^{2-\beta}(e+t_m)
        \notag\\
&=
L_m q_m^\beta(1+q_m)^{2-\beta}
\log^{2-\beta}(e+t_m).
\end{align*}
Hence
\begin{align}
\frac{d_m^2}{m\,g(t_m,L_m,\rho_m)}
&=
(m+k_*)
\left(\frac{q_m}{1+q_m}\right)^{2-\beta}
\log^{-(2-\beta)}(e+t_m).
\label{e:full_main_scaling}
\end{align}
Set
$$
    \zeta_m
    =
    (m+k_*)|\mathfrak D_h(F,F_*)|
    \left(\frac{q_m}{1+q_m}\right)^{2-\beta}
    \log^{-(2-\beta)}(e+t_m).
$$
Then $\zeta_m\to\infty$ by hypothesis. Moreover,
$$
    \zeta_m\leq (m+k_*)|\mathfrak D_h(F,F_*)|,
    \qquad
    \zeta_m\leq d_m|\mathfrak D_h(F,F_*)|.
$$
Thus
$$
    (m+k_*)|\mathfrak D_h(F,F_*)|\to\infty,
    \qquad
    d_m|\mathfrak D_h(F,F_*)|\to\infty.
$$
Using
$$
|\mathfrak D_h(F,F_*)|
=
\|\Delta\|_{\dcal K_h}^{2}
$$
in \eqref{e:full_U_asympt}, together with
$$
(m+k_*)|\mathfrak D_h(F,F_*)|\to\infty,
\qquad
d_m|\mathfrak D_h(F,F_*)|\to\infty,
$$
gives
$$
\mathcal U\big(h;(0,b_m];(b_m,m+k_m]\big)
=
\mathfrak D_h(F,F_*)\{1+o_\P(1)\}.
$$
Hence $\mathcal D_m(k_m)\stackrel{\P}{\longrightarrow}\infty$, and therefore
$
    \P(\tau_m^{\mathrm{full}}\leq k_m)\to1.
$
\end{proof}
\begin{proof}[Proof of Corollary~\ref{c:delay_full_rate}]
Let $r_m=(m+k_*)^{\frac{1-\beta}{2-\beta}}
|\mathfrak D_h(F,F_*)|^{-\frac{1}{2-\beta}}\log m$,
set $b_m=m+k_*$ and $D_m=|\mathfrak D_h(F,F_*)|$,
and take $d_m=\lceil Cr_m\rceil$.
For each fixed $C>0$, the hypothesis gives $d_m/b_m\to0$ and
$(b_m\wedge d_m)D_m\to\infty$, so \eqref{e:full_U_asympt} yields
$\mathcal U\big((0,b_m];(b_m,b_m+d_m]\big)
=\mathfrak D_h(F,F_*)(1+o_\P(1))$.
Since $k_*$ is polynomially bounded and $d_m=o(b_m)$,
$\log(e+(k_*+d_m)/m)\le K\log m$ for some $K$ independent of $C$;
hence \eqref{e:full_detector_lower_start} and
\eqref{e:full_main_scaling} give
$$
\mathcal D_m(k_*+d_m)
\geq (C/K)^{2-\beta}(1+o_\P(1))
$$
Choosing $C$ sufficiently large gives
$\P(\tau_m^{\rm full}\le k_*+d_m)\to1$, proving the claim.
\end{proof}
\section{Proofs for Section \ref{s:implementation}}\label{s:implem_supp}

We begin with a lemma that establishes that the Gram matrix built from empirically centered $\overline h_m(\bx,\by)$ using the data from the historical baseline period serves as a suitable approximation of the analogous matrix built from the true degenerate kernel $\overline h(\bx,\by)$.

\begin{lemma} \label{l:H_empirical_vs_H}Suppose Assumptions~\ref{a:h}, \ref{a:Y_moment}, \ref{a:Lpm}, and \ref{a:ac} hold. Let
\begin{equation}\label{e:def_Hm}
    \overline H_m=m^{-1}(\overline h_m(\bX_i,\bX_j))_{1\leq i,j\leq m},\qquad\widetilde H_m=m^{-1}(\overline h(\bX_i,\bX_j))_{1\leq i,j\leq m}.
\end{equation} Then,
\begin{equation}
\label{e:A_to_H}
\left\|\overline H_m-\widetilde H_m\right\|_{\rm F} = O_\P(m^{-1/2}).
 \end{equation}
\end{lemma}
\begin{proof}
For $\bX,\bY\stackrel{iid} \sim F$, let $h_1(\bx)=\E h(\bx,\bY)$ and $\mu = \E h(\bX,\bY)$. Note
\begin{align*}
&\overline h_m(\bX_i,\bX_j)-\overline h(\bX_i,\bX_j)\\
\quad&= - \left(\frac{1}{m}\sum_{s=1}^m h(\bX_i,\bX_s) -h_1(\bX_i)\right) 
- \left(\frac{1}{m}\sum_{r=1}^m h(\bX_r,\bX_j)- h_1(\bX_j)\right)\\
& \qquad\qquad 
+ \frac{1}{m^2}\sum_{r=1}^m\sum_{s=1}^m (h(\bX_r,\bX_s) - \mu)\\
& = -a_{i,m} -a_{j,m} + b_m,
\end{align*}
which, writing $\mathbf a_m=(a_{1,m},\ldots a_{m,m})^\top$, gives
\begin{align*}
\|\overline H_m-\widetilde H_m\|_{\rm F} &= m^{-1}\|-\mathbf a _m\mathbf 1^\top  -\mathbf 1 \mathbf a_m^\top + \mathbf 1 \mathbf 1^\top b_m\|_{\rm F}\\
&\leq 2m^{-1}\| \mathbf a_m \mathbf 1^\top\|_{\rm F} + m^{-1}|b_m|\| \mathbf 1 \mathbf 1^\top\|_{\rm F}\\
& \leq 2 m^{-1/2} \|\mathbf a_m\| + |b_m|.
\end{align*}
We now proceed to show
\begin{equation}\label{e:am,bm_to0}
    \|\mathbf a_m\|=O_\P(1), \quad |b_m|=O_\P(m^{-1/2}).
\end{equation}
Observe, with $Y_j$ as in \eqref{e:def_Yj}, that for $i\neq j$,
\begin{equation}\label{e:gram_identity}
\langle Y_i,JY_j\rangle_{\cH} = \sum_{\ell =1}^\infty\lambda_\ell \phi(\bX_i)\phi_\ell(\bX_j) = \overline h(\bX_i,\bX_j),
\end{equation}
where the last equality follows from Assumption~\ref{a:ac}.
Now, using $\overline h(\bx,\by)=h(\bx,\by)-h_1(\bx)-h_1(\by)+\mu$, we have
\begin{align*}
a_{i,m}&= \left(\frac1m \sum_{s=1}^m\langle Y_i,J Y_{s}\rangle_{\cH}  + \frac1m\sum_{s=1}^m(h_1(\bX_s) -\mu)\right)  + \frac1m \left( \overline h(\bX_i,\bX_i) -\langle Y_i,J Y_{i}\rangle_{\cH } \right).\\
&= \langle Y_i,J\overline Y_{1,m}\rangle_{\cH} + \frac1m\sum_{s=1}^m(h_1(\bX_s) -\mu)+ \frac1m \left( \overline h(\bX_i,\bX_i) -\langle Y_i,J Y_{i}\rangle_{\cH } \right).
\end{align*}
Hence,
\begin{align}
\notag\E a_{i,m}^2 &\lesssim  \E \|Y_i\|_{\cH}^2 \|\overline Y_{1,m}\|_{\dcal H}^2 + \E \bigg( \frac1m\sum_{s=1}^m(h_1(\bX_s) -\mu)\bigg)^2 + \frac1{m^2} \left(\E \overline h^2(\bX_i,\bX_i)   +\E\|Y_1\|_{\cH}^4\right)\\
&\lesssim \left(\E \|Y_i\|_{\cH}^4\right)^{1/2}\left( \E \|\overline Y_{1,m}\|_{\dcal H}^4\right)^{1/2} +\E\bigg( \frac1m\sum_{s=1}^m(h_1(\bX_s) -\mu)\bigg)^2  +O(m^{-2}).
\label{e:ai_bound}
\end{align}
We proceed to analyze each term in \eqref{e:ai_bound}.  First note Lemma \ref{l:main_invariance_lemma} gives $\E\| \overline Y_{1,m}\|_{\cH}^4 \leq Cm^{-2}$.  Next, using Assumption \ref{a:h}, observe
\begin{align*}
|h_1(\bx)-h_1(\bx')| & \leq \E | h(\bx,\bY)-h(\bx',\bY)|\\
& \leq \left(\E | h(\bx,\bY)-h(\bx',\bY)|^2\right)^{1/2}\\
&= \left(\sum_{\ell=1}^\infty|\lambda_\ell|^2(\phi_\ell(\bx)-\phi_\ell(\bx'))^2 \right)^{1/2}\\
&\leq|\lambda_1| \left(\sum_{\ell=1}^\infty|\lambda_\ell|(\phi_\ell(\bx)-\phi_\ell(\bx'))^2 \right)^{1/2} \\
&\leq C\rho(\bx,\bx')^\alpha.
\end{align*}
With $V_j=h_1(\bX_j)-\mu$, $V_j^{(r)}=h_1(\bX_j^{(r)})-\mu$, and with $p,\eta$ as in Assumption \ref{a:Lpm} we have
$$
\sum_{r=1}^\infty \left(\E | V_1-V_1^{(r)}|^p\right)^{\eta/p} \leq C\sum_{r=1}^\infty \vartheta_{p\alpha}(r)^{\alpha\eta}<\infty,
$$
hence we may apply Lemma \ref{l:main_invariance_lemma} to conclude
$$
\E\Big( \frac1m\sum_{s=1}^m(h_1(\bX_s) -\mu)\Big)^2 \leq C m^{-1}.
$$
Thus, from \eqref{e:ai_bound}, we have
$$
 \E \|\mathbf a_m\|^2 \leq C \left(\frac 1m\sum_{\ell=1}^m(\E\|Y_i\|_{\cH}^4)^{1/2} + 1\right) \leq C,
$$
which gives the first statement in \eqref{e:am,bm_to0}.  For $b_m$, write
\begin{align*}
&\frac1{m^2}\sum_{r=1}^m\sum_{s=1}^m (h(\bX_r,\bX_s) - \mu)\\
&\qquad =  \frac1{m^2}\sum_{r=1}^m \sum_{s=1}^m\left(\overline h(\bX_r,\bX_s) + (h_1(\bX_s)-\mu)+ (h_1(\bX_r) -\mu)\right)\\
& \qquad=\frac1{m^2}\sum_{r=1}^m \sum_{s=1}^m \overline h(\bX_r,\bX_s) +\frac1{m^2}\sum_{i=1}^m
\left( \overline h(\bX_i,\bX_i)-\langle Y_i,JY_i\rangle_{\cH} \right)
+\frac2m\sum_{s=1}^m (h_1(\bX_s)-\mu)\\
&\qquad= \langle \overline Y_{1,m},J\overline Y_{1,m}\rangle_{\cH} + O_\P(m^{-1}) + O_\P(m^{-1/2}).
\end{align*}
Since $\E |\langle \overline Y_{1,m},J\overline Y_{1,m}\rangle_{\cH}|\leq \E\| \overline Y_{1,m}\|_{\cH}^2 \leq Cm^{-1}$, we obtain $b_m=O_\P(m^{-{1/2}})$, which yields the second statement in \eqref{e:am,bm_to0}.  Hence, we have  \eqref{e:A_to_H}.

\end{proof}

\begin{proof}[Proof of Theorem \ref{t:eigen_consistency_theo}]
First, set
\begin{equation}\label{e:Hsp}
H_m^{\mathrm{sp}} = \frac1m \left( \sum_{\ell=1}^\infty\lambda_\ell \phi_\ell(\bX_i)\phi_\ell(\bX_j)\right)_{1\leq i,j\leq m}.
\end{equation}
By Assumption~\ref{a:ac}, $H_m^{\mathrm{sp}}$ agrees with $\widetilde H_m$ off the diagonal.  Hence
\begin{equation}\label{e:Hsp_to_tildeH}
\| H_m^{\mathrm{sp}}-\widetilde H_m\|_{\rm F}^2 =\frac1{m^2} \sum_{i=1}^m\left(\sum_{\ell=1}^\infty\lambda_\ell\phi_\ell(\bX_i)^2 -\overline h(\bX_i,\bX_i)
\right)^2=O_\P(m^{-1}),
\end{equation}
where the last term follows by Assumption~\ref{a:h}(i) and the bound 
$\sum_{\ell\geq1}|\lambda_\ell|\phi_\ell(\bx)^2\leq \|Y(\bx)\|_{\cH}^2$ together with Assumption~\ref{a:Y_moment}.  In view of Lemma~\ref{l:H_empirical_vs_H}, the preceding display shows that
$\overline H_m$, $\widetilde H_m$, and $H_m^{\mathrm{sp}}$ differ only by
$O_\P(m^{-1/2})$  in Frobenius norm. By the
Hoffman--Wielandt inequality  \citep{hoffman:wielandt:1953},  it suffices to study $H_m^{\textrm{sp}}$. We begin along the lines of the proof of \cite[Theorem 3.1]{koltchinskii:gine:2000}. Let
\begin{align}
{\boldsymbol \Phi}_{\ell,m}=m^{-1/2}\left(\phi_\ell(\bX_1),\ldots,\phi_\ell(\bX_m)\right)^\top,\ \ell=1,\ldots, m
\end{align}
$$
\Phi_m = (\boldsymbol \Phi_{1,m},\ldots \mathbf \Phi_{m,m}) \in \R^{m\times m}.
$$
We have:
$$
 H^{\textrm{sp}}_m = \breve H_m + R_m, 
$$
where 
\begin{align*}
\breve H_m &= \sum_{\ell=1}^m \lambda_\ell \mathbf \Phi_{\ell,m}\mathbf \Phi_{\ell,m}^\top =   \Phi_m\Lambda_m\Phi_m^\top,\quad \Lambda_m=\text{diag}\,(\lambda_1,\ldots,\lambda_m).\\
\quad R_m  &= \sum_{\ell=m+1}^\infty \lambda_\ell \mathbf \Phi_{\ell,m}\mathbf \Phi_{\ell,m}^\top.
\end{align*}
Note $\E \phi_\ell(\bX_1)\phi_{\ell'}(\bX_1)=\delta_{\ell,\ell'}$.
Let 
$$
 E_m = \Phi_m^\top\Phi_m - I = \frac1m\sum_{j=1}^m \boldsymbol \xi_{m,j} \boldsymbol \xi_{m,j}^\top - I, \qquad \boldsymbol \xi_{m,j} = (\phi_1(\bX_j),\ldots,\phi_m(\bX_j))^\top.
$$
Now, set
$$
[J]_m = \text{diag}\,(\text{sgn}(\lambda_1),\ldots,\text{sgn}(\lambda_m)),
$$
and consider
$$
M_m =|\Lambda_m|^{1/2}[J]_m \Phi_m^\top\Phi_m |\Lambda_m|^{1/2},\qquad %
$$
Note $M_m$ and $\breve  H_m \in \R^{m\times m}$  have the same spectrum, since, e.g., with $Q_1=\Phi_m|\Lambda_m|^{1/2}$, $Q_2=|\Lambda_m|^{1/2}[J]_m \Phi_m^\top,$ clearly $\breve H_m = Q_1Q_2$  and $M_m=Q_2Q_1$, so, e.g.,
$$
\lambda \mathbf v = \breve H_m \mathbf v = Q_1Q_2\mathbf v \implies \lambda Q_2 \mathbf v = Q_2Q_1Q_2\mathbf v=M_mQ_2 \mathbf v,
$$
and similarly for the opposite direction.  Now,
\begin{align}
\notag\|M_m-\Lambda_m\|^2_F &=\||\Lambda_m|^{1/2}[J]_m (\Phi_m^\top\Phi_m-I) |\Lambda_m|^{1/2}\|_{\rm F}^2 \\
\notag& =\||\Lambda_m|^{1/2}[J]_m E_m|\Lambda_m|^{1/2}\|_{\rm F}^2\\
& =  \frac{1}{m^2}\sum_{\ell,\ell'=1}^m|\lambda_\ell\lambda_{\ell'}|  \left(\sum_{j=1}^m (\phi_\ell(\bX_j)\phi_{\ell'}(\bX_j)-\delta_{\ell,\ell'})\right)^2.\label{e:Mm-L_frob}
\end{align}
Recall that, for any $L\geq1,$ $P_L$ denotes the projection in $\cH$ onto $\text{span} \{\phi_1,\ldots,\phi_L\}.$   With 
\begin{equation}B_i = Y_i\otimes Y_i - \E (Y_i\otimes Y_i ),\label{e:def_Bi_operator}
\end{equation}
note
\begin{align*}
\langle  \phi_\ell,B_j\phi_{\ell'}\rangle_{\cH} &= \langle \phi_\ell,Y_j\rangle_{\cH} \langle \phi_{\ell'},Y_j\rangle_{\cH}- \E \big(\langle \phi_\ell,Y_j\rangle_{\cH} \langle \phi_{\ell'},Y_j\rangle_{\cH}\big)\\
&=  \sqrt{|\lambda_\ell\lambda_{\ell'}|}\big(\phi_\ell(\bX_j)\phi_{\ell'}(\bX_j) - \delta_{\ell,\ell'}\big).
\end{align*}
Hence,
\begin{align}
\notag\sum_{\ell,\ell'=1}^m|\lambda_\ell\lambda_{\ell'}|  \left(\sum_{j=1}^m (\phi_\ell(\bX_j)\phi_{\ell'}(\bX_j)-\delta_{\ell,\ell'})\right)^2%
 \notag & =\sum_{\ell,\ell'=1}^m \bigg( \sum_{j=1}^m\langle \phi_{\ell} , B_j\phi_{\ell'}\rangle_{\cH}\bigg)^2\\
 \notag & =\sum_{\ell,\ell'=1}^\infty \bigg( \sum_{j=1}^m\langle P_m\phi_{\ell} , B_jP_m\phi_{\ell'}\rangle_{\cH}\bigg)^2\\
  \notag& = \Big\| P_m\sum_{j=1}^m B_j P_m\Big\|_{\rm HS}^2\\
  & \leq \Big\| \sum_{j=1}^m B_j \Big\|_{\rm HS}^2\label{e:Bj_HS_norm}
\end{align}
However, note that with $B_i^{(r)}=Y_i^{(r)}\otimes  Y^{(r)}_i-\E (Y_i\otimes Y_i ) ,$ we have the bound
\begin{align*}
\|B_1^{(r)}-B_1\|_{\text{HS}} \leq \| Y_1^{(r)}-Y_1\|_{\cH}\|Y_1^{(r)}\|_{\cH} + \|Y_1\|_{\cH}\| Y_1-Y_1^{(r)}\|_{\cH} .
\end{align*}
Hence, with 
$q=p/2>2$, and $\eta$ as in Assumption \eqref{a:Lpm}, using Cauchy-Schwarz,
\begin{align*}
\sum_{r=1}^\infty \left(\E \|B_1-B_1^{(r)}\|^q_{\text{HS}}\right)^{\eta /q}&  \lesssim (\E\|Y_1\|^{p}_{\cH})^{\eta/p}\sum_{r=1}^\infty  \left(\E\|Y_1-Y_1^{(r)}\|_{\cH}^{p}\right)^{\eta/p}<\infty.
\end{align*}
Hence, viewed as elements of the Hilbert space $\dcal S_2(\cH) =\{T:\cH\to \cH: \|T\|_{\text{HS}}<\infty\}$ with inner product $\langle T_1,T_2\rangle_{\dcal S_2(\cH)} =\tr(T_1^*T_2)$, (recall $T^*$ denotes adjoint of an operator $T$)  the sequence $\{B_j\}$ satisfies the hypotheses of Lemma \ref{l:main_invariance_lemma}.   We readily conclude
$$
\Bigg\|\sum_{i=1}^m B_i\Bigg\|_{\text{HS}}   = O_\P(m^{1/2}).
$$
Hence,  from \eqref{e:Bj_HS_norm} and \eqref{e:Mm-L_frob}
$$
\|M_m-\Lambda_m\| _F=O_\P(1/\sqrt{m}).
$$
On the other hand, with $Y_{L,j}$ as defined in \eqref{e:def_YjL}, taking $L=m$, note
\begin{align*}
    \|R_m\|_{\rm F}^2 & = \frac1{m^2}\sum_{j,j'=1}^m \bigg(\sum_{\ell=m+1}^\infty \lambda_\ell 
     \phi_\ell(\bX_j)\phi_\ell(\bX_{j'})\bigg)^2\\
     & = \frac1{m^2}\sum_{j,j'=1}^m |\langle Y_j-Y_{m,j},J(Y_{j'}-Y_{m,j'})\rangle_{\cH}|^2\\
     & \leq \bigg(\frac1{m}\sum_{j=1}^m\|Y_j-Y_{m,j}\|_{\cH}^2\bigg)^2.
\end{align*}
Now, let $\breve \lambda_{1,m},\ldots,\breve \lambda_{m,m}$ be the eigenvalues of $\breve H_m$ (equivalently, of $M_m$). 
 We next apply the  perturbation inequality for ``nearly Hermitian matrices" of Kahan
\citep{kahan:1975}.  Set
$$
D_m=M_m-\Lambda_m,\qquad
S_m=\frac{D_m+D_m^\top}{2},\qquad
K_m=\frac{D_m-D_m^\top}{2}.
$$
Since $\Lambda_m$ is symmetric, Kahan's inequality \citep[p.12, item (ii), with $X=S_m$ and $Y=K_m$]{kahan:1975} applied to the perturbation $\Lambda_m + D_m$ yields
$$
\left(\min_{\pi\in\Pi_m}
\sum_{\ell=1}^m
\bigl(\breve\lambda_{\pi(\ell),m}-\lambda_\ell\bigr)^2\right)^{1/2}
\leq\|S_m\|_{\rm F}+\|K_m\|_{\rm F}
\leq
\sqrt 2\|D_m\|_{\rm F}.
$$
Consequently, using the Hoffman--Wielandt inequality \citep{hoffman:wielandt:1953} together with the triangle inequality for the rearrangement
distance (see., e.g. \cite{koltchinskii:gine:2000}), we obtain
\begin{align}
&\min_{\pi \in \Pi_m}
\left(
\sum_{\ell=1}^m
(\widehat \lambda_{\pi(\ell),m}-\lambda_{\ell})^2
\right)^{1/2}
\notag\\
&\quad \leq
\|\overline H_m-\widetilde H_m\|_{\rm F}
+\|\widetilde H_m-H_m^{\mathrm{sp}}\|_{\rm F}
+\|H_m^{\mathrm{sp}}-\breve H_m\|_{\rm F}
+\sqrt{2}\|M_m-\Lambda_m\|_{\rm F}
\notag\\
&\quad =
O_\P(m^{-1/2})+O_\P(m^{-1/2})
+\|R_m\|_{\rm F}+O_\P(m^{-1/2})
\notag\\
&\quad =
O_\P\left(\sum_{\ell=m+1}^\infty|\lambda_\ell|\right)
+O_\P(m^{-1/2}),
\label{e:widetildeH_mbound}
\end{align}
where we used that
$
\E\|R_m\|_{\rm F} \leq \sum_{\ell=m+1}^\infty |\lambda_\ell|.$
This completes the proof.
\end{proof}

The proof of Theorem \ref{t:spectral_consistency} requires the following lemma regarding alignment of orthogonal matrices.  The result is standard; for completeness we provide a proof. %
\begin{lemma}\label{l:basic_rotation_lemma}
Suppose $V,U\in \R^{m\times L}$ have orthonormal columns, $P=VV^\top$, $Q=UU^\top$.  Then
\begin{equation}\label{e:frob_id}
\| P-Q\|_{\rm F}=\sqrt{2}\|P(I-Q)\|_\mathrm{F}.
\end{equation}
Moreover, there is an orthogonal matrix $O\in \R^{L\times L}$ such that
\begin{equation}\label{e:rot}
\| V-UO\|_{\rm F}\leq \|P-Q\|_{\rm F}.
\end{equation}
\end{lemma}
\begin{proof}
The identity \eqref{e:frob_id} follows immediately since $\| P-Q\|_{\rm F}^2=\text{tr}((P-Q)^\top(P-Q))=2L -2\text{tr}(PQ)$ and $\|P(I-Q)\|^2_\mathrm{F}=\text{tr}\big( ((I-Q)P)^\top (I-Q)P\big)=\text{tr}\big( P(I-Q)P\big)=L-\mathrm{tr}(PQ).$  For \eqref{e:rot}, write $M=U^\top V\in \R^{L\times L}$ and consider its singular value decompositon $M=W \Sigma_M Z^\top$.  Note $\|\Sigma_M\|_{\mathrm{op}}=\|U^\top V\|_{\mathrm op}\leq \|U\|_{\mathrm{op}}\|V\|_{\mathrm{op}}= 1$.  With $O=WZ^\top\in \R^{L\times L},$
$$
\| V-U O\|_{\rm F}^2  = \mathrm {tr}(V^\top V) +  \mathrm {tr}((UO)^\top UO) - 2 \mathrm {tr}(O^\top U^\top V) = 2L -2\tr (\Sigma_M).
$$
The conclusion follows, since $\|P-Q\|_{\rm F}^2 =2L -2\text{tr}(PQ)=2L - 2\mathrm{tr}(M^\top M)=2L-2\text{tr}(\Sigma_M^2)\geq 2L -2\tr (\Sigma_M)$, since $\|\Sigma_M\|_{\mathrm{op}}\leq 1$ gives $\text{tr}(\Sigma_M^2)\leq \text{tr}(\Sigma_M)$.
\end{proof}
The following Davis–Kahan-type lemma  is used in the proof of Theorem \ref{t:spectral_consistency}.
\begin{lemma}\label{l:DK_lemma}
Let $M_1,M_2\in \R^{m\times m}$ be symmetric, and lets $\sigma(M_i)\subseteq \R$ denote their respective sets of eigenvalues.  Let $\mathcal I_k\subseteq \R$, $k=1,\ldots,b$, be disjoint sets satisfying
$$
\sigma(M_1)\cup \sigma(M_2)\subseteq \bigcup_{k=1}^b \mathcal I_k,\qquad \min_{1\leq k\neq k'\leq b}\mathrm{dist}(\mathcal I_k,\mathcal I_{k'}) >\delta>0.
$$
Let $\cP_k(M_i)$ denote the orthogonal projection onto the subspace spanned by the eigenvectors
of $M_i$ corresponding to those of its eigenvalues that lie in $\mathcal I_k$.  Then, if 
$$
\mathrm{rank}(\cP_k(M_1))=\mathrm{rank}(\cP_k(M_2)),\quad k=1,\ldots,b,
$$
we have
$$
\left(\sum_{k=1}^b \| \cP_k(M_1)-\cP_k(M_2)\|^2_F\right)^{1/2} \leq C\frac{\|M_1-M_2\|_{\rm F}}{\delta }
$$
\end{lemma}
\begin{proof} Write $Q_k(M_i)=I-\cP_k(M_i)=\sum_{k'\neq k}\cP_{k'}(M_i)$. Since the ranks are equal, Lemma \ref{l:basic_rotation_lemma} gives $\|\cP_k(M_1)-\cP_k(M_2)\|_{\rm F}^2=2\|\cP_k(M_1)Q_k(M_2)\|_{\rm F}^2.$ Moreover, applying \cite[p.212]{bhatia:1997},
\begin{align*}
\|\cP_k(M_1)Q_k(M_2)\|_{\rm F}^2 &\leq C \frac{\|\cP_k(M_1)(M_1-M_2)Q_k(M_2)\|_{\rm F}^2}{\delta ^2} \leq C\frac{\|\cP_k(M_1)(M_1-M_2)\|_{\rm F}^2}{\delta^2},
\end{align*}
since $\|Q_k(M_2)\|_{\mathrm{op}}\leq1$.  Since the projections $\cP_k(M_1)$ are orthogonal, we have $\cP_k(M_1) \cP_{k'}(M_1)\equiv0$ whenever $k\neq k'$, which implies  $\sum_{k}\|\cP_k(M_1)(M_1-M_2)\|_{\rm F}^2=\| \sum_k \cP_k(M_1)(M_1-M_2)\|_{\rm F}^2$.  Hence
$$
\sum_{k=1}^b \|\cP_k(M_1)-\cP_k(M_2)\|_{\rm F}^2 \leq C \frac{1}{\delta^2}\Bigg\|\sum_{k=1}^b \cP_k (M_1)(M_1-M_2)\Bigg\|_{\rm F}^2\leq C\frac{\|M_1-M_2\|_{\rm F}^2}{\delta^2},
$$
which after taking square roots gives the statement.
\end{proof}

The following Lemma is used in the proof of Theorem \ref{t:spectral_consistency}. It demonstrates that, under the eigenvalue gap condition of Assumption \ref{a:eigsep},  eigenvalue consistency for $\widehat \lambda_{\ell,m}$ and $\breve \lambda_{\ell,m}$ holds for the top $L$ empirical eigenvalues.  We briefly recall the notation
$$
\varepsilon_m = m^{-1/2} + \sum_{\ell=m+1}^\infty |\lambda_\ell|,
$$
and we recall that the eigenvalue gap $\delta_m$ under Assumption \ref{a:eigsep} satisfies $\delta_m/\varepsilon_m\to\infty.$
\begin{lemma}\label{l:sample_eig_concentration}   Let $L=L(m)$ be chosen such that Assumption \ref{a:eigsep} is satisfied. With 
\begin{equation}\label{def:muj}
\mu_{1},\ldots, \mu_b,\qquad |\mu_1|>\ldots>|\mu_b|
\end{equation}
denoting the distinct values among $\lambda_1,\ldots,\lambda_L$, write $r_0=0$, and for $k=1,\ldots,b$, set
\begin{gather}
d_k =\#\{1\leq\ell\leq L: \lambda_\ell=\mu_k\},\quad \quad r_k = \sum_{j=1}^k d_j, \label{e:def_rk_index}\\
\mathcal B_k=\{r_{k-1}+1,\ldots, r_{k}\},\label{e:def_Bk_block}
\end{gather}
and
\begin{equation}\label{e:Ik_bands}
\mathcal I_{k,m}=\Big(\mu_k-\frac{\delta_m}4,
\,\mu_k+\frac{\delta_m}4\Big),\qquad k=1,\ldots, b,
\end{equation}
where $\delta_m$ is as in Assumption \ref{a:eigsep}.
Then,  %
$$
\P\left( \big\{\widehat \lambda_{\ell,m} ,~\ell \in \mathcal B_k\big\} \subseteq  \mathcal I_{k,m}, ~~ \big\{\breve \lambda_{\ell,m} ,~\ell \in \mathcal B_k\big\} \subseteq \mathcal I_{k,m},\quad k=1,\ldots, b\right) \to 1.
$$
\end{lemma}
\begin{proof}
Write $\zeta_m=\delta_m/\varepsilon_m$ and note $\zeta_m\to\infty$.  In particular, the event
$$
\Omega_m=\min_{\pi \in \Pi_m} \left(\sum_{\ell=1}^m (\widehat \lambda_{\ell,m}-\lambda_{\pi(\ell)})^2\right)^{1/2} \leq \varepsilon_m (\zeta_m)^{1/2}
$$
tends to 1 in probability due to Theorem \ref{t:eigen_consistency_theo}. 
Now, let
$$
\pi_m \in \argmin_{\pi \in \Pi_m} \left(\sum_{\ell=1}^m (\widehat \lambda_{\pi(\ell),m}-\lambda_{\ell})^2\right)^{1/2}
$$
On the event $\Omega_m$,  clearly 
$||\widehat\lambda_{\pi_m(\ell)}|-|\lambda_\ell||\leq |\widehat\lambda_{\pi_m(\ell)}-\lambda_\ell|\leq \varepsilon_m(\zeta_m)^{1/2}$, so that
$$
\begin{aligned}
|\widehat\lambda_{\pi_m(\ell)}|& \geq |\lambda_\ell|-\varepsilon_m(\zeta_m)^{1/2} \geq |\lambda_L| - \varepsilon_m(\zeta_m)^{1/2}, && \text{if}\quad 1\leq \ell \leq L,\\
|\widehat\lambda_{\pi_m(\ell)}| &\leq |\lambda_{\ell}|+\varepsilon_m(\zeta_m)^{1/2} \leq |\lambda_{L+1}| + \varepsilon_m(\zeta_m)^{1/2}, && \text{if}\quad  L< \ell \leq m.
\end{aligned}$$
Since $(|\lambda_L| - \varepsilon_m(\zeta_m)^{1/2})-( |\lambda_{L+1}|+\varepsilon_m(\zeta_m)^{1/2}) \geq \varepsilon_m \zeta_m(1-o(1))=\delta_m(1-o(1))>0$, the largest $L$ eigenvalues $\{\widehat\lambda_{1,m},\ldots,\widehat{\lambda}_{L,m}\}$ coincide with the set $\{\widehat\lambda_{\pi_m(1),m},\ldots,\widehat{\lambda}_{\pi_m(L),m}\}$, i.e,.
\begin{equation}\label{e:pi_m_identity_upto_L}
\pi_m(\{1,\ldots,L\})=\{1,\ldots,L\}.
\end{equation}
Analogously, for any $k=1,\ldots,b-1$
$$
\begin{aligned}
|\widehat\lambda_{\pi_m(\ell)}| &\geq |\lambda_\ell|-\varepsilon_m(\zeta_m)^{1/2} \geq |\mu_k| - \varepsilon_m(\zeta_m)^{1/2}, && \text{if}\quad 1\leq \ell \leq r_k,\\
|\widehat\lambda_{\pi_m(\ell)}| &\leq |\lambda_{\ell}|+\varepsilon_m(\zeta_m)^{1/2} \leq |\mu_{k+1}| + \varepsilon_m(\zeta_m)^{1/2}, && \text{if}\quad r_k< \ell \leq m.
\end{aligned}
$$
Since $(|\mu_k| - \varepsilon_m(\zeta_m)^{1/2})-( |\mu_{k+1}|+\varepsilon_m(\zeta_m)^{1/2}) \geq \varepsilon_m \zeta_m(1-o(1))>0$, we readily conclude  the largest $r_k$ eigenvalues $\{\widehat\lambda_{1,m},\ldots,\widehat{\lambda}_{r_k,m}\}$ coincide with the set $\{\widehat\lambda_{\pi(1),m},\ldots,\widehat{\lambda}_{\pi(r_k),m}\}$, i.e.,
$$
\pi_m(\mathcal B_1 \cup \ldots \cup\mathcal B_k)=\mathcal B_1 \cup \ldots \cup\mathcal B_k, \qquad k=1,\ldots, b.
$$
Hence, $\pi_m(\mathcal B_1)=\mathcal B_1$, $\pi_m(\mathcal B_2)=\pi_m(\mathcal B_1\cup \mathcal B_2) \setminus \pi_m(\mathcal B_1)=\mathcal B_2$, and 
proceeding inductively (noting $\mathcal B_1\cup\ldots\cup\mathcal B_b=\{1,\ldots,L\}$), we obtain $\pi_m(\mathcal B_k)=\mathcal B_k$ for all $k=1,\ldots, b.$ Thus, on the event $\Omega_m,$ it holds that
$$
\{\widehat\lambda_{\ell,m},~\ell \in \mathcal B_k\}=\{\widehat\lambda_{\pi_m(\ell),m},~\ell  \in \mathcal B_k\}.
$$
Hence, we obtain
$$
\zeta_m^{1/2}\varepsilon_m\geq \left(\sum_{\ell=1}^L (\widehat \lambda_{\pi_m(\ell),m}-\lambda_\ell)^2\right)^{1/2} =  \left(\sum_{k=1}^b\sum_{\ell\in \mathcal B_k}(\widehat \lambda_{\ell,m}-\mu_k)^2\right)^{1/2},
$$
which implies $\P\left( \big\{\widehat \lambda_{\ell,m} ,~\ell \in \mathcal B_k\big\} \subseteq  \mathcal I_{k,m}, ~~k=1,\ldots, b\right) \to 1.$
Finally, note the proof of Theorem \ref{t:eigen_consistency_theo} shows the event
$$
\Omega_m'=\min_{\pi \in \Pi_m} \left(\sum_{\ell=1}^m (\breve \lambda_{\ell,m}-\lambda_{\pi(\ell)})^2\right)^{1/2} \leq \varepsilon_m (\zeta_m)^{1/2}
$$
tends to 1 in probability 
(see \eqref{e:widetildeH_mbound}).  Hence, by arguing analagously to the case for the eigenvalues $\widehat \lambda_{\ell,m}$ we obtain  $\P\left( \big\{\breve \lambda_{\ell,m} ,~\ell \in \mathcal B_k\big\} \subseteq  \mathcal I_{k,m}, ~~k=1,\ldots, b\right) \to 1$.
\end{proof}
 We remind the reader that $\mathbf  v_{\ell,m}=(v_{\ell,m}(1),\ldots,v_{\ell,m}(m))^\top$, $\ell=1,\ldots,m$ are orthonormal eigenvectors of $\overline H_m$, ordered according to the decreasing magnitude of their corresponding eigenvalues. %
\begin{proof}[Proof of Theorem \ref{t:spectral_consistency}]
First, note
\begin{align*}
 &\big\||\Lambda_L|^{1/2}(\Phi_m^L)^\top\Phi_m^L
       |\Lambda_L|^{1/2}-|\Lambda_L|\big\|_{\rm F}^2\\
 &\quad=\frac1{m^2}\sum_{\ell,\ell'=1}^L|\lambda_\ell\lambda_{\ell'}|
 \left(\sum_{j=1}^m
 \big(\phi_\ell(\bX_j)\phi_{\ell'}(\bX_j)-\delta_{\ell\ell'}\big)\right)^2\\
 &\quad\leq\left\|\frac1m\sum_{j=1}^m B_j\right\|_{\rm HS}^2
 =O_\P(m^{-1}),
\end{align*}
with $B_j$ as in \eqref{e:def_Bi_operator}, where the final bound was
established in the proof of Theorem~\ref{t:eigen_consistency_theo}.
Hence, if we define
$$
G_L=(\Phi_m^L)^\top\Phi_m^L\in\R^{L\times L},
$$
we have
$$
G_L=I_L+\dcal E_L,\qquad
\big\||\Lambda_L|^{1/2}\dcal E_L|\Lambda_L|^{1/2}\big\|_{\rm F}
=O_\P(m^{-1/2}).
$$
Moreover, since $|\lambda_L|\geq\delta_m$,
\begin{align*}
\|\dcal E_L\|_{\rm F}
&\leq\|\Lambda_L^{-1}\|_{\mathrm{op}}
\big\||\Lambda_L|^{1/2}\dcal E_L|\Lambda_L|^{1/2}\big\|_{\rm F}\\
&=O_\P\left(\frac1{\sqrt m\,\delta_m}\right)=o_\P(1).
\end{align*}
Thus, with probability tending to one, the eigenvalues of $G_L$ belong to
$[1/2,3/2]$, and we may define
$$
\widetilde\Phi_m^L=\Phi_m^LG_L^{-1/2}.
$$
The columns of $\widetilde\Phi_m^L$ form an orthonormal basis for the
column span of $\Phi_m^L$. Also,
$$
\|G_L^{-1/2}-I_L\|_{\rm F}
\leq C\|\dcal E_L\|_{\rm F}
=O_\P\left(\frac1{\sqrt m\,\delta_m}\right),
\qquad
\|\Phi_m^L\|_{\mathrm{op}}^2=\|G_L\|_{\mathrm{op}}=O_\P(1).
$$

Now, let $H_m^{\mathrm{sp}}$ be the matrix \eqref{e:Hsp} defined in
the proof of Theorem~\ref{t:eigen_consistency_theo}. The bound \eqref{e:Hsp_to_tildeH}, together with \eqref{e:A_to_H}, gives
$$
\|\overline H_m-H_m^{\mathrm{sp}}\|_{\rm F}=O_\P(m^{-1/2}).
$$
By the definition of $Y_j$ and orthonormality of the eigenfunctions,
$\langle Y_j,\phi_\ell\rangle_{\cH}
=\sqrt{|\lambda_\ell|}\,\phi_\ell(\bX_j)$ and
$\E(Y_j\otimes Y_j)\phi_\ell=|\lambda_\ell|\phi_\ell$.
Consequently,
$B_j\phi_\ell
=\sqrt{|\lambda_\ell|}\,\phi_\ell(\bX_j)Y_j
-|\lambda_\ell|\phi_\ell$.
Using $J\phi_\ell=\text{sgn}(\lambda_\ell)\phi_\ell$, we obtain
$\langle Y_i,JB_j\phi_\ell\rangle_{\cH}
=\sqrt{|\lambda_\ell|}
\big(\phi_\ell(\bX_j)\langle Y_i,JY_j\rangle_{\cH}
-\lambda_\ell\phi_\ell(\bX_i)\big)$.
Averaging over $j$ and dividing by $\sqrt{m|\lambda_\ell|}$,
with $(H_m^{\mathrm{sp}})_{ij}
=m^{-1}\langle Y_i,JY_j\rangle_{\cH}$ and
$(\Phi_m^L)_{j\ell}=m^{-1/2}\phi_\ell(\bX_j)$, gives
$$
\big(H_m^{\mathrm{sp}}\Phi_m^L-\Phi_m^L\Lambda_L\big)_{i\ell}
=\frac1{\sqrt{m|\lambda_\ell|}}
\left\langle Y_i,
J\left(\frac1m\sum_{j=1}^m B_j\right)\phi_\ell\right\rangle_{\cH}.
$$
Consequently, using $\|J\|_{\mathrm{op}}\leq1$ and
$m^{-1}\sum_{i=1}^m\|Y_i\|_{\cH}^2=O_\P(1)$,
\begin{align*}
\|H_m^{\mathrm{sp}}\Phi_m^L-\Phi_m^L\Lambda_L\|_{\rm F}^2
&\leq\frac1{\delta_m}
\left(\frac1m\sum_{i=1}^m\|Y_i\|_{\cH}^2\right)
\left\|\frac1m\sum_{j=1}^m B_j\right\|_{\rm HS}^2\\
&=O_\P\left(\frac1{m\delta_m}\right).
\end{align*}
It follows that
\begin{align}
 &\|\overline H_m\Phi_m^L-\Phi_m^L\Lambda_L\|_{\rm F}\notag\\
 &\quad\leq\|H_m^{\mathrm{sp}}\Phi_m^L-\Phi_m^L\Lambda_L\|_{\rm F}
 +\|\overline H_m-H_m^{\mathrm{sp}}\|_{\rm F}
   \|\Phi_m^L\|_{\mathrm{op}}\notag\\
 &\quad=O_\P\left(\frac1{\sqrt m\,\delta_m^{1/2}}\right).
 \label{e:full_gram_eigen_residual}
\end{align}
By diagonalizing $G_L$ and the Lipschitz
continuity of $x\mapsto x^{-1/2}$ on $[1/2,3/2]$, we have that $\|\Lambda_LG_L^{-1/2}-G_L^{-1/2}\Lambda_L\|_{\rm F}
 \leq C\|\Lambda_LG_L-G_L\Lambda_L\|_{\rm F}$.  Hence,
\begin{align*}
 \|\Lambda_LG_L^{-1/2}-G_L^{-1/2}\Lambda_L\|_{\rm F}
 &\leq C\|\Lambda_LG_L-G_L\Lambda_L\|_{\rm F}\\
  & = C\|\Lambda_L\dcal  E_L-\dcal E_L\Lambda_L\|_{\rm F} \leq 2C\|\Lambda_L\dcal E_L\|_{\rm F}\\
 &\leq 2C\big\||\Lambda_L|^{1/2}\big\|_{\mathrm{op}}
 \big\||\Lambda_L|^{1/2}\dcal E_L|\Lambda_L|^{1/2}\big\|_{\rm F}
 \big\||\Lambda_L|^{-1/2}\big\|_{\mathrm{op}}\\
 &=O_\P\left(\frac1{\sqrt m\,\delta_m^{1/2}}\right).
\end{align*}
Hence, by \eqref{e:full_gram_eigen_residual} and the identity
\begin{align*}
\overline H_m\widetilde\Phi_m^L-\widetilde\Phi_m^L\Lambda_L
&=(\overline H_m\Phi_m^L-\Phi_m^L\Lambda_L)G_L^{-1/2}+\Phi_m^L(\Lambda_LG_L^{-1/2}-G_L^{-1/2}\Lambda_L),
\end{align*}
we obtain
\begin{equation}\label{e:orthonormal_eigen_residual}
\|\overline H_m\widetilde\Phi_m^L-\widetilde\Phi_m^L\Lambda_L\|_{\rm F}
=O_\P\left(\frac1{\sqrt m\,\delta_m^{1/2}}\right).
\end{equation}

Next, for any symmetric $M\in \R^{m\times m}$, let
$$
\mathcal P_{k}(M)\in \R^{m\times m}
$$
denote
the orthogonal projection in $\R^m$ onto the subspace spanned by the eigenvectors of $M$ whose eigenvalues in the region  $\mathcal I_{k,m}$ as defined in \eqref{e:Ik_bands}.   By Lemma \ref{l:sample_eig_concentration}, with probability tending to 1, these are precisely the  eigenvalues $\big\{\widehat \lambda_{\ell,m},\ell \in \mathcal B_k\}$ of $\overline H_m$ and $\big\{\breve \lambda_{\ell,m},\ell \in \mathcal B_k\}$ of $\breve H_m$.

Moreover, with probability tending to 1, with $d_k$ as in \eqref{e:def_rk_index},
$$
\mathrm{rank}(\cP_k(\overline H_m))=\mathrm{rank}(\cP_k(\breve H_m)) = d_k = \mathrm{rank}\big( \cP_k\big(\widetilde \Phi_m^{L}  \Lambda_L(\widetilde \Phi_m^{L})^\top\big)\big).
$$
 Note that the bands $\mathcal  I_{k,m}$  and $\mathcal  I_{k',m}$ with $k\neq k'$ are separated by a distance of at least $\delta_m/2$.   Moreover, by \eqref{e:def_vareps} and
$\varepsilon_m=o(\delta_m)$, there exists a permutation
$\pi_m\in\Pi_m$ such that
$\max_{1\leq\ell\leq m}
|\widehat\lambda_{\pi_m(\ell),m}-\lambda_\ell|
=o_\P(\delta_m)$.  Together with
Assumption~\ref{a:eigsep}, this implies that every empirical
eigenvalue outside $\mathcal I_{k,m}$ is at distance at least
$3\delta_m/4$ from $\mu_k$, simultaneously for all $k$,
with probability tending to one.   Indeed, with probability tending to one, all approximation errors
are smaller than $\delta_m/4$, so the population gap of $\delta_m$
ensures that empirical eigenvalues outside $\mathcal I_{k,m}$ remain
at distance at least $\delta_m-\delta_m/4=3\delta_m/4$ from $\mu_k$.%

Recall $V_{m}^L=(\mathbf v_{1,m}: \ldots : \mathbf v_{L,m})$ where $\mathbf v_{\ell,m}$ are the eigenvetors of $\overline H_m$ in the usual ordering.   Let $V_m^{(\mathcal I_{k,m})}$ denote the columns of $V_m^L$ corresponding
to eigenvalues in $\mathcal I_{k,m}$, and set
$$
\widetilde\Phi_m^{(\mu_k)}
=(\widetilde\Phi_{r_{k-1}+1,m}:\cdots:\widetilde\Phi_{r_k,m}),
\qquad
\Phi_m^{(\mu_k)}
=(\boldsymbol\Phi_{r_{k-1}+1,m}:\cdots:\boldsymbol\Phi_{r_k,m}).
$$
Since $\widetilde\Phi_m^L$ has orthonormal columns,
$$
\cP_k\big(\widetilde\Phi_m^L\Lambda_L(\widetilde\Phi_m^L)^\top\big)
=\widetilde\Phi_m^{(\mu_k)}(\widetilde\Phi_m^{(\mu_k)})^\top.
$$
Expanding $\widetilde\Phi_m^{(\mu_k)}$ in an orthonormal eigenbasis of
$\overline H_m$, and using the eigenvalue separation from $\mu_k$ for eigenvalues outside $\mathcal I_{k,m}$,  we obtain
\begin{align*}
\|(I_m-\cP_k(\overline H_m))\widetilde\Phi_m^{(\mu_k)}\|_{\rm F}^2
&= \sum_{\ell: \widehat \lambda_{\ell,m}\notin \mathcal I_{k,m}}\left\| \mathbf v_{\ell,m}^\top \widetilde\Phi_m^{(\mu_k)}\right\|_{2}^2\\
& \leq \Big(\frac{4}{3\delta_m}\Big)^2\sum_{\ell: \widehat \lambda_{\ell,m}\notin \mathcal I_{k,m}}|\widehat \lambda_{\ell,m}-\mu_k|^2\left\| \mathbf v_{\ell,m}^\top \widetilde\Phi_m^{(\mu_k)}\right\|_{2}^2\\
&\leq \Big(\frac{4}{3\delta_m}\Big)^2
\|(\overline H_m-\mu_k I_m)\widetilde\Phi_m^{(\mu_k)}\|_{\rm F}^2.
\end{align*}
Thus, using Lemma~\ref{l:basic_rotation_lemma} and
\eqref{e:orthonormal_eigen_residual},
\begin{align}
 \sum_{k=1}^b
 \big\|\cP_k(\overline H_m)
 -\cP_k\big(\widetilde\Phi_m^L\Lambda_L(\widetilde\Phi_m^L)^\top\big)
 \big\|_{\rm F}^2&
 =\sum_{k=1}^b
 \big\|\cP_k(\overline H_m)
 -\widetilde\Phi_m^{(\mu_k)}(\widetilde\Phi_m^{(\mu_k)})^\top
 \big\|_{\rm F}^2\notag\\
 &=2\sum_{k=1}^b
 \|(I_m-\cP_k(\overline H_m))\widetilde\Phi_m^{(\mu_k)}\|_{\rm F}^2
 \notag\\
 &\leq\frac{C}{\delta_m^2}\sum_{k=1}^b
 \|(\overline H_m-\mu_k I_m)\widetilde\Phi_m^{(\mu_k)}\|_{\rm F}^2
 \notag\\
 &=\frac{C}{\delta_m^2}
 \|\overline H_m\widetilde\Phi_m^L-\widetilde\Phi_m^L\Lambda_L\|_{\rm F}^2
 =O_\P\left(\frac1{m\delta_m^3}\right),
 \label{e:H_m_close_to_ortho}
\end{align}
where the second equality follows from \eqref{e:frob_id} since both projections have rank $d_k$, and the second-to-last equality follows by writing
$\overline H_m\widetilde\Phi_m^L-\widetilde\Phi_m^L\Lambda_L
=\big((\overline H_m-\mu_1I_m)\widetilde\Phi_m^{(\mu_1)}
:\cdots:
(\overline H_m-\mu_bI_m)\widetilde\Phi_m^{(\mu_b)}\big)$. 
From \eqref{e:H_m_close_to_ortho}, using
Lemma~\ref{l:basic_rotation_lemma}, for each $\mu_k$ we may find an
orthogonal $O_m^{(\mu_k)}\in\R^{d_k\times d_k}$ such that
\begin{align}
 &\sum_{k=1}^b
 \|V_m^{(\mathcal I_{k,m})}
      -\widetilde\Phi_m^{(\mu_k)}O_m^{(\mu_k)}\|_{\rm F}^2\notag\\
 &\quad\leq\sum_{k=1}^b
 \big\|\cP_k(\overline H_m)
 -\cP_k\big(\widetilde\Phi_m^L\Lambda_L(\widetilde\Phi_m^L)^\top\big)
 \big\|_{\rm F}^2
 =O_\P\left(\frac1{m\delta_m^3}\right).
 \label{e:unif_roation_bound}
\end{align}
On the other hand,
\begin{align}
 \sum_{k=1}^b
 \|\Phi_m^{(\mu_k)}-\widetilde\Phi_m^{(\mu_k)}\|_{\rm F}^2
 &=\|\Phi_m^L-\widetilde\Phi_m^L\|_{\rm F}^2\notag\\
 &=\|\Phi_m^L(I_L-G_L^{-1/2})\|_{\rm F}^2\notag\\
 &\leq\|\Phi_m^L\|_{\mathrm{op}}^2
       \|I_L-G_L^{-1/2}\|_{\rm F}^2\notag\\
 &\leq O_\P(1)\|\dcal E_L\|_{\rm F}^2
 =O_\P\left(\frac1{m\delta_m^2}\right).
 \label{e:unif_ortho_approx}
\end{align}
Combining \eqref{e:unif_roation_bound} and \eqref{e:unif_ortho_approx},
we obtain
\begin{align*}
 &\left(\sum_{k=1}^b
 \|V_m^{(\mathcal I_{k,m})}-\Phi_m^{(\mu_k)}O_m^{(\mu_k)}\|_{\rm F}^2
 \right)^{1/2}\\
 &\quad\leq
 \left(\sum_{k=1}^b
 \|V_m^{(\mathcal I_{k,m})}
      -\widetilde\Phi_m^{(\mu_k)}O_m^{(\mu_k)}\|_{\rm F}^2\right)^{1/2}
 +\left(\sum_{k=1}^b
 \|\Phi_m^{(\mu_k)}-\widetilde\Phi_m^{(\mu_k)}\|_{\rm F}^2\right)^{1/2}\\
 &\quad=O_\P\left(\frac1{\sqrt m\,\delta_m^{3/2}}\right)
        +O_\P\left(\frac1{\sqrt m\,\delta_m}\right)
 =O_\P\left(\frac1{\sqrt m\,\delta_m^{3/2}}\right),
\end{align*}
Finally, note by Lemma \ref{l:sample_eig_concentration}, with probability tending to 1, the columns of $V_{m}^{(\mathcal I_{k,m})}$ belong to the collection $\{\mathbf v_{\ell,m}~\ell\in\mathcal B_k\}=\{\mathbf v_{r_{k-1}+1,m},\ldots ,\mathbf v_{r_k,m}\}$.   Hence for each $k$ we may find a permutation matrix $Q_{m}^{(\mu_k)} \in \R^{d_k\times d_k}$ with
$$
(\mathbf v_{r_{k-1}+1,m}:\cdots : \mathbf v_{r_k,m})=V_m^{(\mathcal I_{k,m })}Q_{m}^{(\mu_k)},
$$
with probability tending to 1. 
Thus,  if we
define the block-orthogonal matrix
$$
O_{L,m}= \text{diag}\Big(O_m^{(\mu_1)}Q_{m}^{(\mu_1)},~ \ldots,~O_m^{(\mu_b)}Q_{m}^{(\mu_b)}  \Big)\in \R^{L\times L},
$$
 in view of \eqref{e:unif_roation_bound} and
  \eqref{e:unif_ortho_approx}, we have
\begin{align*}
\| V_m^L - \Phi_m^LO_{L,m}\|_{\rm F}^2 &=\sum_{k=1}^b\|V_m^{(\mathcal I_{k,m})}Q_{m}^{(\mu_k)} - \Phi_m^{(\mu_k)}O_{m}^{(\mu_k)}Q_{m}^{(\mu_k)}  \|_{\rm F}^2\\
& =\sum_{k=1}^b\|V_m^{(\mathcal I_{k,m})}- \Phi_m^{(\mu_k)}O_{m}^{(\mu_k)} \|_{\rm F}^2 =O_\P\left(\frac1{m\delta_m^3}\right).
\end{align*}
\end{proof}

The next lemma establishes an appropriate intermediate consistency result for estimation of the matrix analog of $\Sigma_L$ (as defined in \eqref{e:def_Sigma}) given by
\begin{equation}
[\Sigma]_L = \{\langle \phi_\ell, \Sigma \phi_{\ell'}\rangle_{\cH}:~ 1\leq \ell,\ell'\leq L\}.\notag
\end{equation}
We briefly recall
\begin{align}\label{e:def_LRC_estimator_supp}
\widehat{\cG}_{L,m}=\sum_{r=-(m-1)}^{m-1}K \left(\frac{r}{\mathfrak b_m}\right)\widehat{\bga}_{m}(r),
\end{align}
where
$$
  \widehat{\bga}_m(r)
  =\frac1m\sum_{i=1}^{m-r}
  \widehat{\bxi}_{L,i}\widehat{\bxi}_{L,i+r}^\top,
  \qquad r\geq0,
$$
with
$$
\widehat{\bxi }_{L,j}=\sqrt m (v_{1,m}(j),\ldots,v_{L,m}(j))^\top.
$$

\begin{lemma}\label{l:LRC_consistency} Suppose Assumptions \ref{a:eigsep}--\ref{ker2} are in place.
Let $\widehat {\cG}_{L,m}$ be as in \eqref{e:def_LRC_estimator_supp}, and write
\begin{align*}
\widehat \Lambda_L=\mathrm{diag}\,(\widehat \lambda_{1,m},\ldots,\widehat \lambda_{L,m}).
\end{align*}
Set
\begin{align}\label{e:LRC_estimator}
[\widehat \Sigma]_{L,m}=  |\widehat \Lambda_L|^{1/2}  \widehat {\cG}_{L,m}|\widehat \Lambda_L|^{1/2}.
\end{align}
Then, with $ O_{L,m}$ as in Theorem \ref{t:spectral_consistency}, 
$$
\| [\widehat \Sigma]_{L,m}- O_{L,m}^\top[\Sigma]_LO_{L,m}\|_{\rm F} =o_\P(1).
$$
\end{lemma}
\begin{proof} 
Let%
$$
\widetilde {\bga }_{m}(r) =\frac1{m}\sum_{i=1}^{m-r}   {\bxi }_{L,i} {\bxi }_{L,i+r}^\top,\qquad \boldsymbol \xi_{L,j}=(\phi_1(\bX_j),\ldots,\phi_L(\bX_j))^\top.
$$
and let $\widetilde {\bga }_{m}(-r)=\widetilde {\bga }_{m}(r)^\top$.
For simplicity write $O= O_{L,m}$.  We first show
\begin{equation}\label{e:LRC_estim_Frob_bound}
\| |\Lambda_L|^{1/2}  \widehat {\cG}_{L,m}|\Lambda_L|^{1/2} -O^\top [\Sigma]_L O\|_{\rm F} = o_\P(1).
\end{equation}
Note
\begin{align*}
\widehat{\bga}_m(r)-O^\top\widetilde\bga_{m}(r)O &=\frac1m \sum_{i=1}^{m-r}   (\widehat{\bxi }_{L,i} -O^\top\bxi_{L,i}) \big( \widehat{\bxi }_{L,i+r})^\top +\frac1m \sum_{i=1}^{m-r} (O^\top\bxi_{L,i})(\widehat \bxi_{L,i+r} -O^\top\bxi_{L,i+r})^\top.
\end{align*}
Also note that
$$
(V_m^L - \Phi_m^L O)^\top = m^{-1/2}(\widehat \bxi_{L,1}-O^\top \bxi_{L,1}:\cdots  :\widehat \bxi_{L,m}-O^\top \bxi_{L,m}) \in \R^{L\times m}.
$$
Hence, with 
$$
M_1=m^{-1/2}(\widehat \bxi_{L,1}-O^\top \bxi_{L,1}:\cdots  :\widehat \bxi_{L,m-r}-O^\top \bxi_{L,m-r}),\qquad M_2 = m^{-1/2}(\widehat \bxi_{L,1+r},\ldots,\widehat \bxi_{L,m})
$$
we clearly have $\|M_1\|_{\rm F}\leq \|V_m^L-\Phi_m^LO\|_{\rm F}$ and $\| M_2\|_{\mathrm{op}}\leq\|V_m^L\|_{\mathrm{op}}$, giving %
$$
\Bigg\|\frac1m \sum_{i=1}^{m-r}   (\widehat{\bxi }_{L,i} -O^\top\bxi_{L,i}) \big( \widehat{\bxi }_{L,i+r})^\top\Bigg\|_{\rm F} = \| M_1 M_2^\top\|_{\rm F} \leq \|M_1\|_{\rm F}\|M_2\|_{\mathrm{op}} \leq\|V_m^L-\Phi_m^LO\|_{\rm F}\|V_m^L\|_{\mathrm{op}}.
$$
Analogously,  
$$
\Bigg\|\frac1m\sum_{i=1}^{m-r} (O^\top\bxi_{L,i})(\widehat \bxi_{L,i+r} -O^\top\bxi_{L,i+r})^\top\Bigg\|_{\rm F} \leq\|V_m^L-\Phi_m^LO\|_{\rm F}\|\Phi_m^L\|_{\mathrm{op}}.%
$$
Hence, by Theorem \ref{t:spectral_consistency}, 
\begin{align*}
\max_{|r|\leq m-1}\|\widehat{\bga}_m(r)-O^\top\widetilde \bga_{m}(r)O \|_{\rm F} &\leq\|V_m^L-\Phi_m^LO\|_{\rm F}(\|V_m^L\|_\mathrm{op}+\|\Phi_m^L\|_{\mathrm{op}})\\
& =O_\P\left(\frac1{\sqrt{m}\,\delta_m^{3/2}}\right).
\end{align*}
Next, with
\begin{align}\label{e:def_widtildeG}
 \widetilde{\mathcal G}_{L,m}=\sum_{r=-(m-1)}^{m-1}K \left(\frac{r}{\mathfrak b_m}\right)\widetilde{\bga}_{m}(r),
\end{align}
 we have
\begin{align}
\notag\|\widehat{\mathcal G}_{L,m}- O^\top \widetilde{\mathcal  G}_{L,m}O \|_{\rm F}& = \left\|\sum_{r=-(m-1)}^{m-1}K \left(\frac{r}{\mathfrak b_m}\right)\left(\widehat{\bga}_m(r)-O^\top\widetilde\bga_{m}(r)O \right)\right\|_{\rm F}\\
\notag& \leq C\mathfrak b_m \sup_u|K(u)| \max_{|r|\leq m-1}\|\widehat{\bga}_m(r)-O^\top\widetilde\bga_{m}(r)O \|_{\rm F} \\
& =O_\P\left(\frac{\mathfrak b_m}{\sqrt m\delta_m^{3/2}}\right) =o_\P(1).\label{e:G_hat_rotation_G}
\end{align}
Now, with $\bga_L(r)=\E\bxi_{L,0}\bxi_{L,r}^\top$, we have $\mathcal G_L = \sum_{r\in \mathbb Z} \bga_L(r)$, so that

\begin{align}
\widetilde{\mathcal G}_{L,m}- \mathcal G_L &= \sum_{r=-(m-1)}^{m-1}K \left(\frac{r}{\mathfrak b_m}\right)(\widetilde \bga_m(r) -\bga_{L}(r)) + \sum_{r=-(m-1)}^{m-1}\left(K \left(\frac{r}{\mathfrak b_m}\right)-1\right) \bga_{L}(r) - \sum_{|r|\geq m}\bga_L(r)\notag\\
& = T_1 + T_2 + T_3.\notag
\end{align}
We now show
\begin{align}\label{e:Ti_123}
\| |\Lambda_L|^{1/2}T_i|\Lambda_L|^{1/2}\|_{\rm F} = o_\P(1),\qquad i=1,2,3.
\end{align}
First, note with 
$$
B_{i,r}=Y_i\otimes Y_{i+r}- \E (Y_i\otimes Y_{i+r}),\
$$
we have
\begin{align}
&\||\Lambda_L|^{1/2}(\widetilde \bga_m(r) -\bga_L(r))|\Lambda_L|^{1/2}\|_{\rm F}\notag\\
&\notag \quad \leq\Bigg\| \frac1{m}\sum_{i=1}^{m-r}  |\Lambda_L|^{1/2}\big( {\bxi }_{L,i} {\bxi }_{L,i+r}^\top-\E\bxi_{L,0}\bxi_{L,r}^\top\big)|\Lambda_L|^{1/2}\Bigg\|_{\rm F}\notag + \frac{r}m\left\||\Lambda_L|^{1/2}\bga_L(r)|\Lambda_L|^{1/2} \right\|_{\rm F}\\
&\quad =\Bigg\|\frac1m\sum_{i=1}^{m-r}P_L B_{i,r}P_L\Bigg\|_{\mathrm{HS}} + \frac{r}m\left\||\Lambda_L|^{1/2}\bga_L(r)|\Lambda_L|^{1/2} \right\|_{\rm F},\label{e:acf_difference_bound}
\end{align}
where $P_L$ denotes the projection in $\cH$ onto $\text{span} \{\phi_1,\ldots,\phi_L\}.$
For $n\geq 1$, let  $Y^{(n)}_j =\sum_{\ell=1}^\infty\sqrt{|\lambda_\ell|} \phi_\ell(\mathbf X_j^{(n)})\phi_\ell$, and for 
$r\geq0$, write
\begin{align*}
 B_{i,r}^{(n)}
 &=Y_i^{(n)}\otimes Y_{i+r}^{(n)}
   -\E\big(Y_i^{(n)}\otimes Y_{i+r}^{(n)}\big),\\
 a_n&=\big(\E\|Y_1-Y_1^{(n)}\|_{\cH}^4\big)^{1/4}.
\end{align*}
Since $p>4$, Cauchy--Schwarz, Lemma~\ref{l:dep_transfer}(i), and
Assumption~\ref{a:Lpm} give
\begin{gather*}
 \sup_{r\geq0}
 \big(\E\|B_{1,r}-B_{1,r}^{(n)}\|_{\mathrm{HS}}^2\big)^{1/2}
 \leq Ca_n,\\
 \sum_{n=1}^\infty a_n^\eta
 \leq C\sum_{n=1}^\infty\vartheta_{p\alpha}(n)^{\eta\alpha}
 <\infty.
\end{gather*}
For each $j\geq0$, we may choose an integer $n_j\in[2^j,2^{j+1})$ such that $a_{n_j}=\min\{a_n:~n\in [2^j,2^{j+1})\}$, giving
$$
 a_{n_j}^\eta
 \leq 2^{-j}\sum_{n=2^j}^{2^{j+1}-1}a_n^\eta
 \leq C2^{-j}.
$$
Then $n_0=1$, $a_{n_j}\leq C2^{-j/\eta}$, and
$B_{i,r}^{(n_j)}\to B_{i,r}$ in $\mathcal L^2$. Moreover, for every $N,K\geq1$, 
\begin{equation}\label{e:telescope_Bnj}
 \sum_{i=1}^N B_{i,r}^{(n_K)}
 =\sum_{i=1}^N B_{i,r}^{(1)}
  +\sum_{j=1}^K\sum_{i=1}^N
    \big(B_{i,r}^{(n_j)}-B_{i,r}^{(n_{j-1})}\big).
\end{equation}
We bound the initial sum and each inner sum in the double sum separately in \eqref{e:telescope_Bnj}.  
By the construction in Assumption~\ref{a:Lpm}, for $j\geq1$, the pair
$(B_{i,r}^{(n_j)},B_{i,r}^{(n_{j-1})})$ depends only on the original
innovations driving $\{\mathbf X_j\}$ with indices in the blocks
$\{i-n_j+1,\ldots,i\}\cup\{i+r-n_j+1,\ldots,i+r\}$, (each of length $n_j$) 
together with the replacement innovations
$\{\varepsilon_t^{(i)}:t\leq i-n_{j-1}\}$ and
$\{\varepsilon_t^{(i+r)}:t\leq i+r-n_{j-1}\}$. Clearly, the pair $(B_{k,r}^{(n_j)},B_{k,r}^{(n_{j-1})})$ uses the analogous
blocks and replacement innovations with $i$ replaced by $k$; and observe that integer index sets $\{a-n_j+1,\ldots,a\}$ and
$\{b-n_j+1,\ldots,b\}$ are disjoint whenever $|a-b|\geq n_j$. Hence, if
\begin{equation}\label{e:n_j_gap}
 \min\{|i-k|,|i-k-r|,|i-k+r|\}\geq n_j.
\end{equation}
the original innovation blocks used by $(B_{i,r}^{(n_j)},B_{i,r}^{(n_{j-1})})$ are disjoint from those used by $(B_{k,r}^{(n_j)},B_{k,r}^{(n_{j-1})})$, and their replacement sequences have distinct indices.  Hence when \eqref{e:n_j_gap} is satisfied,
\begin{align*}
 \E\left\langle
 B_{i,r}^{(n_j)}-B_{i,r}^{(n_{j-1})},
 B_{k,r}^{(n_j)}-B_{k,r}^{(n_{j-1})}
 \right\rangle_{\mathrm{HS}}=\left\langle
 \E\big(B_{i,r}^{(n_j)}-B_{i,r}^{(n_{j-1})}\big),
 \E\big(B_{k,r}^{(n_j)}-B_{k,r}^{(n_{j-1})}\big)
 \right\rangle_{\mathrm{HS}}=0.
\end{align*}
For the remaining terms, Cauchy--Schwarz and stationarity give
\begin{align*}
 &\left|\E\left\langle
 B_{i,r}^{(n_j)}-B_{i,r}^{(n_{j-1})},
 B_{k,r}^{(n_j)}-B_{k,r}^{(n_{j-1})}
 \right\rangle_{\mathrm{HS}}\right|\\
 &\quad\leq
 \big(\E\|B_{i,r}^{(n_j)}-B_{i,r}^{(n_{j-1})}\|_{\mathrm{HS}}^2\big)^{1/2}
 \big(\E\|B_{k,r}^{(n_j)}-B_{k,r}^{(n_{j-1})}\|_{\mathrm{HS}}^2\big)^{1/2}\\
 &\quad=\E\|B_{1,r}^{(n_j)}-B_{1,r}^{(n_{j-1})}\|_{\mathrm{HS}}^2.
\end{align*}
For each $i$, at most $6n_j$ indices $k$ fail to satisfy \eqref{e:n_j_gap}.  
Consequently,
\begin{align*}
 \E\Bigg\|\sum_{i=1}^N
       (B_{i,r}^{(n_j)}-B_{i,r}^{(n_{j-1})})\Bigg\|_{\mathrm{HS}}^2
 &=\sum_{i,k=1}^N\E\left\langle
 B_{i,r}^{(n_j)}-B_{i,r}^{(n_{j-1})},
 B_{k,r}^{(n_j)}-B_{k,r}^{(n_{j-1})}
 \right\rangle_{\mathrm{HS}}\\
 &\leq 6n_jN\E\|B_{1,r}^{(n_j)}-B_{1,r}^{(n_{j-1})}\|_{\mathrm{HS}}^2\\
 &\leq Cn_jN(a_{n_j}+a_{n_{j-1}})^2,
\end{align*}
uniformly in $r$. The same argument, using
$\sup_{r\geq0}\E\|B_{1,r}^{(1)}\|_{\mathrm{HS}}^2\leq C$, gives
$$
 \E\Bigg\|\sum_{i=1}^N B_{i,r}^{(1)}\Bigg\|_{\mathrm{HS}}^2\leq CN.
$$
Applying Minkowski's inequality to \eqref{e:telescope_Bnj} and letting
$K\to\infty$, we obtain
\begin{align*}
 \left(\E\Bigg\|\sum_{i=1}^N B_{i,r}\Bigg\|_{\mathrm{HS}}^2\right)^{1/2}
 &\leq C\sqrt N\bigg(1+\sum_{j=1}^\infty
       \sqrt{n_j}\,(a_{n_j}+a_{n_{j-1}})\bigg)\\
 &\leq C\sqrt N\bigg(1+\sum_{j=1}^\infty
       2^{j(1/2-1/\eta)}\bigg)
 \leq C\sqrt N,
\end{align*}
since $0<\eta<1$. Taking $N=m-r$, we conclude
$$
 \sup_{0\leq r<m}\E\Bigg\|
       \frac1m\sum_{i=1}^{m-r}B_{i,r}\Bigg\|_{\mathrm{HS}}^2
 \leq \frac Cm.
$$

From \eqref{e:acf_difference_bound} we conclude
\begin{align*}
\E\||\Lambda_L|^{1/2}T_1|\Lambda_L|^{1/2}\|_{\rm F}
&\leq \sup_u|K(u)| \sum_{|r|\leq C\mathfrak b_m}
\E\||\Lambda_L|^{1/2}
(\widetilde\bga_m(r)-\bga_L(r))
|\Lambda_L|^{1/2}\|_{\rm F}\\
&\leq C\mathfrak b_m\sup_{0\leq r<m}
\E\Bigg\|\frac1m\sum_{i=1}^{m-r}
P_LB_{i,r}P_L\Bigg\|_{\mathrm{HS}}
+\frac{C\mathfrak b_m}{m}
\sum_{r\in\mathbb Z}\|\Cov(Y_0,Y_r)\|_{\mathrm{HS}}\\
&\leq C\mathfrak b_m\sup_{0\leq r<m}
\left(\E\Bigg\|\frac1m\sum_{i=1}^{m-r}
B_{i,r}\Bigg\|_{\mathrm{HS}}^2\right)^{1/2}
+O(\mathfrak b_m/m)\\
&=O(\mathfrak b_m/m^{1/2})
+O(\mathfrak b_m/m)=o(1).
\end{align*}
Next, for $T_2$, since $\lim_{u\to 0} K(u)= 1$ and %
$$
\sum_{r\in \mathbb Z} \Big\||\Lambda_L|^{1/2}\bga_L(r)|\Lambda_L|^{1/2}\Big\|_{\rm F} = \sum_{r\in \mathbb Z} \Big\|P_L\E (Y_0 \otimes Y_{r})P_L\Big\|_\mathrm{HS} \leq \sum_{r\in \mathbb Z} \Big\|\Cov(Y_0 , Y_{r})\Big\|_\mathrm{HS} <\infty,
$$
dominated convergence gives
$$
\E\||\Lambda_L|^{1/2}T_2|\Lambda_L|^{1/2} \|_{\rm F} \leq \sum_{r\in\mathbb Z} |K(r/\mathfrak b_m)-1|\|  |\Lambda_L|^{1/2}\bga_L(r)|\Lambda_L|^{1/2} \|_{\rm F}\to 0.
$$
Similarly, we have
$$
\E\||\Lambda_L|^{1/2}T_3|\Lambda_L|^{1/2} \|_{\rm F} \leq \sum_{|r|\geq m} \|\Cov(Y_0,Y_r)\|_{\mathrm{HS}}\to 0.
$$
Hence, \eqref{e:Ti_123} is established, from which we deduce
\begin{equation}\label{e:Gtilde_conv}
\||\Lambda_L|^{1/2}\widetilde{\mathcal G}_{L,m}|\Lambda_L|^{1/2}- [\Sigma]_L\|_{\rm F} = o_\P(1).
\end{equation}
Since $\Lambda_LO=O\Lambda_L$, putting \eqref{e:G_hat_rotation_G} and \eqref{e:Gtilde_conv} together, we obtain \eqref{e:LRC_estim_Frob_bound}.  
Finally, Lemma~\ref{l:sample_eig_concentration} and
\eqref{e:pi_m_identity_upto_L} give
$$
\|\widehat\Lambda_L-\Lambda_L\|_{\mathrm{op}}
=O_\P(\varepsilon_m).
$$
Since $|\lambda_L|\geq\delta_m$,
$$
\begin{aligned}
\big\|
|\widehat\Lambda_L|^{1/2}|\Lambda_L|^{-1/2}
-I_L
\big\|_{\mathrm{op}}
&=
\max_{1\leq\ell\leq L}
\frac{
\big|{|\widehat\lambda_{\ell,m}|^{1/2}}
- |\lambda_\ell|^{1/2}\big|}
{ |\lambda_\ell|^{1/2} }\\
&\leq
\max_{1\leq\ell\leq L}
\frac{|\widehat\lambda_{\ell,m}-\lambda_\ell|}
{|\lambda_\ell|}\\
&\leq
\frac{\|\widehat\Lambda_L-\Lambda_L\|_{\mathrm{op}}}
{|\lambda_L|}
=
O_\P\left(\frac{\varepsilon_m}{\delta_m}\right)
=o_\P(1).
\end{aligned}
$$
Moreover, \eqref{e:LRC_estim_Frob_bound} implies
$$
\big\|
|\Lambda_L|^{1/2}\widehat{\cG}_{L,m}
|\Lambda_L|^{1/2}
\big\|_{\rm F}
=O_\P(1),
$$
since
$
\|O^\top[\Sigma]_LO\|_{\rm F} =\|[\Sigma]_L\|_{\rm F} \leq\|\Sigma\|_{\mathrm{HS}}.$
Therefore, using the identity
\begin{align*}
[\widehat\Sigma]_{L,m}
&=|\widehat\Lambda_L|^{1/2}\widehat{\cG}_{L,m}
|\widehat\Lambda_L|^{1/2}\\
&=
\big(|\widehat\Lambda_L|^{1/2}|\Lambda_L|^{-1/2}\big)
\big(|\Lambda_L|^{1/2}\widehat{\cG}_{L,m}|\Lambda_L|^{1/2}\big)
\big(|\widehat\Lambda_L|^{1/2}|\Lambda_L|^{-1/2}\big),
\end{align*}
we get
\begin{align*}
&\big\|
[\widehat\Sigma]_{L,m}-|\Lambda_L|^{1/2}\widehat{\cG}_{L,m} |\Lambda_L|^{1/2}
\big\|_{\rm F}\\
&\quad\leq
\big\|
|\widehat\Lambda_L|^{1/2}|\Lambda_L|^{-1/2}
-I_L \big\|_{\mathrm{op}} \big\|
|\Lambda_L|^{1/2}\widehat{\cG}_{L,m} |\Lambda_L|^{1/2}
\big\|_{\rm F} \left( \big\| |\widehat\Lambda_L|^{1/2}|\Lambda_L|^{-1/2} \big\|_{\mathrm{op}}+1\right)
=o_\P(1).
\end{align*}
Combining this with \eqref{e:LRC_estim_Frob_bound} proves the result.

\end{proof}

The next lemma establishes convergence of $\tr \widehat \Sigma_{L,m}$, which is needed in the proof of Theorem \ref{t:LRC_sim_method}.
\begin{lemma}\label{l:trace_hatSigma}
Suppose the hypotheses of Lemma \ref{l:LRC_consistency} are in force, and let $L=L_m\to\infty$. In addition, assume
\begin{equation}\label{e:extra_trace_condition}
\mathfrak b_m\sqrt{L}\,\varepsilon_m \to 0,
\end{equation}
where $\varepsilon_m$ is given by \eqref{e:def_vareps}. Then
\begin{equation}\label{e:trace_hatSigma_to_traceSigma}
\tr([\widehat \Sigma]_{L,m}) \stackrel\P \to \tr(\Sigma).
\end{equation}
\end{lemma}

\begin{proof}
Set
$$
[\widetilde \Sigma]_{L,m}:=|\Lambda_L|^{1/2}\widetilde{\mathcal G}_{L,m}|\Lambda_L|^{1/2},
$$
where $\widetilde{\mathcal G}_{L,m}$ is given in \eqref{e:def_widtildeG}.
We first show
\begin{equation}\label{e:hatSigma_to_tildeSigma_trace}
\tr([\widehat \Sigma]_{L,m})-\tr([\widetilde \Sigma]_{L,m})=o_\P(1).
\end{equation}
Write
$$
E_m:=\widehat{\mathcal G}_{L,m}-O_{L,m}^\top \widetilde{\mathcal G}_{L,m}O_{L,m}.
$$
Then
\begin{align*}
\tr([\widehat \Sigma]_{L,m})-\tr([\widetilde \Sigma]_{L,m})
&=
\tr\big((|\widehat\Lambda_L|-|\Lambda_L|)\widehat{\mathcal G}_{L,m}\big)
+
\tr\big(|\Lambda_L|E_m\big),
\end{align*}
since $O_{L,m}|\Lambda_L|O_{L,m}^\top=|\Lambda_L|$. For the second term,
\begin{equation}\label{e:weighted_trace_term}
|\tr(|\Lambda_L|E_m)|
\leq
\||\Lambda_L|\|_{\rm F}\,\|E_m\|_{\rm F} =o_\P(1),
\end{equation}
since $\sum_{\ell\geq1}\lambda_\ell^2<\infty$ and Lemma \ref{l:LRC_consistency} yields
$
\|E_m\|_{\rm F}=o_\P(1).$
For the first term, we have
$$
\big|\tr\big((|\widehat\Lambda_L|-|\Lambda_L|)\widehat{\mathcal G}_{L,m}\big)\big|
\leq
\big\||\widehat\Lambda_L|-|\Lambda_L|\big\|_{\mathrm{tr}}
\|\widehat{\mathcal G}_{L,m}\|_{\mathrm{op}}.
$$
Further, by the proof of Lemma \ref{l:sample_eig_concentration} (see \eqref{e:pi_m_identity_upto_L}) and  Theorem \ref{t:eigen_consistency_theo}, we have
\begin{equation}\label{e:l2_topL_bound}
\left(\sum_{\ell=1}^L(\widehat \lambda_{\ell,m}-\lambda_\ell)^2\right)^{1/2}
=
O_\P(\varepsilon_m).
\end{equation}
Moreover, by compact support and boundedness of $K$,
$$
\|\widehat{\mathcal G}_{L,m}\|_{\mathrm{op}}=O_\P(\mathfrak b_m).
$$
Therefore
$$
\tr\big((|\widehat\Lambda_L|-|\Lambda_L|)\widehat{\mathcal G}_{L,m}\big)
=
O_\P(\mathfrak b_m\sqrt{L}\,\varepsilon_m)
=
o_\P(1)
$$
by \eqref{e:extra_trace_condition}. Together with \eqref{e:weighted_trace_term}, this proves \eqref{e:hatSigma_to_tildeSigma_trace}.
It remains to show
\begin{equation}\label{e:tildeSigma_trace_to_traceSigma}
\tr([\widetilde \Sigma]_{L,m})\stackrel \P \to \tr(\Sigma).
\end{equation}
Let
$
Y_{L,j}$ be as in \eqref{e:def_YjL}.  Then
$$
\tr\!\left(|\Lambda_L|^{1/2}\widetilde\bga_m(r)|\Lambda_L|^{1/2}\right)
=
\frac1m\sum_{i=1}^{m-r}\langle Y_{L,i},Y_{L,i+r}\rangle_{\cH},
$$
and
$$
\tr\!\left(|\Lambda_L|^{1/2}\bga_L(r)|\Lambda_L|^{1/2}\right)
=
\E\langle Y_{L,0},Y_{L,r}\rangle_{\cH}.
$$
Hence, writing
$$
\widehat a_{L,m}(r)
:=
\frac1m\sum_{i=1}^{m-r}\langle Y_{L,i},Y_{L,i+r}\rangle_{\cH},
\qquad
a_L(r):=\E\langle Y_{L,0},Y_{L,r}\rangle_{\cH},
$$
(and for negative lags, set
  $\widehat a_{L,m}(-r)=\widehat a_{L,m}(r)$) we have
$$
\tr([\widetilde \Sigma]_{L,m})
=
\sum_{r=-(m-1)}^{m-1}
K\!\left(\frac{r}{\mathfrak b_m}\right)\widehat a_{L,m}(r),
\qquad
\tr(\Sigma_L)=\sum_{r\in\mathbb Z} a_L(r).
$$
Thus
$$
\tr([\widetilde \Sigma]_{L,m})-\tr(\Sigma_L)=S_{1,m}+S_{2,m}+S_{3,m},
$$
where
$$
S_{1,m}
=
\sum_{r=-(m-1)}^{m-1}
K\!\left(\frac{r}{\mathfrak b_m}\right)(\widehat a_{L,m}(r)-a_L(r)),
$$
$$
S_{2,m}
=
\sum_{r=-(m-1)}^{m-1}
\left(K\!\left(\frac{r}{\mathfrak b_m}\right)-1\right)a_L(r),
\qquad
S_{3,m}
=
-\sum_{|r|\ge m} a_L(r).
$$

Since
$$
a_L(r)=\tr\big(P_L\Cov(Y_0,Y_r)P_L\big),
$$
we have
$$
|a_L(r)|
\leq
\|\Cov(Y_0,Y_r)\|_{\tr},
$$
and by the standing assumptions
$$
\sum_{r\in\mathbb Z}\|\Cov(Y_0,Y_r)\|_{\tr}<\infty.
$$
Hence $\sum_{r\in\mathbb Z}|a_L(r)|<\infty$ uniformly in $L$, so dominated convergence yields
$$
S_{2,m}\to 0,
\qquad
S_{3,m}\to 0.
$$

For $S_{1,m}$, write
$$
\widehat a_{L,m}(r)-a_L(r)
=
\frac1m\sum_{i=1}^{m-r} Z_i^{(L,r)}
-\frac rm a_L(r),
\qquad
Z_i^{(L,r)}
:=
\langle Y_{L,i},Y_{L,i+r}\rangle_{\cH}
-
\E\langle Y_{L,0},Y_{L,r}\rangle_{\cH}.
$$
Since $Y_{L,i}$ is the truncated counterpart of $Y_i$, we may use the same $r$-dependent coupling argument used in the proof of Lemma \ref{l:LRC_consistency} to obtain%
$$
\sup_{L\geq 1}\sup_{0\leq r <m}
\E\left|
\frac1m\sum_{i=1}^{m-r} Z_i^{(L,r)}
\right|
\lesssim
m^{-1/2}.
$$
Using the compact support of $K$, we obtain
$$
\E|S_{1,m}|
\lesssim
\frac{\mathfrak b_m}{\sqrt m}\to 0.
$$
Therefore $S_{1,m}\to 0$ in probability, and hence
$$
\tr([\widetilde \Sigma]_{L,m})-\tr(\Sigma_L)\stackrel\P\to 0.
$$
Finally, since $\Sigma$ is trace class and $\Sigma_L=P_L\Sigma P_L$,
$$
\tr(\Sigma_L)\to \tr(\Sigma).
$$
Combining this with \eqref{e:hatSigma_to_tildeSigma_trace} proves \eqref{e:trace_hatSigma_to_traceSigma}.
\end{proof}

\begin{proof}[Proof of Theorem \ref{t:LRC_sim_method}] 

Let $c>0$ be fixed, and let $\mathbb N'\subseteq \mathbb N$ be arbitrary. We will show that there is a further subsequence $\mathbb N''\subseteq \mathbb N'$ along which
$$
\P_\infty\left\{ \sup_{0\leq t <\infty}\widehat \Gamma_{m''}(t)\geq c \right\}
\to
\P\left\{ \sup_{0\leq t <\infty}\Gamma(t)\geq c \right\},
\qquad \text{a.s.}
$$
Since $\mathbb N'$ is arbitrary, convergence in probability will then follow by the subsequences principle.

Let $ \mathcal T_L:\R^L\to \cH$
be the coordinate map defined by
$$
\mathcal T_L \mathbf e_\ell = \phi_\ell,\qquad \ell=1,\ldots,L,
$$
so that $\mathcal T_L\mathcal T_L^*=P_L$ and $\mathcal T_L^*\mathcal T_L=I_L$.  Recalling the process $\widehat \bZ_{L,m}(t)$, define
$$
\widehat G_{L,m}(t)
=
\sum_{\ell=1}^L \widehat Z_{\ell,L,m}(t)\phi_\ell
=
\mathcal T_L\widehat{\mathbf Z}_{L,m}(t),
\qquad t\geq 0.
$$
Under $\P_\infty$, $\widehat G_{L,m}$ is centered Gaussian with covariance
$$
\E_\infty\big[\widehat G_{L,m}(s)\otimes \widehat G_{L,m}(t)\big]
=
(s\wedge t)\mathcal T_L\widehat\Omega_{L,m}\mathcal T_L^* .
$$
Now let
$$
U_{L,m}
=
\mathcal T_L O_{L,m}\mathcal T_L^*+(I-P_L),
$$
where $O_{L,m}$ is as in Theorem \ref{t:spectral_consistency}. Since $O_{L,m}$ is orthogonal, $U_{L,m}$ is unitary on $\cH$. By Lemma~\ref{l:LRC_consistency},
$$
\left\|
|\widehat\Lambda_{L,m}|^{1/2}
\widehat{\mathcal G}_{L,m}
|\widehat\Lambda_{L,m}|^{1/2}
-
O_{L,m}^{\top}[\Sigma]_L O_{L,m}
\right\|_{\rm F}
=o_\P(1).
$$
We claim:
\begin{equation}\label{e:Omega_hat_frob_consistency}
\left\|
\widehat\Omega_{L,m}
-
O_{L,m}^{\top}[\Sigma]_L O_{L,m}
\right\|_{\rm F}
=o_\P(1).
\end{equation}

By Lemma~\ref{l:LRC_consistency},
$$
\left\|
[\widehat\Sigma]_{L,m}
-
O_{L,m}^{\top}[\Sigma]_L O_{L,m}
\right\|_{\rm F}
=o_\P(1).
$$
Moreover, Lemma~\ref{l:trace_hatSigma} gives
$$
\max\{\tr([\widehat\Sigma]_{L,m}),0\}
-\tr([\Sigma]_L)
=o_\P(1).
$$
Suppose first that $\tr(\Sigma)>0$. Then $\tr([\Sigma]_L)>0$ for all
sufficiently large $m$, and
$$
\frac{\max\{\tr([\widehat\Sigma]_{L,m}),0\}}
{\tr([\Sigma]_L)}
O_{L,m}^{\top}[\Sigma]_L O_{L,m}
$$
is feasible in the definition of $\widehat\Omega_{L,m}$. Hence, 
\begin{align*}
&
\left\|
\widehat\Omega_{L,m}
-
O_{L,m}^{\top}[\Sigma]_L O_{L,m}
\right\|_{\rm F}
\\
&\quad\leq
\frac{
\left|
\max\{\tr([\widehat\Sigma]_{L,m}),0\}
-\tr([\Sigma]_L)
\right|
}
{\tr([\Sigma]_L)}
\|[\Sigma]_L\|_{\rm F}
+
2\left\|
[\widehat\Sigma]_{L,m}
-
O_{L,m}^{\top}[\Sigma]_L O_{L,m}
\right\|_{\rm F}
=o_\P(1).
\end{align*}
If $\tr(\Sigma)=0$, positive semidefiniteness of $\Sigma$ implies
$\Sigma=0$, and therefore \eqref{e:Omega_hat_frob_consistency} still holds.
Also note
\begin{align*}
&\left\|
U_{L,m}\mathcal T_L\widehat\Omega_{L,m}\mathcal T_L^*U_{L,m}^*
-\Sigma
\right\|_{\rm HS}
\\
&\quad\leq
\left\|
U_{L,m}\mathcal T_L\widehat\Omega_{L,m}\mathcal T_L^*U_{L,m}^*
-\Sigma_L
\right\|_{\rm HS}
+
\|\Sigma_L-\Sigma\|_{\rm HS}
\\
&\quad=
\left\|
\widehat\Omega_{L,m}
-
O_{L,m}^{\top}[\Sigma]_L O_{L,m}
\right\|_{\rm F}
+o(1)
=o_\P(1).
\end{align*}
We may therefore extract subsequence $\mathbb N''\subseteq \mathbb N'$, along which
\begin{equation}\label{e:cov_subseq_as}
\left\|
U_{L,m''}\mathcal T_L\widehat\Omega_{L,m''}\mathcal T_L^*U_{L,m''}^*
-\Sigma
\right\|_{\rm HS}\to 0.
\end{equation}
and, in view of Lemma \ref{l:trace_hatSigma}, refining the subsequence if necessary, we may also assume
\begin{equation}\label{e:trace_subseq_as}
\tr(\widehat\Omega_{L,m''})
=
\tr(\mathcal T_L\widehat\Omega_{L,m''}\mathcal T_L^*)
=
\max\{\tr([\widehat\Sigma]_{L,m''}),0\}\to \tr(\Sigma)
\qquad \text{a.s.}
\end{equation}
Since $U_{L,m}\widehat G_{L,m}$ is a (finite-rank) $\cH$-valued Brownian motion under $\P_\infty$, it has paths in $C^\gamma_0([0,1+T];\cH)$ almost surely for every $\gamma<1/2$.  For each $m$, let $\mu_m$ denote the $\mathcal F_\infty$-conditional law of $U_{L,m}\widehat G_{L,m}$ on $C^\gamma_0([0,1+T];\cH)$, and let $\mu^G$ denote the law of $G$ on $C^\gamma_0([0,1+T];\cH)$. 
Along the sequence $m''$ the process
$$
X_{m''}:=U_{L,m''} \widehat G_{L,m''}
$$
is centered Gaussian with covariance
$$
\E_\infty\big[X_{m''}(s)\otimes X_{m''}(t)\big]
=
(s\wedge t)U_{L,m''}\mathcal T_L
\widehat\Omega_{L,m''}
\mathcal T_L^*U_{L,m''}^*.%
$$
By \eqref{e:trace_subseq_as},\
$$
\kappa=\sup_{m''\in \mathbb N''} \tr(\widehat\Omega_{L,m''})<\infty.
$$
Fix any $\omega$ on the probability 1 event on which \eqref{e:trace_subseq_as} and \eqref{e:cov_subseq_as} hold. We claim %
\begin{equation}\label{e:ULmG_to_G}
\mu_{m''}(\omega)\Rightarrow \mu^G .
\end{equation}
From \eqref{e:cov_subseq_as} and \eqref{e:trace_subseq_as}, it follows that under $\P_\infty(\omega)$, $\big(X_{m''}(t_1),\ldots,X_{m''}(t_r)\big)\Rightarrow \big(G(t_1),\ldots,G(t_r)\big)$ for any $0\leq t_1<\ldots<t_r\leq T$.  It remains to show tightness in $C^\gamma_0([0,1+T],\cH)$ under $\P_\infty(\omega)$, for which we use the criterion in Corollary 1 of \cite{rackauskas:suquet:2005}.  Write
$$
\eta_h(t)=X_{m''}(t+h)-2X_{m''}(t)+X_{m''}(t-h),
\qquad h\leq t\leq 1+T-h.
$$
i.e. $\eta_h$ is the second difference of $X_{m''}(t)$ with spacing $h>0$. Note since $U_{L,m''}$ is unitary, under $\P_\infty$ we have
$$
\|\eta_h(t)\|_{\cH}
\stackrel d =
\sqrt{2h}\,
\|U_{L,m''}\widehat G_{L,m''}(1)\|_{\cH}
=
\sqrt{2h}\,
\|\widehat G_{L,m''}(1)\|_{\cH}.
$$
for any fixed $t>0$.  Hence,  using Gaussianity of $ \widehat G_{L,m}$, we deduce
\begin{align*}
\P_\infty(\|\eta_h(t)\|_{\cH}>r h^{1/2}) &=\P_\infty\big(\|  \widehat G_{L,m''}(1)\|>r/\sqrt 2\big)\\
& \lesssim r^{-2q}\E_\infty\|  \widehat G_{L,m''}(1)\|^{2q}\\
&  \lesssim r^{-2q}\big(\tr (\mathcal T_L\widehat\Omega_{L,m''}\mathcal T_L^*)\big)^q\\
&  =r^{-2q}\big(\tr (\widehat\Omega_{L,m''})\big)^q\\
&  \leq r^{-2q}\kappa^q
\end{align*}
In particular, applying \cite[Corollary 1]{rackauskas:suquet:2005} with $\sigma(h)=h^{1/2}$, $\rho(h)=h^\gamma$, $\psi(r)=r^{-2q}|\kappa|^q$, by picking a $q$ large enough such that $2q(\gamma-1/2)<-1$, we find
$$
\sum_{j=1}^\infty 2^j \psi\left(u \rho(2^{-j})/\sigma(2^{-j}) \right) =  u^{-2q} |\kappa(\omega)|^q\sum_{j=1}^\infty 2^{j(1+2q(\gamma-1/2))} <\infty,
$$
yielding tightness of $\mu_{m''}(\omega)$.  Hence, \eqref{e:ULmG_to_G} holds by Prohorov's theorem (e.g., \citealp{bil}). Applying the Skorokhod-Dudley-Wichura theorem \citep[p.48]{shorack:wellner:1986}, after extending the probability space if necessary we may find for each $m''\in \mathbb N''$ a $G_{m''}\stackrel d =G$ with 
$$
\P_\infty\left(\|X_{m''}- G_{m''}\|_{\gamma,1+T}>\delta\right) \to 0 \qquad \text{a.s.}
$$
Also,  since $U_{L,m''}$ is unitary,
$$
\|\widehat G_{L,m''}-U_{L,m''}^*G_{m''}\|_{\gamma,1+T}
=
\|U_{L,m''}\widehat G_{L,m''}-G_{m''}\|_{\gamma,1+T}.
$$
For $x\in C^\gamma_0([0,1+T];\cH)$, write
$$
D_{t,u,v}x=
\frac{t-v}{1+u}x(1+u)-\{x(1+t)-x(1+v)\}.
$$
By construction, $U_{L,m}$ commutes with $J$. Hence
$D_{t,u,v}(U_{L,m}^*x)=U_{L,m}^*D_{t,u,v}x$ and
$$
Q_J\left(D_{t,u,v}(U_{L,m}^*x)\right)
=
Q_J\big(D_{t,u,v}x\big).
$$
Thus $\Psi(U_{L,m}^*x)=\Psi x$.
Using Lemma~\ref{l:Psi_continuous} and the identity
$\Psi G_{m''}\equiv\Psi(U_{L,m''}^*G_{m''})$, we obtain
$$
\P_\infty\left(
\|\Psi\widehat G_{L,m''}-\Psi G_{m''}\|_{\infty,T}>\delta
\right)\to0
\qquad \text{a.s.}
$$
Write
$$
    \widehat{\tr}_{L,m}(A)
    :=
    \sum_{\ell=1}^{L}\widehat\lambda_{\ell,m},
$$
and for $x\in C^\gamma_0([0,1+T],\cH)$, set
\begin{align*}
\Psi_{m,L} x(t) =&\sup_{(u,v)\in\mathcal S(t)}\frac{1}{g(t,1+u,t-v)}\\
& \qquad \qquad  \left|  \widehat{\tr}_{L,m}(A) \left( \frac{(t-v)^2}{1+u} + (t-v)\right) - Q_{{\widehat J_L}} \left( \frac{t-v}{1+u} \widetilde x(u) -  (\widetilde x(t)-\widetilde x(v)\big)\right)\right|,
\end{align*}
where $Q_{\widehat J_L}(f)= \langle f, \widehat J_Lf\rangle_{\cH}$, with (c.f. \eqref{e:def_J})
$$
\widehat J_L f = \sum_{\ell=1}^{L(m)} \text{sgn}(\widehat \lambda_{\ell,m}) \langle f,\phi_\ell\rangle_{\cH}\phi_\ell,\qquad f \in \cH.
$$
By construction, 
$$
\widehat \Gamma_m(t) = \Psi_{m,L} \widehat G_{L,m}(t).
$$
Hence,
\begin{align}
\|\widehat \Gamma_m-\Psi G_m\|_{\infty,T} &=\|\Psi_{m,L} \widehat G_{L,m}-\Psi G_m\|_{\infty,T}\notag\\
&\leq \|\Psi_{m,L} \widehat G_{L,m}-\Psi  \widehat G_{L,m}\|_{\infty,T}  +\| \Psi  \widehat G_{L,m}-\Psi G_{m}\|_{\infty,T}.\label{e:hatgamma_bound1}
\end{align}
We now show
\begin{equation}\label{e:PsiGL_convergence_pre}
\P_\infty\left(\|\Psi_{m'',L} \widehat G_{L,{m''}}-\Psi  \widehat G_{L,m''}\|_{\infty,T}>\delta\right) \to 0\quad \text{a.s.}.
\end{equation}
Indeed, writing $J_L=P_LJ$, and $\eta_{L,m}(t,u,v)=\frac{t-v}{1+u}\widehat G_{L,m}(1+u) - \big(\widehat G_{L,m}(1+t)- \widehat G_{L,m}(1+v)\big)$, note that $$Q_J(f)-Q_{\widehat J_L}(f)=\langle f,(J-\widehat J_L)f\rangle_{\cH}=\langle f,(J_L-\widehat J_L)f\rangle_{\cH},\quad \text{for any} \quad f \in \text{span}\{\phi_1,\ldots,\phi_L\}.$$
Hence, for any $m\geq1$, using \(g(t,1+u,t-v)\gtrsim (t-v)^\beta\) uniformly on \(0\leq t\leq T\), we have
\begin{align*}
&\left| \Psi_{m,L}\widehat G_{L,m}(t)-\Psi \widehat G_{L,m}(t)\right|\\
&\quad\leq
C\sup_{(u,v)\in\mathcal S(t)}
(t-v)^{-\beta}
\left(\frac{(t-v)^2}{1+u}+t-v\right)
\big|\widehat{\tr}_{L,m}(A)-\tr A\big|\\
&\qquad+
\sup_{(u,v)\in\mathcal S(t)} \frac1{g(t,1+u,t-v)} \left| \left\langle
\eta_{L,m}(t,u,v),
(\widehat J_L-J_L)\eta_{L,m}(t,u,v)
\right\rangle_{\cH} \right|\\
&\quad\leq
C_T\big|\widehat{\tr}_{L,m}(A)-\tr A\big| +\|\widehat J_L-J_L\|_{\mathrm{op}}
\sup_{(u,v)\in\mathcal S(t)}
\frac{\|\eta_{L,m}(t,u,v)\|_{\cH}^2}
{g(t,1+u,t-v)} .
\end{align*}
Now, let $\pi_m\in\Pi_m$ be a permutation attaining the minimum in
Theorem~\ref{t:eigen_consistency_theo}. On the event
\begin{equation}\label{e:onthisevent}
    \left(\sum_{\ell=1}^{m}
    (\widehat\lambda_{\pi_m(\ell),m}-\lambda_{\ell})^2
    \right)^{1/2}
    \leq \delta_m/3,
\end{equation}
the separation condition $ |\lambda_L|-|\lambda_{L+1}|\geq \delta_m$
implies that the eigenvalues $\widehat\lambda_{\pi_m(1),m},\ldots,\widehat\lambda_{\pi_m(L),m}$
are precisely the $L$ empirical eigenvalues with largest absolute values. On the event \eqref{e:onthisevent}, 
$\widehat J_L=J_L$ on $\operatorname{span}\{\phi_1,\ldots,\phi_L\}$; hence we need only to control the trace term. However, on \eqref{e:onthisevent},
$$
    \widehat{\tr}_{L,m}(A)
    =
    \sum_{\ell=1}^{L}\widehat\lambda_{\pi_m(\ell),m}.
$$
Since \(\varepsilon_m=o(\delta_m)\), Theorem~\ref{t:eigen_consistency_theo}
implies that this event has probability tending to one. Hence
\begin{align*}
\left|\widehat{\tr}_{L,m}(A)-\tr(A)\right|
&\leq
\left|\sum_{\ell=1}^{L}\left(\widehat\lambda_{\pi_m(\ell),m}-\lambda_{\ell}\right)\right|
+\left|\sum_{\ell=L+1}^\infty\lambda_\ell\right|+o_\P(1)\\
&\leq\sqrt L \left( \sum_{\ell=1}^{m}
(\widehat\lambda_{\pi_m(\ell),m}-\lambda_{\ell})^2
\right)^{1/2}
+
\sum_{\ell=L+1}^\infty|\lambda_\ell|+ o_\P(1)\\
&= O_\P(\sqrt L\,\varepsilon_m)
+\sum_{\ell=L+1}^\infty|\lambda_\ell| + o_\P(1)
= o_\P(1).
\end{align*}
Refining $\mathbb N''$ if necessary, we may assume that
$\widehat{\tr}_{L,m''}(A)\to\tr(A)$ almost surely and that
\eqref{e:onthisevent} holds eventually almost surely.
Along this common subsequence, we obtain \eqref{e:PsiGL_convergence_pre}.
Next we bound $\| \Psi  \widehat G_{L,m}-\Psi G_{m}\|_{\infty,T}$.  Since $O_{L,m}$ is block orthogonal, and by construction commutes with $\Lambda_L=\text{diag}(\lambda_1,\ldots,\lambda_L),$ writing $$[J]_L=\mathcal T_L^*J\mathcal T_L=\mathrm{diag}(\text{sgn}(\lambda_1),\ldots,\text{sgn}(\lambda_L)),$$
we have $[J]_LO_{L,m}=O_{L,m}[J]_L$. Hence,
\begin{align*}
U_{L,m}J&= \mathcal T_LO_{L,m}\mathcal T_L^*J + (I-P_L)J\\ &=  \mathcal T_LO_{L,m}[J]_L\mathcal T_L^* + (J-J_L)\\
&=  \mathcal T_L[J]_LO_{L,m}\mathcal T_L^* + J(I-P_L)\\
& = J U_{L,m}.
\end{align*}
Therefore, for any $f \in \cH$, $Q_J(U_{L,m}^* f)=\langle U_{L,m}^*f, J U_{L,m}^*f\rangle_{\cH}= \langle f, U_{L,m} U_{L,m}^*Jf\rangle_{\cH}= \langle f,Jf\rangle_{\cH}=Q_J(f),$ implying
$$
\Psi G_m \equiv  \Psi (U^*_{L,m} G_m).$$
Using Lemma \ref{l:Psi_continuous}, we obtain
\begin{align*}
\|\Psi \widehat G_{L,m}-\Psi (G_m)\|_{\infty,T} & =\|\Psi \widehat G_{L,m}-\Psi (U_{L,m}^*G_m)\|_{\infty,T} \\
&\lesssim  \big(\|\widehat G_{L,m}\|_{\gamma,1+T}+\|G_m\|_{\gamma,1+T}\big)\|\widehat G_{L,m}-U_{L,m}^*G_m\|_{\gamma,1+T}\\
&=o_{\P_\infty}(1).
\end{align*}
 It follows from \eqref{e:hatgamma_bound1} that
\begin{equation}\label{e:finiteint_converg_gammahat}
\P_\infty\left(\|\widehat \Gamma_m -\Psi G_{m}\|_{\infty,T}>\delta\right) \to 0\quad \text{a.s.}
\end{equation}
By repeating arguments analogous to those in Lemma \ref{l:Gamma_tail}, we find
\begin{align}\label{e:Gammahat_tail1}
\lim_{T\to\infty}\limsup_{m\to\infty}\P_\infty\left\{ \sup_{t\geq T} \left|\widehat \Gamma_m(t) \right| >\delta \right\} =0,\quad \text{a.s.}
\end{align}
  Putting together \eqref{e:finiteint_converg_gammahat}, \eqref{e:Gammahat_tail1},  and using Lemma \eqref{l:Gamma_tail}, we obtain
\begin{equation*}
\P_\infty\left(\sup_{t\geq 0}\Big|\widehat \Gamma_m(t)-\Psi G_{m}(t)\Big|>\delta\right)\to 0 \quad \text{a.s.}
\end{equation*}
Since under $\P_\infty$, $\{\Gamma(t),t\geq 0\}\stackrel d =\{\Psi G_{m}(t),t\geq 0\}$ the proof is complete.
\end{proof}

\begin{lemma}\label{l:transformation_gram}
Under the conditions of Lemma~\ref{l:transformation}, with $C_m= I_m - \frac{1}{m}\mathbf 1 \mathbf 1^\top$ let
$$
H_m(\theta)
=
\frac1m
\bigl(h_\theta(\bX_i,\bX_j)\bigr)_{1\leq i,j\leq m},
\qquad
\overline H_m(\theta)=C_mH_m(\theta)C_m.
$$
Then
$$
\big\|
\overline H_m(\widehat\theta_m)
-
\overline H_m(\theta_0)
\big\|_{\rm F}
=
O_\P(m^{-1/2}).
$$
\end{lemma}

\begin{proof}
Since $\|C_m\|_{\mathrm{op}}\leq1$, on the event
$\{\widehat\theta_m\in N\}$,
\begin{align*}
&
\big\|
\overline H_m(\widehat\theta_m)
-
\overline H_m(\theta_0)
\big\|_{\rm F}^2\\
&\quad\leq
\big\|
H_m(\widehat\theta_m)-H_m(\theta_0)
\big\|_{\rm F}^2\\
&\quad=
\frac1{m^2}
\sum_{i,j=1}^m
\big|
h_{\widehat\theta_m}(\bX_i,\bX_j)
-
h_{\theta_0}(\bX_i,\bX_j)
\big|^2\\
&\quad\leq
\frac{C\|\widehat\theta_m-\theta_0\|^2}{m^2}
\sum_{i,j=1}^m
\{q(\bX_i)+q(\bX_j)\}^2\\
&\quad\leq
4C\|\widehat\theta_m-\theta_0\|^2
\frac1m\sum_{i=1}^m q(\bX_i)^2
=
O_\P(m^{-1}).
\end{align*}
Since $\P(\widehat\theta_m\in N)\to1$, the result follows.
\end{proof}

\begin{proof}[Proof of Corollary~\ref{c:transformation}]
Let $\widehat\lambda_{\ell,m}(\theta)$ denote the empirical eigenvalues
of $\overline H_m(\theta)$. By Lemma~\ref{l:transformation_gram}, the
Hoffman--Wielandt inequality, and
Theorem~\ref{t:eigen_consistency_theo},
$$
\min_{\pi\in\Pi_m}
\left(\sum_{\ell=1}^m(\widehat\lambda_{\pi(\ell),m}(\widehat\theta_m)-\lambda_\ell)^2
\right)^{1/2}=O_\P(\varepsilon_m),
$$
where $\lambda_\ell$ and $\varepsilon_m$ are defined using
$h_{\theta_0}$.

The preceding eigenvalue bound and
$\varepsilon_m=o(\delta_m)$ imply, by the argument of
Lemma~\ref{l:sample_eig_concentration}, that
\begin{equation}\label{e:transformation_eig_concentration}
\P\left\{
\big\{\widehat\lambda_{\ell,m}(\widehat\theta_m):
\ell\in\mathcal B_k\big\}
\subseteq\mathcal I_{k,m},
\quad k=1,\ldots,b
\right\}\to1.
\end{equation}
Let $V_m^L(\theta)$ denote the matrix of the first $L$ empirical
eigenvectors of $\overline H_m(\theta)$. The preceding eigenvalue bound
and $\varepsilon_m=o(\delta_m)$ imply, by the same argument as in
Lemma~\ref{l:sample_eig_concentration}, that, with probability tending
to one, $\overline H_m(\widehat\theta_m)$ and
$\overline H_m(\theta_0)$ have exactly $d_k$ eigenvalues in
$\mathcal I_{k,m}$, for each $k=1,\ldots,b$.

Consequently, applying Lemma~\ref{l:DK_lemma} (with the remaining
eigenvalues included as an additional spectral group), gives
\begin{align*}
&\left\{
\sum_{k=1}^b
\left\|
\mathcal P_k\big(\overline H_m(\widehat\theta_m)\big)
-
\mathcal P_k\big(\overline H_m(\theta_0)\big)
\right\|_{\rm F}^2
\right\}^{1/2}\\
&\qquad\leq
\frac{C}{\delta_m}
\left\|
\overline H_m(\widehat\theta_m)
-
\overline H_m(\theta_0)
\right\|_{\rm F}
=
O_\P\left(\frac1{\sqrt m\,\delta_m}\right).
\end{align*}

By Lemma~\ref{l:basic_rotation_lemma}, there therefore exists a
block-orthogonal matrix $R_{L,m}$ such that
$$
\left\|
V_m^L(\widehat\theta_m)
-
V_m^L(\theta_0)R_{L,m}
\right\|_{\rm F}
=
O_\P\left(\frac1{\sqrt m\,\delta_m}\right).
$$
On the other hand, Theorem~\ref{t:spectral_consistency}, applied to
$h_{\theta_0}$, gives a block-orthogonal matrix $O_{L,m}^{(0)}$ such
that
$$
\left\|
V_m^L(\theta_0)-\Phi_m^LO_{L,m}^{(0)}
\right\|_{\rm F}
=
O_\P\left(\frac1{\sqrt m\,\delta_m^{3/2}}\right).
$$
Hence, with $O_{L,m}=O_{L,m}^{(0)}R_{L,m}$,
$$
\left\|
V_m^L(\widehat\theta_m)-\Phi_m^LO_{L,m}
\right\|_{\rm F}
=
O_\P\left(\frac1{\sqrt m\,\delta_m^{3/2}}\right).
$$
Thus the eigenvalue and eigenvector bounds used in
Lemmas~\ref{l:LRC_consistency} and~\ref{l:trace_hatSigma} remain
unchanged. Their proofs, and hence the proof of
Theorem~\ref{t:LRC_sim_method}, therefore apply to the estimators
constructed using $h_{\widehat\theta_m}$.
\end{proof}

\smallskip

\begin{example}\label{ex:lambda}
Suppose $\lambda_\ell=\ell^{-\theta}$ for some $\theta>3/2$.
Then Assumptions~\ref{a:eigsep} and~\ref{ker2} hold when
$L_m\asymp(\log m)^q$ for any fixed $q>0$ and
$1\ll\mathfrak b_m\lesssim m^{1/3}$.

To verify this, note that
$\lambda_L-\lambda_{L+1}\sim\theta L^{-(\theta+1)}$
and $\varepsilon_m\asymp m^{-1/2}$.
We may therefore take
$\delta_m=(\theta/2)L_m^{-(\theta+1)}$,
which satisfies Assumption~\ref{a:eigsep} because $L_m$
grows only logarithmically. With this choice,
$\mathfrak b_m\{(\sqrt m\,\delta_m^{3/2})^{-1}
+\sqrt{L_m}\,\varepsilon_m\}
=O\big(m^{-1/6}(\log m)^{3q(\theta+1)/2}\big)\to0$,
verifying Assumption~\ref{ker2}.
\end{example}
\section{Supplement to Section \ref{s:simulations}}
\label{supp:sim}

\subsection{Details for null simulations in higher dimensions}

In Table \ref{tab:spectral-rank-explainer-ar-var}, we report the empirical effective dimension (\citealp{lopes:jacob:wainwright:2011})
$$
\widehat r_{\rm eff} = \Big(\sum_{\ell=1}^m |\widehat\lambda_\ell|\Big)^2/\Big(\sum_{\ell=1}^m \widehat\lambda_\ell^2\Big),
$$
which measures the diffusiveness of the eigenvalues of $\overline H_m$. Provided at least one $\widehat\lambda_\ell\neq0$, it holds that $1\leq \widehat r_{\rm eff}\leq m$, and  $\widehat r_{\rm eff}=1$ for a completely concentrated spectrum, while
$\widehat r_{\rm eff}=m$ when all $m$ empirical eigenvalues have equal absolute value.  The indicated VAR model is $\bX_t= 
 0.50Q \bX_{t-1} + \sqrt{1-0.50^2}  \boldsymbol \be_t$, where $Q$ is a randomly selected orthogonal matrix, and where $\be_t\sim N(0,I_d)$ are i.i.d.

\begin{table}[H]\footnotesize
\centering
\caption{Median empirical effective rank $\widehat r_{\rm eff}$ for selected kernels.}
\vspace{-4ex}
\label{tab:spectral-rank-explainer-ar-var}\
{\singlespacing
\begin{tabular}{lcccccc}
\toprule
Kernel & \multicolumn{2}{c}{$d=1$ AR(1)} & \multicolumn{2}{c}{$d=2$ VAR(1)} & \multicolumn{2}{c}{$d=50$ VAR(1)}\\
\cmidrule(lr){2-3}\cmidrule(lr){4-5}\cmidrule(lr){6-7}
  & $m=500$ & $m=1000$ & $m=500$ & $m=1000$ & $m=500$ & $m=1000$\\
\midrule
$\|\bx-\by\|_2^{1/2}$ & 9.5 & 9.6 & 20.8 & 21.1 & 305 & 440\\
$\|\bx-\by\|_2$ & 3.5 & 3.6 & 6.8 & 6.8 & 142 & 166\\
\eqref{e:kernelchoice},~$\alpha=1/4$ & 46.7 & 49.3 & 121 & 137 & 473 & 900\\
\eqref{e:kernelchoice},~$\alpha=1$ & 6.8 & 6.8 & 13.9 & 14.1 & 266 & 362\\
\bottomrule
\end{tabular}
}
\end{table}
\normalsize

In this section, we give the full model specification of the 4 high-dimensional models in the calibration study  considered in Section \ref{s:simulations:null}.

The $d=50$ VAR(1) is identical to the model used in Table \ref{tab:spectral-rank-explainer-ar-var}.  The nonlinear VAR(1) is a threshold-type VAR, given by $
\bX_t=\kappa_t Q\bX_{t-1}+(1-\kappa_t^2)^{1/2}\beps_t$, where
$\kappa_t=0.20\mathbf 1_{\{X_{t-1,1}\leq 0\}}+0.60\mathbf 1_{\{X_{t-1,1}>0\}}$, with $X_{t-1,1}$ being the first coordinate of $\bX_{t-1}$.

The remaining two high-dimensional null models are discretized functional time series constructed from  functional principal components (FPC) models.  For $d=500$, each
curve $\{\bX_t(u), u \in[0,1]\}$  is observed on the grid $u_j=(j-1/2)/d$, $j=1,\ldots,d$, and generated as
$$
    \bX_t(u_j)=\frac{1}{\sqrt d}\sum_{\ell=1}^{10}\zeta_{t,\ell}\varphi_\ell(u_j)+0.25 e_{t,j},
    \qquad j=1,\ldots,d.
$$
The functions $\varphi_1,\ldots,\varphi_{10}$ are the first ten elements of the Fourier basis on
$[0,1]$, starting with the constant function and then alternating sine and cosine terms.  The 
$e_{t,j}$ are i.i.d. standard normal.  The vectors 
$\boldsymbol \zeta_t=(\zeta_{t,1},\ldots,\zeta_{t,10})^\top$ in the linear FPC model are generated directly from a $d=10$ version of the VAR(1) described above. In the nonlinear FPC model, the
scores follow
$\boldsymbol \zeta_t=0.60Q^\top f(Q\boldsymbol \zeta_{t-1})+ \widetilde {\boldsymbol e}_t$, with $ \widetilde {\boldsymbol e}_t\stackrel{\mathrm{iid}}{\sim}N(0,(1-0.60^2)I_{10})$, where $Q$ is a fixed orthogonal matrix and $f(\cdot)=\tanh(\cdot)$ is applied
componentwise.

\subsection{Details for alternative simulations at a fixed break location}
In this section, we give details on the full specification of the alternative simulations. The empirical pre-change alarm frequencies for these experiments appear in Table~\ref{tab:fixed-break-prechange-false-alarms}.

\paragraph{Details for the isolated structural breaks scenario.} Here we provide the details for the first family of alternatives where the marginal mean is unchanged and the break occurs through serial dependence or volatility structure. In the scalar AR(1) experiment, the autoregressive parameter changes from $\phi_0=0.30$ to $\phi_1=0.85$, with Laplace innovations and the stationary variance normalized to one before and after the break. In the VAR(1) experiment, the dimension is $d=10$; before the break the transition matrix is $A=0.20I_d$, and after it is $A=0.20I_d+0.70P$, where $P=UU^\top$ where $U\in \R^{d\times 3}$ are 3 randomly chosen orthogonal vectors. Innovations are chosen so that the stationary covariance is approximately the identity both before and after the break. In the GARCH(1,1) experiment, the mean and unconditional variance are unchanged, while the volatility dynamics change from $(\omega,a,b)=(0.45,0.05,0.50)$ to $(\omega,a,b)=(0.03,0.27,0.70)$. The cumulative probability curves are displayed below.

\begin{figure}[h!]
\begin{center}
\includegraphics[width=0.72\linewidth]{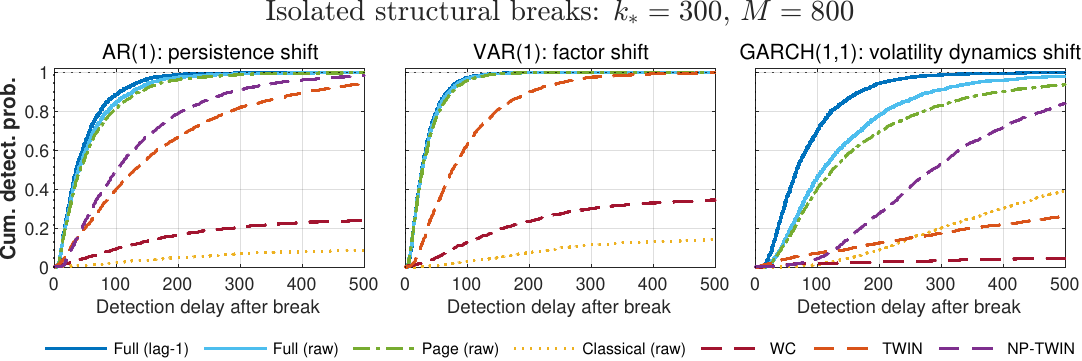}
\end{center}
\caption{
Cumulative detection probabilities in the isolated structural break scenario.%
\label{f:alt_structbreak_cdfs}}
\end{figure}

\paragraph{Details for the pure mean breaks scenario.}
The pure-mean-break scenario considers a scalar AR(1) model with $\phi=0.50$ before and after the break and a change only in the mean. We use Laplace innovations and normalize the stationary variance to one. The mean changes from $\mu=0$ to $\mu\in\{0.5,1\}$.  The cumulative probability curves are displayed below.

\begin{figure}[h!]
\begin{center}
\includegraphics[width=0.72\linewidth]{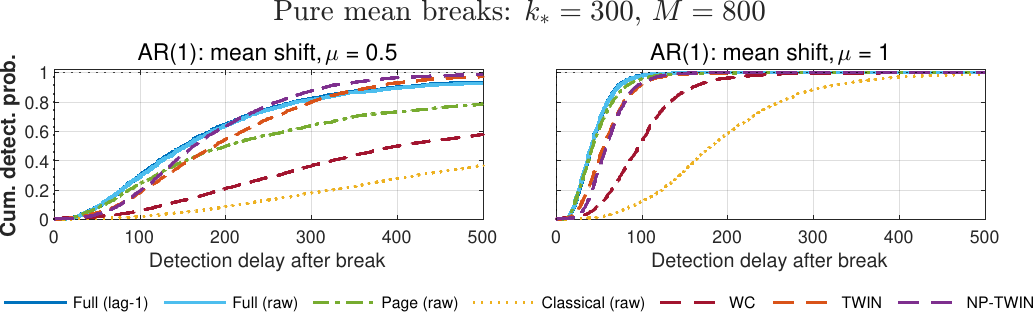}
\end{center}
\caption{
Cumulative detection probabilities in the pure mean break scenario.%
\label{f:alt_fixedbreak_mean_cdfs}}
\end{figure}

\paragraph{Details for the mixed  breaks scenario.}  In the scalar AR(1) experiment, the autoregressive parameter changes from $\phi=0.20$ to $\phi=0.80$, the mean changes from $\mu=0$ to $\mu=0.30$, and the stationary variance increases from 1 to $1.50$. In the VAR(1) experiment, the same low-rank factor dependence shift as in the isolated structural case is combined with a mean shift of Euclidean norm $0.50$ along the first column of $U$, and the stationary covariance changes from $I_d$ to $I_d+0.50P$. In the GARCH(1,1) experiment, the mean changes from $\mu=0$ to $\mu=0.35$,  while the volatility dynamics change from $(\omega,a,b)=(0.45,0.05,0.50)$ to $(\omega,a,b)=(0.0375,0.27,0.70)$.  The innovations are standard Laplace throughout.
\begin{figure}[h!]
\begin{center}
\includegraphics[width=0.72\linewidth]{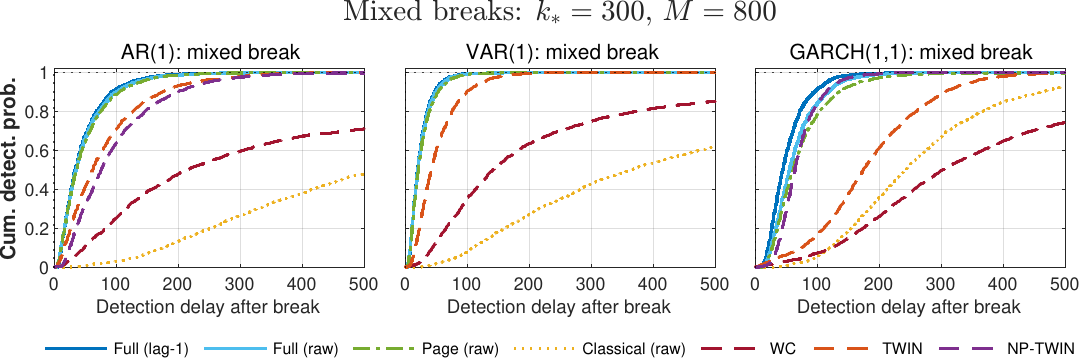}
\end{center}
\caption{
Cumulative detection probabilities in the mixed break scenario.%
\label{f:alt_mixed_cdfs}}
\end{figure}

\begin{table}[h!tb]
\centering
\caption{Empirical pre-change false alarm probabilities $\P\{\tau_m\leq k_*\}$ under the infeasible, model-specific Monte Carlo calibration at nominal level $\alpha=0.05$. Here $m=300$, $k_*=300$, and $M=800$, with 2000 evaluation replications and 2000 separate null calibration paths per scenario. Full lag-1 and Full raw denote the full kernel scans; Page and Classical use the raw representation. A dash indicates that NP-TWIN is not implemented for the multivariate VAR design. Entries concern alarms by the changepoint, rather than full-horizon null rejection.}
\label{tab:fixed-break-prechange-false-alarms}

\small
\setlength{\tabcolsep}{4pt}
\begin{tabular}{lccccccc}
\toprule
& \multicolumn{4}{c}{Kernel procedures} & \multicolumn{3}{c}{Benchmarks} \\
\cmidrule(lr){2-5}\cmidrule(lr){6-8}
Scenario & Full lag-1 & Full raw & Page & Classical & WC & TWIN & NP-TWIN \\
\midrule
Structural AR(1) & .040 & .036 & .032 & .047 & .041 & .021 & .010 \\
Structural VAR(1) & .051 & .064 & .062 & .054 & .048 & .004 & -- \\
Structural GARCH & .049 & .055 & .055 & .054 & .042 & .016 & .008 \\
\midrule
Mean AR(1), $\mu_1=0.5$ & .039 & .036 & .043 & .067 & .047 & .011 & .011 \\
Mean AR(1), $\mu_1=1.0$ & .039 & .037 & .043 & .067 & .047 & .010 & .011 \\
\midrule
Mixed AR(1) & .040 & .042 & .043 & .053 & .037 & .013 & .009 \\
Mixed VAR(1) & .046 & .058 & .055 & .055 & .045 & .005 & -- \\
Mixed GARCH & .042 & .041 & .039 & .055 & .047 & .013 & .009 \\
\bottomrule
\end{tabular}
\end{table}

\subsection{Details for the break location experiments} We consider two alternatives:  a strong, pure mean break, with $\phi=0.5$ fixed with a mean change from $\mu=0$ to $\mu=2$, and a mixed break, with a mean change of $\mu=0$ to $\mu=0.5$, a coefficient change of $\phi:0.2$ to $\phi=0.9$, and with a change in the stationary variance from $1$ to $1.5$.  In all cases, the innovations are rescaled after the changepoint to attain the prescribed stationary variance.  Critical values for all delay experiments are obtained by simulating 500 paths from the corresponding null model; the reported delay distributions are based on 500 simulations for each alternative.

Since the ECDF-based method NP-TWIN is computationally burdensome at long horizons, we implement NP-TWIN using a geometric approximation to its window scan.  Specifically, we use the geometric approximation

$$
\Gamma^{\mathrm{NP},\rm geo}_m(k)
=
\max_{\ell\in\mathcal L_{m,k}^{(\rho)}}
w_{\beta}(\ell,k;m)
\left\|
\min\left(1,\frac{\ell}{m}\right)\widehat F_{\max(\ell,m)}
-
\left(\widehat F_{m+k}-\widehat F_{m+k-\ell}\right)
\right\|_\infty.
$$

Let $a_{m,k}=\min\{k,\lfloor(m+k)/2\rfloor\}$ and define
$$
\ell_0=a_{m,k},\qquad
\ell_{j+1}=\max\{1,\lfloor\rho\ell_j\rfloor\},\qquad
\mathcal L_{m,k}^{(\rho)}=\{\ell_j:j=0,1,2,\ldots\}.
$$
The recursion stops when $\ell_j=1$. We use $\rho=0.9$ and the specified weight
$w_{\beta}(\ell,k;m)$ 
with the recommended choices $\beta=0.6$ and $C_0=20$.

The empirical pre-change false alarm probabilities for the break location experiments are given in Table \ref{tab:delay-location-prechange-false-alarms}.

\begin{table}[h!tb]
\centering
\caption{Empirical pre-change false alarm probabilities $\P\{\tau_m\leq k_*\}$ under the infeasible, model-specific Monte Carlo calibration at nominal level $\alpha=0.05$. Here $m=100$ and $M=2500$ for $k_*\in\{10,100,1000\}$, and $M=12000$ for $k_*=10000$. Each cell uses 500 evaluation replications and 500 separate null calibration paths per configuration; critical values are shared across the first three changepoint locations. Entries concern alarms by the changepoint, rather than full-horizon null rejection.}
\label{tab:delay-location-prechange-false-alarms}

\small
\setlength{\tabcolsep}{4pt}
\begin{tabular}{lcccccc}
\toprule
& & \multicolumn{2}{c}{Full kernel, lag-1} & \multicolumn{3}{c}{Benchmarks} \\
\cmidrule(lr){3-4}\cmidrule(lr){5-7}
Family & $k_*$ & $\beta=.95$ & $\beta=.99$ & WC & TWIN & NP-TWIN \\
\midrule
Pure mean & 10 & .000 & .000 & .000 & .000 & .000 \\
 & 100 & .038 & .040 & .050 & .000 & .000 \\
 & 1000 & .064 & .050 & .064 & .016 & .030 \\
 & 10000 & .044 & .050 & .040 & .030 & .032 \\
\midrule
Mixed & 10 & .000 & .000 & .006 & .000 & .000 \\
 & 100 & .032 & .048 & .056 & .000 & .000 \\
 & 1000 & .076 & .048 & .072 & .024 & .024 \\
 & 10000 & .066 & .052 & .030 & .036 & .044 \\
\bottomrule
\end{tabular}
\end{table}

\section{Supplement to Section \ref{s:data}}
\label{s:data_supplement}
Throughout the data illustrations, we use nominal level $\alpha=0.05$.

\subsection{Details for ERCOT illustration}\label{energy-supp}

We first standardize each hourly coordinate of the bivariate curves
$\bX_t=(\bE_t,\bS_t)^\top$ using its historical mean and standard
deviation. We then compute the first two principal-component scores
of each standardized series and standardize these scores using the
historical sample. The resulting transformation is held fixed during
monitoring. We now describe this procedure in detail and connect it
to Corollary~\ref{c:transformation}.

Write a generic daily observation as
$\bx=(\be,\bs)\in\R^{24}\times\R^{24}$, where $\be$ and $\bs$
are the percentage load forecast-error and price-spread curves,
respectively. We assume that the daily observations lie in a fixed
bounded subset of this space. For $j\in\{E,S\}$, let $Q_j$ be the
diagonal matrix of population hourly standard deviations,
and let $V_j\in\R^{24\times2}$ contain the first two population
principal-component directions of $Q_j^{-1}\bX_{1,j}$.
For $\bz\in\R^{24}$, define
$$
\boldsymbol\xi_j(\bz)=V_j^\top Q_j^{-1}\bz\in\R^2.
$$
Let
$$
\boldsymbol\mu_j
=
\E\{\boldsymbol\xi_j(\bX_{1,j})\},
\qquad
D_j
=
\text{diag}\left\{
\sqrt{\Var(\xi_{j,1})},\sqrt{\Var(\xi_{j,2})}
\right\},
$$
and define the standardized PCA scores
$$
P_j\bz
=
D_j^{-1}
\{\boldsymbol\xi_j(\bz)-\boldsymbol\mu_j\},
\qquad j\in\{E,S\}.
$$
For
$
\theta=(P_E,P_S),
$
define
$$
t_\theta(\bx)
=
\operatorname{vec}\big\{
(P_E\be)(P_S\bs)^\top
\big\}
\in\R^4.
$$
For $\by,\by'\in\R^4$, we use $h_\theta(\bx,\bx')=k\big(t_\theta(\bx),t_\theta(\bx')\big)$, where 
$$
k(\by,\by')
=
\tanh(\by/5)^\top\tanh(\by'/5)
=
\sum_{r=1}^4
\tanh(y_r/5)\tanh(y_r'/5),
$$
where $\tanh$ is applied componentwise. For implementation, $\theta$
is replaced by its historical-sample estimate $\widehat\theta_m$,
obtained from the hourly standard deviations, principal-component
directions, and corresponding score means and standard deviations.
To connect this construction to Corollary~\ref{c:transformation},
take $\mathcal Y=\R^4$ with its Euclidean metric. The kernel $k$
is bounded and satisfies
$$
k(\bu,\bu)+k(\bv,\bv)-2k(\bu,\bv)
=
\|\tanh(\bu/5)-\tanh(\bv/5)\|^2
\leq \frac{1}{25}\|\bu-\bv\|^2.
$$
Thus Lemma~\ref{l:kernel_examples}\textit{(i)} applies with
$\alpha=1$. The boundedness assumption makes all marginal moment
conditions automatic. Moreover, positive population hourly and
score variances ensure that, in a sufficiently small neighborhood
of $\theta_0$, $t_\theta$ satisfies the required Lipschitz conditions
in both its observation and parameter arguments. The kernel's
parameter bound also holds with $q(\bx)\equiv1$.

Each combined score map $P_j$ is affine, with 48 coefficients and
two intercepts, so we may parameterize $\theta$ by these entries
and take $\nu=2(48+2)=100$. Suppose Assumption~\ref{a:Lpm} holds
with $p$ and $\eta$ satisfying $p(1-\eta)>100$, and the first two
eigenvalues of each standardized covariance matrix
$\Cov(Q_j^{-1}\bX_{1,j})$ are simple and separated from the remaining
spectrum. Then the historical parameter estimates are
$\sqrt m$-consistent, with consistent sign choices for the
principal-component directions. Subject to the remaining null
and calibration assumptions, Corollary~\ref{c:transformation}
therefore applies.

For implementation, we use  $\beta=0.95$, and the geometric full scheme with $\rho=0.9$
and a minimum right-window length of two observations.
We retain $L=4$ spectral terms and estimate the long-run covariance
using Bartlett weights $(1-|r|/(q+1))_+$, with
$q=\lceil5m^{1/3}\rceil=29$. Critical values are the empirical
$0.95$ quantiles of $25,000$ simulations of the estimated Gaussian
limit, evaluated over the same observation times, candidate windows,
and monitoring horizon as the detector.

\begin{table}[h!]
\centering
\caption{Monitoring results for selected ERCOT weather events.}
\label{tab:ercot-weather-events}
\setlength{\tabcolsep}{5pt}
\renewcommand{\arraystretch}{1.12}

\begin{tabularx}{\textwidth}{
    @{}
    >{\raggedright\arraybackslash}X
    >{\centering\arraybackslash}p{2.7cm}
    >{\centering\arraybackslash}p{2.7cm}
    >{\centering\arraybackslash}p{2.7cm}
    @{}
}
\toprule
Event info
& Reference date
& Kernel alarm
& WC alarm \\
\midrule

\textbf{Hurricane Nicholas} (Coast zone series)\newline
{\footnotesize
Training: Mar.\ 1--Aug.\ 31, 2021 ($m=183$);\newline
Monitoring: Sep.\ 1--Oct.\ 15, 2021 ($M=45$)}
&
Sep.\ 14, 2021
&
Sep.\ 13, 2021 
&
No alarm
\\

\addlinespace[3pt]

\textbf{Winter Storm Uri} (system-wide series)\newline
{\footnotesize
Training: Aug.\ 1, 2020--Jan.\ 31, 2021 ($m=183$);\newline
Monitoring: Feb.\ 1--Mar.\ 18, 2021 ($M=45$)}
&
Feb.\ 13, 2021
&
Feb.\ 10, 2021 
&
Feb.\ 9, 2021 
\\

\bottomrule
\end{tabularx}

\begin{minipage}{0.96\textwidth}
\footnotesize
The reference dates are the onset 
Winter Storm Uri and the date of landfall for Hurricane Nicholas. The series type (coast zone, or system wide) indicates which series was taken from ERCOT. In the Uri example, 45 days are retained after omitting the 23-hour DST day, March 14.
\end{minipage}
\end{table}

For comparison, we apply WC to the full bivariate curve after standardizing each hourly coordinate by its historical mean and standard deviation, with these quantities held fixed during monitoring. We then set $\bZ_t=24^{-1/2}(\widetilde\bE_t^\top,\widetilde\bS_t^\top)^\top$, so that the Euclidean norm of $\bZ_t$ is the discrete $L^2$ norm of the standardized bivariate curve.  At monitoring time $k$, WC uses
$$
\widehat\Gamma_m^{\beta_{\rm WC}}(k)
=
\max_{0\leq\ell<k}
w_{\beta_{\rm WC}}(\ell,k,m)
\left\|
\frac{k-\ell}{m+\ell}
\sum_{t=1}^{m+\ell}\bZ_t
-
\sum_{t=m+\ell+1}^{m+k}\bZ_t
\right\|_{2},
$$
where
$
w_{\beta_{\rm WC}}(\ell,k,m)
=
\sqrt m\big(
(m+k)^{1-\beta_{\rm WC}}
 (k-\ell)^{\beta_{\rm WC}}
 \log(2+k/m)\big)^{-1}. $
Following the guidance in \cite{kutta:dornemann:2025}, we take $\beta_{\rm WC}=0.30$; critical values are obtained by estimating the long-run covariance matrix of $\bZ_t$ from the historical sample
using a Bartlett kernel with bandwidth 3 and $5,000$ simulations
of the corresponding Gaussian limit.

\subsection{Details for the daily networks example}\label{app:airline}

We consider three daily representations of the airline network, corresponding to the
route-share network, the leading modes of the normalized graph Laplacian, and
carrier-specific cancellation rates. Let $\bA_t$ denote the symmetric weighted adjacency
matrix on day $t$, whose $(i,j)$ entry is the number of operated flights between airports
$i$ and $j$ in either direction. Writing
$c_t=\sum_{i<j}A_{ij,t}$ for the total number of operated flights across all airport pairs
on day $t$, we define the route-share network by $\bP_t=c_t^{-1}\bA_t$.
For the Laplacian representation, let
$\mathbf L_t=\bI-\bD_t^{-1/2}\bA_t\bD_t^{-1/2}$, where
$\bD_t=\text{diag}(\bA_t\boldsymbol 1)$, and write
$0<\ell_{1,t}\leq\cdots\leq\ell_{r_t,t}$ for its positive eigenvalues; we use the first five
modes $(\ell_{1,t},\ldots,\ell_{5,t})^\top$. Finally, letting
$\bC_t=(C_t^A,C_t^D,C_t^U,C_t^S)^\top$, the cancellation-rate representation records
the daily cancellation rates for American, Delta, United, and Southwest, respectively.
We give the full construction, preprocessing steps, kernel choices, and tuning parameters
for these three representations below.

Write $a\in\{P,\mathcal L,C\}$ for the route-share, Laplacian, and
cancellation-rate representations, respectively, and set
$$
\bY_t^{(P)}
=
(P_{ij,t}:1\leq i<j\leq 50)^\top\in\R^{1225},
\qquad
\bY_t^{(\mathcal L)}
=
(\ell_{1,t},\ldots,\ell_{5,t})^\top\in\R^5,
\qquad
\bY_t^{(C)}=\bC_t\in\R^4.
$$
Each representation is
adjusted and monitored separately.

Let $t_0$ denote the first day of the corresponding historical sample.
We use the deterministic calendar vector
$$
\bg_t
=
\left(
1,\,
\cos\left\{\frac{2\pi(t-t_0)}{365}\right\},\,
\sin\left\{\frac{2\pi(t-t_0)}{365}\right\},\,
{\bf 1}_{\mathrm{TG}}(t),\,
 {\bf 1}_{\mathrm{YE}}(t)
\right)^\top .
$$
Here, ${\bf 1}_{\mathrm{TG}}(t)$ equals 1 on the seven-day period from the Monday
preceding Thanksgiving through the following Sunday, and zero otherwise; 
${\bf 1}_{\mathrm{YE}}(t)$ equals 1 on the year-end period from December 18 through January 7, and is zero otherwise. 

For representation $a$, the mean is
$
\boldsymbol\mu_a(t)=\bB_a^\top\bg_t$, where $
\bB_a\in\R^{5\times d_a}.
$
We estimate $\bB_a$ by
ordinary least squares,
$$
\widehat{\bB}_{a,m}
=
\argmin_{\bB\in\R^{5\times d_a}}
\sum_{t=1}^m
\left\|
\bY_t^{(a)}-\bB^\top\bg_t
\right\|_2^2,
$$
and define the mean-adjusted series
$$
\bZ_t^{(a)}
=
\bY_t^{(a)}
-
\widehat{\bB}_{a,m}^{\top}\bg_t.
$$

Accounting for daily and weekly differences, we next fit the autoregression 
$$
(1-\phi_a\mathsf B)(1-\psi_a\mathsf B^7)\bZ_t^{(a)}
=
\boldsymbol\varepsilon_t^{(a)},
$$
where $\mathsf B$ is the backshift operator. Equivalently, for candidate
values $(\phi,\psi)$, let
$$
\boldsymbol\varepsilon_t^{(a)}(\phi,\psi)
=
\bZ_t^{(a)}
-\phi\bZ_{t-1}^{(a)}
-\psi\bZ_{t-7}^{(a)}
+\phi\psi\bZ_{t-8}^{(a)}.
$$
A single pair $(\phi_a,\psi_a)$ is fitted across all coordinates of each
representation. 
The fitted coefficients are
$$
(\widehat\phi_{a,m},\widehat\psi_{a,m})
=
\argmin_{(\phi,\psi)\in\R^2}
\sum_{t=1}^m
\left\|
\boldsymbol\varepsilon_t^{(a)}(\phi,\psi)
\right\|_a^2,
$$
where $\| \cdot\|_a$ denotes the 2-norm for the Laplacian and cancellation-rate cases and denotes the 1-norm in the route-share case, so the route-share fit is matched to the total variation geometry. We then monitor the sequence 
$
\widehat{\boldsymbol\varepsilon}_t^{(a)}
=
\boldsymbol\varepsilon_t^{(a)}
(\widehat\phi_{a,m},\widehat\psi_{a,m}).
$
All calendar and lag parameters are estimated from the historical sample
and held fixed thereafter. 

For the COVID-19 study, the historical fitting period is January 1 through
December 31, 2019, so that $m=365$ and $t_0$ is January 1, 2019.
Observations from December 24--31, 2018 are used only to supply the lagged
values required at the beginning of the historical sample. For the 2025
study, the fitting period is March 1, 2024 through February 28, 2025,
again giving $m=365$, with $t_0$ equal to March 1, 2024. Observations from
February 22--29, 2024 are used only as lagged values.

The kernels used below have the form
$h(\bx,\by)=k(t_\theta(\bx),t_\theta(\by))$
as in Corollary~\ref{c:transformation}, which we now explain in detail. Let
$$
\bX_t^{(a)}
=
\left(
\bY_t^{(a)},\bY_{t-1}^{(a)},\bY_{t-7}^{(a)},\bY_{t-8}^{(a)},
\bg_t,\bg_{t-1},\bg_{t-7},\bg_{t-8}
\right).
$$
Let $s_a>0$ and for $\theta_a=(\bB_a,\phi_a,\psi_a,s_a)$, define
$$
t_{\theta_a}(\bX_t^{(a)})
=
\frac{1}{s_a}
\left[
\bZ_t^{(a)}
-\phi_a\bZ_{t-1}^{(a)}
-\psi_a\bZ_{t-7}^{(a)}
+\phi_a\psi_a\bZ_{t-8}^{(a)}
\right],
\qquad
\bZ_t^{(a)}=\bY_t^{(a)}-\bB_a^\top\bg_t.
$$
Thus,
$
t_{\widehat\theta_{a,m}}(\bX_t^{(a)})
=
\widehat{\boldsymbol\varepsilon}_t^{(a)}/\widehat s_{a,m},
$
where all components of $\widehat\theta_{a,m}$ are estimated from the
historical sample and held fixed during monitoring.  For the route-share representation, we take
$$
\widehat s_{P,m}^2
=
\frac1m
\sum_{t=1}^m
\left(
\frac12
\left\|
\widehat{\boldsymbol\varepsilon}_t^{(P)}
\right\|_1
\right)^2
$$
and use the kernel 
$
k_P(\bu,\bv)
=
\frac12\|\bu-\bv\|_1.$
Hence,
$$
h_P(\bX_t^{(P)},\bX_s^{(P)})
=
\frac{1}{2\widehat s_{P,m}}
\left\|
\widehat{\boldsymbol\varepsilon}_t^{(P)}
-
\widehat{\boldsymbol\varepsilon}_s^{(P)}
\right\|_1.
$$
For $a\in\{\mathcal L,C\}$, we instead set
$$
\widehat s_{a,m}^2
=
\frac{1}{m d_a}
\sum_{t=1}^m
\left\|
\widehat{\boldsymbol\varepsilon}_t^{(a)}
\right\|_2^2
$$
and use the bounded finite-rank kernel 
$$
k_a(\bu,\bv)
=
\frac1{d_a}
\sum_{j=1}^{d_a}
\tanh\left(\frac{u_j}{5}\right)
\tanh\left(\frac{v_j}{5}\right).
$$
The resulting Laplacian and cancellation-rate kernels have rank at most
five and four, respectively. 

For the Spring/Summer 2025 study, we use the full-geo scheme with $\rho=0.9$.   For the COVID-19 study, to disregard weather-related disruptions in January, scanning begins on March 1, 2020; in the notation of
Assumption~\ref{a:admissiblescheme}, we therefore use the shifted full-geometric scheme with 
$
\mathcal S_{\rho,\tau_0}^{\mathrm{full}}(t)
=
\left\{
\bigl(\tau_0+(1-\rho^j)(t-\tau_0),
      \tau_0+(1-\rho^j)(t-\tau_0)\bigr):
j=0,1,\ldots
\right\}
\cup\{(t,t)\}
$
for $ t\geq\tau_0$ and $\mathcal S_{\rho,\tau_0}^{\mathrm{full}}(t)=\varnothing$ otherwise, 
with $\tau_0=60/365$ and $\rho=0.9$.

For all three representations, we use $\beta=0.95$, and a minimum right-window length of two observations.
Long-run covariance estimates use Bartlett weights
$(1-|r|/(q+1))_+$, with $q=\lceil5m^{1/3}\rceil=36$.
We retain $L=4$ spectral terms for cancellation rates, $L=5$ for
Laplacian modes, and $L=\lceil2.5(\log m)^{4/3}\rceil=27$
for route shares. Critical values are the empirical $0.95$
quantiles of $5,000$ simulations of the estimated Gaussian limit,
evaluated over the same observation times, candidate windows,
and monitoring horizon as the corresponding detector.

\begin{table}[h!]
\centering
\caption{Monitoring results for the U.S. air transportation network examples.}
\label{tab:airline-monitoring-results}
\setlength{\tabcolsep}{4.5pt}
\renewcommand{\arraystretch}{1.12}

\begin{tabularx}{\textwidth}{
    @{}
    >{\raggedright\arraybackslash}X
    >{\centering\arraybackslash}p{2.55cm}
    >{\centering\arraybackslash}p{2.55cm}
    >{\centering\arraybackslash}p{2.55cm}
    >{\centering\arraybackslash}p{2.55cm}
    @{}
}
\toprule
Event info
& Reference date
& Cancellations
& Laplacian modes
& Route share \\
\midrule

\textbf{COVID-19 disruption}\newline
{\footnotesize
Training: Jan.\ 1--Dec.\ 31, 2019;\newline
Monitoring: Mar.\ 1--Jun.\ 30, 2020}
&
Mar.\ 11, 2020
&
Mar.\ 18, 2020
&
Mar.\ 26, 2020
&
Mar.\ 24, 2020
\\

\addlinespace[3pt]

\textbf{Spring/Summer 2025}\newline
{\footnotesize
Training: Mar.\ 1, 2024--Feb.\ 28, 2025;\newline
Monitoring: Mar.\ 1--Aug.\ 31, 2025 }
&
--
&
Jul.\ 2, 2025
&
Jun.\ 7, 2025
&
Jun.\ 18, 2025
\\

\bottomrule
\end{tabularx}

\begin{minipage}{0.96\textwidth}
\footnotesize
The COVID-19 reference date is the World Health Organization's March 11, 2020 pandemic declaration.
\end{minipage}
\end{table}

\end{document}